\documentclass[11pt]{article} 
\usepackage{algorithmic}
\usepackage{algorithm}
\usepackage{times} 
\usepackage{amsfonts}
\usepackage{amsmath} 
\usepackage{latexsym} 
\usepackage{epsfig}
\usepackage{latexsym} 
\usepackage{verbatim}
\usepackage{makeidx}
\makeindex
\begin{document}
 
\newtheorem{theorem}{Theorem}
\newtheorem{corollary}[theorem]{Corollary}
\newtheorem{prop}[theorem]{Proposition} 
\newtheorem{problem}[theorem]{Problem}
\newtheorem{lemma}[theorem]{Lemma} 
\newtheorem{remark}[theorem]{Remark}
\newtheorem{observation}[theorem]{Observation}
\newtheorem{defin}{Definition} 
\newtheorem{example}{Example}
\newtheorem{conj}{Conjecture} 
\newenvironment{proof}{{\bf Proof:}}{\hfill$\Box$} 
\newcommand{\PR}{\noindent {\bf Proof:\ }} 
\def\EPR{\hfill $\Box$\linebreak\vskip.5mm} 
 
\def\Pol{{\sf Pol}} 
\def\mPol{{\sf MPol}} 
\def\Polo{{\sf Pol}_1} 
\def\PPol{{\sf pPol\;}} 
\def\Inv{{\sf Inv}}
\def\mInv{{\sf MInv}} 
\def\Clo{{\sf Clo}\;} 
\def\Con{{\sf Con}} 
\def\concom{{\sf Concom}\;} 
\def\End{{\sf End}\;}
\def\Sub{{\sf Sub}\;} 
\def\Im{{\sf Im}} 
\def\Ker{{\sf Ker}\;} 
\def\H{{\sf H}}
\def\S{{\sf S}} 
\def\D{{\sf P}} 
\def\I{{\sf I}} 
\def\Var{{\sf var}} 
\def\PVar{{\sf pvar}} 
\def\fin#1{{#1}_{\rm fin}}
\def\P{{\sf P}} 
\def\Pfin{{\sf P_{\rm fin}}} 
\def\Id{{\sf Id}}
\def\R{{\rm R}} 
\def\F{{\rm F}} 
\def\Term{{\sf Term}}
\def\var#1{{\sf var}(#1)} 
\def\Sg#1{{\sf Sg}(#1)} 
\def\Sgg#1#2{{\sf Sg}_{#1}(#2)} 
\def\Cg#1{{\sf Cg}(#1)}
\def\Cen{{\sf Cen}}
\def\tol{{\sf tol}} 
\def\lnk{{\sf lk}} 
\def\rbcomp#1{{\sf rbcomp}(#1)}
  
\let\cd=\cdot 
\let\eq=\equiv 
\let\op=\oplus 
\let\omn=\ominus
\let\meet=\wedge 
\let\join=\vee 
\let\tm=\times
\def\ldiv{\mathbin{\backslash}} 
\def\rdiv{\mathbin/}
  
\def\typ{{\sf typ}} 
\def\zz{{\un 0}} 
\def\zo{{\un 1}}
\def\one{{\bf1}} 
\def\two{{\bf2}} 
\def\three{{\bf3}}
\def\four{{\bf4}} 
\def\five{{\bf5}}
\def\pq#1{(\zz_{#1},\mu_{#1})}
  
\let\wh=\widehat 
\def\ox{\ov x} 
\def\oy{\ov y} 
\def\oz{\ov z}
\def\of{\ov f} 
\def\oa{\ov a} 
\def\ob{\ov b} 
\def\oc{\ov c}
\def\od{\ov d} 
\def\oob{\ov{\ov b}} 
\def\rx{{\rm x}}
\def\rf{{\rm f}} 
\def\rrm{{\rm m}} 
\let\un=\underline
\let\ov=\overline 
\let\cc=\circ 
\let\rb=\diamond 
\def\ta{{\tilde a}} 
\def\tz{{\tilde z}}
\let\td=\tilde
  
  
\def\zZ{{\mathbb Z}} 
\def\B{{\mathcal B}} 
\def\P{{\mathcal P}}
\def\zL{{\mathbb L}} 
\def\zD{{\mathbb D}}
 \def\zE{{\mathbb E}}
\def\zG{{\mathbb G}} 
\def\zA{{\mathbb A}} 
\def\zB{{\mathbb B}}
\def\zC{{\mathbb C}} 
\def\zM{{\mathbb M}} 
\def\zR{{\mathbb R}}
\def\zS{{\mathbb S}} 
\def\zT{{\mathbb T}} 
\def\zN{{\mathbb N}}
\def\zQ{{\mathbb Q}} 
\def\zW{{\mathbb W}} 
\def\bK{{\bf K}}
\def\C{{\bf C}} 
\def\M{{\bf M}} 
\def\E{{\bf E}} 
\def\N{{\bf N}}
\def\O{{\bf O}} 
\def\bN{{\bf N}} 
\def\bX{{\bf X}} 
\def\GF{{\rm GF}} 
\def\cC{{\mathcal C}} 
\def\cA{{\mathcal A}}
\def\cB{{\mathcal B}} 
\def\cD{{\mathcal D}} 
\def\cE{{\mathcal E}} 
\def\cF{{\mathcal F}} 
\def\cG{{\mathcal G}} 
\def\cH{{\mathcal H}}
\def\cI{{\mathcal I}} 
\def\cL{{\mathcal L}} 
\def\cP{{\mathcal P}} 
\def\cR{{\mathcal R}} 
\def\cRY{{\mathcal RY}}
\def\cS{{\mathcal S}} 
\def\cT{{\mathcal T}} 
\def\cU{{\mathcal U}} 
\def\cV{{\mathcal V}} 
\def\cW{{\mathcal W}} 
\def\oB{{\ov B}}
\def\oC{{\ov C}} 
\def\ooB{{\ov{\ov B}}} 
\def\ozB{{\ov{\zB}}}
\def\ozD{{\ov{\zD}}} 
\def\ozG{{\ov{\zG}}}
\def\tcA{{\widetilde\cA}} 
\def\tcC{{\widetilde\cC}}
\def\tcF{{\widetilde\cF}} 
\def\tcI{{\widetilde\cI}}
\def\tB{{\widetilde B}} 
\def\tC{{\widetilde C}}
\def\tD{{\widetilde D}} 
\def\ttB{{\widetilde{\widetilde B}}}
\def\ttC{{\widetilde{\widetilde C}}}
\def\tba{{\tilde\ba}} 
\def\ttba{{\tilde{\tilde\ba}}}
\def\tbb{{\tilde\bb}} 
\def\ttbb{{\tilde{\tilde\bb}}}
\def\tbc{{\tilde\bc}} 
\def\tbd{{\tilde\bd}}
\def\tbe{{\tilde\be}} 
\def\tbt{{\tilde\bt}}
\def\tbu{{\tilde\bu}} 
\def\tbv{{\tilde\bv}}
\def\tbw{{\tilde\bw}} 
\def\tdl{{\tilde\dl}} 
\def\ocP{{\ov\cP}}
\def\tzA{{\widetilde\zA}} 
\def\tzC{{\widetilde\zC}}
\def\new{{\mbox{\footnotesize new}}}
\def\old{{\mbox{\footnotesize old}}}
\def\prev{{\mbox{\footnotesize prev}}}
\def\oo{{\mbox{\sf\footnotesize o}}}
\def\pp{{\mbox{\sf\footnotesize p}}}
\def\nn{{\mbox{\sf\footnotesize n}}} 
\def\oR{{\ov R}}
  
  
\def\gA{{\mathfrak A}} 
\def\gV{{\mathfrak V}} 
\def\gS{{\mathfrak S}} 
\def\gK{{\mathfrak K}} 
\def\gH{{\mathfrak H}}
  
\def\ba{{\bf a}} 
\def\bb{{\bf b}} 
\def\bc{{\bf c}} 
\def\bd{{\bf d}} 
\def\be{{\bf e}} 
\def\bbf{{\bf f}} 
\def\bg{{\bf g}}
\def\bh{{\bf h}}
\def\bi{{\bf i}} 
\def\bm{{\bf m}} 
\def\bo{{\bf o}} 
\def\bp{{\bf p}} 
\def\bs{{\bf s}} 
\def\bu{{\bf u}} 
\def\bt{{\bf t}} 
\def\bv{{\bf v}} 
\def\bx{{\bf x}}
\def\by{{\bf y}} 
\def\bw{{\bf w}} 
\def\bz{{\bf z}}
\def\ga{{\mathfrak a}} 
\def\oal{{\ov\al}} 
\def\obeta{{\ov\beta}}
\def\ogm{{\ov\gm}} 
\def\oep{{\ov\varepsilon}}
\def\oeta{{\ov\eta}} 
\def\oth{{\ov\th}} 
\def\ovm{{\ov\mu}}
\def\ozero{{\ov0}}
\def\bB{{\bf B}} 
\def\bA{{\bf A}}

  
\def\CCSP{\hbox{\rm c-CSP}} 
\def\CSP{{\rm CSP}} 
\def\NCSP{{\rm \#CSP}} 
\def\mCSP{{\rm MCSP}} 
\def\FP{{\rm FP}} 
\def\PTIME{{\bf PTIME}} 
\def\GS{\hbox{($*$)}} 
\def\ry{\hbox{\rm r+y}}
\def\rb{\hbox{\rm r+b}} 
\def\Gr#1{{\mathrm{Gr}(#1)}}
\def\Grp#1{{\mathrm{Gr'}(#1)}} 
\def\Grpr#1{{\mathrm{Gr''}(#1)}}
\def\Scc#1{{\mathrm{Scc}(#1)}} 
\def\rel{R} 
\def\relo{Q}
\def\rela{S} 
\def\reli{T} 
\def\relp{P} 
\def\dep{\mathsf{dep}}
\def\Filt{\mathrm{Ft}}
\def\Filts{\mathrm{Fts}} 
\def\Agr{$\mathbb{A}$}
\def\Al{\mathrm{Alg}}
\def\Sig{\mathrm{Sig}}
\def\strat{\mathsf{strat}}
\def\relmax{\mathsf{relmax}}
\def\srelmax{\mathsf{srelmax}}
\def\Meet{\mathsf{Meet}}
\def\amax{\mathsf{amax}}
\def\umax{\mathsf{umax}}
\def\emin{\mathsf{Z}}
\def\as{\mathsf{as}}
\def\star{\hbox{$(*)$}}
\def\bmal{{\mathbf m}}
\def\Af{\mathsf{Af}}
\let\sqq=\sqsubseteq
\def\maj{\mathsf{maj}}
\def\razm{\mathsf{size}}
\def\Razm{\mathsf{MAX}}
\def\Centr{\mathsf{Center}}
\def\centr{\mathsf{center}}

  
\let\sse=\subseteq 
\def\ang#1{\langle #1 \rangle}
\def\angg#1{\left\langle #1 \right\rangle}
\def\dang#1{\ang{\ang{#1}}} 
\def\vc#1#2{#1 _1\zd #1 _{#2}}
\def\tms{\tm\dots\tm}
\def\zd{,\ldots,} 
\let\bks=\backslash 
\def\red#1{\vrule height7pt depth3pt width.4pt
\lower3pt\hbox{$\scriptstyle #1$}}
\def\fac#1{/\lower2pt\hbox{$\scriptstyle #1$}}
\def\me{\stackrel{\mu}{\eq}} 
\def\nme{\stackrel{\mu}{\not\eq}}
\def\eqc#1{\stackrel{#1}{\eq}} 
\def\cl#1#2{\arraycolsep0pt
\left(\begin{array}{c} #1\\ #2 \end{array}\right)}
\def\cll#1#2#3{\arraycolsep0pt \left(\begin{array}{c} #1\\ #2\\
#3 \end{array}\right)} 
\def\clll#1#2#3#4{\arraycolsep0pt
\left(\begin{array}{c} #1\\ #2\\ #3\\ #4 \end{array}\right)}
\def\cllll#1#2#3#4#5#6{ \left(\begin{array}{c} #1\\ #2\\ #3\\
#4\\ #5\\ #6 \end{array}\right)} 
\def\pr{{\rm pr}}
\let\upr=\uparrow 
\def\ua#1{\hskip-1.7mm\uparrow^{#1}}
\def\sua#1{\hskip-0.2mm\scriptsize\uparrow^{#1}} 
\def\lcm{{\rm lcm}} 
\def\perm#1#2#3{\left(\begin{array}{ccc} 1&2&3\\ #1&#2&#3
\end{array}\right)} 
\def\w{$\wedge$} 
\let\ex=\exists
\def\NS{{\sc (No-G-Set)}} 
\def\lev{{\sf lev}}
\let\rle=\sqsubseteq 
\def\ryle{\le_{ry}} 
\def\ryprec{\le_{ry}}
\def\os{\mbox{[}} 
\def\zs{\mbox{]}}
\def\link{{\sf link}}
\def\solv{\stackrel{s}{\sim}} 
\def\mal{\mathbf{m}}
\def\precs{\prec_{as}}

  
\def\lb{$\linebreak$}  
  
\def\ar{\hbox{ar}} 
\def\Im{{\sf Im}\;} 
\def\deg{{\sf deg}}
\def\id{{\rm id}}
  
\let\al=\alpha 
\let\gm=\gamma 
\let\dl=\delta 
\let\ve=\varepsilon
\let\ld=\lambda 
\let\om=\omega 
\let\vf=\varphi 
\let\vr=\varrho
\let\th=\theta 
\let\sg=\sigma 
\let\Gm=\Gamma 
\let\Dl=\Delta
\let\kp=\kappa
  
  
\font\tengoth=eufm10 scaled 1200 
\font\sixgoth=eufm6
\def\goth{\fam12} 
\textfont12=\tengoth 
\scriptfont12=\sixgoth
\scriptscriptfont12=\sixgoth 
\font\tenbur=msbm10
\font\eightbur=msbm8 
\def\bur{\fam13} 
\textfont11=\tenbur
\scriptfont11=\eightbur 
\scriptscriptfont11=\eightbur
\font\twelvebur=msbm10 scaled 1200 
\textfont13=\twelvebur
\scriptfont13=\tenbur 
\scriptscriptfont13=\eightbur
\mathchardef\nat="0B4E 
\mathchardef\eps="0D3F

\title{A dichotomy theorem for nonuniform CSPs}
\author{Andrei A.\ Bulatov\\ 
} 
\date{} 
\maketitle

\begin{abstract}
In this paper we prove the Dichotomy Conjecture on the complexity of nonuniform
constraint satisfaction problems posed by Feder and Vardi\footnote{Apart from
correcting numerous typos and inaccuracies, this version 
of the paper is different from the first \emph{Arxiv} version in the following 
ways:\\
-- A self-contained high level presentation of the main results is added.\\
-- Preliminaries section is extended.\\
-- An inconsistency between the definition and the use of the chaining condition 
is fixed.\\
-- The proof of Lemma~\ref{lem:relative-symmetry} is expanded for improved 
readability. Also a gap in this proof is patched.\\
-- Section ``Auxiliary Lemmas'' is split into two sections, 
Sections~\ref{sec:congruence} and~\ref{sec:chaining}, and moved forward.\\
-- The proof of Theorem~\ref{the:non-central} is reorganized and expanded 
for improved readability.
}
\end{abstract}
\tableofcontents

\section{Introduction}

In a Constraint Satisfaction Problem (CSP) the question is to decide whether 
or not it is possible to satisfy a given set of 
constraints. Constraints are often represented by specifying a relation, which is a 
set of allowed combinations of values some variables can take simultaneously. 
If the constraints allowed in a problem have to come from some set $\Gm$ of 
relations, such a restricted problem is referred 
to as a \emph{nonuniform CSP} and denoted $\CSP(\Gm)$. The set $\Gm$ is 
then called a \emph{constraint language}. Nonuniform CSPs not only provide 
a powerful
framework ubiquitous across a wide range of disciplines from theoretical 
computer science to computer vision, but also admit natural and elegant 
reformulations such as the homomorphism problem and a characterization 
as the class of problems equivalent to a logic class MMSNP.
Many different versions of the CSP have been studied across various fields. 
These include CSPs over infinite sets, counting CSPs (and related Holant 
problem and the problem of computing partition functions), several variants
of optimization CSPs, valued CSPs, quantified CSPs, and numerous related 
problems. The 
reader is referred to the recent book \cite{Krokhin17:constraint} for a survey 
of the state-of-the art in some of these areas. In this paper we, however, focus 
on the decision nonuniform CSP and its complexity.

A systematic study 
of the complexity of nonuniform CSPs was started by Schaefer in 1978
\cite{Schaefer78:complexity} who showed that for every constraint language
$\Gm$ over a 2-element set the problem $\CSP(\Gm)$ is either solvable in 
polynomial time or is NP-complete. Schaefer also asked about the complexity of 
$\CSP(\Gm)$ for languages over larger sets. The next step in the study of 
nonuniform CSPs was made in the 
seminal paper by Feder and Vardi \cite{Feder93:monotone,Feder98:monotone}, 
who apart from considering numerous aspects of the problem, posed the 
\emph{Dichotomy Conjecture} that states that for every finite constraint language 
$\Gm$ over a finite set the problem $\CSP(\Gm)$ is either solvable in polynomial 
time or is NP-complete. This conjecture has become a focal point of the CSP 
research and most of the effort in this area revolves to some extent around the 
Dichotomy Conjecture.

The complexity of the CSP in general and the Dichotomy Conjecture in particular
has been studied by several research communities using a variety of methods, 
each contributing an important aspect of the problem.
The CSP has been an established area in artificial intelligence for decades, and 
apart from developing efficient general methods of solving CSPs researchers 
tried to identify tractable fragments of the problem \cite{Dechter03:processing}.
The very important special case of the CSP, the (Di)Graph Homomorphism 
problem and the $H$-Coloring problem have been actively studied in the 
graph theory community, see, e.g.\ \cite{Hell90:h-coloring,Hell04:homomorphism}
and subsequent works by Hell, Feder, Bang-Jensen, Rafiey and others. 
Homomorphism duality introduced in these works has been very useful in
understanding the structure of constraint problems. The CSP
plays a major role and has been successfully studied in database theory, logic and 
model theory \cite{Kolaitis03:csp,Kolaitis00:game,Gottlob14:treewidth}, although 
the version of the problem mostly used there is not necessarily nonuniform. 
Logic games and strategies are a standard tool in most of CSP algorithms. An 
interesting approach to the Dichotomy Conjecture through long codes was
suggested by Kun and Szegedy \cite{Kun16:new}. Brown-Cohen and Raghavendra 
proposed to study the conjecture using techniques based on decay of 
correlations \cite{Brown-Cohen16:correlation}. In this paper we use the 
algebraic structure of the CSP, which is briefly discussed next.

The most effective approach to the study of the CSP turned out to be the 
\emph{algebraic approach} that associates 
every constraint language with its (universal) algebra of polymorphisms. The 
method was first developed in a series of papers by Jeavons and coauthors 
\cite{Jeavons97:closure,Jeavons98:algebraic,Jeavons98:consist} and then refined
by Bulatov, Krokhin, Barto, Kozik, Maroti, Zhuk and others 
\cite{Barto12:absorbing,Barto12:near,Barto14:local,Barto15:constraint,%
Bulatov05:classifying,Bulatov04:graph,Bulatov08:recent,Maroti11:Malcev,%
Maroti10:tree,Zhuk14:key,Zhuk16:5-element,Zhuk16:7-element}. While the 
complexity of $\CSP(\Gm)$ has been already solved for some interesting 
classes of structures such as graphs \cite{Hell90:h-coloring}, the algebraic approach
allowed the researchers to confirm the Dichotomy Conjecture in a number of more 
general  cases: for languages over a set of size up to 7 
\cite{Bulatov02:3-element,%
Bulatov06:3-element,Markovic11:4-element,Zhuk16:5-element,Zhuk16:7-element}, 
so called conservative languages 
\cite{Bulatov03:conservative,Bulatov11:conservative,%
Bulatov16:conservative,Barto11:conservative}, and some classes of digraphs
\cite{Barto09:sources}. It also allowed to design the main classes of CSP 
algorithms \cite{Barto14:local,Bulatov06:simple,Bulatov16:restricted,%
Berman10:varieties,Idziak10:few}, and refine the exact complexity of the CSP 
\cite{Allender05:refining,Barto12:near,Dalmau08:majority,Larose07:first-order}.

In this paper we confirm the Dichotomy Conjecture for arbitrary languages 
over finite sets. More precisely we prove the following

\begin{theorem}\label{the:main}
For any finite constraint language $\Gm$ over a finite set the problem 
$\CSP(\Gm)$ is either solvable in polynomial time or is NP-complete.
\end{theorem}

The proved criterion matches the algebraic form of the Dichotomy Conjecture 
suggested in \cite{Bulatov05:classifying}. The hardness part of the conjecture 
has been known for long time. Therefore the main achievement of this paper 
is a polynomial time algorithm for problems satisfying the tractability 
condition. More specifically, we suggest such an algorithm for languages that 
contain all the constant relations of the form $\{(a)\}$, and this implies a general 
dichotomy due to the results of \cite{Bulatov05:classifying}.

Using the algebraic language we can state the result in a stronger form.
Let $\zA$ be a finite idempotent algebra and let $\CSP(\zA)$ denote the
union of problems $\CSP(\Gm)$ such that every term operation of $\zA$
is a polymorphism of $\Gm$. Problem $\CSP(\zA)$ is no longer a nonuniform
CSP, and Theorem~\ref{the:main} allows for problems $\CSP(\Gm)\sse\CSP(\zA)$
to have different solution algorithms even when $\zA$ meets the tractability 
condition. We show that the solution algorithm only depends on the algebra
$\zA$.

\begin{theorem}\label{the:main2}
For a finite idempotent algebra that satisfies the conditions of the 
Dichotomy Conjecture there is a uniform solution algorithm for $\CSP(\zA)$.
\end{theorem}

An interesting question arising from Theorems~\ref{the:main},\ref{the:main2}
is known as the \emph{Meta-problem}: Given a constraint language or a 
finite algebra, decide whether or not it satisfies the conditions of the theorems.
The answer to this question is not quite trivial, for a thorough study of
the Meta-problem see \cite{Chen16:asking,Freese09:complexity}.

The paper consists of two parts. The first part of the paper aims at a 
self-contained introduction into the main ideas of the solution algorithm. 
The second part mostly concerns with further development of the algebraic 
approach and technical proofs of the results, and is significantly more involved.
We start with introducing the terminology and notation for CSPs that is used 
throughout the paper and reminding the basics of the algebraic approach. 
Then in Section~\ref{sec:separation-central} we introduce the key ingredients 
used in the algorithm: separating congruences and quasi-centralizer. Then in 
Section~\ref{sec:algorithm1} we apply these concepts to CSPs, first, 
to demonstrate how quasi-centralizers help to decompose an instance into smaller
subinstances, and, second, to introduce a new kind of minimality condition for 
CSPs, \emph{block minimality}. After that we state the main results used by
the algorithm and describe the algorithm itself. We complete the first part by
introducing the main technical construction to give an idea of why the algorithm 
works.

\part{Outline of the algorithm}\label{part:outline}

\section{Introduction to CSP}\label{sec:csp-intro}

For a detailed introduction to CSP and the algebraic approach to its structure
the reader is referred to a very recent and very nice survey by Barto et al.\
\cite{Barto17:polymorphisms}. Basics of universal algebra can be learned from
the textbook \cite{Burris81:universal}. In preliminaries to this paper we 
therefore focus on what is needed for our result.

\subsection{CSP, universal algebra and the Dichotomy conjecture}\label{sec:csp-p1}

The `combinatorial' formulation of the CSP best fits our purpose.
Fix a finite set $A$ and let $\Gm$ be a \emph{constraint language} over $A$, 
that is, a set --- not necessarily finite --- of relations over $A$. The 
(\emph{nonuniform}) \emph{Constraint Satisfaction 
Problem}\index{Constraint Satisfaction Problem}\index{nonuniform Constraint 
Satisfaction Problem} (\emph{CSP}) associated with language $\Gm$ 
is the problem $\CSP(\Gm)$, in which, an \emph{instance}\index{instance}  
is a pair $(V,\cC)$, where $V$ is a set of variables; and $\cC$ is a set of 
\emph{constraints}\index{constraints}, 
i.e.\ pairs $\ang{\bs,\rel}$, where $\bs=(\vc vk)$ is a tuple of 
variables from $V$, the \emph{constraint scope}\index{constraint scope}, 
and $\rel\in\Gm$, the $k$-ary 
\emph{constraint relation}\index{constraint relation}. We always assume that 
relations are given explicitly by a list of tuples. The way constraints are 
represented does not matter if $\Gm$ is finite, of course, but it may change the
complexity of the problems for infinite languages.
The goal is to find a \emph{solution}, that is a mapping 
$\vf:V\to A$ such that for every constraint $\ang{\bs,\rel}$,
$\vf(\bs)\in\rel$. 

We will often use the set of solutions of a CSP instance $\cP=(V,\cC)$ or its 
subproblems (to be defined later), viewed either as a $|V|$-ary relation or as
a set of mappings $\vf:V\to A$. It will be denoted by 
$\cS_\cP$, or just $\cS$ if $\cP$ is clear from the context.

Jeavons et al.\ in \cite{Jeavons97:closure,Jeavons98:algebraic} were the first
to observe that higher order symmetries of constraint languages called 
polymorphisms play a significant role in  the study of the complexity of the CSP.
A \emph{polymorphism} of a relation $\rel$ over $A$ is an operation 
$f(\vc xk)$ on $A$ such that for any choice of $\vc\ba k\in\rel$ we have
$f(\vc\ba k)\in\rel$. If this is the case we also say that $f$ 
\emph{preserves}\index{operation preserving a relation} $\rel$, or that $\rel$
is \emph{invariant}\index{invariant relation} with respect to $f$. A polymorphism 
of a constraint language $\Gm$ is an 
operation that is a polymorphism of every $\rel\in\Gm$. 

\begin{theorem}[\cite{Jeavons97:closure,Jeavons98:algebraic}]%
\label{the:algebra-csp}
For constraint languages $\Gm,\Dl$, where $\Gm$ is finite, if every 
polymorphism of $\Dl$ is also a polymorphism of $\Gm$, then 
$\CSP(\Gm)$ is polynomial time reducible to $\CSP(\Dl)$. (In fact, this
can be improved to a log-space reduction.)
\end{theorem}

Listed below are the several types of polymorphisms that occur frequently 
throughout the paper. The presence of each of these polymorphisms
imposes restrictions on the structure of invariant relations that can be used 
in designing a solution algorithm. Some of such results we will mention 
later.\\[1mm]
-- \emph{Semilattice} operation is a binary operation $f(x,y)$ such that $f(x,x)=x$,
$f(x,y)=f(y,x)$, and $f(x,f(y,z))=f(f(x,y),z)$ for all $x,y,z\in A$;\\[1mm]
-- $k$-ary \emph{near-unanimity} operation is a $k$-ary operation $u(\vc xk)$
such that $u(y,x\zd x)=u(x,y,x\zd x)=\dots =u(x\zd x,y)=x$ for all 
$x,y\in A$; a ternary near-unanimity operation $m$ is said to be a \emph{majority} 
operation, it satisfies the equations $m(y,x,x)=m(x,y,x)=m(x,x,y)=x$;\\[1mm]
-- \emph{Mal'tsev} operation is a ternary operation $h(x,y,z)$ satisfying the 
equations
$h(x,y,y)=h(y,y,x)=x$ for all $x,y\in A$; the \emph{affine} operation $x-y+z$ of
an Abelian group is a special case of Mal'tsev operations;\\[1mm]
-- $k$-ary \emph{weak near-unanimity} operation is a $k$-ary operation $w$ that 
satisfies the same equations as a near-unanimity operations
$w(y,x\zd x)=w(x,y,x\zd x)=\dots =w(x\zd x,y)$, except for the last one.

\vspace{2mm}

The next step in discovering more structure behind nonuniform CSPs has
been made in \cite{Bulatov05:classifying}, in which universal algebras 
were brought into the picture.
A \emph{(universal) algebra}\index{universal algebra} is a pair $\zA=(A,F)$ 
consisting of a set $A$, the \emph{universe}\index{universe} of $\zA$, and a 
set $F$ of operations on $A$. Operations from $F$ together with operations that 
can be obtained from them by means of  composition are called the 
\emph{term}\index{term operation} operations of $\zA$.

Algebras allow for a more general definition of CSPs that is used here. Let 
$\CSP(\zA)$ denote the
class of nonuniform CSPs $\{\CSP(\Gm)\mid \Gm\sse\Inv(F)\}$, where $\Inv(F)$
denotes the set of all relations invariant with respect to all operations from $F$. 
Note that 
the tractability of $\CSP(\zA)$ can be understood in two ways: as the existence of 
a polynomial-time algorithm for every $\CSP(\Gm)$ from this class, or as the 
existence of a uniform polynomial-time algorithm for all such problems. One of the
implications of our results is that these two types of tractability are equivalent.
From the formal standpoint we will use the stronger one.

The main structural elements are subalgebras, 
congruences, and quotient algebras.
For $B\sse A$ and an operation $f$ on $A$ by $f\red B$ we denote the 
restriction of $f$ on $B$. Algebra $\zB=(B,\{f\red B\mid f\in F\})$ is called
a \emph{subalgebra} of $\zA$ if $f(\vc bk)\in B$ for any $\vc bk\in B$ and
any $f\in F$.

Congruences play a very significant role in our algorithm, and we discuss them 
in more details. A \emph{congruence} is an equivalence relation $\th\in\Inv(F)$. 
This means that
for any operation $f\in F$ and any $(a_1,b_1)\zd (a_k,b_k)\in\th$ it 
holds $(f(\vc ak),f(\vc bk))\in\th$. Therefore it is possible to define an 
algebra on $A\fac\th$, the set of $\th$-blocks, by setting 
$f\fac\th(a^\th_1\zd a^\th_k)=(f(\vc ak))\fac\th$ for $\vc ak\in A$, where
$a^\th$ denotes the $\th$-block containing $a$. The resulting algebra $\zA\fac\th$ 
is called the \emph{quotient algebra modulo} $\th$. 

The following are examples of congruences and quotient algebras.\\[1mm]
-- Let $\zA$ be any algebra. Then the equality relation $\zz_\zA$ and the full
binary relation $\zo_\zA$ on $\zA$ are congruences of $\zA$. The quotient
algebra $\zA\fac{\zz_\zA}$ is $\zA$ itself, while $\zA\fac{\zo_\zA}$ is a 
1-element algebra.\\[1mm]
-- Let $\zS_n$ be the permutation group on an $n$-element set and binary
relation $\th$ is given by: $(a,b)\in\th$ for $a,b\in\zS_n$ if and only if $a$
and $b$ have the same parity as permutations. Then $\th$ is a congruence of 
$\zS_n$ and $\zS_n\fac\th$ is the 2-element group.\\[1mm]
-- Let $\zL_n$ be an $n$-dimensional vector space and $\zL'$ its $k$-dimensional
subspace. The binary relation $\pi$ given by: $(\ov a,\ov b)\in\pi$ if and 
only if $\ov a,\ov b$ have the same orthogonal projection on $\zL'$, is a
congruence of $\zL_n$ and $\zL_n\fac\pi$ is $\zL'$.

\vspace{1mm}

The (ordered) set of all congruences of $\zA$ is denoted by $\Con(\zA)$. 
This set is actually a lattice. By $\H\S(\zA)$ we denote 
the set of all quotient algebras of all subalgebras of $\zA$. 

The results of \cite{Bulatov05:classifying} reduce the dichotomy conjecture  
to idempotent algebras. An algebra $\zA=(A,F)$ is said to be 
\emph{idempotent}\index{idempotent algebra} if every operation $f\in F$
satisfies the equation $f(x\zd x)=x$. If $\zA$ is idempotent, then all the 
constant relations $\{(a)\}$ are invariant under $F$. Therefore studying
CSPs over idempotent algebras is the same as studying the CSPs that allow all 
constant relations. Another useful property of idempotent algebras is that
every block of every its congruence is a subalgebra. We now can state
the algebraic version of the dichotomy theorem.

\begin{theorem}\label{the:algebra-dichot-p1}
For a finite idempotent algebra $\zA$ the following are equivalent:\\[1mm]
(1) $\CSP(\zA)$ is solvable in polynomial time;\\[1mm]
(2) $\zA$ has a weak near-unanimity term operation;\\[1mm]
(3) every algebra from $\H\S(\zA)$ has a nontrivial term operation 
(that is not a \emph{projection}, an operation of the form 
$f(\vc xk)=x_i$);\\[1mm]
Otherwise $\CSP(\zA)$ is NP-complete.
\end{theorem}

The hardness part of this theorem is proved in \cite{Bulatov05:classifying}; the 
equivalence of (2) and (3) was proved in \cite{Bulatov01:varieties} and 
\cite{Maroti08:existence}. The equivalence of (1) to (2) and (3) is the main result 
of this paper. In the rest of the paper we assume all algebras to satisfy  
conditions (2),(3).

\subsection{Bounded width and the few subpowers algorithm}%
\label{seg:bounded-width-p1}

Leaving aside occasional combinations thereof, there are only two standard 
types of algorithms solving the CSP. In this section we give a brief introduction 
into them.

\paragraph{CSPs of bounded width.}
Algorithms of the first kind are based on the idea of local propagation, that is
formally described below.
%
By $[n]$ we denote the set $\{1\zd n\}$. For sets $\vc An$ tuples 
from $A_1\tms A_n$ are denoted in boldface, say, $\ba$; the $i$th component of 
$\ba$ is referred to as $\ba[i]$. 
An $n$-ary relation $\rel$\index{relation} over sets $\vc An$ is any subset of 
$A_1\tms A_n$. 
For $I=\{\vc ik\}\sse[n]$ by $\pr_I\ba,\pr_I\rel$ we denote the 
\emph{projections}\index{projections} $\pr_I\ba=(\ba[i_1]\zd\ba[i_k])$, 
$\pr_I\rel=\{\pr_I\ba\mid\ba\in\rel\}$ of tuple
$\ba$ and relation $\rel$. If $\pr_i\rel=A_i$ for each $i\in[n]$, relation $\rel$ is 
said to be a \emph{subdirect product}\index{subdirect product} of 
$A_1\tms A_n$. 

Let $\cP=(V,\cC)$ be a CSP instance. For $W\sse V$ by $\cP_W$ 
we denote the \emph{restriction}\index{restriction of CSP} of 
$\cP$ onto $W$, that is, the instance 
$(W,\cC_W)$, where for each $C=\ang{\bs,\rel}\in\cC$, the set $\cC_W$
includes the constraint $C_W=\ang{\bs\cap W,\pr_{\bs\cap W}\rel}$.
The set of solutions of $\cP_W$ will be denoted by $\cS_W$. 

Unary solutions, that is, when $|W|=1$ play a special role. As is easily seen, 
for $v\in V$ the set $\cS_v$ is just the intersections of unary projections 
$\pr_v\rel$ of constraints whose scope contains $v$. One may always 
assume that allowed values for a variable $v\in V$ is the set $\cS_v$. We call 
this set the \emph{domain}\index{domain} of $v$ and assume that CSP
instances may have different domains, which nevertheless are always 
subalgebras or quotient algebras of the original algebra $\zA$. It will be 
convenient to denote the domain of $v$ by $\zA_v$. The domain $\zA_v$
may change as a result of transformations of the instance.
Instance $\cP$ is said to be \emph{1-minimal}\index{1-minimal instance} if 
for every $v\in V$ and every constraint
$C=\ang{\bs,\rel}\in\cC$ such that $v\in\bs$, it holds $\pr_v\rel=\zA_v$.

Instance $\cP$ is said to be \emph{(2,3)-consistent}\index{(2,3)-consistent 
instance} if it has a \emph{(2,3)-strategy}\index{(2,3)-strategy}, that is, a 
collection of relations $\rel^X$, 
$X\sse V$, $|X|=2$ satisfying the following conditions:\\
-- for every $X\sse V$ with $|X|\le2$ and every $C=\ang{\bs,\rel}$,
$\pr_{\bs\cap X}\rel^X\sse\pr_{\bs\cap X}\rel$;\\ 
-- for every $X=\{u,v\}\sse V$, any $w\in V-X$, and any 
$(a,b)\in\rel^X$, there is $c\in\zA_w$ such that $(a,c)\in\rel^{\{u,w\}}$
and $(b,c)\in\rel^{\{v,w\}}$.\\[1mm]
\indent
We will always assume that a (2,3)-consistent instance has a constraint 
$C^X=\ang{X,\cS_X}$ for every $X\sse V$, $|X|=2$. Then clearly
$\rel^X\sse\cS_X$. Let the collection of relations $\rel^X$ be denoted 
by $\cR$. 
%
A tuple $\ba$ whose entries are indexed with elements 
of $W\sse V$ and such that $\pr_X\ba\in\rel^X$ for any $X\sse W$, $|X|=2$,
will be called \emph{$\cR$-compatible}\index{$\cR$-compatible tuple}. If a 
(2,3)-consistent instance $\cP$
with a (2,3)-strategy $\cR$ satisfies the additional condition\\
-- for every constraint $C=\ang{\bs,\rel}$ of $\cP$ every tuple $\ba\in\rel$
is $\cR$-compatible,\\[1mm]
it is called \emph{(2,3)-minimal}\index{(2,3)-minimal instance}. For 
$k\in\nat$,  $(k,k+1)$-strategies, $(k,k+1)$-consistency, and 
$(k,k+1)$-minimality are defined in a similar way replacing 2,3 with $k,k+1$.

Instance $\cP$ is said to be \emph{minimal}\index{minimal instance} (or 
\emph{globally minimal}\index{globally minimal instance})  if for every 
$C=\ang{\bs,\rel}\in\cC$ and every $\ba\in\rel$ there is a solution 
$\vf\in\cS$ such that $\vf(\bs)=\ba$. 

Any instance can be transformed to
a 1-minimal, (2,3)-consistent, or (2,3)-minimal instance in polynomial
time using the standard constraint propagation algorithms (see, e.g.\ 
\cite{Dechter03:processing}). These algorithms work by changing the constraint
relations and the domains of the variables eliminating some tuples and elements 
from them. We call such a process \emph{tightening}\index{tightening} the 
instance. It is important to notice that if the original instance belongs to 
$\CSP(\zA)$ for some algebra $\zA$, that is, all its constraint relations are invariant
under the term operations of $\zA$, the constraint relations obtained by 
propagation algorithms are also invariant under term operations of $\zA$, and 
so the resulting instance also belongs to $\CSP(\zA)$. Establishing minimality
amounts to solving the problem and therefore not always can be easily 
done.

If a constraint propagation algorithm solves a CSP, the problem is said to be of 
bounded width. More precisely, $\CSP(\Gm)$ (or $\CSP(\zA)$) is said to have 
\emph{bounded width}\index{bounded width} if for some $k$
every $(k,k+1)$-minimal instance from $\CSP(\Gm)$ (or $\CSP(\zA)$) has a 
solution. Problems 
of bounded width are very well studied, see the older survey 
\cite{Bulatov08:dualities} and a more recent paper \cite{Barto16:collapse}.

\begin{theorem}%
[\cite{Barto16:collapse,Bulatov16:restricted,Kozik16:characterization}]%
\label{the:bounded-width-p1}
For an idempotent algebra $\zA$ the following are equivalent:\\[1mm]
(1) $\CSP(\zA)$ has bounded width;\\[1mm]
(2) every (2,3)-minimal instance from $\CSP(\zA)$ has a solution;\\[1mm]
(3) $\zA$ has a weak near-unanimity term of arity
$k$ for every $k\ge3$;\\[1mm]
(4) every algebra $\H\S(\zA)$ has a nontrivial operation, and none of them 
is equivalent to a module.
\end{theorem}

\paragraph{Omitting semilattice edges and the few subpowers property.}
%
The second type of CSP algorithms can be viewed as a generalization of 
Gaussian elimination, although, it utilizes just one property also used by 
Gaussian elimination: the set of solutions of a system of linear equations
or a CSP has a set of generators of size linear in the number of variables. 
The property 
that for every instance $\cP$ of $\CSP(\zA)$ its solution space $\cS_\cP$ has 
a set of generators of polynomial size is nontrivial, because there are only 
exponentially many such sets, while,
as is easily seen CSPs with $n$ variables may have up to double exponentially 
many different sets of solutions. Formally, an algebra $\zA=(A,F)$ has 
\emph{few subpowers} if for every $n$ there are only exponentially many 
$n$-ary relations in $\Inv(F)$. 

Algebras with few subpowers are well studied, completely characterized, and
the CSP over such an algebra has a polynomial-time solution algorithm, see, 
\cite{Berman10:varieties,Idziak10:few}. In particular, such algebras admit
a characterization in terms of the existence of a term operation with special
properties, an \emph{edge} term. We however need only a subclass of
algebras with few subpowers that appeared in \cite{Bulatov16:restricted}
and is defined as follows.

A pair of elements $a,b\in\zA$ is said to be a \emph{semilattice edge}
if there is a binary term operation $f$ of $\zA$ such that $f(a,a)=a$ and
$f(a,b)=f(b,a)=f(b,b)=b$, that is, $f$ is a semilattice operation on $\{a,b\}$.

\begin{prop}[\cite{Bulatov16:restricted}]\label{pro:semilattice-edges-p1}
If an idempotent algebra $\zA$ has no semilattice edges, it has few subpowers,
and therefore $\CSP(\zA)$ is solvable in polynomial time.
\end{prop}

Semilattice edges have other useful properties including the following one 
that we use for reducing a CSP to smaller problems.

\begin{lemma}[\cite{Bulatov16:connectivity}]\label{lem:multiplication-p1}
For any idempotent algebra $\zA$ there is a binary term operation $xy$ 
of $\zA$ (think multiplication) such that $xy$ is a semilattice operation 
on any semilattice edge and for any $a,b\in\zA$ either $ab=a$ or 
$\{a,ab\}$ is a semilattice edge.
\end{lemma}

\section{Solving CSPs}\label{sec:solving-p1}
\subsection{Congruence separation and centralizers}\label{sec:separation-central}

In this section we introduce two of the key ingredients of the algorithm.

\paragraph{Separating congruences.}
%
Unlike the vast majority of the literature on the algebraic approach to the CSP
we use not only term operations, but also polynomial operations of an algebra.
It should be noted however that the first to use polynomials for CSP algorithms
was Maroti in \cite{Maroti10:tree}. We make use of some ideas from that paper
in the next section.

Let $f(\vc xk,\vc y\ell)$ be a $k+\ell$-ary term operation of an algebra $\zA$ and
$\vc b\ell\in\zA$. The operation $g(\vc xk)=f(\vc xk,\vc b\ell)$ is called a
\emph{polynomial} of $\zA$. The name `polynomial' refers to usual polynomials. 
Indeed, if $\zA$ is a ring, its polynomials as just defined are the same as 
polynomials in the regular sense. A polynomial that depends on only one variable
is said to be a \emph{unary} polynomial.

While polynomials of $\zA$ do not have to be polymorphisms of relations from
$\Inv(F)$, congruences and unary polynomials are in a special relationship.
More precisely, an equivalence relation over $\zA$ is a congruence if and only 
if it is preserved by all the unary polynomials of $\zA$.

Let $\zA$ be an algebra. For $\al,\beta\in\Con(\zA)$ 
we write $\al\prec\beta$ if $\al<\beta$ and $\al\le\gm\le\beta$ in 
$\Con(\zA)$ implies $\gm=\al$ or $\gm=\beta$; if this is the case we 
call $(\al,\beta)$ a \emph{prime interval}\index{prime interval} in $\Con(\zA)$. 
Let $\al\prec\beta$ and $\gm\prec\dl$ be prime
intervals in $\Con(\zA)$. We say that $\al\prec\beta$ can be 
\emph{separated}\index{separated intervals} from 
$\gm\prec\dl$  if there is a unary polynomial $f$ of $\zA$ such that 
$f(\beta)\not\sse\al$, but $f(\dl)\sse\gm$. The polynomial $f$ in this case is 
said to \emph{separate}\index{separating polynomial}  
$\al\prec\beta$ from $\gm\prec\dl$. 

In a similar way separation can be defined for prime intervals in different coordinate 
positions of a relation. Let $\rel$ be a subdirect product of 
$\zA_1\tm\dots\tm\zA_n$. Then $\rel$ is also an algebra and its polynomials 
can be defined in the same way. Let $i,j\in[n]$ and let
$\al\prec\beta$, $\gm\prec\dl$ be prime intervals in $\Con(\zA_i)$ and 
$\Con(\zA_j)$, respectively. Interval $\al\prec\beta$ can be separated from 
$\gm\prec\dl$ if there is a unary polynomial $f$ of $\rel$ such that 
$f(\beta)\not\sse\al$ but $f(\dl)\sse\gm$ (note that the actions of $f$ on 
$\zA_i,\zA_j$ are polynomials of those algebras). 

The binary relation `cannot be separated' on the set of prime intervals of an 
algebra or factors of a relation is easily seen to be reflexive and transitive.
Under certain mild conditions it can also be shown to be symmetric in a
certain sense (Lemma~\ref{lem:relative-symmetry}), and so for the 
purpose of our algorithm it can be treated as an equivalence relation.

\paragraph{Quasi-Centralizers.}
The second ingredient introduced here is the notion of quasi-centralizer of a prime 
interval of congruences. It is similar to the centralizer as it is defined in 
commutator theory, albeit the exact relationship between the two concepts 
is not quite clear, and so we name differently for safety.

For an algebra $\zA$, a term operation $f(x,\vc yk)$, and $\ba\in\zA^k$, let
$f^\ba(x)=f(x,\ba)$. Let $\al,\beta\in\Con(\zA)$, $\al\le\beta$, 
and let $\zeta(\al,\beta)\sse\zA^2$ denote the following binary 
relation: 
$(a,b)\in\zeta(\al,\beta)$ if an only if, for any term operation $f(x,\vc yk)$, any 
$i\in[k]$, and any $\ba,\bb\in\zA^k$
such that $\ba[i]=a$, $\bb[i]=b$, and $\ba[j]=\bb[j]$ for $j\ne i$, it holds
$f^\ba(\beta)\sse\al$ if and only if $f^\bb(\beta)\sse\al$. 
The relation $\zeta(\al,\beta)$ is always a congruence of $\zA$ 
(Lemma~\ref{lem:delta-properties}) and its effect on the structure of algebra 
$\zA$ is illustrated by the following statement.

\begin{lemma}\label{lem:centralizer-multiplication-p1}
Let $\zeta(\al,\beta)=\zo_\zA$, $a,b,c\in\zA$ and 
$(b,c)\in\beta$. Then $(ab,ac)\in\al$.
\end{lemma}

Fig.~\ref{fig:central-p1}(a),(b) shows the effect of large quasi-centralizers on the 
structure
of algebra $\zA$. Dots there represent $\al$-blocks (assume $\al$ is the 
equality relation), ovals represent $\beta$-blocks, let they be $B$ and $C$,
and such that there is at least one semilattice edge between $B$ and $C$.
If $\zeta(\al,\beta)$ is the full relation, Lemmas~\ref{lem:multiplication-p1} 
and~\ref{lem:centralizer-multiplication-p1} imply that for any $a\in B$ and
any $b,c\in C$ we have $ab=ac$, and so $ab$ is the only element of $C$
such that $\{a,ab\}$ is a semilattice edge (represented by arrows). In other 
words, we have a
mapping that can also be shown injective from $B$ to $C$. We will use
this mapping to lift any solution with a value from $B$ to a solution with 
a value from $C$.

\begin{figure}[ht]
\centerline{\includegraphics[totalheight=3cm,keepaspectratio]{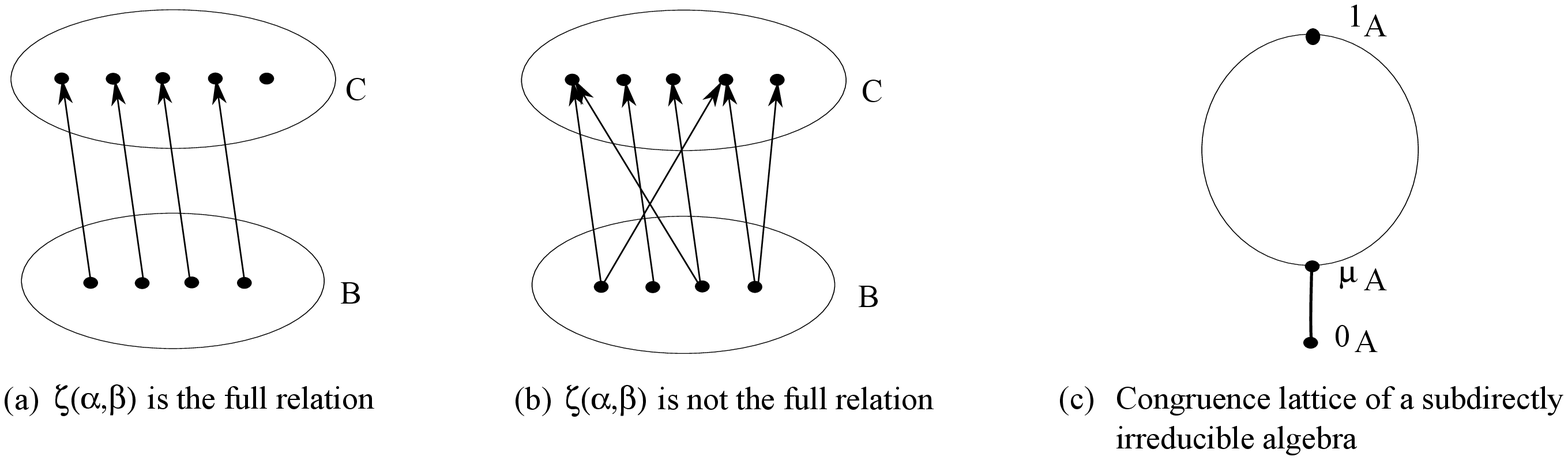}}
\caption{}\label{fig:central-p1}
\end{figure}

\subsection{The algorithm}\label{sec:algorithm1}

We have seen in the previous section that big centralizers impose strong 
restrictions on the structure of an algebra. We start this section showing that
small centralizers restrict the structure of CSPs.

\paragraph{Decomposition of CSPs.}
Let $\rel$ be a binary relation, a subdirect product of $\zA\tm\zB$, and 
$\al\in\Con(\zA)$, $\gm\in\Con(\zB)$. Relation $\rel$ is said to be 
\emph{$\al\gm$-aligned}\index{$\al_i\al_j$-aligned coordinates} 
if, for any $(a,c),(b,d)\in\rel$, $(a,b)\in\al$ if and only if $(c,d)\in\gm$. 
This means that if $\vc Ak$ are the $\al$-blocks of $\zA$, then there are
also $k$ $\gm$-blocks of $\zB$ and they can be labeled $\vc Bk$ in such 
a way that 
$$
\rel=(\rel\cap(A_1\tm B_1))\cup\dots\cup(\rel\cap(A_k\tm B_k)).
$$

\begin{lemma}\label{lem:delta-alignment-p1}
Let $\rel,\zA,\zB$ be as above and $\al,\beta\in\Con(\zA)$, $\gm,\dl\in\Con(\zB)$, 
with $\al\prec\beta$, $\gm\prec\dl$. If $(\al,\beta)$ and $(\gm,\dl)$ cannot be 
separated, then $\rel$ is $\zeta(\al,\beta)\zeta(\gm,\dl)$-aligned.
\end{lemma}

Lemma~\ref{lem:delta-alignment-p1} provides a way to decompose 
CSP instances.
Let $\cP=(V,\cC)$ be a (2,3)-minimal instance from $\CSP(\zA)$, in particular, 
$\cC$ contains a constraint $C^{\{v,w\}}=\ang{(v,w),\rel^{\{v,w\}}}$ for 
every $v,w\in V$, and these relations form a (2,3)-strategy for $\cP$. Due to
(2,3)-minimality the domain of variables from $V$ do not have to be $\zA$ 
itself, but can be subalgebras of $\zA$. Recall that $\zA_v$ denotes the domain of 
$v\in V$. Also, let $W\sse V$ and congruences $\al_v,\beta_v\in\Con(\zA_v)$ 
for $v\in W$ be such that $\al_v\prec\beta_v$, and for any $v,w\in W$
the intervals $(\al_v,\beta_v)$ and $(\al_w,\beta_w)$ cannot be separated
in $\rel^{\{v,w\}}$. 

Denoting $\zeta_v=\zeta(\al_v,\beta_v)$ we see that there is a one-to-one
correspondence between $\zeta_v$ and $\zeta_w$ blocks of $\zA_v$ and 
$\zA_w$. Moreover, by (2,3)-minimality these correspondences are consistent,
that is, if $u,v,w\in W$ and $B_u,B_v,B_w$ are $\zeta_u$-, $\zeta_v$-
and $\zeta_w$-blocks, respectively, such that 
$\rel^{\{u,v\}}\cap(B_u\tm B_v)\ne\eps$ and 
$\rel^{\{v,w\}}\cap(B_v\tm B_w)\ne\eps$, then 
$\rel^{\{u,w\}}\cap(B_u\tm B_w)\ne\eps$. This means that $\cP_W$
can be split into several instances, whose domains are $\zeta_v$-blocks.

\begin{lemma}\label{lem:central-decomposition-p1}
Let $\cP,W,\al_v,\beta_v$ be as above. Then $\cP_W$ can be decomposed
into a collection of instances $\cP_1\zd\cP_k$, $\cP_i=(W,\cC_i)$ such that 
every solution of $\cP_W$ is a solution of one of the $\cP_i$ and for every 
$v\in V$ its domain in $\cP_i$ is a $\zeta_v$-block.
\end{lemma}

\paragraph{Irreducibility.}
In order to formulate the algorithm properly we need one more transformation 
of algebras. An algebra $\zA$ is said to be 
\emph{subdirectly irreducible}\index{subdirectly irreducible algebra}
if the intersection of all its nontrivial (different from the equality relation)
congruences is nontrivial. This smallest nontrivial congruence $\mu_\zA$ is called 
the \emph{monolith}\index{monolith} of $\zA$, see Fig.~\ref{fig:central-p1}(c).
It is a folklore observation that any CSP instance can be transformed in 
polynomial time to an instance, in which the domain of every variable 
is a subdirectly irreducible algebra. We will assume this property of all the
instances we consider.

\paragraph{Block-minimality.}
Lemma~\ref{lem:central-decomposition-p1} allows one to establish a much stronger
version of local consistency, block-minimality; in fact, it is not local anymore. 
The definitions below are designed in such a way that to allow for an efficient 
procedure to establish block-minimality. This is achieved either by allowing for 
decomposing a subinstance into instances over smaller domains as in 
Lemma~\ref{lem:central-decomposition-p1}, or by replacing large domains with 
their quotient algebras.

Let $\al_v$ be a congruence of $\zA_v$ for each $v\in V$. By 
$\cP\fac{\ov\al}$ we denote
the instance $(V,\cC^{\ov\al})$ constructed as follows: the domain of 
$v\in V$ is $\zA_v\fac{\al_v}$; for every constraint $C=\ang{\bs,\rel}\in\cC$,
the set $\cC^{\ov\al}$ includes the constraint $\ang{\bs,\rel\fac{\ov\al_\bs}}$,
where $\bs=(\vc vk)$ and 
$\rel\fac{\ov\al}=\{(\ba[v_1]^{\al_{v_1}}\zd\ba[v_k]^{\al_{v_k}})
\mid \ba\in\rel\}$.

Let $\cP=(V,\cC)$ be a (2,3)-minimal instance and $\{\rel^X\mid X\sse V, 
|X|=2\}$ is its (2,3)-strategy. Let $\ov\beta=(\beta_v)_{v\in V}$, 
$\beta_v\in\Con(\zA_v)$, $v\in V$, be a collection of congruences. Let 
$\cW^\cP(\ov\beta)$ denote the set of triples $(v,\al,\beta)$ 
such that $v\in V$, $\al,\beta\in\Con(\zA_v)$, and $\al\prec\beta\le\beta_v$. 
Also, $\cW^\cP$ denotes $\cW^\cP(\ov\beta)$ when $\beta_v$ is the full 
relation for all $v\in V$. We will omit the superscript $\cP$ whenever it is clear 
from the context. 

For every $(v,\al,\beta)\in\cW(\ov\beta)$, let $W_{v,\al\beta,\ov\beta}$ denote 
the set of variables $w\in V$ such that $(\al,\beta)$ and $(\gm,\dl)$
cannot be separated in $\rel^{vw}$ for some $\gm,\dl\in\Con(\zA_w)$ with 
$(w,\gm,\dl)\in\cW^\cP(\ov\beta)$. 
Let $\cW'(\ov\beta)$ (and respectively $\cW'$) denote the 
set of triples $(v,\al,\beta)\in\cW(\ov\beta)$ (respectively, from $\cW$), for
which $\zeta(\al,\beta)$ is the full relation. 

We say that algebra $\zA_v$ is 
\emph{semilattice free}\index{semilattice free algebra} if it does not contain 
semilattice edges. Let $\razm(\cP)$ denote the 
maximal size of domains of $\cP$ that are not semilattice free and 
$\Razm(\cP)$ be the set of variables $v\in V$ such that
$|\zA_v|=\razm(\cP)$ and $\zA_v$ is not semilattice free. For an instance $\cP$ 
we say that an instance $\cP'$ is \emph{strictly smaller} than instance $\cP$ if 
$\razm(\cP')<\razm(\cP)$. For $Y\sse V$ 
let $\mu^Y_v=\mu_v$ if $v\in Y$ and $\mu^Y_v=\zz_v$ otherwise. 

Instance $\cP$ is said to be \emph{block-minimal}\index{block-minimal instance} 
if for every $(v,\al,\beta)\in\cW$ (here $\beta_v=\zo_v$, $v\in V$) the following
conditions hold:
\begin{itemize}
\item[(BM1)]
for every $C=\ang{\bs,\rel}\in\cC$ the problem $\cP_{W_{v,\al\beta,\ov\beta}}$
if $(v,\al,\beta)\not\in\cW'$, and the problem 
$\cP_{W_{v,\al\beta,\ov\beta}}\fac{\ov\mu^Y}$ otherwise, where 
$Y=\Razm(\cP)-\bs$, is minimal;
\item[(BM2)]
if $(v,\al,\beta)\in\cW'$, then for every $(w,\gm,\dl)\in\cW-\cW'$
the problem $\cP_{W_{v,\al\beta,\ov\beta}}\fac{\ov\mu^Y}$, where
$Y=\Razm(\cP)-(W_{v,\al\beta,\ov\beta}\cap W_{w,\gm\dl,\ov\beta})$
is minimal.
\end{itemize}
Observe that $W_{v,\al\beta,\ov\beta}$ can be large, even
equal to $V$. However if $(v,\al,\beta)\not\in\cW'$ by 
Lemma~\ref{lem:central-decomposition-p1} the problem 
$\cP_{W_{v,\al\beta,\ov\beta}}$ splits into a union of disjoint problems over 
smaller domains, and so its minimality can be established by recursing
to strictly smaller problems. On the other hand, if $(v,\al,\beta)\in\cW'$ then 
$\cP_{W_{v,\al\beta,\ov\beta}}$ may not be decomposable. Since we need 
an efficient procedure of establishing block-minimality, this explains the 
complications introduced in (BM1),(BM2). In the first case 
$\cP_{W_{v,\al\beta,\ov\beta}}\fac{\ov\mu^Y}$ can be solved for each tuple
$\ba\in\rel$. Taking the quotient algebras of the domains guarantees that
we recurse to strictly smaller instances. In the second case 
$\cP_{W_{v,\al\beta,\ov\beta}\cap W_{w,\gm\dl,\ov\beta}}\fac{\ov\mu^Y}$
is decomposable, and we branch on those strictly smaller subproblems.

\begin{lemma}\label{lem:to-block-minimality-p1}
Let $\cP=(V,\cC)$ be a (2,3)-minimal instance. 
Then $\cP$ can be transformed to an equivalent block-minimal 
instance $\cP'$ by solving a quadratic number of strictly smaller CSPs.
\end{lemma}

\paragraph{The algorithm.}
%
In the algorithm we distinguish three cases depending on semilattice edges and 
quasi-centralizers of the domains of variables. In each case we employ different 
methods of solving or reducing the instance to a strictly smaller one.

Let $\cP=(V,\cC)$ be a subdirectly irreducible (2,3)-minimal instance. Let 
$\Centr(\cP)$ denote the set of variables $v\in V$ such 
that $\zeta(\zz_v,\mu_v)=\zo_v$. Let $\mu^*_v=\mu_v$ if 
$v\in\Razm(\cP)\cap\Centr(\cP)$ and $\mu^*_v=\zz_v$ otherwise. 

\vspace{2mm}

\emph{Semilattice free domains.}
If no domain of $\cP$ contains a semilattice edge
then by Proposition~\ref{pro:semilattice-edges-p1} $\cP$ can be solved in polynomial 
time, using the few subpowers algorithm, as shown in 
\cite{Idziak10:few,Bulatov16:restricted}.

\emph{Trivial centralizers.}
If $\mu^*_v=\zz_v$ for all $v\in V$, block-minimality guarantees the existence 
of a solution, and we can use Lemma~\ref{lem:to-block-minimality-p1} to solve 
the instance. 

\begin{theorem}\label{the:non-central-p1}
If $\cP$ is subdirectly irreducible, (2,3)-minimal, block-minimal, and \lb
$\Razm(\cP)\cap\Centr(\cP)=\eps$, then $\cP$ has a solution.
\end{theorem}

\emph{Nontrivial centralizers.}
Suppose that $\Razm(\cP)\cap\Centr(\cP)\ne\eps$. In this case we consider the 
problem $\cP\fac{\ov\mu^*}$. For this problem either 
$\razm(\cP\fac{\ov\mu^*})<\razm(\cP)$, or 
$\Razm(\cP\fac{\ov\mu^*})\cap\Centr(\cP\fac{\ov\mu^*})=\eps$; in either
case it can be solved by the previous case or by recursion to a strictly smaller 
problem.
We find a solution $\vf$ of $\cP\fac{\ov\mu^*}$ satisfying the following 
conditions. For every $v\in V$ such that $\zA_v$ is not semilattice free 
there is $a\in\zA_v$ such that 
$\{a,\vf(v)\}$ is a semilattice edge if $\mu^*_v=\zz_v$, or, if $\mu^*_v=\mu_v$,
there is $b\in\vf(v)$ such that $\{a,b\}$ is a semilattice edge. Then we apply the
transformation of $\cP$ suggested by Maroti in \cite{Maroti10:tree}. By $\cP\cdot\vf$
we denote the instance $(V,\cC_\vf)$ given by the rule: for every 
$\C=\ang{\bs,\rel}\in\cC$ the set $\cC_\vf$ contains a constraint 
$\ang{\bs,\rel\cdot\vf}$. To construct $\rel\cdot\vf$ choose a tuple $\bb\in\rel$
such that $\bb[v]^{\mu^*_v}=\vf(v)$ for all $v\in\bs$; this is possible 
because $\vf$ is a solution of $\cP\fac{\ov\mu^*}$. Then set 
$\rel\cdot\vf=\{\ba\cdot\bb\mid \ba\in\rel\}$. By the results of 
\cite{Maroti10:tree} and Lemma~\ref{lem:centralizer-multiplication-p1} the instance 
$\cP\cdot\vf$ has a solution if and only if $\cP$ does and 
$\razm(\cP\cdot\vf)<\razm(\cP)$.

\begin{theorem}\label{the:central-p1}
If $\cP\fac{\ov\mu^*}$ is 1-minimal, 
then $\cP$ can be reduced in polynomial time to a strictly smaller instance.
\end{theorem}

\begin{algorithm}
\caption{Procedure {\sf SolveCSP}}
\label{alg:csp-p1}
\begin{algorithmic}[1] 
\REQUIRE A CSP instance $\cP=(V,\cC)$ from $\CSP(\zA)$
\ENSURE A solution of $\cP$ if one exists, `NO' otherwise
\IF{all the domains are semilattice free}
\STATE Solve $\cP$ using the few subpowers algorithm and 
RETURN the answer
\ENDIF
\STATE Transform $\cP$ to a subdirectly irreducible, block-minimal and 
(2,3)-minimal instance
\STATE $\mu^*_v=\mu_v$ for $v\in\Razm(\cP)\cap\Centr(\cP)$ and 
$\mu^*_v=\zz_v$ otherwise
\STATE $\cP^*=\cP\fac{\ov\mu^*}$
\STATE  \textit{/*  Check the 1-minimality of $\cP^*$}
\FOR{every $v\in V$ and $a\in\zA_v\fac{\mu^*_v}$} 
\STATE $\cP'=\cP^*_{(v,a)}$ \ \ \ \textit{/* Add a constraint $\ang{(v),\{a\}}$ 
fixing the value of $v$ to $a$}
\STATE Transform $\cP'$ to a subdirectly irreducible, (2,3)-minimal instance
$\cP''$
\STATE If $\razm(\cP'')<\razm(\cP)$ call {\sf SolveCSP} on $\cP''$ and flag $a$ if $\cP''$ has no solution
\STATE Establish block-minimality of $\cP''$; if the problem changes, return to Step~10
\STATE If the resulting instance is empty, flag the element $a$
\ENDFOR
\STATE If there are flagged values, tighten the instance by removing 
the flagged elements and start over
\STATE Use Theorem~\ref{the:central-p1} to reduce $\cP$ to an
instance $\cP'$ with $\razm(\cP')<\razm(\cP)$
\STATE Call {\sf SolveCSP} on $\cP'$ and RETURN the answer
\end{algorithmic}
\end{algorithm}

\paragraph{Comments on the algorithm.}
Using Lemma~\ref{lem:to-block-minimality-p1} and 
Theorems~\ref{the:non-central-p1},\ref{the:central-p1} it is not difficult to see 
that the algorithm runs in polynomial time. Indeed, every time it makes 
a recursive call it calls on a problem whose non-semilattice free domains
have strictly smaller size, and therefore the depth of recursion is
bounded by $|\zA|$ if we are dealing with $\CSP(\zA)$.

In order to prove Theorem~\ref{the:non-central-p1} we introduce 
$\ov\beta$-strategies that are somewhat similar to (2,3)-strategies
in the sense that they are also collections of relations defined through 
some sort of minimality condition and are consistent. We show how such 
constructions can be used to prove Theorem~\ref{the:non-central-p1}.

Let $\cP=(V,\cC)$ be a subdirectly irreducible, (2,3)-minimal and block-minimal
instance. Let $\zA_v$ denote the domain of $v\in V$. Also, let 
$\beta_v\in\Con(\zA_v)$ and $B_v$ a $\beta_v$-block. Let 
$\cR$ be a collection of relations $\rel_{C,v,\al\beta}$ for every 
$C=\ang{\bs,\rel}\in\cC$, $(v,\al,\beta)\in\cW(\ov\beta)$ and such that 
$S(C,v,\al\beta)=\bs\cap W_{v,\al\beta,\ov\beta}$ is its set of coordinate 
positions. Similar to (2,3)-minimality a tuple $\ba\in\prod_{w\in X}\zA_x$ for
some $X\sse V$, is called \emph{$\cR$-compatible} if for any $C\in\cC$ and 
$(v,\al,\beta)\in\cW(\ov\beta)$ it holds $\pr_T\ba\in\pr_T\rel_{C,v,\al\beta}$,
where $T=X\cap S(C,v,\al\beta)$. Collection $\cR$ is said to be a 
\emph{$\ov\beta$-strategy} with respect to $(B_v)_{v\in V}$ if the following 
conditions hold for every $C=\ang{\bs,\rel}\in\cC$ and 
$(v,\al,\beta)\in\cW(\ov\beta)$ (let 
$W=W_{v,\al\beta,\ov\beta}$)\footnote{These are
the most important conditions that a $\ov\beta$-strategy has to satisfy. A 
complete and precise list of conditions can be found in 
Section~\ref{sec:strategies} or \cite{Bulatov17:cspd}.}:
\begin{itemize}
\item[(S1)]
the relations $\rel^{X,\cR}$, where $\rel^{X,\cR}$ consists 
of $\cR$-compatible tuples from $\rel^X$ 
for $X\sse V$, $|X|\le 2$, form a nonempty $(2,3)$-strategy for $\cP$; 
\item[(S2)]
for every $(w,\gm,\dl)\in\cW(\ov\beta)$ (let $U=W_{w,\gm\dl}$)
and every $\ba\in\pr_{\bs\cap W\cap U}\rel_{C,v,\al\beta}$ it holds:
if $(w,\gm,\dl)\not\in\cW'$ then $\ba$ extends
to an $\cR$-compatible solution of $\cP_U$; otherwise
if $(v,\al,\beta)\not\in\cW'$ then $\ba$ extends to an $\cR$-compatible 
solution of $\cP_U\fac{\ov\mu^{Y_1}}$ with $Y_1=\Razm(\cP)-(W\cap U)$; 
and if $(v,\al,\beta)\in\cW'$ then $\ba$ extends to an $\cR$-compatible 
solution of $\cP_U\fac{\ov\mu^{Y_2}}$, where $Y_2=\Razm(\cP)-\bs$;
\item[(S3)]
$\rel\cap\prod_{w\in\bs} B_w\ne\eps$ and for any $I\sse\bs$ any 
$\cR$-compatible tuple $\ba\in\pr_I\rel$ extends to an $\cR$-compatible tuple 
$\bb\in\rel$.
\end{itemize}

Let $\cP$ be a block-minimal instance, $\beta_v=\zo_v$ and $B_v=\zA_v$
for $v\in V$. Then as it is not hard to see the collection of relations 
$\cR=\{\rel_{C,v,\al\beta}\mid (v,\al,\beta)\in\cW(\ov\beta), C\in\cC\}$ given by
$\rel_{C,v,\al\beta}=\pr_{\bs\cap W_{v,\al\beta,\ov\beta}}\rel$ for 
$C=\ang{\bs,\rel}\in\cC$ is a $\ov\beta$-strategy with respect to $\ov B$.
%
Also, by (S3) a $\ov\gm$-strategy with $\gm_v=\zz_v$ gives a solution of 
$\cP$. 
%
Our goal is therefore to show that a $\ov\beta$-strategy for any
$\ov\beta$ can be `reduced', that is, transformed to a $\ov\beta'$-strategy 
for some $\ov\beta'<\ov\beta$. Note that this reduction of strategies is where 
the condition $\Razm(\cP)\cap\Centr(\cP)=\eps$ is used. Indeed, suppose that 
$\beta_v=\mu^*_v$. Then by conditions (S1)--(S3) we only have information 
about solutions to problems of the form $\cP_W\fac{\ov\mu^*}$ or something 
very close to that. Therefore this barrier cannot be penetrated.

\part{Technicalities}\label{part:technicalities}

We start this part with preliminaries, where apart from 
additional definitions and notation, we remind some of the results of 
\cite{Bulatov04:graph,Bulatov16:connectivity,Bulatov16:restricted} related to
colored graphs of algebras and relational structures and also some of their 
properties. In Sections~\ref{sec:maximality}--\ref{sec:rectangularity} we 
advance these results a little further. Then in 
Section~\ref{sec:polynomials-maximality} we introduce in a more detailed 
way the method of separating factors in congruence lattices using polynomial 
operations of the algebra. This method 
constitutes the basis for our algorithm. Some preliminary versions of this approach 
can be found in \cite{Bulatov02:maltsev-3-element,Bulatov17:semilattice}. 
In Section~\ref{sec:centralizer} we introduce again and study the
quasi-centralizer operator on congruence lattices that is similar to the well studied 
centralizer operator, although the precise relationship between the two is not 
quite clear. In particular, it allows to split certain CSPs into smaller ones. 
In Sections~\ref{sec:congruence} and~\ref{sec:chaining} we prove two important 
technical results

In Section~\ref{sec:maximal-solutions} we give a description of the algorithm, 
and prove its running time and partially soundness. In very broad strokes the 
algorithm works as follows. If none of the domains of $\cP$ contains a semilattice
edge in the sense of colored graphs of algebras, then $\cP$ can be solved 
by the few subpowers algorithm \cite{Berman10:varieties,Idziak10:few}, as
shown in \cite{Bulatov16:restricted}. Otherwise in most cases the problem can 
be solved by establishing \emph{block-minimality} similar to that in 
\cite{Bulatov17:semilattice}. 
The problematic case when block-minimality does not provide a solution,
or rather when it cannot be established is roughly speaking when the domains of 
the instance have nontrivial centers in the sense of the commutator theory.
In this case we show in Section~\ref{sec:theorem-47} that a solution of a 
problem $\cP'$ obtained from $\cP$ by replacing some of its domains with quotient 
algebras modulo their centers allows one to 
reduce the number of semilattice edges in those domains, and we can recurse 
to an instance with smaller domains.

The key ingredient of our result is presented 
in Section~\ref{sec:strategies}. There for block-minimal instances we introduce 
strategies that are in certain aspects similar to strategies 
used to solve problems of bounded width, but allow us to approach general 
CSPs. Then in Section~\ref{sec:affine-consistency} we show, 
Theorem~\ref{the:non-central}, 
that if for a CSP instance $\cP$ satisfying the block-minimality conditions such a 
strategy exists, one can improve (tighten) the strategy to obtain a solution 
of the quotient problem $\cP'$ needed to reduce semilattice edges. This 
theorem is the most difficult
and technically involved part of the proof. Tightening of a strategy works 
by (effectively) reducing domains of the CSP to a class of a maximal congruence, 
and then repeating the process as long as possible. The main cases of tightening 
considered are: when the interval formed by the maximal congruence used and the 
full congruence is Abelian, and when it is non-Abelian, Sections~\ref{sec:type-2} 
and~\ref{sec:type-not-2}, respectively. In the two cases we use quite different 
transformations of the strategy. In the Abelian case the argument 
is based on the rectangularity of relations understood in a general sense, 
while in the non-Abelian case the transformation is similar to that used for 
bounded width CSPs in \cite{Bulatov16:restricted}.

\section{Preliminaries}\label{sec:preliminaries}

We expand upon many of the definitions and notation given in 
Part~\ref{part:outline}. For the sake of convenience we also repeat some of the 
definitions given in Part~\ref{part:outline}.

\subsection{Universal algebra and CSP: notation and agreements}

We assume familiarity with the basics of universal algebra and the algebraic 
approach to the CSP. For reference on universal algebra please use 
\cite{Burris81:universal,Mckenzie87:algebras}; for the algebraic approach
see the recent survey \cite{Barto17:polymorphisms} and earlier papers  
\cite{Barto12:absorbing,Barto15:constraint,Bulatov05:classifying,%
Bulatov08:recent,Bulatov08:dualities,Bulatov16:connectivity}. 

By $[n]$\label{not:n} we denote the set $\{1\zd n\}$. For sets $\vc An$ tuples 
from $A_1\tms A_n$ are denoted in boldface, say, $\ba$; the $i$th component of 
$\ba$ is referred to as $\ba[i]$\label{not:bai}. An $n$-ary relation 
$\rel$\index{relation} over sets $\vc An$ is any subset of $A_1\tms A_n$. For 
$I=\{\vc ik\}\sse[n]$ by $\pr_I\ba,\pr_I\rel$ we denote the \emph{projections} 
\index{projections} $\pr_I\ba=(\ba[i_1]\zd\ba[i_k])$, 
$\pr_I\rel=\{\pr_I\ba\mid\ba\in\rel\}$\label{not:pr} of tuple
$\ba$ and relation $\rel$. If $\pr_i\rel=A_i$ for each $i\in[n]$, relation $\rel$ is 
said to be a \emph{subdirect product}\index{subdirect product} of 
$A_1\tms A_n$. It will be convenient to use $\ov A$\label{not:ov-A} for 
$A_1\tms A_n$, or for $\prod_{v\in V}A_v$ if the sets $V$ and $A_v$ are 
clear from the context. For $I\sse[n]$ or $I\sse V$ we will use 
$\ov A_I$\label{not:ov-A-I}, for $\prod_{i\in I}A_i$, or if $I$ is clear from the 
context just $\ov A$.

Algebras will be denoted by $\zA,\zB$ etc.; we often do not distinguish between 
subuniverses and subalgebras. For $B\sse\zA$ the subalgebra generated by
$B$ is denoted $\Sg B$\label{not:Sg}. For $C\sse\zA^2$ the congruence 
generated by $C$ is denoted $\Cg{C}$\label{not:Cg}. The equality relation and 
the full congruence of algebra $\zA$ are denoted $\zz_\zA$ and 
$\zo_\zA$\label{not:zz-zo}, respectively. Often when we need to 
use one of these trivial congruences of an algebra indexed in some way, say, 
$\zA_i$, we write $\zz_i,\zo_i$ for $\zz_{\zA_i}, \zo_{\zA_i}$.
The set of all polynomials (unary polynomials) of $\zA$ 
is denoted by $\Pol(\zA)$ and $\Polo(\zA)$\label{not:Pol}, respectively. We 
frequently use operations on subalgebras of direct products of algebras, say, 
$\rel\sse\zA_1\tms \zA_n$. If $f$ is such an operation (say, $k$-ary) then we 
denote its component-wise action 
also by $f$, e.g.\ $f(\vc ak)$ for $\vc ak\in\zA_i$. In the same way we denote the
action of $f$ on projections of $\rel$, e.g.\ $f(\vc\ba k)$ for $I\sse[n]$ and 
$\vc\ba k\in\pr_I\rel$. What we mean will always be clear from the context.
We use similar agreements for collections of congruences. If $\al_i\in\Con(\zA_i)$
then $\ov\al$\label{not:ov-alpha} denotes the congruence $\al_1\tms\al_n$ of 
$\rel$. If $I\sse[n]$ we use $\ov\al_I$\label{not:ov-alpha-I} to denote 
$\prod_{i\in I}\al_i$. If it does not lead to a confusion we write 
$\ov\al$ for $\ov\al_I$. Sometimes $\al_i$ are specified for $i$ from a certain 
set $I\sse[n]$, then by $\ov\al$ we mean the congruence $\prod_{i\in[n]}\al'_i$
where $\al'_i=\al_i$ if $i\in I$ and $\al'_i$ is the equality relation otherwise.
For example, if $\al\in\Con(\zA_1)$ then $\rel\fac\al$ means the factor of $\rel$
modulo $\al\tm\zz_2\tm\dots\tm\zz_n$. For $\al,\beta\in\Con(\zA)$ we write
$\al\prec\beta$ if $\al<\beta$ and $\al\le\gm\le\beta$ in $\Con(\zA)$ implies 
$\gm=\al$ or $\gm=\beta$.
In this paper all algebras are finite, idempotent and omit type \one.

The (\emph{nonuniform}) \emph{Constraint Satisfaction 
Problem}\index{Constraint Satisfaction Problem}\index{nonuniform Constraint 
Satisfaction Problem} (\emph{CSP}) associated with a relational 
structure $\bB$ is the problem $\CSP(\bB)$, in which, given a structure 
$\bA$ of the same signature as $\bB$, the goal is to decide whether or not 
there is a homomorphism from $\bA$ to $\bB$. For a class of similar algebras 
$\cA=\{\zA_i\mid i\in I\}$ for some set $I$ 
an \emph{instance}\index{instance} of  $\CSP(\cA)$ 
is a triple $(V,\dl,\cC)$, where $V$ is a set of variables; $\dl:V\to\cA$ is a 
\emph{type function}\index{type function}
that associates every variable with a \emph{domain}\index{domain} in $\cA$. 
Finally, $\cC$ is a set of \emph{constraints}\index{constraints}, 
i.e.\ pairs $\ang{\bs,\rel}$\label{not:ang}, where $\bs=(\vc vk)$ is a tuple of 
variables from $V$, the \emph{constraint scope}\index{constraint scope}, 
and $\rel\in\Inv(\cA)$, a subset of $\zA_{\dl(v_1)}\tm\dots\tm\zA_{\dl(v_k)}$, 
the \emph{constraint relation}\index{constraint relation}.
The goal is to find a \emph{solution}, that is a mapping 
$\vf:V\to\bigcup\cA$ such that $\vf(v)\in\zA_{\dl(v)}$ and for every constraint 
$\ang{\bs,\rel}$, $\vf(\bs)\in\rel$. It is easy to see that if $\cA$ is a class 
containing just one algebra $\zA$, then $\CSP(\cA)$ can be viewed as the union 
of $\CSP(\bA)$ for all relational structures $\bA$ invariant under the operations of 
$\zA$. To simplify the notation we always write $\zA_v$ rather than 
$\zA_{\dl(v)}$, because the mapping $\dl$ is always clear 
from the context. This also allows us to simplify the notation for instances 
to $\cP=(V,\cC)$\label{not:cP}. To allow for transformations of CSP described 
below we assume that $\cA$ is closed under taking subalgebras and quotient 
algebras.

The set of solutions of a CSP instance $\cP=(V,\cC)$ will be denoted by 
$\cS_\cP$\label{not:cS}, or just $\cS$ if $\cP$ is clear from the context. For 
$W\sse V$ by $\cP_W$\label{not:cP-W} we denote the 
\emph{restriction}\index{restriction of CSP} of $\cP$ onto $W$, that is, the 
instance $(W,\cC_W)$, where for each $C=\ang{\bs,\rel}\in\cC$, the set $\cC_W$
includes the constraint $C_W=\ang{\bs\cap W,\pr_{\bs\cap W}\rel}$.
The set of solutions of $\cP_W$ will be denoted by $\cS_W$\label{not:cS-W}. 
For $v\in V$ and a subalgebra $\zB$ of $\zA_v$ by 
$\cP_{(v,\zB)}$\label{not:P-v-B} we denote the instance
$\cP$ with an extra constraint $\ang{\{v\},\zB}$; note that this is essentially 
equivalent to reducing the domain of $v$, and this is how we usually consider 
this construction. For $C=\ang{\bs,\rel}\in\cC$ let $\rel'$ be a subalgebra of 
$\rel$ and $C'=\ang{\bs,\rel'}$. The instance obtained from $\cP$ replacing 
$C$ with $C'$ is denoted by $\cP_{C\to C'}$\label{not:C-to-C}. The 
transformation of $\cP$ by reducing the domain of a variable
$v\in V$ or reducing a constraint $C\in\cC$, that is, transforming $\cP$ into 
$\cP_{(v,\zB)}$ or $\cP_{C\to C'}$ in such a way that the new instance has a 
solution if and only if $\cP$ does, will be called 
\emph{tightening}\index{tightening of a CSP} of $\cP$. Let $\al_v$ be a
congruence of $\zA_v$ for each $v\in V$. By $\cP\fac{\ov\al}$\label{not:P-fac} 
we denote the instance $(V,\cC^{\ov\al})$ constructed as follows: the domain 
of $v\in V$ is $\zA_v\fac{\al_v}$; for every constraint $C=\ang{\bs,\rel}\in\cC$,
the set $\cC^{\ov\al}$ includes the constraint $\ang{\bs,\rel\fac{\ov\al_\bs}}$.
Note that if $\cA$ is closed under taking subalgebras and quotient algebras, then 
applying a transformation of one of these kinds to an instance from 
$\CSP(\cA)$ results again in an instance from $\CSP(\cA)$.

Instance $\cP$ is said to be \emph{minimal}\index{minimal instance} (or 
\emph{globally minimal}\index{globally minimal instance})  if for every 
$C=\ang{\bs,\rel}\in\cC$ and every $\ba\in\rel$ there is a solution 
$\vf\in\cS$ such that $\vf(\bs)=\ba$. Instance $\cP$ is said to be
\emph{1-minimal}\index{1-minimal instance} if for every $v\in V$ and every 
constraint $C=\ang{\bs,\rel}\in\cC$ such that $v\in\bs$, $\pr_v\rel=\cS_v$.
Instance $\cP$ is said to be 
\emph{(2,3)-consistent}\index{(2,3)-consistent instance} if it has a 
\emph{(2,3)-strategy}\index{(2,3)-strategy}, that is, a collection of relations 
$\rel^X$\label{not:R-X}, $X\sse V$, $|X|=2$ satisfying the following conditions:\\
-- for every $X\sse V$ with $|X|\le2$ and any constraint $C=\ang{\bs,\rel}$,
$\pr_{\bs\cap X}\rel^X\sse\pr_{\bs\cap X}\rel$;\\ 
-- for every $X=\{u,v\}\sse V$, any $w\in V-X$ and any 
$(a,b)\in\rel^X$, there is $c\in\zA_w$ such that $(a,c)\in\rel^{\{u,w\}}$
and $(b,c)\in\rel^{\{v,w\}}$.\\[1mm]
We will always assume that a (2,3)-consistent instance has a constraint 
$C^X=\ang{X,\cS_X}$ for every $X\sse V$, $|X|=2$. Then clearly
$\rel^X\sse\cS_X$. Let the collection of relations $\rel^X$ be denoted 
by $\cR$\label{not:cR}. 
A tuple $\ba$ whose entries are indexed with elements 
of $W\sse V$ such that $\pr_X\ba\in\rel^X$ for any $X\sse W$, $|X|=2$,
will be called \emph{$\cR$-compatible}\index{$\cR$-compatible tuple}. If a 
(2,3)-consistent instance $\cP$ with a (2,3)-strategy $\cR$ satisfies the 
additional condition\\[1mm]
-- for every constraint $C=\ang{\bs,\rel}$ of $\cP$ every tuple $\ba\in\rel$
is $\cR$-compatible,\\[1mm]
it is called \emph{(2,3)-minimal}\index{(2,3)-minimal instance}. Any instance 
can be transformed to a 1-minimal, (2,3)-consistent, or (2,3)-minimal instance 
in polynomial time using the standard constraint propagation algorithms 
(see, e.g.\ \cite{Dechter03:processing} or \cite{Bulatov08:dualities}). 
These algorithms tighten the instance.

\subsection{Minimal sets and polynomials}

We will use the following basic facts from the tame congruence theory
\cite{Hobby88:structure}, often without further notice. 

Let $\zA$ be a finite algebra and $\al,\beta\in\Con(\zA)$ with $\al\prec\beta$. 
An \emph{$(\al,\beta)$-minimal set}\index{minimal set} is a set minimal
with respect to inclusion among the sets of the form $f(\zA)$, where 
$f\in\Polo(\zA)$ is such that $f(\beta)\not\sse\al$. Sets $B,C$ are said
to be \emph{polynomially isomorphic}\index{polynomially isomorphic sets}
in $\zA$ if there are $f,g\in\Polo(\zA)$ such that $f(B)=C$, $g(C)=B$, and
$f\circ g, g\circ f$ are identity mappings on $C$ and $B$, respectively.

\begin{lemma}[Theorem 2.8, \cite{Hobby88:structure}]%
\label{lem:minimal-sets}
Let $\al,\beta\in\Con(\zA)$, $\al\prec\beta$. Then the following hold.\\[1mm]
(1) Any $(\al,\beta)$-minimal sets $U,V$ are polynomially isomorphic.\\[1mm]
(2) For any $(\al,\beta)$-minimal set $U$ and any $f\in\Polo(\zA)$, if
$f(\beta\red U)\not\sse\al$ then $f(U)$ is an $(\al,\beta)$-minimal set, $U$ 
and $f(U)$ are polynomially isomorphic,  and $f$ witnesses this fact.\\[1mm]
(3) For any $(\al,\beta)$-minimal set $U$ there is $f\in\Polo(\zA)$ such that
$f(\zA)=U$, $f(\beta)\not\sse\al$, and $f$ is idempotent, in particular, 
$f$ is the identity mapping on $U$.\\[1mm]
(4) For any $(a,b)\in\beta-\al$ and an $(\al,\beta)$-minimal set $U$ there is
$f\in\Polo(\zA)$ such that $f(\zA)=U$ and $(f(a),f(b))\in\beta\red U-\al\red U$.
Moreover, $f$ can be chosen to satisfy the conditions of item~(3).\\[1mm]
(5) For any $(\al,\beta)$-minimal set $U$, $\beta$ is the transitive closure of
$$
\al\cup\{(f(a),f(b))\mid (a,b)\in\beta\red U, f\in\Polo(\zA)\}.
$$
In fact, as $\al\prec\beta$ this claim can be strengthen to the following. 
For any $(a,b)\in\beta-\al$, $\beta$ is the transitive closure of 
$$
\al\cup\{(f(a),f(b))\mid f\in\Polo(\zA)\}.
$$
(6) For any $f\in\Polo(\zA)$ such that $f(\beta)\not\sse\al$ there is an
$(\al,\beta)$-minimal set $U$ such that $f$ witnesses that $U$ and $f(U)$
are polynomially isomorphic.
\end{lemma}

For an $(\al,\beta)$-minimal set $U$ and a $\beta$-block $B$ such that 
$\beta\red{U\cap B}\ne\al\red{U\cap B}$, the set $U\cap B$ is said
to be an \emph{$(\al,\beta)$-trace}\index{trace}. A 2-element set 
$\{a,b\}\sse U\cap B$ such that $(a,b)\in\beta-\al$, is called an 
\emph{$(\al,\beta)$-subtrace}\index{subtrace}. The union $Q$ of the 
traces from $U$ is called the \emph{body}\index{body} of $U$, and 
$U-Q$ is called the \emph{tail}\index{tail} of $U$. Depending on the 
structure of its minimal sets the interval $(\al,\beta)$ can be of one of the 
five types, \one--\five. Since we assume the tractability conditions of the
Dichotomy Conjecture, type~\one\ does not occur in algebras we deal with.

\begin{lemma}[Section~4 of \cite{Hobby88:structure}]\label{lem:traces}
Let $\al,\beta\in\Con(\zA)$ and $\al\prec\beta$. Then the following hold.\\[1mm]
(1) If $\typ(\al,\beta)=\two$ then every $(\al,\beta)$-trace is polynomially 
equivalent to a 1-dimensional vector space.\\[1mm]
(2) If $\typ(\al,\beta)\in\{\three,\four,\five\}$ then every $(\al,\beta)$-minimal 
set $U$ contains exactly one trace $T$, and if $\typ(\al,\beta)\in\{\three,\four\}$,
$T$ contains only 2 elements. Also, $T\fac\al$ is polynomially equivalent
to a Boolean algebra, 2-element lattice, or 2-element semilattice, respectively.
\end{lemma}

Intervals $(\al,\beta),(\gm,\dl)$, $\al,\beta,\gm,\dl\in\Con(\zA)$ and 
$\al\prec\beta,\gm\prec\dl$ are said to be 
\emph{perspective}\index{perspective intervals} if $\beta=\al\join\dl,
\gm=\al\meet\dl$, or $\dl=\beta\join\gm,\al=\beta\meet\gm$.

\begin{lemma}[Lemma~6.2, \cite{Hobby88:structure}]%
\label{lem:perspective-intervals}
Let $\al,\beta,\gm,\dl\in\Con(\zA)$ be such that $\al\prec\beta,\gm\prec\dl$
and intervals $(\al,\beta),(\gm,\dl)$ are perspective. Then
$\typ(\al,\beta)=\typ(\gm,\dl)$ and a set $U$ is $(\al,\beta)$-minimal
if and only if it is $(\gm,\dl)$-minimal.
\end{lemma}

We will also use polynomials that behave on a minimal set in a particular way.

\begin{lemma}[Lemmas 4.16, 4.17, \cite{Hobby88:structure}]%
\label{lem:pseudo-meet}
Let $\al,\beta\in\Con(\zA)$, $\al\prec\beta$, and
$\typ(\al,\beta)\in\{\three,\four,\five\}$. Let $U$ be an $(\al,\beta)$-minimal
set and $T$ its only trace. Then there is element $1\in T$ and a binary polynomial 
$g$ of $\zA$ such that\\[1mm]
(1) $(1,a)\not\in\al$ for any $a\in U-\{1\}$;\\[1mm]
(2) for all $a\in U-\{1\}$, the algebra $(\{a,1\},g)$ is a semilattice with 
neutral element~$1$, that is, $g(1,1)=1$ and $g(1,a)=g(a,1)=g(a,a)=a$.\\[1mm]
(3) for any $a\in U-\{1\}$ and any $b\in T-\{1\}$, 
$g(a,b)\eqc\al g(b,a)\eqc\al a$;\\[1mm]
(4) for all $a,b\in U$, $g(a,g(a,b))=g(a,b)$.\\[1mm]
Polynomial $g$ is said to be a 
\emph{pseudo-meet}\index{pseudo-meet operation} on $U$.
\end{lemma}

\subsection{Coloured graphs}


In \cite{Bulatov04:graph,Bulatov08:recent} we introduced a local approach to the
structure of finite algebras. As we use this approach throughout the paper, 
we present it here in some details, see also \cite{Bulatov16:connectivity}. 
For the sake of the definitions below we slightly abuse terminology 
and by a module\index{module} mean the full idempotent reduct of a module.

For an algebra $\zA$ graph $\cG(\zA)$\label{not:cG-A} is defined as follows. 
The vertex set is the universe $A$ of $\zA$. A pair $ab$ of vertices is an 
\emph{edge}\index{edge} if and only if there exists a congruence $\th$ of 
$\Sg{a,b}$, other than the full congruence and a term operation $f$ of $\zA$ 
such that either $\Sg{a,b}\fac\th$ is a module and $f$ is an affine 
operation on it, or $f$ is a semilattice operation on
$\{a^\th,b^\th\}$, or $f$ is a majority operation on
$\{a^\th,b^\th\}$. (Note that we use the same operation symbol in this case.)
Usually, $\th$ is chosen to be a maximal congruence of $\Sg{a,b}$.

If there are a congruence $\th$ and a term operation $f$ of $\zA$ such that $f$ 
is a semilattice operation on $\{a^\th,b^\th\}$ then $ab$ is said to have the
\emph{semilattice type}\index{semilattice type}. An edge $ab$ is of 
\emph{majority type}\index{majority type} if there are 
a congruence $\th$ and a term operation $f$ such that $f$ is a majority
operation on $\{a^\th,b^\th\}$ and there is no semilattice 
term operation on $\{a^\th,b^\th\}$. Finally, $ab$ 
has the \emph{affine type}\index{majority type} if there are $\th$ and $f$ 
such that $f$ is an affine operation on $\Sg{a,b}\fac\th$ and 
$\Sg{a,b}\fac\th$ is a module; in particular it implies that 
there is no semilattice or majority operation on $\{a^\th,b^\th\}$.  In all 
cases we say that congruence $\th$ \emph{witnesses}\index{witness of a type} 
the type of edge $ab$. Observe that a pair $ab$ can still be an edge of more 
than one type as witnessed by different congruences, although this has 
consequences in this paper.

Omitting type \one\ can be characterized as follows.

\begin{theorem}[\cite{Bulatov04:graph,Bulatov16:connectivity}]%
\label{the:connectedness}
An idempotent algebra $\zA$ omits type \one\ (that is, the variety generated
by $\zA$ omits type if and only if $\cG(\zB)$ is connected
for every subalgebra $\zB$ of $\zA$.

Moreover, a finite class $\cA$ of similar idempotent algebras closed under 
subalgebras and quotient algebras omit type~\one\ if and only if $\cG(\zA)$ 
is connected for any $\zA\in\cA$.
\end{theorem}

For the sake of the dichotomy conjecture, it suffices to consider 
\emph{reducts}\index{reduct} of an algebra $\zA$ omitting type \one, 
that is, algebras with the same universe but 
reduced set of term operations, as long as reducts also omit type \one. 
In particular, we are interested in reducts of $\zA$, in which semilattice and 
majority edges are subalgebras.

\begin{theorem}[\cite{Bulatov04:graph,Bulatov16:connectivity}]%
\label{the:adding}
Let $\zA$ be an algebra such that $\cG(\zB)$ is connected for all subalgebras of 
$\zB$ of $\zA$, and let $ab$ be an edge of $\cG(\zA)$ of the semilattice or 
majority type witnessed by congruence $\th$, and $\rel_{ab}=a^\th\cup b^\th$.
Let also $F_{ab}$ denote set of term operations of $\zA$ preserving $\rel_{ab}$, 
and $\zA'=(A,F_{ab})$. Then $\cG(\B')$ is connected for all subalgebras $\zB'$ 
of $\zA'$.
\end{theorem}

An algebra $\zA$ such that $a^\th\cup b^\th$ is a subuniverse of $\zA$ 
for every semilattice or majority edge $ab$ of $\zA$ is called 
\emph{sm-smooth}\index{sm-smooth algebra}. 
If $\cA$ is the class of all quotient algebras of subalgebras of an sm-smooth 
algebra $\zA$, it is easy to see that that every $\zB\in\cA$ is sm-smooth. 
Although it is not needed in this paper, for any finite class $\cA$ omitting 
type \one\ there is a class $\cA'$ of sm-smooth algebras which are reducts of
algebras from $\cA$, and  such that $\cA'$ omits type \one, as well. In the rest 
of the paper all algebras are assumed to be sm-smooth.

The next statement uniformizes the operations witnessing the type of edges.

\begin{theorem}[\cite{Bulatov04:graph,Bulatov16:connectivity}]%
\label{the:uniform}
Let $\cA$ be a class of similar idempotent algebra closed under taking 
subalgebras and quotient algebras. There are term operations $f,g,h$ of 
$\cA$ such that for any $\zA\in\cA$ and any $a,b\in\zA$ operation $f$ is a 
semilattice operation on $\{a^\th,b^\th\}$ if $ab$ is a semilattice edge;
$g$ is a majority operation on $\{a^\th,b^\th\}$ if $ab$ is a majority edge;
$h$ is an affine operation on $\Sg{a,b}\fac\th$ if $ab$ is an affine 
edge, where $\th$ witnesses the type of the edge.
Moreover,  $f,g,h$ can be chosen such that
\begin{itemize}
\item[(1)]
$f(x,f(x,y))=f(x,y)$ for all $x,y\in\zA$, $\zA\in\cA$;
\item[(2)]
$g(x,g(x,y,y),g(x,y,y))=g(x,y,y)$ for all $x,y\in\zA$, $\zA\in\cA$;
\item[(3)]
$h(h(x,y,y),y,y)=h(x,y,y)$ for all $x,y\in\zA$, $\zA\in\cA$.
\end{itemize}
There is a term operation $t$ such that for any affine edge $ab$ and a majority, 
edge $cd$ witnessed by congruences $\eta$ and $\th$, respectively, 
$t(a,b)\eqc\eta a$ and $t(c,d)\eqc\th d$.
\end{theorem}

Unlike majority and affine operations, for a semilattice edge $ab$ and a 
congruence $\th$ of $\Sg{a,b}$ witnessing that, there can be semilattice 
operations acting differently on $\{a^\th,b^\th\}$, which corresponds to 
the two possible orientations of $ab$. In every such case by fixing operation 
$f$ from Theorem~\ref{the:uniform} we effectively choose one of the two 
orientations. In this paper we do not really care about what 
orientation is preferable. 

In \cite{Bulatov16:connectivity} we introduced a stronger notion of edge. 
A pair $ab$ of elements of algebra $\zA$ is called
a \emph{thin semilattice edge}\index{thin semilattice edge} if $ab$ is a 
semilattice edge, and the congruence 
witnessing that is the equality relation. In other words, $f(a,a)=a$ and 
$f(a,b)=f(b,a)=f(b,b)=b$. We denote the fact that $ab$ is a thin semilattice edge 
by $a\le b$\label{not:le}. Thin semilattice edges allow us to introduce a directed 
graph $\cG_s(\zA)$, whose vertices are the elements of $\zA$, and the arcs are 
the thin semilattice edges. We then can define 
\emph{semilattice-connected}\index{semilattice-connected component} and 
\emph{strongly semilattice-connected}%
\index{strongly semilattice-connected component}
components of $\cG_s(\zA)$. We will also use the natural order on the
set of strongly semilattice-connected components of $\cG_s(\zA)$: for
components $A,B$, we write $A\le B$ if there is a directed path in $\cG_s(\zA)$
connecting a vertex from $A$ with a vertex
from $B$. Elements from the maximal strongly connected components (or simply
\emph{maximal components}\index{maximal components}) of $\cG_s(\zA)$
are called \emph{maximal}\index{maximal element} elements of $\zA$ and 
the set of all such elements is denoted by $\max(\zA)$\label{not:max}. A 
directed path in $\cG_s(\zA)$ is called a 
\emph{semilattice path}\index{semilattice path} or 
\emph{s-path}\index{s-path}. If there is an s-path from $a$ to $b$
we write $a\sqq b$\label{not:sqq}. 

\begin{prop}[\cite{Bulatov04:graph,Bulatov16:connectivity}]%
\label{pro:good-operation}
Let $\cA$ be a finite class of similar idempotent algebras closed under taking 
subalgebras and quotient algebras. There is a binary
term operation $f$ of $\cA$ such that $f$ is a semilattice operation on
$\{a^\th,b^\th\}$ for every semilattice edge $ab$ of any $\zA\in\cA$, 
where congruence $\th$ witnesses the type of $ab$, and, for any $a,b\in\zA$, 
either $a=f(a,b)$ or the pair $(a,f(a,b))$ is a thin semilattice edge of $\zA$. 
Operation $f$ with this property will be denoted by a dot (think multiplication).
\end{prop}

Let operations $g,h$ be as in Theorem~\ref{the:uniform}.
A pair $ab$ from $\zA\in\cA$ is called a 
\emph{thin majority edge}\index{thin majority edge} if (a) it is a majority 
edge, let congruence $\th$ witness this, (b) for any $c\in b^\th$, 
$b\in\Sg{a,c}$, (c) $g(a,b,b)=b$, and (d) there exists a ternary term operation 
$g'$ such that $g'(a,b,b)=g'(b,a,b)=g'(b,b,a)=b$.  Finally, a pair $ab$ is called a 
\emph{thin affine edge}\index{thin affine edge} if (a) it is an affine edge, let 
congruence $\th$ witness this, (b) for any 
$c\in b^\th$, $b\in\Sg{a,c}$, (c) $h(b,a,a)=b$, (d) there exists a ternary 
term operation $h'$ such that $h'(b,a,a)=h'(a,a,b)=b$, and (e) $a$ is maximal in 
$\Sg{a,b}$.  Note that the operations $h,g$ from 
Theorem~\ref{the:uniform} do not have to be majority or affine operations 
on thin edges; thin edges do not have to be even closed under $g,h$. 
Thin edges of all types are oriented. We therefore can define yet another 
directed graph, $\cG'(\zA)$, in which the arcs are the thin edges of all types. 

\begin{lemma}[\cite{Bulatov16:connectivity}]\label{lem:thin}
Let $\zA$ be an algebra.\\[1mm]
(1) Let $ab$ be a semilattice or majority edge in $\zA$, and $\th$ the 
congruence of 
$\Sg{a,b}$ witnessing that. Then there is $b'\in b^\th$ such that $ab'$ is a thin 
semilattice or majority edge, respectively.\\[1mm]
(2)  Let $ab$ be an affine edge, and $\th$ the congruence of 
$\Sg{a,b}$ witnessing that. Then there are $a'\in a^\th$ and $b'\in b^\th$ such 
that $a\sqq a'$ in $a^\th$ and $a'b'$ is a thin affine edge.
\end{lemma}

The following simple properties of thin edges will be useful. Note that a 
subdirect product of algebras (a relation) is also an algebra, and so edges 
and thin edges can be defined for relations as well.

\begin{lemma}[\cite{Bulatov16:connectivity}]\label{lem:thin-properties}
(1) Let $\zA$ be an algebra and $ab$ a thin edge. Then $ab$
is a thin edge in any subalgebra of $\zA$ containing $a,b$, and $a^\th b^\th$
is a thin edge in $\zA\fac\th$ for any congruence $\th$.\\[1mm]
(2) Let $\rel$ be a subdirect product of $\vc\zA n$, $I\sse[n]$, and $\ba\bb$ a
thin edge in $\rel$. Then $\pr_I\ba\pr_I\bb$ is a thin edge in $\pr_I\rel$ of 
the same type as $\ba\bb$.
\end{lemma} 

We will need stronger versions of Lemmas~18 and~20 of 
\cite{Bulatov16:connectivity}. Let $\cA$ be a finite class of similar idempotent 
algebras closed under taking subalgebras and quotient algebras.

\begin{lemma}\label{lem:op-s-on-affine}
(1) Let $ab$ be a thin majority edge of algebra $\zA\in\cA$. There is a term
operation $t_{ab}$\label{not:1-ab} such that $t_{ab}(a,b)=b$ and $t_{ab}
(c,d)\eqc{\th_{cd}} c$ for all affine edges $cd$ of all $\zA'\in\cA$, where the 
type of $cd$ is witnessed by congruence $\th_{cd}$.\\[1mm]
(2) Let $ab$ be a thin affine edge of algebra $\zA\in\cA$. There is a term
operation $h_{ab}$\label{not:h-ab} such that $h_{ab}(a,a,b)=b$ and 
$h_{ab}(d,c,c)\eqc{\th_{cd}} d$ for all affine edges $cd$ of all 
$\zA'\in\cA$, where the type of $cd$ is witnessed by congruence $\th_{cd}$.
Moreover, $h_{ab}(x,c',d')$ is a permutation of $\Sg{c,d}\fac{\th_{cd}}$
for any $c',d'\in\Sg{c,d}$.\\[1mm]
(3) Let $ab$ and $cd$ be thin edges in algebras $\zA,\zA'\in\cA$, respectively.
If they have different types there is a binary term operation $p$ such that 
$p(a,b)=a$, $p(c,d)=d$. If both edges are affine then there is a term operation 
$h'$ such that $h'(a,a,b)=b$ and $h'(d,c,c)=d$.
\end{lemma}

\begin{proof}
(1) Let $c_1d_1\zd c_\ell d_\ell$ be a list of all affine edges of algebras in
$\cA$, $c_i,d_i\in\zA_i$ and $\th_{c_id_i}$ the corresponding congruences. Set 
$\bc=(\vc c\ell), \bd=(\vc d\ell)$. Let $\rel$ be the subalgebra of 
$\zA\tm\prod_{i=1}^\ell\zA_i$ generated by $(a,\bc),(b,\bd)$. 
Pair $ab$ is also a majority edge, let it be witnessed by a congruence $\th$.
By Theorem~\ref{the:uniform}
$$
\cl {b'}{\bc'}=t\left(\cl a\bc,\cl b\bd\right)\in\rel,
$$
where $b'\in b^\th$ and $\bc'[i]\eqc{\th_{c_id_i}}\bc[i]$, as $t$ is the 
first projection on $\Sg{c_i,d_i}\fac{\th_{c_id_i}}$ and a second 
 projection on $\Sg{a,b}\fac\th$.
Then as $b\in\Sg{a,b'}$, we get $(b,\bc'')\in\rel$ for some $\bc''$ such
that $\bc''[i]\eqc{\th_{c_id_i}}\bc[i]$. This means there is a binary term 
operation $t_{ab}$ such that 
$$
t_{ab}\left(\cl a\bc,\cl b\bd\right)=\cl b{\bc''}.
$$
The result follows.

(2) We use the notation from item (1) except $ab$ now is a thin affine edge
and $\rel$ is generated by $(a,\bd),(a,\bc),(b,\bc)$.
By condition (a) of the definition of thin affine edges,
$$
\cl {b'}{\bd'}=h\left(\cl a\bd,\cl a\bc,\cl b\bc\right)\in\rel,
$$
where $b'\in b^\th$ and $\bd'[i]\eqc{\th_{c_id_i}}\bd[i]$, as $h$ is a 
Mal'tsev operation on $\Sg{a,b}\fac\th$ and on 
$\Sg{c_i,d_i}\fac{\th_{c_id_i}}$. Then as $b\in\Sg{a,b'}$, by condition (b)
we get $(b,\bd'')\in\rel$ for some $\bd''$ such
that $\bd''[i]\eqc{\th_{c_id_i}}\bd[i]$. The first result follows.

Let now $h_{ab}$ be the term operation we constructed
and $c',d'\in\Sg{c_i,d_i}$, $i\in[\ell]$. Since 
$\zB=\Sg{c_i,d_i}\fac{\th_{c_id_i}}$ is a module, in particular, it is an Abelian
algebra and $h_{ab}(x,c^*,c^*)=x$ for all $c^*\in\zB$, the second result
follows.

(3) Follows from \cite{Bulatov16:connectivity}, Lemmas~15,18,19,20.
\end{proof}

\subsection{Maximality}\label{sec:maximality}

A directed path in $\cG'(\zA)$ is called an \emph{asm-path}\index{asm-path},  
if there is an asm-path from $a$ to $b$ we write 
$a\sqq_{asm} b$\label{not:sqq-asm}. If all edges of this path 
are semilattice or affine, it is called an 
\emph{affine-semilattice path}\index{as-path} or 
an \emph{as-path}, if there is an as-path from $a$ to $b$ we write 
$a\sqq_{as} b$\label{not:sqq-as}. Similar to maximal components, we 
consider strongly connected components of $\cG'(\zA)$ with majority edges 
removed, and the natural partial order on such components. The maximal 
components will be called \emph{as-components}\index{as-components}, and 
the elements from as-components are called 
\emph{as-maximal}\index{as-maximal}; the set of all
as-maximal elements of $\zA$ is denoted by $\amax(\zA)$\label{not:amax}. 
If $a$ is an as-maximal element, the as-component containing $a$ is denoted 
$\as(a)$\label{not:as}. An alternative way to define as-maximal elements is as 
follows: $a$ is as-maximal if for every $b\in\zA$ such that $a\sqq_{as} b$ it also 
holds that $b\sqq_{as}a$. Finally, element $a\in\zA$ is said to be 
\emph{universally maximal} (or \emph{u-maximal}\index{u-maximal} for short) 
if for every $b\in\zA$ such that $a\sqq_{asm} b$ 
it also holds that $b\sqq_{asm}a$. The set of all u-maximal elements of 
$\zA$ is denoted $\umax(\zA)$\label{not:umax}.

\begin{prop}[\cite{Bulatov16:connectivity}]\label{pro:as-connectivity}
Let $\zA$ be an algebra. Then\\[1mm]
(1) any $a,b\in\zA$ are connected in $\cG'(\zA)$ with an undirected path;\\[1mm]
(2) any $a,b\in\max(\zA)$ (or $a,b\in\amax(\zA)$, or $a,b\in\umax(\zA)$) are 
connected in $\cG'(\zA)$ with a directed path.
\end{prop}

\begin{proof}
Item (2) is only proved in \cite{Bulatov16:connectivity} for maximal and 
as-maximal elements; so we prove it here for u-maximal elements as well.
Let $a',b'\in\zA$ be maximal elements of $\zA$ such that $a\sqq a'$ and 
$b\sqq b'$. Then by Proposition~\ref{pro:as-connectivity} for maximal elements 
$a'\sqq_{asm}b'$, and, as $b$ is u-maximal, $b'\sqq_{asm}b$.
\end{proof}

Since for every $a\in\zA$ there is a maximal $a'\in\zA$ such that $a\sqq a'$, 
Proposition~\ref{pro:as-connectivity} implies that there is only one u-maximal 
component. U-maximality has an additional useful property, it is somewhat 
hereditary, as it made precise in the following

\begin{lemma}\label{lem:u-max-congruence}
Let $\zB$ be a subalgebra of $\zA$ containing a u-maximal element of $\zA$. 
Then every element u-maximal in $\zB$ is also 
u-maximal in $\zA$. In particular, if $\al$ is a congruence of $\zA$ and $\zB$ 
is a u-maximal $\al$-block, that is $\zB$ is a u-maximal element in 
$\zA\fac\al$, then $\umax(B)\sse\umax(\zA)$. 
\end{lemma}

\begin{proof}
Let $a\in \zB$ be an element u-maximal in $\zA$, let $b\in\umax(\zB)$. For any 
$c\in\zA$ with $b\sqq_{asm} c$ we also have $c\sqq_{asm}a$. Finally, since 
$b\in\umax(\zB)$ and $a\in\zB$, we have $a\sqq_{asm}b$. For the second 
part of the lemma we need to find a u-maximal element in $\zB$. Let 
$b\in\umax(\zA)$. Then as $\zB$ is u-maximal in $\zA\fac\al$ applying 
Lemma~\ref{lem:thin} we get that there is
 $a'\in\zB$ such that $b\sqq_{asm}a'$. Clearly, $a'\in\umax(\zA)$.
\end{proof}

Let $\zA$ be an algebra and $a\in\zA$. By $\Filt_\zA(a)$\label{not:Filt} we 
denote the set of elements $a$ is connected to (in terms of semilattice paths); 
similarly, by $\Filt^{as}_\zA(a)$\label{not:Filt-as} and 
$\Filt^{asm}_\zA(a)$\label{not:Filt-asm} we denote the set of elements 
$a$ is as-connected and asm-connected to. Also, 
$\Filt_\zA(C)=\bigcup_{a\in C}\Filt_\zA(a)$ 
($\Filt^{as}_\zA(C)=\bigcup_{a\in C}\Filt^{as}_\zA(a)$, 
$\Filt^{asm}_\zA(C)=\bigcup_{a\in C}\Filt^{asm}_\zA(a)$, respectively) 
for $C\sse \zA$. Note that if $a$ is an as-maximal element then 
$\as(a)=\Filt^{as}_\zA(a)$, and $a\in\Filt^{asm}_\zA(b)$ for any 
$b\in\zA$. We will need the following statements.

\begin{lemma}[The Maximality Lemma,\cite{Bulatov16:restricted}]%
\label{lem:to-max}
Let $\rel$ be a subdirect product of $\zA_1\tms\zA_n$, $I\sse[n]$.\\[1mm]
(1) For any $\ba\in\rel$, $\bb\in\pr_I\rel$ with $\pr_I\ba\le\bb$, 
there is $\bb'\in\rel$ such that $\ba\le\bb'$ and $\pr_I\bb'=\bb$.\\[1mm]
(2) For any $\ba\in\rel$, $\bb\in\pr_I\rel$ such that $(\pr_I\ba)\bb$  is a thin  
majority edge there is $\bb'\in\rel$ such that $\ba\bb'$ is a 
thin majority edge, and $\pr_I\bb'=\bb$.\\[1mm]
(3) For any $\ba\in\rel$, $\bb\in\pr_I\rel$ such that $(\pr_I\ba)\bb$  is a 
thin affine edge there are $\ba',\bb'\in\rel$ such that $\ba\sqq\ba'$, $\ba'\bb'$ 
is a thin affine edge, and $\pr_I\ba'=\pr_I\ba$, $\pr_I\bb'=\bb$.\\[1mm]
(4) For any $\ba\in\rel$, and an s-path (as-path, asm-path) 
$\vc\bb k\in\pr_I\rel$ with $\pr_I\ba=\bb_1$, there is an
s-path (as-path,asm-path, respectively) $\vc{\bb'}\ell\in\rel$ 
such that $\pr_I\bb'_\ell=\bb_\ell$.\\[1mm]
(5) For any $\bb\in\max(\pr_I\rel)$ ($\bb\in\amax(\pr_I\rel)$, 
$\bb\in\umax(\pr_I\rel)$) there is $\bb'\in\max(\rel)$ 
($\bb'\in\amax(\rel)$, $\bb'\in\umax(\rel)$, respectively), such that 
$\pr_I\bb'=\bb$. In particular, $\pr_{[n]-I}\bb'\in\max(\pr_{[n]-I}\rel)$ 
($\pr_{[n]-I}\bb'\in\amax(\pr_{[n]-I}\rel)$, 
$\pr_{[n]-I}\bb'\in\umax(\pr_{[n]-I}\rel)$, respectively).
\end{lemma}

\begin{proof}
Items (1) and (3) are proved in \cite{Bulatov16:restricted}, and items
of (4) and (5) are only proved for
s- and as-paths, and, respectively, for maximal and as-maximal elements. 
Items (4) and (5) for asm-paths and u-maximal elements follow 
from~(1)--(3).

(2) Observe that it suffices to consider binary relations $\rel$. Indeed, $\rel$ 
can be viewed as a subdirect product of $\pr_I\rel\tm\pr_{[n]-I}\rel$. So, 
suppose $n=2$ and $I=\{1\}$. We have $\ba=(a_1,a_2)$ and $\bb=b_1$. 
Let $\th$ be a maximal congruence of $\Sg{a_1,b_1}$ witnessing that 
$a_1b_1$ is an majority edge. Choose $\bb''=(b'_1,b_2)\in\rel$ 
such that $b'_1\in b_1^\th$ and $\rel'=\Sg{\ba,\bb''}$ is minimal possible 
with this condition. It suffices to prove the lemma for $\rel'$, since 
$b_1\in\Sg{a_1,b'_1}$ and so $b_1\in\pr_1\rel'$, and $a_1b_1$ is still a 
thin majority edge. This means that $(b'_1,b_2)$ can be chosen such that 
$b'_1=b_1$. Also, by taking 
$\cl{b_1}{b'_2}=g\left(\cl{a_1}{a_2},\cl{b_1}{b_2},\cl{b_1}{b_2}\right)$ 
we may assume by Theorem~\ref{the:uniform} that $g(a_2,b_2,b_2)=b_2$.
As is easily seen, the pair $(a_1,a_2)(b_1,b_2)$ is a majority edge as 
witnessed by congruence $\th'=\th\tm\zo_{\zA'_2}$ where 
$\zA'_2=\Sg{a_2,b_2}$. By the choice of $\bb''$ the pair $(b_1,b_2)$ belongs to 
$\Sg{(a_1,a_2),(c_1,c_2)}$ for any $(c_1,c_2)\in(b_1,b_2)^{\th'}$, and it 
only remains to prove condition (d) of the definition of thin majority edges.

Let $g'$ be the operation from condition (d) for $a_1b_1$. Then
$$
g'\left(\cll{(a_1,a_2)}{(b_1,b_2)}{(b_1,b_2)},
\cll{(b_1,b_2)}{(a_1,a_2)}{(b_1,b_2)},
\cll{(b_1,b_2)}{(b_1,b_2)}{(a_1,a_2)}\right)=
\cll{(b_1,b'_2)}{(b_1,b''_2)}{(b_1,b'''_2)}.
$$
Since $(b_1,b_2)\in\Sg{(a_1,a_2),(b_1,b'_2)}$ by the choice of $(b_1,b_2)$, 
there is a term operation $r_1$ such that
$$
r_1\left(\cll{(a_1,a_2)}{(b_1,b_2)}{(b_1,b_2)},
\cll{(b_1,b'_2)}{(b_1,b''_2)}{(b_1,b'''_2)}\right)
=\cll{(b_1,b_2)}{(b_1,b^*_2)}{(b_1,b^{**}_2)}.
$$
Repeating this for the second and third coordinate positions by finding 
$r_2,r_3$ with
\begin{eqnarray*}
r_2\left(\cll{(b_1,b_2)}{(a_1,a_2)}{(b_1,b_2)},
\cll{(b_1,b_2)}{(b_1,b^*_2)}{(b_1,b^{**}_2)}\right)
&=&\cll{(b_1,b_2)}{(b_1,b_2)}{(b_1,b^\dagger_2)},\\  
r_3\left(\cll{(b_1,b_2)}{(b_1,b_2)}{(a_1,a_2)},
\cll{(b_1,b_2)}{(b_1,b_2)}{(b_1,b^\dagger_2)}\right)
&=& \cll{(b_1,b_2)}{(b_1,b_2)}{(b_1,b_2)}, 
\end{eqnarray*}
we obtain a ternary operation $g''$ such that
\begin{eqnarray*}
g''\left(\cl{a_1}{a_2},\cl{b_1}{b_2},\cl{b_1}{b_2}\right) &=&
g''\left(\cl{b_1}{b_2},\cl{a_1}{a_2},\cl{b_1}{b_2}\right)\\
&=& g''\left(\cl{b_1}{b_2},\cl{b_1}{b_2}\cl{a_1}{a_2}\right)=\cl{b_1}{b_2},
\end{eqnarray*}
confirming property (d).
\end{proof}

The following lemma considers a special case of as-components in subdirect 
products, and is straightforward.

\begin{lemma}\label{lem:as-product}
Let $\rel$ be a subdirect product of $\zA_1\tm\zA_2$, $B,C$ as components of 
$\zA_1,\zA_2$, respectively, and $B\tm C\sse\rel$. Then $B\tm C$ is an
as-component of $\rel$.
\end{lemma}

We complete this section with an auxiliary statement that will be needed
later.

\begin{lemma}\label{lem:as-type-2}
Let $\al\prec\beta$, $\al,\beta\in\Con(\zA)$, let $B$ be a $\beta$-block and 
$\typ(\al,\beta)=\two$. Then $B\fac\al$ is term equivalent to a module. 
In particular,
every pair of elements of $B\fac\al$ is a thin affine edge in $\zA\fac\al$.
\end{lemma}

\begin{proof}
As $\zA$ is an idempotent algebra that generates a variety omitting type \one, 
and $(\alpha,\beta)$ is a simple interval in $\Con(\zA)$ of type \two, by 
Theorem~7.11 of \cite{Hobby88:structure}  there is a term operation of 
$\zA$ that is Mal'tsev on $B\fac\al$. Since $\beta$ is Abelian on $B\fac\al$, 
we get the result.
\end{proof}

\subsection{Quasi-decomposition and quasi-majority}

We make use of the property of 
quasi-2-decomposability proved in\index{quasi-2-decomposability}
\cite{Bulatov16:restricted}.

\begin{theorem}[The 2-Decomposition Theorem, \cite{Bulatov16:restricted}]%
\label{the:quasi-2-decomp}
If $\rel$ is an $n$-ary relation, $X\sse[n]$, tuple $\ba$ is such that
$\pr_J\ba\in\pr_J\rel$ for any $J\sse[n]$, $|J|=2$, and $\pr_X\ba
\in\amax(\pr_X\rel)$, there is a tuple $\bb\in\rel$ with
$\pr_J\bb\in\Filt^{as}_{\pr_J\rel}(\pr_J\ba)$ for any $J\sse[n]$, $|J|=2$, and
$\pr_X\bb=\pr_X\ba$. 
\end{theorem}

One useful implication of the 2-Decomposition Theorem~\ref{the:quasi-2-decomp} is the 
existence of term operation resembling a majority function. We state this 
theorem for finite classes of algebras rather than a single algebra, because
it concerns as-components that in subalgebras of products may have 
complicated structure.

\begin{theorem}\label{the:pseudo-majority}
Let $\cA$ be a finite class of finite similar sm-smooth algebras
omitting type~\one.
There is a term operation $\maj$\label{not:maj} of $\cA$ such that for any $\zA\in\cA$ and any 
$a,b\in\zA$, $\maj(a,a,b),\maj(a,b,a),\maj(b,a,a)\in\Filt_\zA^{as}(a)$.

In particular, if $a$ is as-maximal, then $\maj(a,a,b),\maj(a,b,a),\maj(b,a,a)$ 
belong to the as-component of $\zA$ containing $a$.
\end{theorem}

\begin{proof}
Let $\{a_1,b_1\}\zd\{a_n,b_n\}$ be a list of all pairs of elements from 
algebras of $\cA$,
let $a_i,b_i\in\zA_i$. Define relation $\rel$ to be a subdirect product of 
$\zA_1^3\tms\zA_n^3$ generated by $\ba_1,\ba_2,\ba_3$, where for every 
$i\in[n]$, $\pr_{3i-2,3i-1,3i}\ba_1=(a_i,a_i,b_i)$, 
$\pr_{3i-2,3i-1,3i}\ba_2=(a_i,b_i,a_i)$, $\pr_{3i-2,3i-1,3i}\ba_3=(b_i,a_i,a_i)$.
In other words the triples $(\ba_1[3i-2],\ba_2[3i-2],\ba_3[3i-2])$,
$(\ba_1[3i-1],\ba_2[3i-1],\ba_3[3i-1])$, $(\ba_1[3i],\ba_2[3i],\ba_3[3i])$ 
have the form $(a_i,a_i,b_i),(a_i,b_i,a_i),(b_i,a_i,a_i)$, respectively. Therefore 
it suffices to show that $\rel$ contains a tuple $\bb$ such that 
$a_i\sqq_{as}\bb[j]$, where $j\in\{3i,3i-1,3i-2\}$. However, since 
$(a_{i_1},a_{i_2})\in\pr_{j_1j_2}\rel$ for any $i_1,i_2\in[n]$ and
$j_1\in\{3i_1,3i_1-1,3i_1-2\}$, $j_2\in\{3i_2,3i_2-1,3i_2-2\}$, this 
follows from the 2-Decomposition Theorem~\ref{the:quasi-2-decomp}.
\end{proof}

A function $\maj$ satisfying the properties from 
Theorem~\ref{the:pseudo-majority} will be called a 
\emph{quasi-majority function}\index{quasi-majority function}.

\subsection{Rectangularity}\label{sec:rectangularity}

Let $\rel$ be a subdirect product of $\zA_1,\zA_2$. By 
$\rel[c], \rel^{-1}[c']$\label{not:rel-of-a} for $c\in\zA_1,c'\in\zA_2$ we denote 
the sets $\{b\mid (c,b)\in\rel\}, \{a\mid (a,c')\in\rel\}$, respectively, and for 
$C\sse\zA_1,C'\sse\zA_2$ we use $\rel[C]=\bigcup_{c\in C}\rel[c]$, 
$\rel^{-1}[C']=\bigcup_{c'\in C'}\rel^{-1}[c']$, respectively. Binary relations 
$\tol_1,\tol_2$ on $\zA_1,\zA_2$ given by 
$\tol_1(\rel)=\{(a,b)\mid \rel[a]\cap\rel[b]\ne\eps\}$ and
$\tol_2(\rel)=\{(a,b)\mid \rel^{-1}[a]\cap\rel^{-1}[b]\ne\eps\}$\label{not:tol}, 
respectively, are called \emph{link tolerances}\index{link tolerances} of $\rel$. 
They are tolerances of $\zA_1$, $\zA_2$, respectively, that is invariant reflexive 
and symmetric relations. The transitive closures $\lnk_1,\lnk_2$\label{not:lnk} 
of $\tol_1(\rel),\tol_2(\rel)$ are called 
\emph{link congruences}\index{link congruences}, and they are, indeed,
congruences. Relation $\rel$ is said to be \emph{linked}\index{linked relation} 
if the link congruences are full congruences.

\begin{lemma}[\cite{Bulatov16:restricted}]\label{lem:as-rectangularity} 
Let $\rel$ be a subalgebra of $\zA_1\tm\zA_2$ and let $a\in\zA_1$ and 
$B=\rel[a]$. For any $b\in\zA_1$ such that $ab$ is thin edge, and any 
$c\in\rel[b]\cap B$, $\Filt^{as}_B(c)\sse\rel[b]$.
\end{lemma}

\begin{proof}
The case when $a\le b$ or $ab$ is affine is considered in 
\cite{Bulatov16:restricted}, so suppose that $ab$ is majority.
Let $D=\Filt^{as}_B(c)\cap\rel[b]$. Set $D$
is nonempty, as $c\in D$. If $D\ne \Filt^{as}_B(c)$, there are $b_1\in D$ 
and $b_2\in\Filt^{as}_B(c)-D$
such that $b_1b_2$ is a thin edge. By Lemma~\ref{lem:op-s-on-affine}(3) 
there is a term operation $p$ such that $p(a,b)=b$ and $p(b_2,b_1)=b_2$. 
Then $\cl b{b_2}=p\left(\cl a{b_2},\cl b{b_1}\right)\in\rel$.
The result follows.
\end{proof}

\begin{prop}[\cite{Bulatov16:restricted}]\label{pro:max-gen} 
Let $\rel\le\zA_1\tm \zA_2$ be a linked subdirect product and let $B_1,B_2$ be 
as-components of $\zA_1,\zA_2$, respectively, such that 
$\rel\cap(B_1\tm B_2)\ne\eps$. Then $B_1\tm B_2\sse\rel$.
\end{prop}

\begin{corollary}[The Rectangularity Corollary]\label{cor:linkage-rectangularity}
Let $\rel$ be a subdirect product of $\zA_1$ and $\zA_2$, $\lnk_1,\lnk_2$ the link 
congruences, and let $B_1,B_2$ be as-components of a $\lnk_1$-block and a 
$\lnk_2$-block, respectively, such that $\rel\cap(B_1\tm B_2)\ne\eps$. Then 
$B_1\tm B_2\sse\rel$.
\end{corollary}

\begin{prop}\label{pro:umax-rectangular}
Let $\rel$ be a subdirect product of $\zA_1$ and $\zA_2$, $\lnk_1,\lnk_2$ the link 
congruences, and let $B_1$ be an as-component of a $\lnk_1$-block and 
$B'_2=\rel[B_1]$; let $B_2=\umax(B'_2)$. Then $B_1\tm B_2\sse\rel$.
\end{prop}

\begin{proof}
Let $B'_2$ be a subset of a $\lnk_2$-block $C$. By the Maximality 
Lemma~\ref{lem:to-max}(5) 
$B'_2$ contains an as-maximal element $a$ of $C$. By the Rectangularity 
Corollary~\ref{cor:linkage-rectangularity} $B_1\tm\{a\}\sse\rel$. It then
suffices to show that $B_1\tm\Filt^{asm}_{B'_2}(a)\sse\rel$.

Suppose for $D\sse\Filt^{asm}_{B'_2}(a)$ it holds $B_1\tm D\sse\rel$. 
If $D\ne \Filt^{asm}_{B'_2}(a)$, there are $b_1\in D$ and 
$b_2\in\Filt^{asm}_{B'_2}(a)-D$
such that $b_1b_2$ is a thin edge. By Lemma~\ref{lem:as-rectangularity}
$B_1\tm\{b_2\}\sse\rel$; the result follows.
\end{proof}

We complete this section with a technical lemma that will be useful later.

\begin{lemma}\label{lem:as-square-congruence}
Let $\zA$ be an algebra and $C$ its as-component such that $\zA=\Sg{C}$,
let $\rel=\zA\tm\zA=\Sg{C\tm C}$, and let $\beta$ be a congruence 
of $\rel$. Then for some $a,b\in C$, $a\ne b$ 
the pair $(a,b)$ is as-maximal in a $\beta$-block.
\end{lemma}

\begin{proof}

We start with a general claim.

\medskip

{\sc Claim.} If $\beta,\gm\in\Con(\rel)$ are such that 
$\beta\join\gm=\zo_\rel$,  then 
$\beta\red{C^2}\circ\gm\red{C^2}=\gm\red{C^2}\circ\beta\red{C^2}
=C^2\tm C^2$. 

\smallskip

Let $\rel_1\sse\rel\fac\beta\tm\rel$, $\rel_2\sse\rel\fac\gm\tm\rel$
be given by
$$
\rel_1=\{(a^\beta,a)\mid a\in\rel\},\qquad 
\rel_2=\{(a^\gm,a)\mid a\in\rel\}.
$$ 
Consider a subdirect product of $\rel\fac\beta\tm\rel\fac\gm$ defined 
as follows
$$
\relo(x,y)=\exists z\in\rel\ \ \rel_1(x,z)\meet\rel_2(y,z).
$$
As is easily seen, for a $\beta$-block $B$ and a $\gm$-block $D$,
$(B,D)\in\relo$ if and only if $B\cap D\ne\eps$. As 
$\beta\join\gm=\zo_\rel$, relation $\relo$ is linked. Since 
$C^2\fac\beta$ is an as-component of $\rel\fac\beta$ and $C^2\fac\gm$
is an as-component of $\rel\fac\gm$, Proposition~\ref{pro:max-gen}
implies that $C^2\fac\beta\tm C^2\fac\gm\sse\relo$. Therefore 
for any $\beta$- and $\gm$-blocks $B,D$ such that $B\cap C^2\ne\eps$,
$D\cap B^2\ne\eps$ we have $B\cap D\ne\eps$. Using as-connectivity and
Lemma~\ref{lem:thin} it can also be inferred that $B\cap D\cap C^2\ne\eps$.
The result follows.

\medskip

Let $\al$ be a maximal congruence of $\zA$ and $\gm_1=\al\tm\zo_\zA$,
$\gm_2=\zo_\zA\tm\al$. As is easily seen, $\gm_1,\gm_2$ are 
maximal congruences of $\rel$. There are two cases.

\medskip

{\sc Case 1.} $\beta\join(\gm_1\meet\gm_2)=\zo_\rel$.

\smallskip

By the Claim for any
$\al$-blocks $B_1,B_2$ such that $B_1\cap C,B_2\cap C\ne\eps$ we also
have $B\cap(B_1\tm B_2)\cap C^2\ne\eps$. Also, using 
Lemma~\ref{lem:thin} $B\cap(B_1\tm B_2)\cap C^2$ contains a pair 
as-maximal in $B$. Choosing $B_1\ne B_2$ we get the result.

\medskip

{\sc Case 2.} $\beta\join(\gm_1\meet\gm_2)\ne\zo_\rel$.

\smallskip

In this case consider $\zA'=\zA\fac\al$, $\rel'=\rel\fac{\al\tm\al}$, 
$\beta'=\beta\fac\al$; note that $\zA'$ is a simple idempotent algebra, and as
$\rel=\zA\tm\zA$, we have $\rel'=\zA'\tm\zA'$.
By \cite{Kearnes96:idempotent} either $\zA'$ has an \emph{absorbing 
element} $a$, that is, $f(\vc ak)=a$ for any term operation $f$ of $\zA'$,
whenever for some essential variable $x_i$ of $f$, $a_i=a$, or $\zA'$ is a
module, or the only nontrivial congruences of $\zA'$ are 
$\gm'_1=\gm_1\fac{\al\tm\al}$, $\gm'_2=\gm_2\fac{\al\tm\al}$. 
Since $C$ is a nontrivial as-component, the first option is impossible. 
If $\zA'$ is a simple module, the only congruence that is different from
$\gm'_1,\gm'_2$ is the skew congruence with 
$\Dl=\{(a,a)\mid a\in\zA'\}$ as a congruence block. If $\beta$ is the 
skew congruence, let 
$D$ be a $\beta\join(\gm_1\meet\gm_2)$-block different from $\Dl$
and $B\sse D$ a $\beta$-block. Then for any as-maximal pair 
$(a,b)\in B$ we have $a\ne b$. 

So, suppose $\beta\le\gm_1$. If $\beta\le\gm_1\meet\gm_2$, 
choose a $\gm_1\meet\gm_2$-block $B_1\tm B_2$ such that $B_1\ne B_2$;
clearly $B_1,B_2$ are $\al$-blocks. Then for any $\beta$-block 
$D\sse B_1\tm B_2$ and as-maximal element $(a,b)\in D$ we have 
$a\ne b$ as required.

Finally, suppose $\beta\not\le\gm_2$, then $\beta\join\gm_2=\zo_\rel$.
Take a $\beta$-block $B$, $B\cap C^2\ne\eps$, we have $B\sse B_1\tm\zA$ 
for some $\al$-block $B_1$. Moreover, by the Claim for any $\al$-block 
$B_2$ with $B_2\cap C\ne\eps$ there is $(a,b)\in B\cap C^2$ such 
that $b\in B_2$.
By Lemma~\ref{lem:thin} there is an as-component $E$ of $B$ such that
$(a,b)$ above can be chosen from $E$. Choosing $B_2\ne B_1$ we obtain
a required pair.
\end{proof}

\section{Separating congruences}\label{sec:polynomials-maximality}

In this section we introduce and study the relationship between prime intervals
in the congruence lattice of an algebra, or in congruence lattices of factors in a
subdirect products. It was first introduced in \cite{Bulatov02:maltsev-3-element}
and used in the CSP research in~\cite{Bulatov17:semilattice}.

\subsection{Special polynomials, mapping pairs}\label{special-polynomials}

We start with several technical results. They demonstrate the connection between 
minimal sets of an algebra $\zA$ and the structure of its graph $\cG'(\zA)$.
Let $\zA$ be an algebra and let 
$\relo_{ab}^\zA$\label{not:Qab}, $a,b\in\zA$, denote the subdirect 
product of $\zA^2$ generated by $\{(x,x)\mid x\in\zA\}\cup\{(a,b)\}$.

\begin{lemma}\label{lem:Qab-tolerance}
(1) $\relo^\zA_{ab}=\{(f(a),f(b))\mid f\in\Polo(\zA)\}$.\\
(2) For any $f\in\Polo(\zA)$, $(f(a),f(b))\in\tol_1(\relo^\zA_{ab})$. 
In particular, $\lnk_1(\relo^\zA_{ab})=\Cg{a,b}$; denote this congruence 
by $\al$.\\
(3) $\relo^\zA_{ab}\sse\Cg{a,b}$.\\
(4) Let $B_1,B_2$ be $\al$-blocks, and $C_1,C_2$ as-components of 
$B_1,B_2$, respectively, such that $f(a)\in C_1$ and $f(b)\in C_2$ for a 
polynomial $f$. Then $C_1\tm C_2\sse\relo_{ab}^\zA$.
\end{lemma}

\begin{proof}
(1) follows directly from the definitions.

(2) Take $f\in\Polo(\zA)$ and let $f(x)=g(x,\vc ak)$ for a term operation $g$ of 
$\zA$. Then $\cl{f(a)}{f(b)}=
g\left(\cl ab,\cl{a_1}{a_1}\zd\cl{a_k}{a_k}\right)\in\rel$ 
and $\cl{f(b)}{f(b)}=g\left(\cl bb,\cl{a_1}{a_1}\zd\cl{a_k}{a_k}\right)\in\rel$. 
Thus $(f(a),f(b))\in\tol_1(\relo^\zA_{ab})$.

(3) follows from (1), and 
(4) follows from (2),(3), and the Rectangularity Corollary~\ref{cor:linkage-rectangularity}.
\end{proof}

We say that $a$ is \emph{$\al$-maximal}\index{$\al$-maximal element} 
for a congruence $\al\in\Con(\zA)$ if $a$ is as-ma\-xi\-mal in the 
subalgebra $a^\al$.

\begin{corollary}\label{cor:mapping-pairs}
Let $\al\in\Con(\zA)$ and $\zz\prec\al$. Then for any $a,b\in\zA$ with 
$a\eqc\al b$ and any $\al$-maximal $c,d\in\zA$, $c\ne d$, with 
$c\eqc\al d$, belonging to the same
as-component of $c^\al$, there is $f\in\Polo(\zA)$ 
such that $c=f(a)$, $d=f(b)$.
\end{corollary}

\begin{proof}
The result follows from Lemma~\ref{lem:Qab-tolerance}(4).
\end{proof}

Recall that for $\al,\beta\in\Con(\zA)$ with $\al\prec\beta$ a pair $\{a,b\}$
is called an \emph{$(\al,\beta)$-subtrace}\index{subtrace} if 
$(a,b)\in\beta-\al$ and $a,b\in U$ for some $(\al,\beta)$-minimal set $U$.

\begin{corollary}\label{cor:max-min-set}
Let $\al\in\Con(\zA)$ and $\zz\prec\al$, and let $c,d\in\zA$, 
$c\eqc\al d$, be $\al$-maximal. \\[2mm]
(1) If $c,d$ belong to the same as-component of $c^\al$, 
then $\{c,d\}$ is a $(\zz,\al)$-subtrace.\\[1mm]
(2) If there is a $(\zz,\al)$-subtrace $\{c',d'\}$ such that $c'\in\as(c)$
and $d'\in\as(d)$ then $\{c,d\}$ is a $(\zz,\al)$-subtrace as well.
\end{corollary}

\begin{proof}
(1) Take any $(\zz,\al)$-minimal set $U$, and 
$a,b\in U$ with $a\eqc\al b$. By Corollary~\ref{cor:mapping-pairs} there is 
$f\in\Polo(\zA)$ with $c=f(a)$, $d=f(b)$. By Lemma~\ref{lem:minimal-sets}(3)
$U'=f(U)$ is a $(\zz,\al)$-minimal set.

(2) As in (1) one can argue that $(c',d')\in\relo^\zA_{ab}$, that is, 
$\relo^\zA_{ab}\cap(\as(c)\tm\as(d))\ne\eps$. We then complete by 
Lemma~\ref{lem:Qab-tolerance}(4).
\end{proof}

\begin{lemma}\label{lem:type23}
For any $\al\in\Con(\zA)$ with $\zz\prec\al$ such that $|D|>1$ for some 
as-component $D$ of an $\al$-block, the prime interval $\zz\prec\al$ has type 
\two\ or \three.
\end{lemma}

\begin{proof}
Let $a,b\in D$ for an as-component $D$ of an $\al$-block. Then by 
Corollary~\ref{cor:mapping-pairs} there is a polynomial $f$ such that 
$f(a)=b$ and 
$f(b)=a$. Also, $a,b$ belong to some $(\zz,\al)$-minimal set. This rules out types 
\four\ and \five. Since $\zA$ omits type $\one$, this only leaves types \two\ and \three.
\end{proof}

\begin{lemma}\label{lem:type2-condition}
Let $\al\in\Con(\zA)$ with $\zz\prec\al$ be such that some $\al$-block contains 
a semilattice or majority edge. Then the prime interval $(\zz,\al)$ has type 
\three, \four\, or \five.
\end{lemma}

\begin{proof}
We need to show that $(\zz,\al)$ does not have type \two. Let $B$ the 
$\al$-block containing a semilattice or majority edge. Then $B$ contains a 
non-Abelian subalgebra, which implies $(\zz,\al)$ is also non-Abelian. 
\end{proof}

\subsection{Separation}\label{sec:separation}

Let $\zA$ be an algebra, and let $\al\prec\beta$ and $\gm\prec\dl$ be prime
intervals in $\Con(\zA)$. We say that $\al\prec\beta$ can be 
\emph{separated}\index{separated intervals} from 
$\gm\prec\dl$  if there is a unary polynomial $f\in\Polo(\zA)$ such that 
$f(\beta)\not\sse\al$, but $f(\dl)\sse\gm$. The polynomial $f$ in this case is 
said to \emph{separate}\index{separating polynomial}  
$\al\prec\beta$ from $\gm\prec\dl$. 

Since we often consider relations rather than single algebras, we also introduce 
separability in a slightly different way. Let $\rel$ be a subdirect product of 
$\zA_1\tm\dots\tm\zA_n$. 
Let $i,j\in[n]$ and let
$\al_i\prec\beta_i$, $\al_j\prec\beta_j$ be prime intervals in $\Con(\zA_i)$ and 
$\Con(\zA_j)$, respectively. Interval $\al_i\prec\beta_i$ can be separated from 
$\al_j\prec\beta_j$ 
if there is a unary polynomial $f$ of $\rel$ such that $f(\beta_i)\not\sse\al_i$ but 
$f(\beta_j)\sse\al_j$. Similarly, the polynomial $f$ in this case is said to 
\emph{separate} $\al_i\prec\beta_i$ from $\al_j\prec\beta_j$

First, we observe a connection between separation in a single algebra and in 
relations.

\begin{lemma}\label{lem:separation-separation}
Let $\rel$ be the binary equality relation on $\zA$. Let $\al_1=\al,\beta_1=\beta$
be viewed as congruences of the first factor of $\rel$, and 
$\al_2=\gm,\beta_2=\dl$ as congruences of the second factor of $\rel$. 
Prime interval $\al\prec\beta$ can be separated from $\gm\prec\dl$ as 
intervals in $\Con(\zA)$ if and only if $\al_1\prec\beta_1$ can be separated 
from $\al_2\prec\beta_2$ in $\rel$.
\end{lemma}

\begin{proof}
Note that for any polynomial $f$ its action on the first and second projections of 
$\rel$ is the same polynomial of $\zA$. Therefore $\al\prec\beta$ can be 
separated from $\gm\prec\dl$ in $\Con(\zA)$ if and only if, there is 
$f\in\Polo(\zA)$, $f(\beta)\not\sse\al$ while $f(\dl)\sse\gm$. This condition 
can be expressed as follows: there is $f\in\Polo(\rel)$, 
$f(\beta_1)\not\sse\al_1$ while $f(\beta_2)\sse\al_2$, which precisely 
means that $\al_1\prec\beta_1$ can be separated from $\al_2\prec\beta_2$ 
in $\rel$.
\end{proof}

In what follows when proving results about separation we will always assume 
that we deal with a relation --- a subdirect product ---  and that the prime 
intervals in question are from congruence lattices of different factors of the 
subdirect product. If this is not the case, one
can duplicate the factor containing the prime intervals and apply 
Lemma~\ref{lem:separation-separation}.

Let $\rel$ be a subdirect product of $\zA_1\tm\dots\tm\zA_n$, $I\sse[n]$,
and let $f$ be a polynomial of $\pr_I\rel$, that is, there are a term operation 
$g$ of $\rel$ and $\vc\ba k\in\pr_I\rel$ such that 
$f(\vc x\ell)=g(\vc x\ell,\vc\ba k)$. The tuples $\ba_i$ can be extended
to tuples $\ba'_i\in\rel$. Then the polynomial of $\rel$ given by
$f(\vc x\ell)=g(\vc x\ell,\vc{\ba'}k)$ is said to be an 
\emph{extension}\index{extension of a polynomial} of $f$
to a polynomial of $\rel$.

\begin{lemma}\label{lem:min-set-separation}
Let $\rel$ be a subdirect product of $\zA_1\tm\dots\tm\zA_n$, $i,j\in[n]$, and 
$\al_i\prec\beta_i$, $\al_j\prec\beta_j$ for $\al_i,\beta_i\in\Con(\zA_i)$, 
$\al_j,\beta_j\in\Con(\zA_j)$. Let also a unary polynomial $f$ of $\rel$ separate 
$\al_i\prec\beta_i$ from $\al_j\prec\beta_j$. Then $f$ can be 
chosen idempotent and such that $f(\zA_i)$ is a $(\al_i,\beta_i)$-minimal set.
\end{lemma}

\begin{proof}
Let $g$ be a polynomial separating $(\al_i,\beta_i)$ from $(\al_j,\beta_j)$. 
Since $g(\beta_i)\not\sse\al_i$, by Lemma~\ref{lem:minimal-sets}(6) there 
is an $(\al_i,\beta_i)$-minimal set $U$ such that $g(\beta_i\red U)\not\sse\al_i$.
Let $V=g(U)$, by Lemma~\ref{lem:minimal-sets}(2) $V$ is a 
$(\al_i,\beta_i)$-minimal set. Let $h$ be a unary 
polynomial such that $h$ maps $V$ onto $U$ and $h\circ g\red U$ is the identity
mapping. Let also $h'$ be an extension of $h$ to a polynomial of $\rel$. Then 
$h'\circ g$ separates $i$ from $j$. Now $f$ can be chosen to be an appropriate 
power of $h'\circ g$.
\end{proof}

For a subdirect product $\rel\sse\zA_1\tms\zA_n$ the relation `cannot be 
separated' on prime intervals of the $\zA_i$s  is clearly reflexive and
transitive. If the algebras $\zA_i$ are Mal'tsev, it is also symmetric (for 
partial results see \cite{Bulatov02:maltsev-3-element,Bulatov17:semilattice}). 
Moreover, it can be shown that it remains `almost' symmetric when the 
$\zA_i$s contain no majority edges. In the general case however the situation 
is more complicated. Next we introduce conditions that make the 
`cannot be separated' relation to some extent symmetric, at least in what 
concerns our algorithm, as it will be 
demonstrated in Lemma~\ref{lem:relative-symmetry}.

Let $\beta_i\in\Con(\zA_i)$, let $B_i$ be a $\beta_i$-block for $i\in[n]$,
and let $B'_i=\pr_i(\rel\cap\ov B)$. 
Let also $\cU$ be a set of unary polynomials of $\rel$, $i$ an element 
from $[n]$, and $\al,\beta\in\Con(\zA_i)$ with $\al\prec\beta\le\beta_i$. 
Let $T_{\zA_i}(a',b';\al,\beta,\cU)\sse\beta\fac\al\sse(\zA_i\fac\al)^2$ 
for $a',b'\in B'_i\fac\al$, $(a',b')\in\beta-\al$, denote the set of pairs 
$(a,b)\in\beta\fac\al$ such that there is a polynomial $g\in\cU$ 
satisfying the following conditions: 
$g(\{a',b'\})=\{a,b\}$ and $g(\zA_i)$ is a $(\al,\beta)$-minimal set.
Note that these conditions imply that $\{a,b\}$ is a $(\al,\beta)$-subtrace or $a=b$. 
We say that $\al$ and $\beta$ are \emph{$\cU$-chained}\index{$\cU$-chained 
congruences} in $\rel$ with respect to $\ov\beta,\ov B$ if for 
any $a',b'\in B'_i\fac\al$ with $(a',b')\in\beta-\al$, the following conditions 
hold:\\[2mm]
(G1) For a $\beta\fac\al$-block $E$ such that $E\cap\umax(B'_i)\ne\eps$
let $E'=E\cap B'_i\fac\al$. Then $(a,b)\in T_{\zA_i}(a',b';\al,\beta,\cU)$ 
for any $a,b$ from the same as-component of $E'$.\\[1mm]
(G2) For any $\beta\fac\al$-block $E$ such that $E\cap\umax(B'_i)\ne\eps$, 
and any $a,b\in\umax(E')$, where $E'=E\cap B'_i\fac\al$, there is a sequence $a=a_1\zd a_k=b$ in $E'$ such that 
$\{a_i,a_{i+1}\}\in T_{\zA_i}(a',b';\al,\beta,\cU)$ for any $i\in[k-1]$.\\[2mm] 
Also,  if for elements $a,b$ and any $a',b'$ there is a sequence of elements
satisfying (G2), we say that $a$ and $b$ are 
\emph{subtrace connected}\index{subtrace connected};
congruences $\al,\beta,\ov\beta$, congruence classes $\ov B$, and 
set of polynomials $\cU$ will always be clear from the context in this case.
Observe that $\cU$-chaining amounts to saying that polynomials from $\cU$ 
do not allow any congruences of $\beta$-blocks viewed as subalgebras between 
$\al$ and $\beta$, at least where u-maximal elements are concerned. 

A unary polynomial $f$ is said to be 
\emph{$\ov B$-preserving}\index{$\ov B$-preserving} if $f(\ov B)\sse\ov B$. 
We call relation $\rel$ \emph{chained}\index{chained} with respect to 
$\ov\beta,\ov B$ if\\[2mm]
(Q1) for any $\al,\beta\in\Con(\zA_i)$, $i\in[n]$, such that 
$\al\prec\beta\le\beta_i$, congruences $\al$ and $\beta$ are $\cU_B$-chained in 
$\rel$, where $\cU_B$ is the set of all $\ov B$-preserving polynomials of 
$\rel$\\[1mm]
(Q2) for any $\al,\beta\in\Con(\zA_i)$, $\gm,\dl\in\Con(\zA_j)$, $i,j\in[n]$, such 
that $\al\prec\beta\le\beta_i$, $\gm\prec\dl\le\beta_j$, and $(\al,\beta)$ can be
separated from $(\gm,\dl)$, congruences $\al$ and $\beta$ are 
$\cU^*$-chained in $\rel$, where $\cU^*$ is the set of all $\ov B$-preserving 
polynomials of $\rel$ such that $g(\dl)\sse\gm$.\\[2mm]
Polynomials from $\cU^*$ in condition (Q2) will be called 
\emph{$(\gm,\dl,\ov B)$-good}\index{$(\gm,\dl,\ov B)$-good polynomial}.

\begin{lemma}\label{lem:good-polys}
(1) Any constant polynomial from $\ov B\cap\rel$ is 
$(\gm,\dl,\ov B)$-good.\\[1mm]
(2) If $f$ is a $k$-ary term function of $\rel$ and $\vc gk$ are 
$(\gm,\dl,\ov B)$-good polynomials, then $f(g_1(x)\zd g_k(x))$ is 
$(\gm,\dl,\ov B)$-good.\\[1mm]
(3) Let $T(a',b')$ denote $T_{\zA_i}(a',b';\al,\beta,\cU)$ 
for $\cU\in\{\cU_{\ov B},\cU^*\}$. If 
$\{a,b\}\in T(a',b')$ then $T(a,b)\sse T(a',b')$.\\[1mm]
(4) If there is a $\beta\fac\al$-block $E$ such that 
$E\cap\umax(B'_i)\ne\eps$,   
let $E'=E\cap B'_i\fac\al$, and if $E'$ contains a nontrivial as-component, 
then there is a set $T\sse\beta\fac\al$ such that $T\sse T(a',b')$ for any 
$a',b'\in B'_i\fac\al$, $a'\eqc{\beta\fac\al} b'$ and $T$ satisfies conditions 
(G1),(G2) for $T(a',b')$.\\[1mm]
(5) Let $a',b'\in B'_i\fac\al$, $a'\eqc{\beta\fac\al} b'$ be such that 
$T(a',b')$ is minimal among sets of this form. Then for any 
$(a,b)\in T(a',b')$ there is $h\in\cU$ such that $h$ is idempotent and 
$h(a)\eqc\al a,h(b)\eqc\al b$.
\end{lemma}

\begin{proof}
Items (1),(2) are straightforward.

(3) Let $\{a'',b''\}\in T(a,b)$. Then there are polynomials $f,g\in\cU$ with
$\{a,b\}=f(\{a',b'\})$ and $\{a'',b''\}=g(\{a,b\})$. Then $g\circ f\in\cU$ by 
item (2) or definition and $g\circ f(\{a',b'\})=\{a'',b''\}$.

(4) Take $a,b\in C$ where $C$ is a nontrivial as-component in $E'$. 
By (G2) $\{a,b\}\in T(a',b')$ for any appropriate $a',b'$. Therefore by (3)
$T=T(a,b)\sse T(a',b')$.

(5) Let $\{a,b\}\in T(a',b')$. Then by (3) $T(a,b)\sse T(a',b')$, and therefore 
by the minimality of $T(a',b')$ we get $T(a,b)=T(a',b')$. The result 
follows by definition of $T(a',b')$.
\end{proof}

The next lemma shows how we will use the property of being chained.

\begin{lemma}\label{lem:link-separability}
Let $\rel$ be a subdirect product of $\vc\zA n$ chained 
with respect to $\ov\beta,\ov B$, where $\beta_i\in\Con(\zA_i)$
and $B_i$ is a $\beta_i$-block, and $\rel'=\rel\cap(B_1\tms B_n)$,
$B'_i=\pr_i\rel'$. Let also $\lnk$ be the link congruence 
of $B'_i$ with respect to $\pr_{ij}\rel'$ for some $i,j\in[n]$, and 
$\dl=\Cg\lnk$ the congruence of $\zA_i$ generated by 
$\lnk$. Then for any $\gm\in\Con(\zA_i)$ with $\gm\prec\dl$ 
it holds $(\dl\fac\gm)\red{\umax(E)}=(\lnk\fac\gm)\red{\umax(E)}$
for every $\dl\red{B'_i}$-block $E$ u-maximal in $B'_i\fac\dl$.
\end{lemma}

\begin{proof}
If $\gm\prec\dl$ then by the choice of $\dl$
there are $a,b\in B'_i$ with $(a,b)\in\dl-\gm$. Let $E$ be the intersection
of $B'_i$ with a $\dl$-block. We apply condition (Q1) to $\al=\gm\meet\beta_i$ 
and $\beta=\dl\meet\beta_i$. By condition (Q1) for any $a',b'\in\umax(E)$
there is a sequence $a'=a_1\zd a_k=b'$ such that 
$\{a_\ell,a_{\ell+1}\}=f_\ell(\{a,b\})$ for some $\ov B$-preserving polynomial
$f_\ell$ for each $\ell\in[k-1]$. This means that $(a_\ell,a_{\ell+1})\in\lnk$,
and so $(a',b')\in\lnk$.
\end{proof}

The following lemma establishes the weak symmetricity of separability
relation mentioned before.

\begin{lemma}\label{lem:relative-symmetry}
Let $\rel$ be a subdirect product of $\zA_1\tms\zA_n$, $\beta_i\in\Con(\zA_i)$, 
$B_i$ a $\beta_i$-block such that $\rel$ is chained with respect to 
$\ov\beta,\ov B$; $\rel'=\rel\cap\ov B$, 
$B'_i=\pr_i\rel'$. Let also $\al\prec\beta\le\beta_1$, $\gm\prec\dl=\beta_2$, 
where $\al,\beta\in\Con(\zA_1)$, $\gm,\dl\in\Con(\zA_2)$. 
If $B'_2\fac\gm$  has a 
nontrivial as-component $C'$ and $(\al,\beta)$ can be separated from $(\gm,\dl)$,
then there is a $\ov B$-preserving polynomial $g$ such that 
$g(\beta\red{B'_1})\sse\al$ and $g(\dl)\not\sse\gm$. Moreover, for any
$c,d\in C'$ polynomial $f$ can be chosen such that $f(c)=c,f(d)=d$.
\end{lemma}

\begin{proof}
As is easily seen, we can assume that  $\al,\gm$ are equality relations. 
We need to show that there is $g$ such that $g$ collapses $\beta$
but does not collapse $\beta_2=\dl$. 

First we show that there are $c,d\in B'_2$ such that for any 
$(a,b)\in\beta\red{B'_1}$ there is a polynomial $h^{ab}$ 
of $\rel$ such that 
\begin{itemize}
\item[(1)]
$h^{ab}$ is idempotent;
\item[(2)]
$h^{ab}(a)=h^{ab}(b)$;
\item[(3)]
$h^{ab}(c)=c$, $h^{ab}(d)=d$.
\end{itemize}

We consider two cases.
 
\medskip

{\sc Case 1.}
There is an element $c$ from a nontrivial as-component of $B'_2$ such that 
$(a,c,\bb)\in\rel'$ for some $a\in B'$, a $\beta$-block such that 
$B'\cap\umax(B'_1)\ne\eps$ and $|\umax(B'\cap B'_1)|>1$.

\medskip

First, we choose $d$ to be any element other than $c$ of the nontrivial 
as-component $C'$ of $B'_2$ containing $c$. 
Let $T_1$ be a minimal set of $(\al,\beta)$-subtraces 
as in Lemma~\ref{lem:good-polys}(4) for $\cU^*$, the set of 
$(\gm,\dl,\ov B)$-good polynomials. We start with the case when $(a,b)\in T_1$.
Even more specifically, as $c$ is as-maximal in $B'_2$, by the 
Maximality Lemma~\ref{lem:to-max}(5) $a$ can be chosen from 
$\umax(B'\cap B'_1)$. Take an $(\al,\beta)$-subtrace 
$\{a,b\}\in T_1$, such a subtrace exists by condition (Q1).


By $\relo^*\sse\zA_1^2\tm\zA_2^2\tm\rel$ we denote the relation 
generated by 
$\{(a,b,c,d,\ba)\}\ \cup\{(x,x,y,y,\bz)\mid \bz\in\rel,\bz[1]=x,\bz[2]=y\}$,
where $\ba$ is an arbitrary element from $\rel'$. Let 
$\relo=\pr_{1234}\relo^*$ and 
$\relo'=\pr_{1234}(\relo^*\cap(B'_1\tm B'_1\tm B'_2\tm B'_2\tm\ov B))$.
Observe that $\relo$ is exactly the set of quadruples $(f(a),f(b),f(c),f(d))$ 
for unary polynomials $f$ of $\rel$ and $\relo'$ is exactly the set of quadruples 
$(f(a),f(b),f(c),f(d))$ for $\ov B$-preserving unary polynomials $f$ of $\rel$.
We prove that $\relo'$ contains a quadruple of the form $(a',a',c,d)$; 
the result then follows.

Let also $\relo_1=\pr_{1,2}\relo=\relo^{\zA_1}_{ab}$,  
$\relo_2=\pr_{3,4}\relo=\relo^{\zA_2}_{cd}$; set 
$\relo'_1=\pr_{1,2}\relo'$, $\relo'_2=\pr_{3,4}\relo'$. Note that 
$\pr_1\relo'=\pr_2\relo'=B'_1$ and $\pr_3\relo'=\pr_4\relo'=B'_2$,
because $\pr_{12}\rel'\sse\pr_{13}\relo',\pr_{24}\relo'$.  
Let $\lnk_1,\lnk_2$
denote the link congruences of $\relo'$ viewed as a subdirect 
product of $\relo'_1$ and $\relo'_2$. Note that 
these congruences may be different from the link congruences of $\relo$ 
restricted to $\relo_1\cap(B'_1\tm B'_1)$, $\relo_2\cap(B'_2\tm B'_2)$, 
respectively. We show that $(a',a')$ for some $a'\in B'_1$ is as-maximal in a
$\lnk_1$-block, $(c,d)$ is as-maximal in a $\lnk_2$-block, and 
$\relo'\cap(\as(a',a')\tm\as(c,d))\ne\eps$. By the Rectangularity 
Corollary~\ref{cor:linkage-rectangularity} this implies the result.

\medskip

{\sc Claim 1.}
$(\al\tm\beta)\red{\relo'_1}\sse\lnk_1$ and 
$(\gm\tm\dl)\red{\relo'_2}\sse\lnk_2$.

\medskip

Relation $\relo'$ contains tuples $(a,b,c,d)$, $(a,b,c',c')$, $(a,a,c',c')$, $(a,a,c,c)$
for some $c'\in B'_2$. Indeed, $(a,b,c,d)\in\relo'$ by definition, 
$(a,a,c,c)\in\relo$ because $(a,c,\bb)\in\rel$, and $(a,b,c',c'),(a,a,c',c')$
can be chosen to be the images of $(a,b,c,d)$ and $(a,a,c,c)$, respectively,
under a $\ov B$-preserving polynomial $g^{ab}$ such that $g^{ab}(a)=a$, 
$g^{ab}(b)=b$ and $g^{ab}(\dl)\sse\gm$. Such a polynomial exists because 
$\rel$ is chained and because $(\al,\beta)$ can be separated from 
$(\gm,\dl)$. This implies that $(c,d)\eqc{\lnk_2}(c,c)$. Let $\eta_1,\eta_2$ be
congruences of $\relo_1,\relo_2$ generated by $((a,b),(a,a))$ and 
$((c,d),(c,c))$, respectively. Then 
$$
\eta_1\red{\relo'_1}=(\al\tm\beta)\red{\relo'_1},\quad
\text{and}\quad 
\eta_2\red{\relo'_2}=(\gm\tm\dl)\red{\relo'_2}.
$$
Indeed, in the case of, say, $\al\tm\beta$, relation $\relo'_1$ consists of pairs 
$(g(a),g(b))$ for a $\ov B$-preserving unary polynomial $g$ of $\zA_1$. Since 
$(a,b)\eqc{\al\tm\beta}(a,a)$, for any $(a',b')\in\relo'_1$ it holds that 
$$
(a',b')=(g(a),g(b))\eqc{\eta_1}(g(a),g(a))=(a',a').
$$
Since $(a,b),(a,a)$ are in the same $\lnk_1$-block, 
$(\al\tm\beta)\red{\relo'_1}\sse\lnk_1$. 

For $\relo'_2$ and $\gm\tm\dl$ the argument is similar.
Observing from the same tuples as 
before that $(c,d)\eqc{\lnk_2}(c,c)$, we prove 
$(\gm\tm\dl)\red{\relo'_2}\sse\lnk_2$ by a similar argument. 
Claim~1 is proved.

\medskip

{\sc Claim 2.}
Let $E=B'\cap B'_1$, where $B'$ is the $\beta$-block containing $a,b$.
Then $(\beta\tm\beta)\red{\umax(E)\tm\umax(E)}\sse\lnk_1$.

\medskip

By the assumption for any $(\al,\beta)$-subtrace $(a',b')\in T_1\cap E^2$ there is  
a $\ov B$-preserving polynomial $g^{a'b'}$ satisfying $g^{a'b'}(a')=a'$, 
$g^{a'b'}(b')=b'$, and $g^{a'b'}(B'_2)=\{c'\}\sse B'_2$.
Applying $g^{a'b'}$ to tuples $(a,b,c,d),(a,a,c,c)$, and $(b,b,d',d')$ for
any $d'$ such that $(b,d)\in\pr_{12}\rel'$, we obtain 
$(a',b',c',c'),(a',a',c',c'),(b',b',c',c')\in\relo'$. The second two tuples 
imply that $(a',a')\eqc{\lnk_1}(b',b')$, and therefore $(a'',a'')\eqc{\lnk_1}(b'',b'')$
for any $a'',b''\in\umax(E)$. Along with Claim~1 this proves the result.

\medskip

{\sc Claim 3.}
$(c,d)$ is as-maximal in a $\lnk_2$-block.

\medskip

If for some $e,e'\in B'_2$ we have $(e,e)\eqc{\lnk_2}(e',e')$, 
then, as $(e,e')$ generates $\dl$, for any $(\gm,\dl)$-subtrace 
$\{e'',e'''\}\in T_{\zA_2}(e,e')
=T_{\zA_2}(e,e';\gm,\dl,\cU_{\ov B})$ 
there is a $\ov B$-preserving polynomial $f'$ with 
$f'(\{e,e'\})=\{e'',e'''\}$. Applying this polynomial to 
the tuples witnessing that $(e,e)\eqc{\lnk_2}(e',e')$ we get 
$(e'',e'')\eqc{\lnk_2}(e''',e''')$. Therefore by condition (Q1) all tuples of the 
form $(x,x)$, 
$x\in\umax(B'_2)$, are $\lnk_2$-related. Since by condition (G1) $\{c,d\}$ is a 
$(\gm,\dl)$-subtrace from $T_{\zA_2}(c,d)\sse T_{\zA_2}(e,e')$, 
using Claim~1 this implies that $\lnk_2\red{\relo''}=(\dl\tm\dl)\red{\relo''}$, 
where $\relo''=\relo'_2\cap(\umax(B'_2)\tm \umax(B'_2))$. In particular, 
$C'\tm C'$, where $C'$ is the as-component of $B'_2$ containing $c,d$, is 
contained in $\relo'_2$, and is contained in a $\lnk_2$-block. All elements 
of $C'\tm C'$ are as-maximal in $\relo''$. 

Otherwise, since the inclusion $(\gm\tm\dl)\red{\relo'_2}\sse\lnk_2$ implies 
that  if $(c_1,d_1)\eqc{\lnk_2}(c_2,d_2)$ then 
$(c_1,c_1)\eqc{\lnk_2}(c_2,c_2)$, by Claim~1 
we have $\lnk_2\red{\relo''}=(\gm\tm\dl)\red{\relo''}$. 
In particular, $\{c\}\tm C'$ is contained in a $\lnk_2$-block. Since $c,d$ are 
as-maximal, $(c,d)$ is as-maximal in this $\lnk_2$-block. Claim~3 is proved.

\medskip

By the Maximality Lemma~\ref{lem:to-max}(5) there is an element 
$(a',b')$ as-maximal in a $\lnk_1$-block $D$ such that $(a',b',c,d)\in\relo'$. 
If $a'=b'$ then we are done. Otherwise by condition (G1) and 
Lemma~\ref{lem:good-polys}(3)
$\{a',b'\}$ is an $(\al,\beta)$-subtrace from $T_1$, also $(a',c)\in\rel$ because 
$\pr_{1,3}\relo\sse\rel$, and we can replace $a,b$ with $a',b'$. 
Observe that if we show the existence of a
polynomial $g$ such that $g(a')=g(b')$ and $g(c)=c$, $g(d)=d$, this will
witness the existence of $g'$ with $g'(a)=g'(b)$ and $g'(c)=c$, $g'(d)=d$. 
Let $E=a'^\beta\cap B'_1$. Note that by Claim~2 
$(\beta\tm\beta)\red{\umax(E)\tm\umax(E)}\sse\lnk_1$ and by the Rectangularity
Corollary~\ref{cor:linkage-rectangularity} $\umax(E)\tm\as(c,d)\sse\relo'$. 
Therefore, again $a',b'$ can be chosen as-maximal in $E$. We use $a,b$ for 
$a',b'$ from now on.

\medskip

{\sc Claim 4.}
$(a,a)$ is as-maximal in $\relo''_1=\relo'_1\cap(E\tm E)$.

\medskip

Let $\eta_1,\eta_2$ be the link congruences of 
$B'_1,B'_2$, respectively, with respect to $\relo'_1$; as 
$\relo'_1\sse\relo^{\zA_1}_{ab}$ we have $\eta_1,\eta_2\le\beta$. On the
other hand, since $\relo'_1$ consists of pairs of the form $(x,x)$ and 
$(\al,\beta)$-subtraces, and since $\umax(E)$ belongs to a block of the 
transitive closure of $T_1$, it is easy to see that $\umax(E)$ is a subset of 
both a $\eta_1$- and $\eta_2$-blocks. Indeed, let $e,e'\in\umax(E)$ and
$e=e_1\zd e_k=e'$ are such that $\{e_i,e_{i+1}\}\in T_1$. This means that 
either $(e_i,e_{i+1})\in\relo'_1$ or $(e_{i+1},e_i)\in\relo'_1$. Since
$(e_i,e_i),(e_{i+1},e_{i+1})\in\relo'_1$ by construction, in either case 
we have $(e_i,e_{i+1})\in\eta_1,\eta_2$.

Let $E'$ be the as-component of $E$ containing $a$; such an as-component 
exists by the choice of $a,b$. As $(a,a)\in\relo'_1\cap(E'\tm E')\ne\eps$, 
by the Rectangularity Corollary~\ref{cor:linkage-rectangularity}  
$E'\tm E'\sse\relo'_1$. Since $E'$ is an as-component in $E$, by
Lemma~\ref{lem:as-product} $E'\tm E'$ is an as-component in $\relo''_1$.
In particular $(a,a)$ is as-maximal in $\relo''_1$.
Claim~4 is proved.

\medskip

{\sc Claim 5.} $(a,a,c,d)\in\relo'$. 

\medskip

To prove this claim we find a subalgebra 
$\relo''$ of $\relo'$ such that it is linked enough and both $(a,a)$ and $(c,d)$
belong to as-components of $\pr_{12}\relo'',\pr_{34}\relo''$, respectively, and
then apply the Rectangularity Corollary~\ref{cor:linkage-rectangularity}. 

Let $F$ be the as-component of the 
$\lnk_2$-block containing $(c,d)$. By Claim~3 it is either $\{c\}\tm C'$
or $C'\tm C'$. Since $\relo''_1$ belongs to a $\lnk_1$-block and 
$(a,a,c,c)\in\relo'$, by the Maximality Lemma~\ref{lem:to-max}(4)
for any $(a',b')\in E'\tm E'$ there are $(c',d')\in F$ such that 
$(a',b',c',d')\in\relo'$. Now consider $\relo''=\relo'\cap(E\tm E\tm B'_2\tm B'_2)$.
Clearly, $\relo'_1\sse\pr_{12}\relo''$. Also, since $(a,b)$ is as-maximal in 
a $\lnk_1$-block, by the Rectangularity 
Corollary~\ref{cor:linkage-rectangularity} $\{(a,b)\}\tm F\sse\relo''$,
implying $F\sse\pr_{34}\relo''$. If $\th_1,\th_2$ denote the
link congruences of $\pr_{12}\relo'',\pr_{34}\relo''$ with respect to $\relo''$,
the observation above implies that the as-components of $\pr_{12}\relo''$
containing $(a,a)$ and $(a,b)$ belong to the same $\th_1$-block, and $F$
belongs to a $\th_2$-block. Therefore again by the Rectangularity 
Corollary~\ref{cor:linkage-rectangularity} we get 
$(E'\tm E')\tm(\{c\}\tm C')\sse\relo'$, in particular $(a,a,c,d)\in\relo$. 
Thus, there is a polynomial $h$ such that $h(a)=h(b)=a$ and $h(c)=c, h(d)=d$. 

\medskip

So far we have proved that for any subtrace $(a,b)\in T_1$, where $a$ is
a fixed element such that $(a,c)\in\pr_{12}\rel'$, there is a polynomial 
$h^{ab}$ satisfying the conditions stated in the beginning of the proof. 

\medskip

{\sc Claim 6.} For every $(\al,\beta)$-subtrace $\{a',b'\}$ from 
$T_1\cap(E\tm E)$ (recall that $E=a^\beta\cap B'_1$) there is a polynomial 
$h$ such that $h(a)=h(b)$ and $h(c)=c, h(d)=d$. 

Let us consider $T_1$ 
as a graph; we can introduce the distance $r(x)$ of element $x$ from $a$. 
In particular, all elements from $\umax(E)$ belong to the connected component 
containing $a$. Let $D(i)\sse E$ denote the set of elements at distance at 
most $i$ from $a$. By what is proved above there is a composition $h^*$ of 
polynomials $h^{ab}$ for $b\in D(1)$ such that $h^*(D(1))\sse\{a\}$. 
By Lemma~\ref{lem:minimal-sets}(2) every $\ov B$-preserving polynomial 
maps every $(\al,\beta)$-subtrace either to a singleton, or to a 
$(\al,\beta)$-subtrace from $T_1$. Therefore by induction we
also get $h^*(D(i+1))\sse D(i)$. Therefore composing several copies
of $h^*$ collapses $a'$ and $b'$ and leaves $c,d$ unchanged.

\medskip

We now can prove the result in Case~1. We have proved that for 
any $(\al,\beta)$-subtrace $(a',b')\in T_1$ from a 
$\beta$-block $E$ such that $\pr_{12}\rel'\cap(E\tm\{c\})\ne\eps$, a polynomial 
$h^{a'b'}$ with the required properties exists. Suppose now that 
$a',b'\in B'_1$ be any such that $(a',b')\in\beta$.
Take any $c',d'\in B'_2$ such that $(a,c')\in\pr_{12}\rel'$.
By Lemmas~\ref{lem:Qab-tolerance} and~\ref{lem:good-polys}(5)
there is an idempotent $\ov B$-preserving 
polynomial $g$ such that $g(c')=c,g(d')=d$. If $g(a')=g(b')$, 
we are done, as $g$ may serve as $h^{a'b'}$. Otherwise, as before
consider the relation $\relo^\dagger\sse\zA_1^2\tm\zA_2^2\tm\rel$ 
generated by $\{(g(a'),g(b'),c,d,\ba)\}\ \cup\{(x,x,y,y,\bz)\mid 
\bz\in\rel,\bz[1]=x,\bz[2]=y\}$, where $\ba$ is an arbitrary element 
from $\rel'$, and let 
$\relo^\ddagger=\pr_{1234}(\relo^\dagger\cap(B'_1\tm B'_1\tm B'_2\tm 
B'_2\tm\ov B))$, $\relo''_1=\pr_{12}\relo^\ddagger$, 
$\relo''_2=\pr_{34}\relo^\ddagger$,
and the link congruences $\lnk''_1,\lnk''_2$ of $\relo''_1,\relo''_2$ with respect
to $\relo^\ddagger$. Recall that $C'$ is the as-component of $B'_2$ containing
$c,d$. Consider relation 
$\rela=\relo^\ddagger\cap(B_1\tm B_1\tm\Sg{C'}\tm\Sg{C'})$. 
Since $C'\tm C'\sse\relo''$ by (Q1), by Lemma~\ref{lem:as-square-congruence} 
there is $(c'',d'')\in C'\tm C'$
such that $(c'',d'')$ is as-maximal in a $\lnk''_2$-block and $c''\ne d''$.
By the Maximality Lemma~\ref{lem:to-max}(4) there is u-maximal 
$(a'',b'')\in\relo''_1$ such that $(a'',b'',c'',d'')\in\relo'$. This means that
for some $\ov B$-preserving polynomial $g'$ it holds $g'(g(a'))=a'',
g'(g(b'))=b'', g'(c)=c'', g'(d')=d''$. By what was proved 
there is a polynomial $h^{a''b''}$ with $h^{a''b''}(a'')=h^{a''b''}(b'')$ and 
$h^{a''b''}(c'')=c'', h^{a''b''}(d'')=d''$. Also, there is a $\ov B$-preserving
polynomial $g''$ such that $g''(c'')=c, g''(d'')=d$. Finally, this all implies that  
$(a^*,a^*,c,d)\in\relo^\ddagger$ for some $a^*\in B'_1$, that is, 
there is a polynomial $h^{a'b'}=g''\circ h^{a''b''}\circ g'\circ g$ with 
$h^{a'b'}(a')=h^{a'b'}(b')=a^*$, and $h^{a'b'}(c)=c$, $h^{a'b'}(d)=d$. 

\medskip

{\sc Case 2.}
For every element $c$ from a nontrivial as-component of $B'_2$ and any 
$a\in B'_1$ such that $(a,c)\in\rel$ element $a$ belongs to a $\beta$-block 
$B'$ such that $B'\cap\umax(B'_1)=\eps$ or $|\umax(B'\cap B'_1)|=1$.

\medskip

We use the same elements $c,d\in C'$, an as-component of $B'_2$.
For any $(a,b)\in\beta\red{B'_1}$ choose $c',d'\in B'_2$ 
such that $(a,c'),(b,d')\in\pr_{12}\rel'$. (Recall that we are assuming $\al$ 
and $\gm$ to be equality relations.) If $c'=d'$, that is, $(b,c')\in\rel$, 
choose $d'$ to be an arbitrary element from $B'_2$. 
By Lemmas~\ref{lem:Qab-tolerance},~\ref{lem:good-polys}(5) and because 
$\rel$ is chained there is an idempotent $\ov B$-preserving polynomial 
$g$ such that $g(c')=c,g(d')=d$. Let $g(a)=a'$, $g(b)=b'$. 
Then $(a',c)\in\rel$ and $b'\eqc\beta a'$. Since $g$ is $\ov B$-preserving, 
$b'\in B'_1$. We again as in Case~1 consider the relation 
$\relo^*\sse\zA_1^2\tm\zA_2^2\tm\rel$ generated by 
$\{(a',b',c,d,\ba)\}\ \cup\{(x,x,y,y,\bz)\mid \bz\in\rel,\bz[1]=x,\bz[2]=y\}$,
where $\ba$ is an arbitrary element from $\rel'$,  
$\relo^{**}=\pr_{1234}(\relo^*\cap(B'_1\tm B'_1\tm B'_2\tm B'_2\tm\ov B))$,
$\relo^*_1=\pr_{12}\relo^{**}$, $\relo^*_2=\pr_{34}\relo^{**}$, 
and the link congruences $\lnk^*_1,\lnk^*_2$ of $\relo^*_1,\relo^*_2$ 
with respect to $\relo'$. Now we complete the proof as in the end of Case~1.

\medskip

Finally, we use polynomials $h^{ab}$ to construct a single polynomial
that collapses $\beta$ on $E=B'\cap B'_1$ for
every $\beta$-block $B'$. 
Fix $c,d$ and $h^{ab}$ for every pair $a,b\in B'_1$, $a\eqc\beta b$. 
Let $\vc Vk$ be the list of all such pairs, and if $V_\ell=\{a,b\}$ is the pair 
number $\ell$, $h^\ell$ denotes $h^{ab}$. Take a sequence 
$1=\ell_1,\ell_2,\ldots$ such 
that $h^{(1)}=h_{\ell_1}$, $V_{\ell_2}$ is a subset of $h^{(1)}(\zA)$, 
and, for $s>2$, $V_{\ell_s}$ is a subset of 
the range of $h^{(s-1)}=h^{\ell_{s-1}}\circ\ldots\circ h^{\ell_1}$. Since 
$|\Im(h^{(s)})|<|\Im(h^{(s-1)})|$, there is $r$ such that 
$\Im(h^{(r)})$ contains no pair $V_\ell$ for any $\ell$. Therefore setting 
$h(x)=h_{\ell_r}\circ\ldots\circ h_{\ell_1}(x)$ we have that $h$ collapses 
all the pairs $V_\ell$, and $h$ acts identically on $\{c,d\}$. 
The result follows.
\end{proof}

\subsection{Collapsing polynomials}

We say that prime intervals $(\al,\beta)$ and $(\gm,\dl)$ \emph{cannot be 
separated}\index{cannot be separated} if $(\al,\beta)$ cannot be separated from 
$(\gm,\dl)$ and $(\gm,\dl)$ cannot be separated from $(\al,\beta)$. 
In this section we introduce and prove the existence of polynomials that 
collapse all prime intervals in congruence lattices of factors of a subproduct,
except for a set of factors that cannot be separated from each other.

\begin{lemma}\label{lem:e-related}
(1) Let $\zA$ be an algebra. If prime intervals $\al\prec\beta$ and $\gm\prec\dl$ 
in $\Con(\zA)$ are perspective, then they cannot be separated. \\
(2) If $\al\prec\beta$ and $\gm\prec\dl$ from $\Con(\zA)$ cannot be separated, 
then a set $U$ is a $(\al,\beta)$-minimal set if and only if it is a 
$(\gm,\dl)$-minimal set.\\
(3) Let $\rel$ be a subdirect product of $\zA$ and $\zB$, $\al,\beta\in\Con(\zA)$,
$\gm,\dl\in\zB$ such that $\al\prec\beta$, $\gm\prec\dl$, and let $\al\prec\beta$ 
and $\gm\prec\dl$ cannot be separated. Then for any $(\al,\beta)$-minimal 
set $U$ there is a unary idempotent polynomial $f$ such that $f(\zA)=U$ and 
$f(\zB)$ is a $(\gm,\dl)$-minimal set.
\end{lemma}

\begin{proof}
(1) Follows from Lemma~\ref{lem:perspective-intervals}.

(2) Let $f$ be a  polynomial of $\zA$ 
such that $f(\zA)=U$ and $f(\beta)\not\sse\al$. Since $(\al,\beta)$ cannot be 
separated from $(\gm,\dl)$, we have $f(\dl)\not\sse\gm$ and therefore $U$ 
contains a $(\gm,\dl)$-minimal set $U'$. If $U'\ne U$, there is a polynomial $g$
with $g\circ f(\dl)\not\sse\gm$ and $g\circ f(\zA)=U'$. In particular, 
$|g(U)|<|U|$, and so $g\circ f(\beta)\sse\al$; a contradiction with the assumption
that $(\gm,\dl)$ cannot be separated from $(\al,\beta)$.

(3) Take an idempotent polynomial $g$ of $\rel$ such that $g(\zB)$ is a 
$(\gm,\dl)$-minimal 
set. Then, as $(\gm,\dl)$ cannot be separated from $(\al,\beta)$, 
$g(\beta)\not\sse\al$. By Lemma~\ref{lem:minimal-sets}(6) there is an 
$(\al,\beta)$-minimal set $U'\sse g(\zA)$. Let $g',h$ be polynomials 
of $\rel$ such that $g'(U)=U'$, $h(U')=U$ and $h(\zA)=U$, which exist by
Lemma~\ref{lem:minimal-sets}(1). Then $h'=h\circ g\circ g'$ is such that 
$h'(\zA)=h'(U)=U$, $h'(\beta)\not\sse\al$ and therefore $h'(\dl)\not\sse\gm$.
Then iterating $h'$ sufficiently many times we get an idempotent polynomial $f$
satisfying the same properties.
\end{proof}

Let $\rel$ be a subdirect product of $\zA_1\tms\zA_n$, and choose 
$\beta_j\in\Con(\zA_j)$, $j\in[n]$. Let also $i\in[n]$, and 
$\al,\beta\in\Con(\zA_i)$ be such that $\al\prec\beta\le\beta_i$; 
let also $B_j$ be a $\beta_j$-block. We call an idempotent unary polynomial 
$f$ of $\rel$ \emph{$\al\beta$-collapsing for 
$\ov\beta,\ov B$}\index{collapsing polynomial} if $f$ is $\ov B$-preserving,
$f(\beta)\not\sse\al$, $f(\dl\red{B_j})\sse\gm\red{B_j}$ for every 
$\gm,\dl\in\Con(\zA_j)$, $j\in[n]$, with $\gm\prec\dl\le\beta_j$, and such 
that $(\al,\beta)$ can be separated from $(\gm,\dl)$ or $(\gm,\dl)$
can be separated from $(\al,\beta)$, and $|f(\rel)|$ is minimal possible. 

\begin{lemma}\label{lem:collapsing}
Let $\rel$, $\al,\beta$, and $\beta_j$, $j\in[n]$, be as above and $\rel$ 
chained with respect to $\ov\beta,\ov B$. Let also $\rel'=\rel\cap\ov B$. Then
if $\beta=\beta_i$ and $\pr_i\rel'\fac\al$ contains a nontrivial as-component, 
then there exists an $\al\beta$-collapsing polynomial for $\ov\beta,\ov B$.
\end{lemma}

\begin{proof}
Suppose $i=1$, let $B'_1=\pr_1\rel'$ and $C$ be a nontrivial 
as-component of $B'_1\fac\al$. Take a $(\al,\beta)$-subtrace $\{a,b\}\sse B'_1$
such that $a^\al,b^\al\in C$. Since $\rel$ is chained with respect to
$\ov\beta,\ov B$, by (Q1) and Lemma~\ref{lem:good-polys}(5) there is a 
$\ov B$-preserving idempotent polynomial $f$
of $\rel$ such that $f(\zA_1)$ is an $(\al,\beta)$-minimal set and 
$a^\al,b^\al\in f(\zA_1)\fac\al$. Let polynomial $f$ be such that 
$f(\rel)$ is minimal possible. We show that $f$ is $\al\beta$-collapsible. 

Let $j\in[n]$ and $\gm,\dl\in\Con(\zA_j)$ be such that 
$\gm\prec\dl\le\beta_j$, and $(\al,\beta),(\gm,\dl)$ can be separated. 
Since $\rel$ is chained, by Lemma~\ref{lem:relative-symmetry} 
there is an idempotent unary $\ov B$-preserving polynomial 
$f_{j\gm\dl}$ of $\rel$ such that $f_{j\gm\dl}(\zA_1)$ is an 
$(\al,\beta)$-minimal set with 
$a^\al,b^\al\in f_{j\gm\dl}(\zA_1)\fac\al$ and 
$f_{j\gm\dl}(\dl\red{B_j})\sse\gm\red{B_j}$. 
Then if $f(\dl\red{B_j})\not\sse\gm$, then let $g=f_{j\gm\dl}\circ f$.
We have $g(\beta)\not\sse\al$, but $g(\dl\red{B_j})\sse\gm$ implying 
$|g(\rel)|<|f(\rel)|$, a contradiction with minimality of $f(\rel)$.
\end{proof}

\subsection{Separation and minimal sets}

In this section we show a connection between the fact that two prime 
intervals cannot be separated, their types, and link congruences.

\begin{lemma}\label{lem:34-links}
Let $\rel$ be a subdirect product of $\zA$ and $\zB$ and let 
$\al,\beta\in\Con(\zA)$, $\gm,\dl\in\Con(\zB)$ be such that $\al\prec\beta$, 
$\gm\prec\dl$, and $(\al,\beta),(\gm,\dl)$ cannot be separated. Let also 
$\lnk_1,\lnk_2$ be the link congruences of $\zA,\zB$, respectively. If 
$\typ(\al,\beta)\ne\two$ then $\lnk_1\meet\beta\le\al$, $\lnk_2\meet\dl\le\gm$.
\end{lemma}

\begin{proof}
Assume as usual $\al=\zz_\zA$, $\gm=\zz_\zB$.
By Lemma~\ref{lem:e-related}(3) there is a unary polynomial $f$ 
such that $f(\zA_1)=U_1$, $f(\zA_2)=U_2$ are $(\zz_1,\al_1)$- and 
$(\zz_2,\al_2)$-minimal sets, respectively. We first study the structure
of $\rel\cap(U_1\tm U_2)$ and then show how it can be used to prove the lemma.
By $N_1=\{0,1\}$ we 
denote the only trace of $U_1$; by $T_1$ we denote the tail of $U_1$. 
By Lemma~\ref{lem:pseudo-meet} there is a 
polynomial $p(x,y)$ with $p(\zA_1,\zA_1)=U$ and such that $p$ is a semilattice 
operation on $N$, say, $p(0,1)=0$, and $p$ is a semilattice operation
on $\{0,a\},\{1,a\}$ with $p(a,0)=p(a,1)=a$ for any $a\in T_1$.
There are two cases.

\medskip

{\sc Case 1.}  $\typ(\gm,\dl)\ne \two$.

\smallskip

Let $N_2=\{0',1'\}$ be the trace of $U_2$ and $T_2$ the tail of $U_2$. 
We may assume 
$p(x,p(x,y))=p(x,y)$. Observe first that $p$ preserves $N_2$. Indeed, otherwise 
$p(x,x)$ is not a permutation, as $p(0',0'),p(0',1'),p(1',0'),p(1',1')$ belong 
to the same $\beta_2$-block, and if they do not belong to $N_2$ then they 
are all equal, a contradiction with the assumption that $(\al_1,\beta_1)$ and 
$(\al_2,\beta_2)$ cannot be separated. A binary operation on a 2-element
set is either a projection, or a semilattice operation.

Suppose first that $p$ is a projection, say, the first projection on $N_2$. If 
$(\{0\}\tm N_2)\cap\rel\ne\eps$, say, $(0,a)\in\rel$, then 
$f'(x)=p\left(x,\cl0a\right)$ satisfies the conditions: $f'(N_1)=\{0\}$, that is, 
$f'(\al_1)\sse\zz_1$, and $f'(x)=x$ on $N_2$; a contradiction that $(\gm,\dl)$ 
cannot be separated from $(\al,\beta)$. If $(\{1\}\tm N_2)\cap\rel\ne\eps$, 
say, $(1,a)\in\rel$, then $f'(x)=p\left(\cl1a,x\right)$ satisfies the conditions: 
$f'(x)=x$ on $N_1$, that is, $f'(\beta_1)\not\sse\zz_1$, and $f'(N_2)=\{a\}$ 
on $N_2$; a contradiction that $(\al,\beta)$ cannot be separated from 
$(\gm,\dl)$. Therefore, for some $a\in T_1$, $(a,1')\in\rel$. The operation 
$f'=p\left(x,\cl a{1'}\right)$ is the projection on $N_2$ and $f'(N_1)=\{a\}$; 
a contradiction again.

Suppose now that $p$ is a semilattice operation on $N_2$. Let $1'$ be the 
neutral element of $p$. If $(a,1')\in\rel$ for some $a\in T_1$, then 
$f'(x)=p\left(x,\cl a{1'}\right)$ is the projection on $N_2$, and 
$f'(N_1)=\{a\}$. If $(a,1')\in\rel$ for no $a\in T_1$, then we continue as 
follows. If $(0,1')$ or $(1,0')$ belong to $\rel$, then one of the 
operations $p\left(x,\cl 0{1'}\right)$ and $p\left(x,\cl 1{0'}\right)$ contradicts the 
assumption that $i,j$ cannot be separated. Therefore 
$\rel\cap(U_1\tm U_2)\sse\{(1,1')\}\cup((\{0\}\cup T_1)\tm(\{0'\}\cup T_2))$. 

Suppose that either $\lnk_1\cap\beta\ne\zz_\zA$ or 
$\lnk_2\cap\dl\ne\zz_\zB$, where $\lnk_1,\lnk_2$ are the link congruences
of $\zA,\zB$ with respect to $\rel$. Assume the latter. Then, as 
$\zz_\zB\le\lnk_2$, there is a $(\zz_\zB,\dl)$-trace $N$ such that $(a,b)\in\lnk_2$.
This means there are $\vc ak\in\zA$ and $\vc b{k+1}\in\zB$ with 
$a=b_1, b=b_{k+1}$, and $(a_i,b_i),(a_i,b_{i+1})\in\rel$. Take a polynomial 
$f$ of $\rel$ such that 
$U_1=f(\zA)$, $U_2=f(\zB)$ are $(\zz_\zA,\beta)$-, and $(\zz_\zB,\dl)$-minimal
sets, respectively, and such that $U_2$ is a $(\zz_\zB,\dl)$-minimal set
containing $N$ as a trace and $f(a)=a,f(b)=b$. Then, as is easily seen, 
$\rel\cap(U_1\tm U_2)$ does not have the form described above. Thus, 
$\lnk_1\meet\beta=\zz_\zA$ and $\lnk_2\meet\dl=\zz_\zB$.

\medskip

{\sc Case 2.} $\typ(\zz_\zB,\dl)=\two$.

\smallskip

As in Case~1, since $p(x,x)=x$ on $U_2$, operation $p$ preserves every trace 
of $U_2$. Let $N_2$ be a trace in $U_2$. Then $N_2$ is polynomially equivalent to 
a one-dimensional vector space over $\GF(q)$ where $q$ is a prime power. Since 
$p$ is idempotent, it can be represented in the form $\gm x+(1-\gm)y$, 
$\gm\in\GF(q)$. We may assume that $\gm=1$. Indeed, if $\gm=0$ 
then consider $p(y,x)$ instead of $p(x,y)$. Otherwise, the operation
$$
\underbrace{p(p\ldots p}_{q-1\hbox{\footnotesize\ times}}
(x,y),y\ldots, y)
$$
satisfies the required conditions. Now we can complete the proof as in Case~1.
\end{proof}

The proof of Lemma~\ref{lem:34-links} also implies

\begin{corollary}\label{cor:type-equal}
Let $\rel$ be a subdirect product of $\zA$ and $\zB$ and let 
$\al,\beta\in\Con(\zA)$, $\gm,\dl\in\Con(\zB)$ be such that $\al\prec\beta$, 
$\gm\prec\dl$, and $(\al,\beta),(\gm,\dl)$ cannot be separated. Then 
$\typ(\al,\beta)=\typ(\gm,\dl)$.
\end{corollary}

\section{Centralizers and decomposition of CSPs}\label{sec:centralizer}

In this section we introduce an operator on congruence lattices similar to 
the centralizer in commutator theory, and study its 
properties and its connection to decompositions of CSPs.

\subsection{Quasi-Centralizer}\label{sec:definition-centralizer}

For an algebra $\zA$, a term operation $f(x,\vc yk)$, and $\ba\in\zA^k$, let
$f^\ba(x)=f(x,\ba)$\label{not:f-a}. Let $\al,\beta\in\Con(\zA)$, $\al\le\beta$, 
and let $\zeta(\al,\beta)\sse\zA^2$\label{not:zeta} denote the following binary 
relation: 
$(a,b)\in\zeta(\al,\beta)$ if an only if, for any term operation $f(x,\vc yk)$, any 
$i\in[k]$, and any $\ba,\bb\in\zA^k$
such that $\ba[i]=a$, $\bb[i]=b$, and $\ba[j]=\bb[j]$ for $j\ne i$, it holds
$f^\ba(\beta)\sse\al$ if and only if $f^\bb(\beta)\sse\al$. 

\begin{lemma}\label{lem:delta-properties}
For any $\al,\beta\in\Con(\zA)$, $\al\le\beta$:\\[1mm]
(1) $\zeta(\al,\beta)$ is an equivalence relation.\\[1mm]
(2) $\zeta(\al,\beta)$ is the greatest binary relation $\dl$ satisfying the condition:
for any term operation $f(x,\vc yk)$, and any $\ba,\bb\in\zA^k$
such that $(\ba[j],\bb[j])\in\dl$ for $j\in[k]$, it holds
$f^\ba(\beta)\sse\al$ if and only if $f^\bb(\beta)\sse\al$.\\[1mm]
(3) $\zeta(\al,\beta)$ is a congruence of $\zA$.\footnote{Congruence 
$\zeta(\al,\beta)$ appeared in \cite{Hobby88:structure}, but completely 
inconsequentially, they did not study it at all. It is easy to see,
thanks to K.Kearnes, that $\zeta(\al,\beta)$ is greater than the centralizer
of $\al$ and $\beta$, but the reverse inclusion is unclear.}
\end{lemma}

\begin{proof}
(1) $\zeta(\al,\beta)$ is clearly reflexive and symmetric. Suppose 
$(a,b),(b,c)\in\zeta(\al,\beta)$. Let $f(x,\vc yk)$ be a term operation, 
$i\in[k]$, and $\ba,\bc\in\zA^k$ such that $\ba[i]=a,\bc[i]=c$ and 
$\ba[j]=\bc[j]$ for $j\ne i$. Let $\bb\in\zA^k$ be such that $\bb[i]=b$ and 
$\bb[j]=\ba[j]$ for $j\ne i$. Then $f^\ba(\beta)\sse\al$
if and only if $f^\bb(\beta)\sse\al$, which is if and only if $f^\bc(\beta)\sse\al$.

(2) As is easily seen, $\dl$ is reflexive. Choosing $\ba,\bb$ that differ in 
only one position, we show that $\dl\sse\zeta(\al,\beta)$. 

Let us show the reverse inclusion. Let $f,\ba,\bb$ be as in item (2) of the 
lemma, except $(\ba[i],\bb[i])\in\zeta(\al,\beta)$, rather than $\dl$. 
Set $\ba_i\in\zA^k$, $i\in\{0\zd k\}$, 
as follows: $\ba_i[j]=\ba[j]$ for $j\le i$ and $\ba_i[j]=\bb[j]$ for $j>i$. Then
$f^{\ba_i}(\beta)\sse\al$ if and only if $f^{\ba_{i+1}}(\beta)\sse\al$. Thus,
$f^\ba(\beta)\sse\al$ if and only if $f^\bb(\beta)\sse\al$.

(3) By (1) $\zeta(\al,\beta)$ is an equivalence relation, so, we only need to show 
it is preserved by term operations. Let $g(\vc zm)$ be a term operation and 
$\ba,\bb\in\zA^m$ such that $(\ba[i],\bb[i])\in\zeta(\al,\beta)$ for $i\in[m]$.
Let also $a=g(\ba)$ and $b=g(\bb)$. We show that $(a,b)\in\zeta(\al,\beta)$.
Take a term operation $f(x,\vc yk)$, $i\in[k]$, and $\ba',\bb'\in\zA^k$ such 
that $\ba'[i]=a,\bb'[i]=b$, and $(\ba'[j],\bb'[j])\in\zeta(\al,\beta)$ for $j\ne i$.
Without loss of generality, $i=k$. Let also 
$$
h(x,\vc y{k-1},\vc zm)=f(x,\vc y{k-1},g(\vc zm)),
$$
and $\ba''=(\ba'[1]\zd\ba'[k-1],\ba[1]\zd\ba[m])$, 
$\bb''=(\bb'[1]\zd\bb'[k-1],\bb[1]\zd\bb[m])$. Then 
$(\ba''[j],\bb''[j])\in\zeta(\al,\beta)$
for all $j\in[k+m-1]$. Therefore $f^{\ba'}(\beta)=h^{\ba''}(\beta)\sse\al$ if
and only if $f^{\bb'}(\beta)=h^{\bb''}(\beta)\sse\al$.
\end{proof}

The congruence $\zeta(\al,\beta)$ will be called the 
\emph{quasi-centralizer}\index{quasi-centralizer} of $\al,\beta$. Next we 
prove several properties of quasi-centralizer similar to some
extent to the properties of the regular centralizer. The following statement is 
one of the key ingredients of the algorithm.

\begin{prop}\label{pro:full-centralizer}
If $\zeta(\al,\beta)\ge\beta$, then $(\al,\beta)$ has type \two, and for any 
$\beta$-blocks $B,C$ such that $B\le C$ in $\zA\fac\beta$ (that is $BC$ is 
a thin semilattice edge in $\zA\fac\beta$) and they belong to the same 
$\zeta(\al,\beta)$-block,  there is an injective mapping 
$\sg\colon B\fac\al\to C\fac\al$ such that for any $a\in B\fac\al$, $a\le\sg(a)$ and
$a\not\le b$ for any other $b\in C$.
\end{prop}

\begin{proof}
Clearly, we may assume $\al=\zz$. Suppose that $\zeta(\al,\beta)\ge\beta$,
and suppose first that $\typ(\zz,\beta)\ne\two$. Take any 
$(\zz,\beta)$-minimal set 
$U$, its only trace $N$, and a pseudo-meet operation $p$ on $U$. 
Then the polynomial $p(x,0)$ does not collapse $\beta$, as $f(0,0)=0$, 
$f(1,0)=1$, while the polynomial $p(x,1)$ does, a contradiction with the 
assumption $\zeta(\al,\beta)\ge\beta$.

Suppose now that $\typ(\zz,\beta)=\two$. Then by 
Corollary~\ref{cor:max-min-set}(1) for any $a,b\in\zA$ with $a\eqc\beta b$, 
there is a $(\zz,\beta)$-minimal set $U$ such that $a,b\in U$.

Let $B,C$ be $\beta$-blocks and $B\le C$ in $\zA\fac\beta$. By 
Lemma~\ref{lem:thin}(1) for any $a\in\zA$ there is $b\in C$ with $a\le b$. 
Suppose the statement of the proposition is not
true. Then there are two possibilities.

1. For some $a\in B$ and $b,c\in C$, $a\le b, a\le c$. Let $f$ be a polynomial 
such that $U=f(\zA)$ is a $(\zz,\beta)$-minimal set and $b,c\in U$. 
Consider $g_1(x)=a\cdot f(x)$ and $g_2(x)=b\cdot f(x)$. Clearly, 
$g_1(b)=b, g_1(c)=c$, so $g_1(\beta)\not\sse\zz$. 
On the other hand, $g_2(b)=b=g_2(c)$, that is, $g_2(\beta)\sse\zz$, as 
$|g_2(\zA)|\le|U|$, a contradiction with the assumption $(a,b)\in\zeta(\zz,\beta)$.

2. For some $a\in C$ and $b,c\in B$, $b\le a, c\le a$. Let $f$ be a polynomial 
such that $U=f(\zA)$ is a $(\zz,\beta)$-minimal set and $b,c\in U$. 
Consider $g_1(x)=f(x)\cdot a$ and $g_2(x)=f(x)\cdot b$. Clearly, 
$g_1(b)=g_1(c)=a$, so $g_1(\beta)\sse\zz$, as $|g_1(\zA)|\le|U|$. 
On the other hand, $g_2(b)=b, g_2(c)=c$, that is, $g_2(\beta)\not\sse\zz$ 
a contradiction again.
\end{proof}

\begin{corollary}\label{cor:centralizer-multiplication}
Let $\zeta(\al,\beta)=\zo_\zA$, $a,b,c\in\zA$ and $b\eqc\beta c$.
Then $ab\eqc\al ac$.
\end{corollary}

\begin{proof}
We have $ab\eqc\beta ac$ and $a\le ab,a\le ac$. 
By Proposition~\ref{pro:full-centralizer} $ab\eqc\al ac$.
\end{proof}

\begin{remark}
Recently, Payne \cite{Payne16:product}
designed a polynomial time algorithm for the following class of algebras:
Every algebra $\zA$ from this class has a congruence $\al$ such that $\zA\fac\al$ 
is a semilattice, and the interactions between $\al$-blocks satisfy a certain 
condition. It seems that Lemma~\ref{pro:full-centralizer} is
similar to what this condition can provide.
\end{remark}

An injective mapping between $\beta$-blocks $B,C$ inside a 
$\zeta(\al,\beta)$-block can also be established whenever $BC$ is any
thin edge in $\zA\fac\beta$, as the following lemma shows.

\begin{lemma}\label{lem:central-mapping}
Let $\al,\beta\in\Con(\zA)$ such that $\al\prec\beta\le\zeta=\zeta(\al,\beta)$
and let $B,C$ be $\beta$-blocks from the same $\zeta$-block such that $BC$ is 
a thin edge in $\zA\fac\beta$. For any $b\in B,c\in C$ such that $bc$ is a thin 
edge the polynomial $f(x)=x\cdot c$ if $b\le c$, $f(x)=t_{bc}(x,c)$ if $bc$ is 
majority, and $f(x)=h_{bc}(x,b,c)$ if $bc$ is affine, where $t_{ab},h_{ab}$ 
are the operations from Lemma~\ref{lem:op-s-on-affine}, is an injective 
mapping from $B\fac\al$ to $C\fac\al$.
\end{lemma}

\begin{proof}
We can assume that $\al$ is the equality relation. Suppose $f(a_1)=f(a_2)$ 
for some $a_1,a_2\in B$. Since $\typ(\al,\beta)=\two$, by 
Corollary~\ref{cor:max-min-set}(1) every pair of elements
of $B$ is an $(\al,\beta)$-subtrace. Let $f'$ be an idempotent unary 
polynomial such that $f'(a_1)=a_1$, $f'(a_2)=a_2$, and $f'(\zA)$
is an $(\al,\beta)$-minimal set. 

If $b\le c$, let $g(x,y)=f'(x)\cdot y$. Then $g^c=g(x,c)=f(x)$ on $\{a_1,a_2\}$,
that is, $g^c(a_1)=g^c(a_2)$ implying $g^c(\beta)\sse\al$. On the other hand, 
$g^b(x)=f'(x)$ on $\{a_1,a_2\}$ implying $g^b(\beta)\not\sse\al$, 
a contradiction with the assumption $b\eqc\zeta c$. 

If $bc$ is a thin majority edge, set $g(x,y)=t_{bc}(f'(x),y)$. Then 
$g^c(a_1)=f(a_1)=f(a_2)=g^c(a_2)$, and so $g^c(\beta)\sse\al$. On the
other hand, since $B\fac\al$ is a module, $a_1b,a_2b$ are affine edges 
and $\al$ witnesses that. Therefore $g^b(a_1)=a_1$ and $g^b(a_2)=a_2$,
implying $g^b(\beta)\not\sse\al$, and we have a contradiction again.

Finally, if $bc$ is a thin affine edge, we consider the polynomials 
$g(x,y,z)=h_{bc}(f'(x),y,z)$ and $g^{bc}(x)=g(x,b,c)$, 
$g^{a_1a_1}(x)=g(x,a_1,a_1)$. Again, 
$g^{bc}(a_1)=f(a_1)=f(a_2)=g^{bc}(a_2)$, while 
$$
g^{a_1a_1}(a_1)=h_{bc}(f'(a_1),a_1,a_1)=a_1\ne
h_{bc}(f'(a_2),a_1,a_1)=g^{a_1a_1}(a_2),
$$
since by Lemma~\ref{lem:op-s-on-affine} $h_{bc}(x,a_1,a_1)$ is a permutation.
This implies that $g^{bc}(\beta)\sse\al$ and $g^{a_1a_1}(\beta)\not\sse\al$,
a contradiction.
\end{proof}

\begin{lemma}\label{lem:equal-centralizer}
Let $\al,\beta\in\Con(\zA)$ be such that $\al\prec\beta$ and 
$\typ(\al,\beta)=\two$, and $\zeta=\zeta(\al,\beta)$. Then for any 
$\beta$-blocks $B_1,B_2$ that belong to the same $\zeta$-block $C$ and 
such that $B_1\sqq_{asm}B_2$ and $B_2\sqq_{asm} B_1$ in $C\fac\beta$, 
$|B_1\fac\al|=|B_2\fac\al|$.
\end{lemma}

\begin{proof}
Since there is an asm-path from $B_1$ to $B_2$ and back, the result 
follows from Lemma~\ref{lem:central-mapping}.
\end{proof}

Let $\zA$ be an algebra and $\al,\beta\in\Con(\zA)$, $\al\prec\beta$.
Element $a\in\zA$ is said to be 
\emph{$\al\beta$-minimal}\index{$\al\beta$-minimal element} if it belongs 
to an $(\al,\beta)$-trace. Let $\emin_\zA(\al,\beta)$\label{not:emin} denote 
the set of all $\al\beta$-minimal elements of $\zA$. By 
Lemma~\ref{lem:minimal-sets}(5) $\emin_\zA(\al,\beta)$ intersects
every $\al$-block from a nontrivial $\beta$-block. The following lemma 
shows that $\emin_\zA(\al,\beta)$ can be much larger than that. Due to
the way we will use it in the future the statement of the lemma is not
quite straightforward.

\begin{lemma}\label{lem:maximal-minimal}
Let $\al,\beta,\gm,\dl\in\Con(\zA)$ be such that $\gm\prec\dl\le\beta$, 
$\al\prec\beta$, intervals $(\al,\beta),(\gm,\dl)$ cannot be separated,  
and $\typ(\al,\beta)=\two$; let $B$ be a $\beta$-block. 
If $a\in \emin_\zA(\gm,\dl)\cap B$ then for any 
$b\in B$ such that $a\sqq_{asm}b$ in $B$, $b\in \emin_\zA(\gm,\dl)$.

In particular, $\umax(B)\sse \emin_\zA(\gm,\dl)$.

Moreover, if $a\eqc\al b$, $a\sqq_{asm}b$ in $a^\al$, and $f$ is a 
polynomial such that $f(a)=a$, $f(\zA)$ is an $(\al,\beta)$-minimal
set, and $N$ its trace with $a\in N$, then there is a polynomial $g$
such that $g(b)=b$, $g(\zA)$ is an $(\al,\beta)$-minimal set,
$N'$ is its trace containing $b$ and $N'\fac\al=N\fac\al$.
\end{lemma}

\begin{proof}
Let $f$ be an idempotent unary polynomial of $\zA$ such that $a\in N$, 
a trace in $U=f(\zA)$, a $(\gm,\dl)$-minimal set. Note that 
$f(B)\sse B$ and $f(\beta)\not\sse\al$. It suffices
to consider the case when $ab$ is a thin edge.

Depending on the type of the edge $ab$ we set $f'(x)=f(x)\cdot b$,
$f'(x)=t_{ab}(f(x),b)$, or $f'(x)=h_{ab}(f(x),a,b)$, if $ab$ is semilattice,
majority or affine, respectively. Note also that by 
Lemma~\ref{lem:op-s-on-affine} $f'(a)=b$, and therefore if 
$f'(\dl)\not\sse\gm$ we have $f'(\zA)$ is a $(\gm,\dl)$-minimal set, and
$b$ belongs to it.

Since $(\al,\beta)$ and $(\gm,\dl)$ cannot be separated, by 
Lemma~\ref{lem:e-related}(2) $U$ is an $(\al,\beta)$-minimal set.
Hence, there are $a_1,a_2\in B\fac\al$ such that $a_1\ne a_2$ and
$f(a_1)=a_1,f(a_2)=a_2$. Since $a_1a_2$ is an affine edge in $\zB\fac\al$,
depending on the type of $ab$ we have:\\
-- if $a\le b$, then $f'(a_i)=a_i\cdot b^\al=a_i$ for $i=1,2$;\\
-- if $ab$ is majority, then $f'(a_i)=t_{ab}(a_i,b^\al)=a_i$, as 
$a_i\eqc{\beta\fac\al}b^\al$ for $i=1,2$ and by 
Lemma~\ref{lem:op-s-on-affine}(1);\\
-- if $ab$ is affine, then by Lemma~\ref{lem:op-s-on-affine}(2)
$h_{ab}(x,a^\al,b^\al)$ is a permutation on
$B\fac\al$, in particular, $f'(a_1)\ne f'(a_2)$.\\
In either case we obtain $f'(\beta)\not\sse\al$, implying $f'(\dl)\not\sse\gm$.

For the last claim of the lemma it suffices to notice that if $a\eqc\al b$
then $f'(x)\eqc\al f(x)$ for $x\in B$.
\end{proof}

\subsection{Decomposition of CSPs}\label{sec:csp-decomposition}

In this section we show that if intervals in congruence lattices of domains 
in a CSP instance cannot be separated, they induce certain decomposition
of the instance or its subinstances. The components of this decomposition 
are instances over smaller domains, which are, actually, blocks of the 
corresponding quasi-centralizers.

Let $\rel$ be a subdirect product of $\zA_1\tm\dots\tm\zA_n$, $i,j\in[n]$, 
and $\al_i\in\Con(\zA_i)$, 
$\al_j\in\Con(\zA_j)$. The coordinate positions $i,j$ are said to be 
\emph{$\al_i\al_j$-aligned in $\rel$}\index{$\al_i\al_j$-aligned coordinates} 
if, for any $(a,c),(b,d)\in\pr_{ij}\rel$, 
$(a,b)\in\al_i$ if and only if $(c,d)\in\al_j$. Or in other words, the link 
congruences of $\zA_i,\zA_j$ with respect to $\pr_{ij}\rel$ are no greater
than $\al_i,\al_j$, respectively.

\begin{lemma}\label{lem:delta-alignment}
Let $\rel$ be a subdirect product of 
$\zA_1\tm\zA_2$, $\al_i,\beta_i\in\Con(\zA_i)$, $\al_i\prec\beta_i$, 
for $i=1,2$. If $(\al_1,\beta_1)$ and $(\al_2,\beta_2)$ cannot be separated 
from each other, then the coordinate positions 1,2 are 
$\zeta(\al_1,\beta_1)\zeta(\al_2,\beta_2)$-aligned in $\rel$.
\end{lemma}

\begin{proof}
Let us assume the contrary, that is, without loss of generality there are 
$a,b\in\zA_1$ and $c,d\in\zA_2$
with $(a,c),(b,d)\in\rel$, $(a,b)\in\zeta(\al_1,\beta_1)$, but 
$(c,d)\not\in\zeta(\al_2,\beta_2)$. Therefore there is $g(x,\vc yk)$, a term 
operation of $\zA_2$, $i\in[k]$, and $\bc,\bd\in\zA^k_2$ with $\bc[i]=c$,
$\bd[i]=d$ and $\bc[j]=\bd[j]$ for $j\ne i$, such that $g^\bc(\beta_2)\sse\al_2$
but $g^\bd(\beta_2)\not\sse\al_2$, or the other way round.
Extend $g$ to a term operation $g$ of $\rel$, and choose $\ba,\bb\in\zA_1^k$
such that $\ba[i]=a,\bb[i]=b$, $\ba[j]=\bb[j]$ for $j\ne i$, and $(\ba[j],\bc[j]),
(\bb[j],\bd[j])\in\rel$ for $j\in[k]$. Then $g^\ba(\beta_1)\sse\al_1$ if and only if
$g^\bb(\beta_1)\sse\al_1$. Therefore, there is a polynomial of $\rel$ that 
separates $(\al_1,\beta_1)$ from $(\al_2,\beta_2)$ or the other way round, a 
contradiction.
\end{proof}

Let $\cP=(V,\cC)$ be a (2,3)-minimal instance, in particular,
for every $X\sse V$, $|X|=2$, it contains a constraint $C^X=\ang{X,\rel^X}$.
Let $w_1,w_2,\in V$. We say that $w_1,w_2$ are
\emph{$\al_1\al_2$-aligned} in $\cP$, where $\al_1\in\Con(\zA_{w_1})$, 
$\al_2\in\Con(\zA_{w_2})$, if they are $\al_1\al_2$-aligned 
in $\rel^{w_1w_2}$. For $\al_v\in\Con(\zA_v)$, $v\in V$, instance $\cP$ is 
said to be \emph{$\ov\al$-aligned}\index{$\ov\al$-aligned instance} 
if every $w_1,w_2$ are $\al_{w_1}\al_{w_2}$-aligned. This means that 
there are one-to-one mappings 
$\vf_{w_1w_2}:\zA_{w_1}\fac{\al_{w_1}}\to \zA_{w_1}\fac{\al_{w_1}}$
such that whenever $(a,b)\in\rel^{w_1w_2}$, 
$b^{\al_{w_2}}=\vf_{w_1w_2}(a^{w_1})$. Observe that since $\cP$
is (2,3)-minimal, these mappings are consistent, that is, for any $u,v,w\in V$,
$\vf_{vw}\circ\vf_{uv}=\vf_{uw}$. Therefore $\cP$ can be represented as a 
disjoint union of instances $\cP_1\zd\cP_k$, where $k$ is the number of 
$\al_v$-blocks for any $v\in V$ and the domain of $v\in V$ of $\cP_i$ is the $i$-th
$\al_v$-block.

Let again $\cP=(V,\cC)$ be a (2,3)-minimal instance and let
$\ov\beta$, $\beta_v\in\Con(\zA_v)$, $v\in V$, be a collection of congruences. 
Let $\cW^\cP(\ov\beta)$\label{not:cW} denote the set of triples $(v,\al,\beta)$ 
such that $v\in V$, $\al,\beta\in\Con(\zA_v)$, and $\al\prec\beta\le\beta_v$. Also, 
$\cW^\cP$ denotes $\cW^\cP(\ov\beta)$ when $\beta_v=\zo_v$ for all $v\in V$. 
We will omit the superscript $\cP$ whenever it is clear from the context. 
For every $(v,\al,\beta)\in\cW(\ov\beta)$, let $Z$ denote the set of triples 
$(w,\gm,\dl)\in\cW(\ov\beta)$ such that $(\al,\beta)$ and $(\gm,\dl)$
cannot be separated in $\rel^{vw}$. Slightly abusing the terminology we 
will also say that $(\al,\beta)$ and $(\gm,\dl)$ cannot be separated in $\cP$. 
Then let $W_{v,\al\beta,\ov\beta}=\{w\in V\mid (w,\gm,\dl)\in Z
\text{ for some $\gm,\dl\in\Con(\zA_w)$}\}$\label{not:W-v-albeta-beta}.
We will omit the subscript $\ov\beta$ whenever possible.
The following statement is an easy corollary of 
Lemma~\ref{lem:delta-alignment}.

\begin{theorem}\label{the:centralizer-alignement}
Let $\cP=(V,\cC)$ be a (2,3)-minimal instance and $(v,\al,\beta)\in\cW$. 
For $w\in W_{v,\al\beta,\ov\beta}$, where $\beta_v=\zo_v$ for $v\in V$, 
let $(w,\gm,\dl)\in\cW$ be such that $(\al,\beta)$ and 
$(\gm,\dl)$ cannot be separated and $\zeta_w=\zeta(\gm,\dl)$. Then 
$\cP_{W_{v,\al\beta,\ov\beta}}$ is $\ov\zeta$-aligned.
\end{theorem}

\section{The Congruence Lemma}\label{sec:congruence}

This section contains several technical results that will be used when 
proving of the soundness of the algorithm. The main result is the Congruence
Lemma~\ref{lem:affine-link}. Lemmas~\ref{lem:max-nontrivial} 
and~\ref{lem:dirprod-propagation} are auxiliary and only used in the 
proof of the Congruence Lemma~\ref{lem:affine-link}. We start with 
introducing two closure properties of algebras and their subdirect products, 
this time under certain polynomials.

We say that a set $A$ is \emph{as-closed}\index{as-closed set} in algebra 
$\zB$, $A\sse\zB$, if $A\cap\umax(\zB)\ne\eps$ and, 
for every $a,b\in\zB$ such that $a\sqq_{as}b$ in 
$\zB$ and $a\in A\cap\umax(\zB)$, element $b$ also belongs to $A$. 

Let $\rel$ be a subdirect product of $\vc\zA n$ and $\relo$ a subalgebra of 
$\rel$. We say that $\relo$ is \emph{polynomially closed} in $\rel$ if 
for any polynomial $f$ of $\rel$ the following condition
holds: for any $\ba,\bb\in\umax(\relo)$ such that $f(\ba)=\ba$ and 
for any $\bc\in\Sg{\ba,f(\bb)}$ such that 
$\ba\sqq_{as}\bc$ in $\Sg{\ba,f(\bb)}$, the tuple $\bc$ belongs 
to $\relo$. 

\begin{remark}
Polynomially closed subalgebras of Mal'tsev algebras are congruence blocks.
In the general case the structure of polynomially closed subalgebras is more 
intricate. The intuition 
(although not entirely correct) is that if for some block $B$ of a congruence 
$\beta$ and a congruence $\al$ with $\al\prec\beta$ the set $B\fac\al$ 
contains several as-components, a polynomially closed subalgebra contains 
some of them and has empty intersection with the rest. However, since this
is true only for factor sets, and we do not even consider non-as-maximal
elements, the actual structure is more `fractal'.
\end{remark}

The following lemma follows from the definitions, Lemma~\ref{lem:thin}, 
and the fact that congruences are invariant under polynomials.

\begin{lemma}\label{lem:poly-closed}
(1) Let $\rel$ be a subdirect product of $\vc\zA n$ and $\relo_1,\relo_2$ relations 
polynomially closed in $\rel$, and $\relo_1\cap\relo_2\cap\umax(\relo)\ne\eps$. 
Then $\relo_1\cap\relo_2$ is polynomially closed in $\rel$. 

In particular, let $\beta_i\in\Con(\zA_i)$ and $B_i$ a u-maximal $\beta_i$-block. 
Then $\relo_1\cap\ov B$ is polynomially closed in $\rel$.

If $\relo_1,\relo_2$ are as-closed in $\rel$, then $\relo_1\cap\relo_2$ 
is as-closed in $\rel$.\\[2mm]
(2) Let $\relo_i$ be polynomially closed in $\rel_i$, $i\in[k]$, and let
$\rel,\relo$ be pp-defined through $\vc\rel k$ and $\vc\relo k$, respectively, 
by the same pp-formula $\Phi$; that is, $\rel=\Phi(\vc\rel k)$ and 
$\relo=\Phi(\vc\relo k)$. Let also for every atom $\rel_i(\vc x\ell)$ and any 
$\ba\in\umax(\rel_i)$ there is $\bb\in\rel$ with $\pr_{\{\vc x\ell\}}\bb=\ba$,
and also $\umax(\relo)\cap\umax(\rel)\ne\eps$. 
Then $\relo$ is polynomially closed in $\rel$.

If $\relo_i$ is as-closed in $\rel_i$ then $\relo$ is as-closed in $\rel$.\\[2mm]
(3) Let $\rel$ be a subdirect product of $\vc\zA n$, $\beta_i\in\Con(\zA_i)$, 
$i\in[n]$, and let $\relo$ be polynomially closed in $\rel$. Then 
$\relo\fac{\ov\beta}$ is polynomially closed in $\rel$.

If $\relo$ is as-closed in $\rel$ then $\relo\fac{\ov\beta}$ is as-closed
in $\rel\fac{\ov\beta}$.
\end{lemma}

The first two lemmas in the rest of this section study the structure of binary 
relations that have
in their domains a pair of prime intervals of type~\two\ that cannot be 
separated. They show that if we restrict ourselves to blocks of the link 
congruences then this structure is very uniform. The third lemma, 
Lemma~\ref{lem:affine-link} (the Congruence Lemma), is an important
technical result. To explain what it amounts to saying consider this: let 
$\relo\sse\zA'\tm\zB'$ be a subdirect product and the link congruence of 
$\zA'$ is the equality relation. Then, clearly, $\relo$ is the graph of 
a mapping $\sg:\zB'\to\zA'$, and the kernel of this mapping is the link
congruence $\eta$ of $\zB'$ with respect to $\relo$. Suppose now that $\relo$
is a subalgebra of $\rel$, a subdirect product of $\zA\tm\zB$ such that 
$\zA'$ is a subalgebra of $\zA$ and $\zB'$ is a subalgebra of $\zB$. Then
the restriction of the link congruence of $\zA$ with respect to $\rel$ to $\zA'$
does not have to be the equality relation, and similarly the restriction of
the link congruence of $\zB$ to $\zB'$ does not have to be $\eta$. Most 
importantly, the restriction of $\Cg\eta$, the congruence of $\zB$ generated by 
$\eta$, to $\zB'$ does not have to be $\eta$. The Congruence 
Lemma~\ref{lem:affine-link} shows, however, that this is exactly what 
happens when $\relo$ and $\zA',\zB'$ satisfy some additional conditions, 
such as being chained and polynomially closed.

In the next two lemmas
let $\rel$ be a subdirect product of $\zA_1\tm\zA_2$, $\beta_1,\beta_2$
congruences of $\zA_1,\zA_2$ and $B_1,B_2$ $\beta_1$- and 
$\beta_2$-blocks, respectively; $\rel$ is chained with respect to 
$(\beta_1,\beta_2),(B_1,B_2)$ and $\rel^*=\rel\cap(B_1\tm B_2)$,
$B^*_1=\pr_1\rel^*,B^*_2=\pr_2\rel^*$. Let
$\al,\beta\in\Con(\zA_1)$, $\gm,\dl\in\Con(\zA_2)$ be such that
$\al\prec\beta\le\beta_1$, $\gm\prec\dl\le\beta_2$, 
$\typ(\al,\beta)=\typ(\gm,\dl)=\two$, and $(\al,\beta),(\gm,\dl)$ cannot 
be separated. Let also $\zeta_1=\zeta(\al,\beta)\red{B^*_1}$,
$\zeta_2=\zeta(\gm,\dl)\red{B^*_2}$ and $\lnk^*_1,\lnk^*_2$ the link 
congruences of $B^*_1,B^*_2$, respectively, with respect to $\rel^*$. 
Let $F,G$ be $\zeta_1$-, $\zeta_2$-blocks such that $\rel^*\cap(F\tm G)\ne\eps$ 
and $F,G$ contain nontrivial $\beta$- and $\dl$-blocks $A,B$, respectively 
(that is, $|A\fac\al|,|B\fac\gm|>1$). 

By Lemma~\ref{lem:equal-centralizer} all the $\beta$-blocks $A'\in F\fac\beta$, 
$A\sqq_{asm}A'$ in $F$ (respectively, all $\dl$-blocks $B'\in G\fac\dl$, 
$B\sqq_{asm}B'$ in $G$) are also nontrivial. Note that by 
Lemma~\ref{lem:delta-alignment} $\lnk^*_1\le\zeta_1$ and
$\lnk^*_2\le\zeta_2$. Let also $D\sse F,E\sse G$ be blocks of 
$\lnk^*_1,\lnk^*_2$ such that $\rel^*\cap(D\tm E)\ne\eps$. 

The first lemma claims that for the link congruence $\lnk^*_2$ there are
only two options: either it is a subset of $\gm$ on the $\zeta_2$-block
$G$, or it contains $\dl\fac\gm$ on $\umax(G)$.  

\begin{lemma}\label{lem:max-nontrivial}
Suppose that $B\cap E\ne\eps$ and that $\lnk^*_2$ is nontrivial on the 
$\dl$-block $B$, that is, there are distinct $a,b\in (B\cap B^*_2)\fac\gm$ with 
$(a,b)\in\lnk^*_2$, or equivalently $\lnk^*_2\meet\dl$ is not a subset of 
$\gm$ on $B\cap B^*_2$. Then\\[1mm]
(1) if $G\cap\umax(B^*_2)\ne\eps$ then 
$\dl\red{\umax(G)}\le\lnk^*_2\join\gm\red{B^*_2}$; 
and\\[1mm] 
(2) any $B'\in G\fac\dl$ with 
$B\sqq_{asm}B'$ in $G\fac\dl$ is nontrivial, that is, $|B'\fac\gm|>1$. In particular, 
$\umax(D),\umax(E)$ and $\umax(F),\umax(G)$ 
do not intersect any trivial $\beta$- and $\dl$-blocks, respectively.
\end{lemma}

\begin{proof}
Since $\lnk^*_1\le\zeta_1$ and $\lnk^*_2\le\zeta_2$, (2) follows by 
Lemma~\ref{lem:equal-centralizer}. Also, as $\lnk^*_2$ is nontrivial on a 
$\dl$-block, we obtain (1) by Lemma~\ref{lem:link-separability}.
\end{proof}

The second lemma amounts to saying that if $\lnk^*_2$ is nontrivial
on $G$ then it not only contains $\dl$ (modulo $\gm$), but also that
if an element of $F$ is related by $\rel$ to some element of a $\dl$-block
$B$, it is also related to the entire $B$ (again, modulo $\gm$).

\begin{lemma}\label{lem:dirprod-propagation}
Suppose $\dl\red{\umax(G)}\le\lnk^*_2\join\gm\red{B^*_2}$ and sets 
$A,B$ satisfy one of the following two conditions:\\[1mm]
(1) let $A\sse F,B\sse G$ be $\beta$- and $\dl$-blocks, respectively, 
such that $(A,B)\in\umax((\rel\cap(F\tm G))\fac{\beta\tm\dl})$, or \\[1mm]
(2) let $A,B$ be $\beta$-, and $\dl$-blocks, respectively, and such that 
$A'=A\cap D\ne\eps$ and $(A',B)\in\umax((\rel\cap(D\tm G))\fac{\beta\tm\dl})$. 
\\[1mm]
Then either $\rel\cap(A\tm B)=\eps$, or for any $c\in A$ with 
$B\cap\rel[c]\ne\eps$ we have $\{c\}\tm B\fac\gm\sse\rel\fac\gm$.
\end{lemma}

Note that in condition (2) $D$ does not have to contain u-maximal 
elements of $F$, and similarly $E$ does not have to contain u-maximal 
elements of $G$. Thus, (1) is not necessarily a more general condition.

\begin{proof}
We prove the lemma for condition (2), that is, when 
$(A',B)\in\umax((\rel\cap(D\tm G))\fac{\beta\tm\dl})$.
It will be clear that case (1) follows from the same argument. 

We assume $\gm=\zz_2$. Since $B$ is a module, it is as-connected. 
Therefore if some element of $B$ belongs to an as-component of $E$, 
the whole set $B$ is contained in that as-component. By the Rectangularity 
Corollary~\ref{cor:linkage-rectangularity}, this means that if 
$\rel\cap(\{c\}\tm B)\ne\eps$ for $c\in\amax(D)$, then 
$\{c\}\tm B\sse\rel$, and the result for $B$ follows.
 
Now we show that if $\{c\}\tm B\sse\rel$ for some $c\in D$ and $B\in E\fac\dl$ 
then $\{d\}\tm b'^\dl\sse\rel$ for any $d\in D$ and $b'\in E$ such that 
$(c,b)(d,b')$ is a thin edge in $\rel\cap(D\tm E)$ for some $b\in B$. 
As is easily seen this implies the result. There are 3 possible cases. 

\medskip

{\sc Case 1.} $b'^\dl=B$, that is $cd$ is a thin edge and $(d,b')\in\rel$,
$b'\in B$. Then $\{d\}\tm B\sse\rel$.

This case follows from Lemma~\ref{lem:as-rectangularity}.

\medskip

{\sc Case 2.} $c=d$, that is, $BB'$ is a thin edge in $\zA_2\fac\dl$ where 
$B'=b'^\dl$ and $(c,b')\in\rel$. 

\medskip

Let $f(x)$ be the unary polynomial of $\rel$ constructed as in 
Lemma~\ref{lem:central-mapping}, that is, $f(x)=x\cdot \cl c{b'}$, 
$f(x)=t_{bb'}\left(x,\cl c{b'}\right)$, or 
$f(x)=h_{bb'}\left(x,\cl cb,\cl c{b'}\right)$, 
depending on the type of $bb'$. Then by Lemma~\ref{lem:central-mapping} 
$f:B\to B'$ is a bijection, and therefore 
maps $\{c\}\tm B\sse\rel$ onto $\{c\}\tm B'$, implying $\{c\}\tm B'\sse\rel$.

\medskip

{\sc Case 3.} 
$cd$ and $bb'$ are thin edges of the same type.

\medskip

Let $B'=b'^\dl$. Similar to Case~2 depending on the type of $cd$ we 
consider polynomial
$f(x)=x\cdot \cl d{b'}$, $f(x)=t_{cd}\left(x,\cl d{b'}\right)$, or 
$f(x)=h_{cd}\left(x,\cl cb,\cl d{b'}\right)$ for some $b\in B$. We have $f(c)=d$
and $f:B\to B'$ is a bijection by Lemma~\ref{lem:central-mapping}, thus 
proving that $\{d\}\tm B'\sse\rel$.

\medskip

If condition (1) holds the prove is essentially the same, except we need to 
use the same starting point as above, and consider pairs from 
$\umax(\rel\cap(F\tm G))$.
\end{proof}

We are now in a position to state and prove the main result of the section.
Let again $\rel$ be a subdirect product of $\zA_1\tm\zA_2$, 
$\beta_1,\beta_2$ congruences of $\zA_1,\zA_2$ and $B_1,B_2$ 
$\beta_1$- and $\beta_2$-blocks, respectively. Also, let $\rel$ be chained 
with respect to $(\beta_1,\beta_2),(B_1,B_2)$ and 
$\rel^*=\rel\cap(B_1\tm B_2)$,
$B^*_1=\pr_1\rel^*,B^*_2=\pr_2\rel^*$. Let $\al,\beta\in\Con(\zA_1)$ be 
such that $\al\prec\beta\le\beta_1$. This time we do not assume that 
$\typ(\al,\beta)=\two$.

\begin{lemma}[The Congruence Lemma]\label{lem:affine-link}
Suppose $\al=\zz_1$ and let $\rel'$ be a polynomially closed subalgebra of 
$\rel^*$ and such that $B'_1=\pr_1\rel'$ contains an as-component $C$ of 
$B^*_1$ and $\rel'\cap\umax(\rel^*)\ne\eps$.
Then either \\[2mm]
(1) $C\tm\umax(B''_2)\sse\rel'$, where $B''_2=\rel'[C]$, or\\[1mm] 
(2) there is $\eta\in\Con(\zA_2)$ with $\eta\prec\beta_2$ 
such that intervals $(\al,\beta_1)$ and $(\eta,\beta_2)$ cannot be separated. \\
Moreover, in case (2) $\rel'\cap(C\tm B''_2)$ is the graph of a mapping 
$\vf: B''_2\to C$ such that the kernel of $\vf$ is the restriction 
of $\eta$ on $B''_2$.
\end{lemma}

\begin{proof}
Note that if $|C|=1$, the lemma is trivially true. Let $B'_2=\pr_2\rel'$. 
We assume $\beta_2\red{B'_2}\ne\ld\red{B'_2}$ for any 
congruence $\ld\le\beta_2$; otherwise replace $\beta_2$ with $\ld$.
Let  $\lnk'_1,\lnk'_2$ be the link congruences of $B'_1,B'_2$ with 
respect to $\rel'$. 
Let $\eta\le\beta_2$ be such that $\eta\red{\umax(B'_2)}\sse\lnk'_2$ and 
$\eta$ is maximal among congruences of $\zA_2$ with this property.
We show that either $\eta=\beta_2$ or it is one of the congruences in item (2) 
of the lemma. If $\eta$ is the full relation on $\umax(B'_2)$, we are done
by Proposition~\ref{pro:umax-rectangular}; otherwise there are two cases.

\medskip

{\sc Case 1.} For some $\th\in\Con(\zA_2)$ with $\eta\prec\th\le\beta_2$ 
the intervals $(\zz_1,\beta_1),(\eta,\th)$ can be separated. 

\medskip

In this case we prove that $\eta$ has to be $\beta_2$ and we have 
option~(1) of the lemma.
Since $\rel$ is chained, by Lemmas~\ref{lem:good-polys}(4) 
and~\ref{lem:collapsing} there is a set 
$T\sse B^*_1\tm B^*_1$ of $(\zz_1,\beta_1)$-subtraces such that any pair 
of elements from $\umax(B^*_1)$ belongs to the transitive closure of $T$, 
and for any $(a,b)\in T$ there is a $(B_1,B_2)$-preserving polynomial $f$ 
such that $f(a)=a, f(b)=b$, and $f(\th\red{B^*_2})\sse\eta$.
This means that $C$ belongs to the $\lnk^*_1$-block of 
$B^*_1$, where $\lnk^*_1$ is the link congruence with respect to 
$\rel^*\fac\eta$. 
Therefore $C\tm\umax(\rel^*[C])\fac\eta\sse\rel^*\fac\eta$. Observe
that as $\rel'\sse\rel^*$, the link congruence of $B^*_1$ with respect to
$\rel^*$ restricted to $C$ contains $\lnk'_1\red C$. Therefore, we also 
have $C\tm\umax(\rel^*[C])\sse\rel^*$. Note that by the assumption
$\rel'\cap\umax(\rel^*)\ne\eps$ of the lemma 
both $\rel^*[C]$ and $B'_2$ contain a u-maximal element from $B^*_2$.
Since $B''_2\sse\rel^*[C]$, by Lemma~\ref{lem:u-max-congruence} we have 
$\umax(B''_2)\sse\umax(\rel^*[C])$. Therefore $C\tm\umax(B''_2)\sse\rel^*$.

We are going to argue that the same inclusion holds for $\rel'$. But first 
we show that for any thin semilattice or affine edge $ab$ of $C$ 
and any $c\in\umax(\rel^*[C])$ there is a polynomial $g$ such that 
$g(a)=a, g(b)=b$, $f(\th\red{B^*_2})\sse\eta$, and $g(c)=c$. 
Note that since $\rel$ is chained, all such pairs $\{a,b\}$ belong to $T$.
Since every pair of elements of $C$
is a $(\zz_1,\beta_1)$-subtrace, again, as $\rel$ is chained, and 
by Lemma~\ref{lem:good-polys}(5) this is true for some 
$c\in\rel^*[C]$. Suppose $cc'$ is a thin edge in $\rel^*[C]$; 
by Lemma~\ref{lem:as-rectangularity} this implies that 
$(a,c),(b,c),(a,c'),(b,c')\in\rel$. Then as in 
Lemma~\ref{lem:maximal-minimal} we find a polynomial satisfying 
the required properties for $c'$. Specifically, $g'(x)=g(x)\cdot\cl a{c'}$, 
$g'(x)=t\left(g(x),\cl a{c'}\right)$, and 
$g'(x)=h'\left(g(x),\cl ac,\cl a{c'}\right)$, where $t$ and 
$h$ are the operations from Lemma~\ref{lem:op-s-on-affine}(3), 
depending on the type of $cc'$ and $ab$. 

Now we are back to proving that $C\tm\umax(B''_2)\sse\rel'$.
Observe that $\rel'\fac{\lnk'_1\tm\lnk'_2}$ is the graph of a bijective mapping 
$\vf: B'_1\fac{\lnk'_1}\to B'_2\fac{\lnk'_2}$. 
Take $a,b\in C$ and $c\in\umax(B''_2)$ 
such that $(a,c)\in\rel'$, $ab$ is a thin semilattice or affine edge and 
$(a,b)\not\in\lnk'_1$. Let also $(b,d)\in\rel'$.
By what is proved there is a polynomial $f$ of 
$\rel$ such that $f(a)=a,f(b)=b$, $f(c)=c\eqc\eta d'=f(d)$,
and $f(\th\red{B^*_2})\sse\eta$. In particular, $(a,d'),(b,d')\in\rel$.
Since $(a,c)\in\rel'$ and $(a,c)\le(b,c')$ or $(a,c)(b,c')$ is an affine edge for 
some $c'\in\Sg{c,d'}$, we obtain $(b,c')\in\rel'$, as $\rel'$ is polynomially 
closed. Since $c'\eqc\eta c$ and $\eta\le\lnk'_2$, we get a contradiction with 
$(a,b)\not\in\lnk'_1$.

\medskip

{\sc Case 2.}
For all $\th\in\Con(\zA_2)$ with $\eta\prec\th\le\beta_2$ the intervals 
$(\zz_1,\beta_1),(\eta,\th)$ cannot be separated.

\medskip

Suppose $\lnk'_2\red{B''_2}\not\sse\eta\red{B''_2}$. Without 
loss of generality let $\eta=\zz_2$. Then there are $a,b\in B''_2$ 
and $c\in C$ such that $(c,a),(c,b)\in\rel'$. Let $\th$ any congruence
with $\eta\prec\th\le\Cg{\eta\cup\{(a,b)\}}\le\beta_2$.
If $\typ(\zz_1,\beta)=\three$ then by Lemma~\ref{lem:34-links} 
such $a,b$ do not exists, as long as $(\zz_1,\beta_1),(\eta,\th)$ cannot be 
separated. Finally, if $\typ(\zz_1,\beta)\in\{\four,\five\}$, $C$ is a singleton
by Lemma~\ref{lem:type23}, and the 
result is trivial.

Suppose now that $\typ(\zz_1,\beta_1)=\typ(\zz_2,\th)=\two$. In this case $B^*_1$ is a module, $C=B^*_1$ implying $B'_2=B''_2$. Since 
$\rel$ is chained $a,b$ can be assumed to be from $\umax(B''_2)$, and 
so $\eta\red{\umax(B''_2)}<\lnk'_2\red{(B''_2)}$. Also, this implies by
Proposition~\ref{pro:umax-rectangular} that for any $\lnk'_1$-block $E$ 
and the corresponding $\lnk'_2$-block $E'$ it holds $E\tm\umax(E')\sse\rel'$. 
Since by the choice of $\eta$, $\lnk'_2\meet\th\not\le\eta$, pairs $(c,a),(c,b)$
can be chosen such that $a\eqc\th b$. We prove that 
$\th\red{\umax(B''_2)}\sse\lnk'_2$ producing a contradiction with the 
choice of $\eta$.

By Lemma~\ref{lem:delta-alignment} $\rel[B^*_1]$ is a subset of 
a $\zeta(\zz_2,\th)$-block. Then $\th\le\lnk''_2\join\eta$, where $\lnk''_2$ 
is the link congruence of $\zA_2$ with respect to $\rel$, 
and as $\rel$ is chained, by Lemma~\ref{lem:link-separability} 
$\th\red{\umax(E)}\le\lnk^*_2\red{B^*_2}$ for any 
$\lnk^*_2$-block $E\sse B^*_2$, since $B^*_2\fac{\lnk^*_2}$ is a 
module, and therefore $E$ is u-maximal in this set. Thus we are in
the conditions of Lemma~\ref{lem:dirprod-propagation}. Therefore if 
$(c,d)\in\rel^*$ then $(c,e)\in\rel^*$ for any $e\eqc\th d$ for any $c\in C$
and $d,e\in\umax(\rel[C])$.

Again, we now extend this property to $\rel'$ using the assumption that $\rel$ 
is polynomially closed.  
Since any pair $\{a',b'\}\sse B^*_2$ with $a'\eqc\th b'$ is a 
$(\eta,\th)$-subtrace, as $\rel$ is chained, there is a 
$(B_1,B_2)$-preserving polynomial $f$ such that
$f(a)=a'$ and $f(b)=b'$. Now, use the pairs $(c,a),(c,b)\in\rel'$. For any
$b'\in\umax(B''_2)$ with $b\eqc\th b'$, let $b''\in b^\th$ be such that
$h(b,b'',a)=b'$, where $h$ is the function from Theorem~\ref{the:uniform}; 
such $b''$ exists because $h(b,x,a)$ is a permutation on every $\th$-block
(recall that a $\th$-block is a module in this case). 
Since $\rel$ is chained, there is a polynomial $f$ such that
$f(a)=a, f(b)=b''$ and $f(c)=d$ for some $d\in B^*_1=C$. The mapping 
$g(x)=h\left(x,f(x),\cl da\right)$ is such that $g\cl ca=\cl ca$ and 
$g\cl cb=\cl c{b'}$, because, again, $B^*_1$ is a module. Since $\rel'$ is 
polynomially closed and 
$(c,b)\sqq_{as}(c,b')$ we have $(c,b')\in\rel'$; and as $b'$ is 
arbitrary from $a^\th$, we have $\{c\}\tm a^\th\sse\rel'$.
Thus, we have proved the property for a specific $\th$-block; next 
we extend it to other $\th$-blocks.

Suppose $\{c\}\tm E\sse\rel'$ for some $\th$-block $E$ and a 
$\th$-block $E'$ is such that for some $a\in E, b\in E'\cap B''_2$, 
$ab$ is a thin edge and $(d,b)\in\rel'$ for some $d\in C$. Then
by Lemma~\ref{lem:central-mapping} mapping $g(x)$ that is 
defined as $x\cdot\cl db$, $t_{ab}\left(x,\cl db\right)$, 
$h_{ab}\left(x,\cl ca,\cl db\right)$ depending on the type 
of $ab$ is injective on $E$. In particular, if $ab$ is semilattice or majority 
then $g$ maps $\{c\}\tm E$ to $\{c\}\tm E'$, $g(c,a)=(c,b), g(c,a')=(c,b')$ 
and $b\ne b'$ whenever $a\ne a'$; and
since $t_{ab},h_{ab}$ are term operations and all the tuples involved 
belong to $\rel'$, $(c,b),(c,b')\in\rel'$.
If $ab$ is affine then $g$ maps $\{c\}\tm E$ to $\{d\}\tm E'$, and 
$g(c,a)=(d,b), g(c,a')=(d,b')$ and 
$b\ne b'$ whenever $a\ne a'$, and $(d,b),(d,b')\in\rel'$. In either case,
$\lnk'_2$ is nontrivial on $E'$, and applying the argument from the
previous paragraph we obtain $\{c\}\tm E'\sse\rel'$ or 
$\{d\}\tm E'\sse\rel'$. Therefore 
$\th\red{\umax(B'_2)}\sse\lnk'_2\red{\umax(B'_2)}$,
a contradiction with the choice of $\eta$.
\end{proof}

\section{Chaining}\label{sec:chaining}

In this section we first introduce a property of relations which is slightly 
stronger than chaining; this is the property that will be used in further proofs.
Then we show in Lemma~\ref{lem:S7} that this 
property is preserved under certain transformations of the relation.

We call relation $\rel$ \emph{strongly chained}\index{strongly chained relation} 
with respect to $\ov\beta,\ov B$, where $\beta_i\in\Con(\zA_i)$ and $B_i$
is a $\beta_i$-block for $i\in[n]$, if\\[2mm]
(Q1s) for any $I\sse[n]$ and $\al,\beta\in\Con(\pr_I\rel)$ such that 
$\al\prec\beta\le\ov\beta_I$, $\al$ and $\beta$ are $\cU_B$-chained in 
$\rel$, where $\cU_B$ is the set of all $\ov B$-preserving polynomials of 
$\rel$\label{not:Q1s}\\[2mm]
(Q2s) for any $\al,\beta\in\Con(\pr_I\rel)$, $\gm,\dl\in\Con(\zA_j)$, 
$j\in[n]$, such that $\al\prec\beta\le\ov\beta_I$, $\gm\prec\dl\le\beta_j$, 
and $(\al,\beta)$ can be separated from $(\gm,\dl)$, the congruences 
$\al$ and $\beta$ are $\cU^*$-chained in $\rel$, where $\cU^*$ is the set 
of all $\ov B$-preserving polynomials $g$ of $\rel$ such that 
$g(\dl)\sse\gm$\label{not:Q2s}\\[2mm]
As in the definition of chained relations a polynomial from $\cU^*$ in 
condition (Q2s) will be called \emph{$(\gm,\dl,\ov B)$-good}.

We now can state and prove Lemma~\ref{lem:S7} that the property to be 
strongly chained is preserved 
under certain transformations of $\ov\beta$ and $\ov B$.
We will use it prove that one of the conditions, (S7), of a $\beta$-strategy 
(see Section~\ref{sec:strategies}) remains true when the $\beta$-strategy 
is being transformed.

\begin{lemma}\label{lem:S7}
Let $\rel$ be a subdirect product of $\vc\zA n$, $\beta_i\in\Con(\zA_i)$
and $B_i$ a $\beta_i$-block, $i\in[n]$, such that $\rel$ is strongly chained
with respect to $\ov\beta,\ov B$. Let $\rel'=\rel\cap(B_1\tms B_n)$
and $B'_i=\pr_i\rel'$. Fix $i\in[n]$, $\beta'_i\prec\beta_i$,
and let $D_i$ be a $\beta'_i$-block that is a member of a nontrivial 
as-component of $B'_i\fac{\beta'_i}$. Let also $\beta'_j=\beta_j$ 
and $D_j=B_j$ for $j\ne i$. Then
$\rel$ is strongly chained with respect to $\ov\beta',\ov D$.
\end{lemma}

\begin{proof}
Let $\rel''=\rel\cap(D_1\tms D_n)$ and $D'_i=\pr_i\rel''$. Take $I$, $j$ 
from the definition of being strongly chained. Let $I=[\ell]$; if $|I|>1$ we 
may consider $\rel$ as a subdirect product of $\pr_I\rel$ and 
$\zA_{\ell+1}\zd\zA_n$, so we assume $|I|=1$ and $j=n$ in (Q2s).
Let $\al,\beta\in\Con(\zA_1)$, $\gm,\dl\in\Con(\zA_n)$ be such that 
$\al\prec\beta\le\beta_1$, $\gm\prec\dl\le\beta_n$. Clearly, we may 
assume $\al=\zz_1$, $\gm=\zz_n$, and $\beta'_i=\zz_i$. Note 
that replacing $\rel$
with the $n+1$-ary relation $\{(\ba,\ba[i])\mid \ba\in\rel\}$ we 
may assume that $i\not\in I\cup\{j\}$. Without loss of generality assume 
$i=2$. By the assumption $\beta'_2=\zz_2$, the classes of $\beta'_2$ are
just elements of $\zA_2$, so let $B'_2$ be denoted by $c$. Let $C$ be the 
as-component of $B'_2$ containing $c$. 

To prove the lemma for every $a,b\in D'_1$ with $a\eqc\beta b$ we have
to identify a set $T(a,b,\gm,\dl,\cU^*)$ as in conditions (G1),(G2), and 
for every $\{a',b'\}\in T(a,b,\gm,\dl,\cU^*)$ we need to find a 
$(\gm,\dl,\ov D)$-good polynomial $f$ such that $f(a)=a',f(b)=b'$. In fact,
we rather find all the sets $T$ minimal among the sets of the form 
$T(a,b,\gm,\dl,\cU^*)$ and that satisfy the conditions of 
Lemma~\ref{lem:good-polys}(4). 
Note that such minimal sets exist for $\ov\beta,\ov B$, as well as, 
$(\gm,\dl,\ov B)$-good polynomials by the assumption that $\rel$ is
strongly chained with respect to $\ov\beta,\ov B$. We need to change such a set 
$T$ and change the polynomials so that they fit the new requirements. 
We divide the proof into two cases,
depending on whether or not $\relo=\pr_{12}\rel'$ is linked. First, we consider 
the case when $\relo$ is not linked, this case is relatively easy.

\medskip

{\sc Claim 1.}
Let $\relo'=\relo\cap(\umax(\pr_1\relo)\tm C)$ be not 
linked and $\lnk_1,\lnk_2$ link congruences of $\relo$.
Then $\lnk_2=\zz_2$ and either $\beta\le\lnk_1$ or $\beta$ is 
trivial on $D_1$.

\medskip

Relation $\relo$ is a subalgebra of $\rel\cap(B_1\tm B_2)$ and is 
polynomially closed in $\pr_{12}\rel$ by Lemma~\ref{lem:poly-closed}. By 
the Congruence Lemma~\ref{lem:affine-link} if $\relo'$ is not linked then 
$\relo$ is the graph of a mapping $\vf:\pr_1\relo\to C$. 
This means $\lnk_2=\zz_2$ and $\lnk_1$ is the restriction of a congruence 
$\eta$ of $\zA_1$ onto $\pr_1\relo$. If $\beta\le\eta$ then obtain the 
first option of the conclusion of the claim, otherwise $\lnk_1\cap\beta=\zz_1$ 
and we have the second option.

\medskip

Note that if $\beta\le\lnk_1$ then any $\ov B$-preserving polynomial
that maps a pair of $\beta$-related elements from $D'_1$ on a 
$(\al,\beta)$-subtrace from $D'_1$ is also $\ov D$-preserving, because 
$\lnk_2=\zz_2$; the result follows. If $\beta$ is trivial on $D'_1$, 
there is nothing to prove. Therefore we may assume $\relo'$ is linked.

We start with choosing a $\beta$-block required in the chaining conditions, 
and studying some of its properties. Observe that since $c$ is as-maximal 
in $B'_2$, the set $D'_1$ also contains as-maximal elements of $B'_1$. 
Therefore by Lemma~\ref{lem:u-max-congruence} 
$\umax(D'_1)\sse\umax(B'_1)$. Let $E$ be a $\beta$-block such that 
$E''=E\cap D'_1\ne\eps$, $E\cap\umax(D'_1)\ne\eps$ (and so $E''$ satisfies
the requirements of the chaining conditions), and let $E'=E\cap B'_1$. Consider 
$\rel^*=\rel'\cap(B_1\tm C\tm B_3\tms B_n)$. Note that 
$\rel^*$ is not necessarily a subalgebra. Let $B^*_i=\pr_i\rel^*$, 
$i\in[n]$, and $E^*=E\cap B^*_1$. By the Maximality 
Lemma~\ref{lem:to-max}(4) $\amax(E^*)$ is a union of as-components 
of $E'$. Indeed, let $a\in E^*$ and let $\ba\in\rel^*$ be such that 
$\ba[1]=a$ and $\ba[2]\in C$; let also $b\in E'$ with $a\sqq_{as}b$ in
$E'$. Then by the Maximality Lemma~\ref{lem:to-max}(4) there is 
$\bb\in\rel'$ such that $\bb[1]=b$ and $\ba\sqq_{as}\bb$ in $\rel'$.
In particular, $\ba[2]\sqq_{as}\bb[2]$ implying $\bb[2]\in C$.
Also, by Proposition~\ref{pro:umax-rectangular}, since $\relo$ is linked
and $\umax(E^*)\sse\umax(B'_1)$, we have 
$\umax(E^*)\tm C\sse\relo$, and therefore 
$\umax(E^*)=\umax(E'')\sse\umax(E')$. In particular, $\amax(E'')$ is a union 
of as-components of $E'$. The last inclusion here is because $E^*$ contains 
some as-maximal elements of $E'$.

First we prove condition (Q1s) for $\ov\beta'$ and $\ov D$.

\medskip

{\sc Claim 2.}
For any $a,b,a',b'\in E''$ such that $a,b$ belong to the same as-component of $E''$
there is a $(\gm,\dl,\ov D)$-good polynomial $f$ with 
$f(\{a',b'\})=\{a,b\}$.

\medskip

Consider relation $\rela$, a subdirect product of 
$\zA_1\tm\zA_1\tm\zA_2\tms\zA_n$, produced from by $(a',b',\ba)$, where 
$\ba\in\pr_{\{2\zd n\}}\rel''$, as follows: 
$$
\rela=\{f(a),f(b),f(\ba))\mid \text{$f$ is a unary polynomial of $\rel$
with $f(\dl)\sse\gm$}\}.
$$ 
It is not difficult to see that $\rela$ is a subalgebra, and, in particular it contains 
all the tuples of the form $(\bb[1],\bb[1],\bb[2]\zd\bb[n])$ for 
$\bb\in\rel$. Let $\rela'=\rela\cap\ov B$, and $\rela''=\rela\cap\ov D$. 
Every tuple from $\rela'$ or from $\rela''$ corresponds to a $\ov B$- or 
$\ov D$-preserving polynomial. Therefore it suffices to prove that 
$(a,b)\in\pr_{12}\rela''$. Let $F$ be the as-component of $E''$ containing 
$a,b$; as observed above $F$ is also an as-component of $E'$. 
By the assumption of (Q2s) $F^2\sse\pr_{12}\rela'$ and
$(e,e)\in\pr_{12}\rela''$ for any $e\in F$, since $F\tm C\sse\relo'$.
We consider relation $\relp=\pr_{123}\rela'$. As $F^2\sse\relp'=\pr_{12}\relp$,
$(a,b)$ is as-maximal in $\relp'$. Therefore it suffices to show
that $\amax(\relp)$ is linked when considered as subdirect product of $\relp'$
and $B'_2$. Since $(e,e)\in\pr_{12}\rela''$ for any $e\in F$, all pairs
of this form are linked in $\relp$. 
Then $(e,d,a'')\in\relp$ for any $e,d\in F$ and some $a''\in B'_2$,
and $(e,e,c'')\in\relp$ for some $c''\in C$. Since $F^2\sse\relp'$, 
$(e,e)\sqq_{as}(e,d)$, and by the Maximality Lemma~\ref{lem:to-max}(4) 
$a''$ can be chosen from $C$, and so this implies that $(e,d)$ and $(e,e)$ 
are also linked. Claim~2 is proved.

\medskip

Now we extend the result above to pairs from $\umax(E^*)$. We prove the result 
in two steps. First, we show that for any $a',b'\in E^*$ and
any $a,b\in\umax(E^*)$ there is a sequence of $\ov B$-preserving polynomials 
$\vc fk$ such that $f_1(\{a',b'\})\zd f_k(\{a',b'\})\sse E^*$ form a chain 
connecting $a$ and $b$, $f_i(\zA_1)$ is an $(\al,\beta)$-minimal set, and 
$f_i(c)\in C$ for $i\in[k]$. Then we prove 
that $\vc fk$ can be chosen in such a way that 
$f_1(\{a',b'\})\zd f_k(\{a',b'\})\sse E''$ and $f_1(c)=\dots=f_k(c)=c$.
Clearly, it suffices to prove in the case when $b$ is as-maximal in $E^*$.

By the assumption there are $a=a_1,a_2\zd a_k=b$, $\vc ak\in E'$ and\lb 
$(\gm,\dl,\ov B)$-good polynomials $\vc f{k-1}$ such that
$f_i(\zA_1)$ is a $(\al,\beta)$-minimal set and $f_i(\{a',b'\})=\{a_i,a_{i+1}\}$,
and also $f_i(c)\in B'_2$. We need to show that $\vc a{k-1}$ and $\vc f{k-1}$
can be chosen such that $f_i(c)\in C$. Choose $\ba,\bb\in\rel''$
such that $\ba[1]=a,\bb[1]=b$ and $\ba[2]=\bb[2]=c$. This is possible 
because $\umax(E^*)=\umax(E'')$. Now let 
$g_i(x)=\maj(\ba,f_i(x),\ba)$ and $h_i(x)=\maj(\ba,\bb,f_i(x))$. 
By Lemma~\ref{lem:good-polys} $g_i,h_i$ are 
$(\gm,\dl,\ov B)$-good polynomials, and for each of them
either $\{b_i,b_{i+1}\}=g_i(\{a',b'\})$ ($\{c_i,c_{i+1}\}=h_i(\{a',b'\})$) 
is an $(\al,\beta)$-subtrace, or $g_i(\beta)\sse\al$ ($h_i(\beta)\sse\al$), 
that is $g_i(a')=g_i(b')$ (respectively, $h_i(a')=h_i(b')$). The polynomials
$g_i,h_i$ satisfying the first option form a sequence of 
$(\al,\beta)$-subtraces connecting $a$ with $\maj(a,b,a)$ --- by subtraces
of the form $\{b_i,b_{i+1}\}$, --- and $\maj(a,b,a)$ with $\maj(a,b,b)$
--- by subtraces of the form $\{c_i,c_{i+1}\}$. Also, by 
Theorem~\ref{the:pseudo-majority} $\maj(a,b,b)$ belongs to the 
as-component of $E^*$ (and therefore of $E'$ and $E''$) containing $b$.
Therefore by Claim~2 this sequence of polynomials and subtraces can be 
continued to connect $\maj(a,b,b)$ to $b$. Finally, by the same theorem
$g_i(c)=\maj(c,f_i(c),c)\in C$ and $h_i(c)=\maj(c,c,f_i(c))\in C$.

For the second step we assume that $a$ and $b$ are connected with
$(\al,\beta)$-subtraces $\{a_i,a_{i+1}\}$, $i\in[k-1]$ witnessed by 
$(\gm,\dl,\ov B)$-good polynomials 
$f_i$ such that $c_i=f_i(c)\in C$. We need to show that $f_i$ can be chosen 
such that $f_i(c)=c$. Suppose that $c_i\ne c$ for some $i\in[k-1]$.
Since $c_i$ and $c$ belong to the same as-component, there is an
as-path $c_i=d_1\zd d_\ell=c$ in $C$. We show that if there is 
a sequence of $(\al,\beta)$-subtraces $\{b_j,b_{j+1}\}$ witnessed
by polynomials $g_j$ such that $g_j(c)=c$ whenever $f_j(c)=c$, 
and $f_i(c)=d_t$, there are also $(\al,\beta)$-subtraces $\{b'_j,b'_{j+1}\}$
such that $b'_1=a$ and $b'_k$ is in the as-component containing $b$,
witnessed by polynomials  $\vc{g'}k$ such that $g'_i(c)=d_{t+1}$
and $g'_j(c)=c$ whenever $g_j(c)=c$.

As is easily seen, it suffices to find a ternary term operation $p$
such that $p(a,a,b)$ belongs to the as-component containing $b$,
and $p(d_{t+1},d_t,d_t)=d_{t+1}$. Indeed, if such a term operation exists, 
then we set $g'_j(x)=p(\ba,\ba,g_j(x))$, where $\ba$ is as in the first step 
above, for $j\in[k-1]-\{i\}$, and 
$\{b'_j,b'_{j+1}\}=g'_j(\{a',b'\})$. We have $g'_1(a')=p(a,a,g_1(a'))=a$
and $g'_j(c)=p(c,c,g_j(c))=c$ whenever $g_j(c)=c$. Finally, since
$g'_\ell(b)=p(a,a,b)$ belongs to the as-component containing $b$,
we can use Claim~2 as before to connect $p(a,a,b)$ to $b$. For $g'_i$ we set
$g'_i(x)=p(\ba',\ba'',g_i(x))$ where $\ba',\ba''\in\rel''$ are such that
$\ba'[1]=\ba''[1]=a$ and $\ba'[2]=d_{t+1},\ba''[2]=d_t$. Note that 
such $\ba',\ba''$ exist, because $\umax(E^*)\tm C\sse\relo$. It follows 
from the assumption about $p$ that $g'_i$ is as required.

If $d_t\le d_{t+1}$, then $p(x,y,z)=z\cdot x$ fits the requirements. 
If $d_td_{t+1}$ is an affine edge, consider the relation 
$\rela\sse\zA_1\tm\zA_2$ generated by $\{(a,d_t),(a,d_{t+1}),(b,d_t)\}$.
Let $\zB=\Sg{a,b}$ and $\zC=\Sg{d_t,d_{t+1}}$; then 
$\zB\tm\{d_t\},\{a\}\tm\zC\sse\rela$. By Lemma~\ref{lem:as-rectangularity},
as $d_td_{t+1}$ is a thin affine edge, $\umax(\zB)\tm\{d_{t+1}\}\sse\rela$.
There is $b'$ with $b\sqq_{as}b'$ in $\zB$ such that $b'\in\umax(\zB)$.
Therefore there is a term operation $p$ with $p(a,a,b)=b'$  
and $p(d_{t+1},d_t,d_t)=d_{t+1}$, as required.
\end{proof}

\section{Strategies and solutions}\label{sec:maximal-solutions}

\subsection{The grand scheme}

In this section we describe the `grand scheme' of solving CSPs. We
start with introducing two preprocessing steps for our algorithm.

We call a CSP instance $\cP=(V,\cC)$ 
\emph{subdirectly irreducible}\index{subdirectly irreducible instance} if it is
1-minimal and $\zA_v$ is subdirectly irreducible for every $v\in V$.

\begin{lemma}[Folklore]\label{lem:sub-irreducible}
Every CSP instance can be reduced in polynomial time to an equivalent 
subdirectly irreducible one.
\end{lemma}

In this section all instances we consider are assumed subdirectly irreducible. The 
monolith of $\zA_v$ is denoted by $\mu_v$. 

Let $\cP=(V,\cC)$ be a (2,3)-minimal instance and 
for $X\sse V$, $|X|=2$, there is a constraint $C^X=\ang{X,\rel^X}$, where 
$\rel^X$ is the set of partial solutions on $X$. We use the notation from 
the end of Section~\ref{sec:csp-decomposition}. Recall that $\cW^\cP(\ov\beta)$
denotes the set of triples $(v,\al,\beta)$ 
such that $v\in V$, $\al,\beta\in\Con(\zA_v)$, and $\al\prec\beta\le\beta_v$. 
If $\beta_v=\zo_v$ for all $v\in V$, we set $\cW^\cP\cW^\cP(\ov\beta)$. 
Also, let $W_{v,\al\beta,\ov\beta}$ to be the set of $w\in V$ such that for some
$(w,\gm,\dl)\in\cW^\cP(\ov\beta)$ the prime intervals $(\al,\beta)$ and  $(\gm,\dl)$
cannot be separated in $\rel^{\{v,w\}}$.
Let $\ov\beta$, $\beta_v\in\Con(\zA_v)$, $v\in V$, be a collection of congruences.
Let $\cW'(\ov\beta)$\label{not:cWprime} (and respectively $\cW'$) denote the 
set of triples $(v,\al,\beta)\in\cW(\ov\beta)$ (respectively, from $\cW$) with 
$\zeta(\al,\beta)=\zo_v$. 

We say that algebra $\zA_v$ is 
\emph{semilattice free}\index{semilattice free algebra} if it does not contain 
semilattice edges. Let $\razm(\cP)$\label{not:size} denote the 
maximal size of domains of $\cP$ that are not semilattice free and 
$\Razm(\cP)$\label{not:MAX} be the set of variables $v\in V$ such that
$|\zA_v|=\razm(\cP)$ and $\zA_v$ is not semilattice free. Finally, for $Y\sse V$ 
let $\mu^Y_v=\mu_v$ if $v\in Y$ and $\mu^Y_v=\zz_v$ otherwise. 
Recall that by $\cP\fac{\ov\mu^Y}$ we denote
the instance $(V,\cC^{\ov\mu^Y})$ constructed as follows: the domain of 
$v\in V$ is $\zA_v\fac{\mu^Y_v}$; for every constraint $C=\ang{\bs,\rel}\in\cC$,
the set $\cC^{\ov\mu^Y}$ includes the constraint 
$\ang{\bs,\rel\fac{\ov\mu^Y_\bs}}$.

Instance $\cP$ is said to be \emph{block-minimal}\index{block-minimal instance} 
if for every $(v,\al,\beta)\in\cW$ (here $\beta_v=\zo_v$, $v\in V$)
\begin{itemize}
\item[(BM1)]
for every $C=\ang{\bs,\rel}\in\cC$ the problem $\cP_{W_{v,\al\beta,\ov\beta}}$
if $(v,\al,\beta)\not\in\cW'$, and the problem 
$\cP_{W_{v,\al\beta,\ov\beta}}\fac{\ov\mu^Y}$ otherwise, where 
$Y=\Razm(\cP)-\bs$, is minimal;
\item[(BM2)]
if $(v,\al,\beta)\in\cW'$, then for every $(w,\gm,\dl)\in\cW-\cW'$
the problem $\cP_{W_{v,\al\beta,\ov\beta}}\fac{\ov\mu^Y}$, where
$Y=\Razm(\cP)-(W_{v,\al\beta,\ov\beta}\cap W_{w,\gm\dl,\ov\beta})$
is minimal.
\end{itemize}
The definition of block-minimality is designed in such a way that 
block-minimality can 
be efficiently established. Observe that $W_{v,\al\beta,\ov\beta}$ can be
large, even equal to $V$. However if $(v,\al,\beta)\not\in\cW'$ by 
Theorem~\ref{the:centralizer-alignement} the problem 
$\cP_{W_{v,\al\beta,\ov\beta}}$ splits into a union of disjoint problems over 
smaller domains. On the other hand, if $(v,\al,\beta)\in\cW'$ then 
$\cP_{W_{v,\al\beta,\ov\beta}}$ may not be decomposable. Since we need 
an efficient procedure of establishing block-minimality, this explains the 
complications introduced in (BM1),(BM2).

For an instance $\cP$ we say that an instance $\cP'$ is \emph{strictly smaller} 
than instance $\cP$ if $\razm(\cP')<\razm(\cP)$. 

\begin{lemma}\label{lem:to-block-minimality}
Let $\cP=(V,\cC)$ be a (2,3)-minimal instance. 
Then $\cP$ can be transformed to an equivalent block-minimal 
instance $\cP'$ by solving a quadratic number of strictly smaller CSPs.
\end{lemma}

\begin{proof}
To establish block-minimality of $\cP$, for every 
$(v,\al,\beta)\in\cW$ (let $W=W_{v,\al\beta}$), 
we need to check if the problems given in conditions (BM1),(BM2) are minimal.
If they are then $\cP$ is block-minimal, otherwise some tuples can be 
removed from some constraint relation $\rel$ (the set of tuples that remain 
in $\rel$ is always a subalgebra, as is easily seen), and the instance $\cP$ 
tightened, in which case we need to repeat the procedure with the tightened 
instance. Therefore we just need to show how to reduce solving those subproblems 
to solving strictly smaller CSPs. 

For $C=\ang{\bs,\rel}\in\cC$ and $\ba\in\rel$ let $\cP'$ be the 
problem obtained as follows: fix the values of 
variables from $\bs\cap W$, or from $\bs\cap W\cap W_{w,\gm\dl}$
in the case of (BM2) to those of $\ba$. If the resulting problem is
$\cP''$ then set $\cP'=\cP''\fac{\ov\mu^Y}$, where $Y$ is either
empty, if $(v,\al,\beta)\not\in\cW'$, or $Y=\Razm(\cP)-\bs$, if 
$(v,\al,\beta)\in\cW'$ in (BM1), or $Y=\Razm(\cP)-(W\cap M_{w,\gm\dl})$  
in (BM2). In the first case, by Theorem~\ref{the:centralizer-alignement}
$\cP'$ is a disjoint union of instances $\vc\cP\ell$ and 
$\razm(\cP_i)<\razm(\cP)$. In the second case the domains of
variables from $\bs\cap W$ have cardinality 1, and the domain
of each of the remaining variables either is semilattice free, or is smaller 
than $\razm(\cP)$. Finally, in the last case the domain of each of the variables 
outside of $W\cap W_{w,\gm\dl}$ is either semilattice free or smaller than 
$\razm(\cP)$. Also, by Theorem~\ref{the:centralizer-alignement} 
$\cP_{W\cap W_{w,\gm\dl}}$ is a disjoint union of
instances with domains of smaller size. Let $\cP_1\zd\cP_k$ be these 
disjoint instances. Then $\cP'$ can be reduced to solving the instances
$\cP'_1\zd\cP'_k$ obtained from $\cP'$ by restricting 
$\cP'_{W\cap W_{w,\gm\dl}}$ to $\cP_i$. This completes the proof.
\end{proof}

Let $\cP=(V,\cC)$ be a subdirectly irreducible (2,3)-minimal instance. Let 
$\Centr(\cP)$\label{not:Centr} denote the set of variables $v\in V$ such 
that $\zeta(\zz_v,\mu_v)=\zo_v$. Let $\mu^*_v=\mu_v$ if 
$v\in\Razm(\cP)\cap\Centr(\cP)$ and $\mu^*_v=\zz_v$ otherwise. 

We consider several cases and indicate what kind of reductions or solution 
algorithms we intend to use in each case.

\paragraph{Case 1: Semilattice free domains.}
If all domains of $\cP$ are semilattice free then $\cP$ can be solved in 
polynomial time, using the few subpowers algorithm, as shown in 
\cite{Bulatov16:restricted}.

\paragraph{Case 2: Collapsing trivial centralizers.}
If $\mu^*_v=\zz_v$ for all $v\in V$, block-minimality guarantees that a 
solution exists, and we can use Lemma~\ref{lem:to-block-minimality} to solve
the instance. 

\begin{theorem}\label{the:non-central}
If $\cP$ is subdirectly irreducible, (2,3)-minimal, block-minimal, and 
$\Razm(\cP)\cap\Centr(\cP)=\eps$, then $\cP$ has a solution.
\end{theorem}

\paragraph{Case 3: Nontrivial centralizers.}
If $\Razm(\cP)\cap\Centr(\cP)=\eps$, we first solve the problem 
$\cP\fac{\mu^*}$, and then use Theorem~\ref{the:central} to reduce $\cP$ to
a strictly smaller instance. An efficient way to establish 1-minimality of 
$\cP\fac{\mu^*}$ is given in Theorem~\ref{the:algorithm}.

\begin{theorem}\label{the:central}
If $\cP\fac{\ov\mu^*}$ is 1-minimal, 
then $\cP$ can be reduced in polynomial time to a strictly smaller instance.
\end{theorem}

With the reductions above a solution algorithm goes as shown in 
Algorithm~\ref{alg:csp}, we reproduce it here for convenience.
\begin{algorithm}
\caption{Procedure {\sf SolveCSP}}
\label{alg:csp}
\begin{algorithmic}[1] 
\REQUIRE A CSP instance $\cP=(V,\cC)$ over $\cA$
\ENSURE A solution of $\cP$ if one exists, `NO' otherwise
\IF{all the domains are semilattice free}
\STATE Solve $\cP$ using the few subpowers algorithm and RETURN the answer
\ENDIF
\STATE Transform $\cP$ to a subdirectly irreducible, block-minimal and 
(2,3)-minimal instance
\STATE $\mu^*_v=\mu_v$ for $v\in\Razm(\cP)\cap\Centr(\cP)$ and 
$\mu^*_v=\zz_v$ otherwise
\STATE $\cP^*=\cP\fac{\ov\mu^*}$
\STATE /* the global 1-minimality of $\cP^*$
\FOR{every $v\in V$ and $a\in\zA_v\fac{\mu^*_v}$}
\STATE $\cP'=\cP^*_{(v,a)}$ \ \ \ \ /* Add constraint $\ang{(v),\{a\}}$
fixing the value of $v$ to $a$
\STATE Transform $\cP'$ to a subdirectly irreducible, (2,3)-minimal instance
$\cP''$
\STATE If $\razm(\cP'')<\razm(\cP)$ call {\sf SolveCSP} on $\cP''$ and 
flag $a$ if $\cP''$ has no solution
\STATE Establish block-minimality of $\cP''$; if the problem changes,
return to Step~10
\STATE If the resulting instance is empty, flag the element $a$ 
\ENDFOR
\STATE If there are flagged values, tighten the instance by removing the 
flagged elements and start over
\STATE Use Theorem~\ref{the:central} to reduce $\cP$ to an
instance $\cP'$ with $\razm(\cP')<\razm(\cP)$
\STATE Call {\sf SolveCSP} on $\cP'$ and RETURN the answer
\end{algorithmic}
\end{algorithm}

\begin{theorem}\label{the:algorithm}
Algorithm SolveCSP (Algorithm~\ref{alg:csp}) correctly solves every instance 
from $\CSP(\cA)$ and runs in polynomial time.
\end{theorem}

\begin{proof}
By the results of \cite{Bulatov16:restricted} the algorithm correctly solves
the given instance $\cP$ in polynomial time if the conditions of Step~1 are 
true. Lemma~\ref{lem:to-block-minimality} implies
that Steps~4 and~12 can be completed by recursing to strictly smaller 
instances. 

Next we show that the for-loop in Steps 8-14 checks if 
$\cP^*=\cP\fac{\ov\mu^*}$ is globally 1-minimal. For this we need to 
verify that a value $a$ is flagged if and only if $\cP^*$ has no solution 
$\vf$ with $\vf(v)=a$, and therefore if no values are flagged then 
$\cP^*$ is globally 1-minimal.
If $\vf(v)=a$ for some solution $\vf$ of $\cP^*$, then $\vf$
is a solution $\cP'$ constructed in Step~9. In this case Steps~11,12 cannot
result in an empty instance. 
Suppose $a\in\zA_v\fac{\mu^*_v}$ is not flagged. If
$\razm(\cP'')<\razm(\cP)$ this means that $\cP''$ and therefore $\cP'$
has a solution. Otherwise this means that establishing block-minimality of 
$\cP''$ is successful. In this case $\cP''$ has a solution by 
Theorem~\ref{the:non-central}, because $\Razm(\cP'')\cap\Centr(\cP'')=\eps$. 
This in turn implies that $\cP'$ has a solution.  
Observe also that the set of unflagged values for each variable $v\in V$
is a subalgebra of $\zA\fac{\mu^*}$. Indeed, the set of solutions of $\cP^*$
is a subalgebra $\cS^*$ of $\prod_{v\in V}\zA\fac{\mu^*}$, and the set of 
unflagged values is the projection of $\cS^*$ of the coordinate position $v$.

Finally, if Steps~8--15 are completed without restarts, Steps~16,17 
can be completed by Theorem~\ref{the:central} and recursing on $\cP'$ 
with $\razm(\cP')<\razm(\cP)$.

To see that the algorithm runs in polynomial time it suffices to observe that\\[1mm] 
(1) The number of restarts in Steps~4 and~15 is at most linear, as the 
instance becomes smaller after every restart; therefore the number of
times Steps~4--15 are executed together is at most linear.\\[1mm]
(2) The number of iterations of the for-loop in Steps 8--14 is linear.\\[1mm]
(3) The number of restarts in Steps~10 and~12 is at most linear, as the
instance becomes smaller after every iteration.\\[1mm]
(4) Every call of SolveCSP when establishing block-minimality in Steps~4, 
and~12 is made on an instance strictly smaller than $\cP$, and therefore
depth of recursion is bounded by $\razm(\cP)$ in Step~4,11,12 and~17. \\[2mm]
Thus a more thorough estimation gives a bound on the running time of 
$O(n^{3k})$, where $k$ is the maximal size of an algebra in $\cA$.
\end{proof}

\subsection{Proof of Theorem~\ref{the:central}}\label{sec:theorem-47}

Following \cite{Maroti10:tree} let $\cP=(V,\cC)$ be an instance and 
$p_v\colon\zA_v\to\zA_v$, $v\in V$. 
Mappings $p_v$, $v\in V$, are said to be 
\emph{consistent}\index{consistent mappings} if for any
$\ang{\bs,\rel}\in\cC$, $\bs=(\vc vk)$, and any tuple $\ba\in\rel$ the
tuple $(p_{v_1}(\ba[1])\zd p_{v_k}(\ba[k]))$ belongs to $\rel$. 
It is easy to see that the composition of two families of consistent mappings 
is also a consistent mapping. 
For consistent idempotent mappings $p_v$ by $p(\cP)$\label{not:p-cP} 
we denote the \emph{retraction}\index{retraction} of $\cP$, that is, $\cP$ 
restricted to the images of $p_v$. In this case $\cP$ has a solution if and only if 
$p(\cP)$ has, see \cite{Maroti10:tree}. 

Let $\vf$ be a solution of $\cP\fac{\ov\mu^*}$. We define 
$p^\vf_v:\zA_v\to\zA_v$ as follows: $p^\vf_v=q_v^k$, where 
$q_v(a)=a\cdot b_v$, element $b_v$ is any 
element of $\vf(v)$, and $k$ is such that $q_v^k$ is idempotent for all $v\in V$. 
Note that by Corollary~\ref{cor:centralizer-multiplication} this mapping is 
properly defined even if $\mu^*_v\ne\zz_v$.

\begin{lemma}\label{lem:consistent-mapping}
Mappings $p^\vf_v$, $v\in V$, are consistent. 
\end{lemma}

\begin{proof}
Take any $C=\ang{\bs,\rel}\in\cC$. Since $\vf$ is a solution of 
$\cP\fac{\ov\mu^*}$, there is $\bb\in\rel$ such that $\bb[v]\in\vf(v)$ 
for $v\in\bs$. Then for any $\ba\in\rel$, 
$q(\ba)=\ba\cdot\bb\in\rel$, and this product does not depend on the choice
of $\bb$, as it follows from Corollary~\ref{cor:centralizer-multiplication}. 
Iterating this operation also produces a tuple from $\rel$.
\end{proof}

We would like to use the above reduction to reduce $\cP$ to a problem $\cP'$
such that $\razm(\cP')<\razm(\cP)$. If $\vf$ is such that for $v\in\Razm(\cP)$ 
there is $a\in\zA_v$ with $a^{\mu^*_v}\le\vf(v)$ and $a\not\in\vf(v)$, then 
$|p^\vf_v(\zA_v)|<|\zA_v|$. Also, observe that if $|p^\vf_v(\zA_v)|=|\zA_v|$,
then $p^\vf_v$ is the identity mapping, that is $p^\vf_v(\zA_v)=\zA_v$. 
If $\zA_v$ is semilattice free then $p^\vf_v$ is the identity mapping by
Proposition~\ref{pro:good-operation}.
Let $V^*$ be the set of variables $v\in V$ such that $\zA_v\fac{\mu^*_v}$ 
is not semilattice free. 

\begin{lemma}\label{lem:good-retraction}
There are consistent mappings $p_v$, $v\in V$, such that 
for any $v\in V^*$ we have $|p_v(\zA_v)|<|\zA_v|$. Moreover, such  
mappings can be found solving a linear number of instances of the form
$(\cP_{(v,a^{\mu^*_v})})\fac{\ov\mu^*}$.
\end{lemma}

\begin{proof}
Since $\cP\fac{\ov\mu^*}$ is globally 1-minimal, for any 
$a\in\zA_v\fac{\mu^*_v}$ there is a solution $\vf$ with $\vf(v)=a$, and it
can be found solving the instance $(\cP_{(v,a^{\mu^*_v})})\fac{\ov\mu^*}$. 
For every $v\in V^*$ choose $a\in\zA_v$ such that there is $b\in\zA_v$ and 
$b\le a$, $a\not\eqc{\mu^*_v}b$, and let $\vf_v$ be a solution of 
$\cP\fac{\ov\mu^*}$ with $\vf_v(v)=a^{\mu^*_v}$. Then 
$|p^{\vf_v}_v(\zA_v)|<|\zA_v|$ and $|p^{\vf_v}_w(\zA_w)|<|\zA_w|$ or
$p^{\vf_v}_w$ is the identity mapping for any $w\in V^*$. 
Therefore the composition of the $p^{\vf_w}$ for all $w\in V^*$ is as 
required.
\end{proof}  

Theorem~\ref{the:central} now follows by observing that if 
$\zA_v\fac{\mu^*_v}$ is semilattice free then $\zA_v$ itself is semilattice 
free.

In order to use Theorem~\ref{the:central} we however need to argue that
$p(\cP)$ is a problem over a class of algebras omitting type~\one. Let
$f$ be a weak near-unanimity term of the class $\cA$. Then $p\circ f$
is a weak near-unanimity term of $p(\cA)=\{p(\zA)\mid \zA\in\cA\}$.
Moreover, if $\zA$ is semilattice free then $p(\zA)=\zA$.

\subsection{Strategies}\label{sec:strategies}

In this section similar to strategies related to the concepts of consistency and
minimality we introduce strategies of some sort that will be used to prove 
Theorem~\ref{the:non-central}. We start with some necessary definitions.

Let $\cP=(V,\cC)$ be a (2,3)-minimal and block-minimal instance over $\cA$.
For $(v,\al,\beta)\in\cW$, if 
$(v,\al,\beta)\not\in\cW'$, then $\cS_{W_{v,\al\beta,\ov\beta}}$ denotes the 
set of solutions of $\cP_{W_{v,\al\beta,\ov\beta}}$, and if $(v,\al,\beta)\in\cW'$,
then $\cS_{W_{v,\al\beta,\ov\beta},Y}$ denotes the 
set of solutions of $\cP_{W_{v,\al\beta,\ov\beta}}\fac{\ov\mu^Y}$ for an
appropriate $Y$.

Let $\beta_v\in\Con(\zA_v)$ and let $B_v$ be a $\beta_v$-block, 
$\ov\beta=(\beta_v\mid v\in V)$, $\ov B=(B_v\mid v\in V)$.  
Let $\cR=\{\rel_{C,v,\al\beta}\mid 
C=\ang{\bs,\rel}\in\cC, (v,\al,\beta)\in\cW(\ov\beta)\}$ be a collection 
of relations such that $\rel_{C,v,\al\beta}$ is a subalgebra of 
$\pr_{\bs\cap W_{v,\al\beta,\ov\beta}}\rel$. Let $C=\ang{\bs,\rel}$, 
$(v,\al,\beta)\in\cW$, and $W=W_{v,\al\beta,\ov\beta}$. Let $\ba$ be a tuple 
from $\pr_X\rel$ for $X\sse\bs$, or from $\pr_X\cS_W$, $X\sse W$, if 
$(v,\al,\beta)\not\in\cW'$, or from $\pr_X\cS_{W,Y}$ if $(v,\al,\beta)\in\cW'$, 
where $X\sse W$ and $Y$ is a set specified in the condition of block-minimality. 
Tuple $\ba$ is said to be \emph{$\cR$-compatible}\index{$\cR$-compatible tuple} 
if for any $(w,\gm,\dl)\in\cW(\ov\beta)$, (let $U=W_{w,\gm\dl,\ov\beta}$)
$\pr_{X\cap U}\ba\in\pr_{X\cap U}\rel_{C,w,\gm\dl}$ or 
$\pr_{X\cap U}\ba\in\pr_{X\cap U}\rel_{C,w,\gm\dl}\fac{\ov\mu^Y}$
for an appropriate set $Y$. By 
$\rel^\cR,\cS^\cR_W,\cS^\cR_{W,Y}$\label{not:S-R-S-R-prime} we
denote the set of all $\cR$-compatible tuples from the corresponding 
relation. Also, let $\cP^\cR=(V,\cC^\cR)$\label{not:P-cR} denote the 
problem instance obtained from $\cP$ replacing every constraint 
$\ang{\bs,\rel}\in\cC$ with $\ang{\bs,\rel^\cR}$.

The collection $\cR$ 
is called a \emph{$\ov\beta$-strategy}\index{$\ov\beta$-strategy} with 
respect to $\ov B$ if it satisfies the following conditions for every 
$(v,\al,\beta)\in\cW(\ov\beta)$, 
and every $C=\ang{\bs,\rel}\in\cC$ (let $W=W_{w,\al\beta,\ov\beta}$):
\begin{itemize}
\item[(S1)]
the relations $\umax(\rel^{X,\cR})$, where $\rel^{X,\cR}$ consists 
of $\cR$-compatible tuples from $\rel^X$ 
for $X\sse V$, $|X|\le 2$, form a nonempty $(2,3)$-strategy for $\cP^\cR$; 
\item[(S2)]
for every $(w,\gm,\dl)\in\cW(\ov\beta)$ (let $U=W_{w,\gm\dl}$)
and every $\ba\in\umax(\pr_{\bs\cap W\cap U}\rel_{C,v,\al\beta})$ it holds:
if $(w,\gm,\dl)\not\in\cW'$ then $\ba$ extends
to an $\cR$-compatible solution $\vf$ of $\cP_U$; otherwise
if $(v,\al,\beta)\not\in\cW'$ then $\ba$ extends to an $\cR$-compatible 
solution of $\cP_U\fac{\ov\mu^{Y_1}}$ with $Y_1=\Razm(\cP)-(W\cap U)$; 
and if $(v,\al,\beta)\in\cW'$ then $\ba$ extends to an $\cR$-compatible 
solution of $\cP_U\fac{\ov\mu^{Y_2}}$, where $Y_2=\Razm(\cP)-\bs$;
\item[(S3)]
$\rel\cap\ov B_\bs\ne\eps$ and for any $I\sse\bs$ any $\cR$-compatible 
tuple $\ba\in\umax(\pr_I\rel)$ extends to an $\cR$-compatible tuple 
$\bb\in\rel$.
\item[(S4)]
the relation $\rel_{C,v,\al\beta}$ is a subalgebra of $\pr_{\bs\cap W}\rel$, and
$\umax(\rel_{C,v,\al\beta})\sse\umax(\pr_{\bs\cap W}\rel)$;
if $(v,\al,\beta)\not\in\cW'$ then the relation $\cS^\cR_W$
is a subalgebra of $\cS_W$, and $\umax(\cS^\cR_W)\sse\umax(\cS_W)$;
if $(v,\al,\beta)\in\cW'$ then for any $(w,\gm,\dl)\in\cW(\ov\beta)-\cW'$ the
relations $\cS^\cR_{W,Y_1},\cS^\cR_{W,Y_2}$ are subalgebras of 
$\cS_{W,Y_1},\cS_{W,Y_2}$, respectively, and 
$\umax(\cS^\cR_{W,Y_1})\sse\umax(\cS_{W,Y_1})$, 
$\umax(\cS^\cR_{W,Y_2})\sse\umax(\cS_{W,Y_2})$, where
$Y_1=\Razm(\cP)-\bs$ and $Y_2=\Razm(\cP)-(W\cap W_{w,\gm\dl})$;
\item[(S5)]
for every $w\in\bs$ and every $(w,\gm,\dl)\in\cW(\ov\beta)$ 
(let $U=W_{w,\gm\dl}$) with $w\in\bs\cap U$ it holds 
$\umax(\pr_w\rel_{C,w,\gm\dl})=\umax(\pr_w\rel_{C,v,\al\beta})$, let 
$A_{\cR,w}$ denote the subalgebra generated by this set, 
$\umax(A_{\cR,w})$ is as-closed in $\umax(\pr_w(\rel\cap\ov B))$;
\item[(S6)]
for every $(w,\gm,\dl)\in\cW(\ov\beta)$ with $\bs\cap W_{w,\gm\dl}\ne\eps$ 
the set of $\cR$-compatible tuples from $\rel_{C,w,\gm\dl}$ is polynomially 
closed in $\pr_{\bs\cap W_{w,\gm\dl}}\rel$;
\item[(S7)]
relation $\rel$ is strongly chained with respect to $\ov\beta,\ov B$; 
if $(v,\al,\beta)\not\in\cW'$, relation $\cS_W$ is strongly chained with 
respect to $\ov\beta,\ov B$; if $(v,\al,\beta)\in\cW'$, for any 
$(w,\gm,\dl)\in\cW(\ov\beta)-\cW'$ the relations $\cS_{W,Y_1},\cS_{W,Y_2}$,
$Y_1=\Razm(\cP)-\bs$, $Y_2=\Razm(\cP)-(W\cap W_{w,\gm\dl})$, are 
strongly chained with respect to $\ov\beta,\ov B$.
\end{itemize}

Conditions (S1)--(S3) are the conditions we actually want to maintain when
transforming a strategy, and these are the ones that provide the desired results.
However, to prove that (S1)--(S3) are preserved under transformations of 
a strategy we also need more technical conditions (S4)--(S7).

We now show how we plan to use $\ov\beta$-strategies.
Let $\cP$ be a subdirectly irreducible, (2,3)-minimal, and block-minimal instance, 
$\beta_v=\zo_v$ and $B_v=\zA_v$ for $v\in V$. Then as is easily seen the 
collection of relations 
$\cR=\{\rel_{C,v,\al\beta}\mid (v,\al,\beta)\in\cW(\ov\beta),C\in\cC\}$ 
given by $\rel_{C,v,\al\beta}=\pr_{\bs\cap W_{v,\al\beta,\ov\beta}}\rel$ for 
$C=\ang{\bs,\rel}\in\cC$ is a $\ov\beta$-strategy with respect to $\ov B$.
Also, by (S3) a $\ov\gm$-strategy with $\gm_v=\zz_v$ for all $v\in V$ gives a 
solution of 
$\cP$. Our goal is therefore to show that a $\ov\beta$-strategy for any
$\ov\beta$ can be `reduced', that is, transformed to a $\ov\beta'$-strategy 
for some $\ov\beta'<\ov\beta$. Note that this reduction of strategies is where 
the condition $\Razm(\cP)\cap\Centr(\cP)=\eps$ is used. Indeed, suppose that 
$\beta_v=\mu^*_v$. Then by conditions (S1)--(S7) we only have information 
about solutions to problems of the form $\cP_W\fac{\ov\mu^*}$ or something 
very close to that. Therefore this barrier cannot be penetrated. We  consider 
two cases.

\medskip
{\sc Case 1.} There are $v\in V$ and $\al\prec\beta_v$ nontrivial on $B_v$, 
$\typ(\al,\beta_v)=\two$. This case is considered in 
Section~\ref{sec:type-2}.

\smallskip
{\sc Case 2.} For all $v\in V$ and $\al\prec\beta_v$ nontrivial on $B_v$, 
$\typ(\al,\beta_v)\in\{\three,\four,\five\}$. This case is considered 
in Section~\ref{sec:type-not-2}.

\section{Proof of Theorem~\ref{the:non-central}}\label{sec:affine-consistency}

In the remaining part of the paper we prove Theorem~\ref{the:non-central}.

\subsection{Tightening affine factors}\label{sec:type-2}

In this section we consider Case~1 of tightening strategies: there is 
$\al\in\Con(\zA_v)$ for some $v\in V$ such that $\al\prec\beta_v$ 
and $\typ(\al,\beta_v)=\two$.

\subsubsection{Transformation of the strategy and induced congruences}%
\label{sec:transformation1}

Let $\cP=(V,\cC)$ be a block-minimal instance with subdirectly irreducible 
domains, $\ov\beta=(\beta_v\in\Con(\zA_v)\mid v\in V)$ and 
$\ov B=(B_v\mid B_v \text{ is a $\beta_v$-block, } v\in V)$. 
We use notation from Section~\ref{sec:maximal-solutions}.
Let also $\cR=\{\rel_{C,v,\al\beta}\}$ be a $\ov\beta$-strategy for 
$\ov B$. We select $v\in V$ and $\al,\beta\in\Con(\zA_v)$ with 
$\al\prec\beta=\beta_v$, $\typ(\al,\beta)=\two$, and 
an $\al$-block $B\in B_v\fac\al$. Note that since $\typ(\al,\beta)=\two$, 
$B_v\fac\al$ is a module, and therefore $B'_v$ is as-maximal in this
set. In this section we show how $\cR$
can be transformed to a $\ov\beta'$-strategy $\cR'$ for $\ov B'$ such
that $\beta'_w\le\beta_w$, $B'_w\sse B_w$ for $w\in V$, and 
$\beta'_v=\al$, $B'_v=B$.

First of all we identify variables $w\in V$ for which $\beta'_w$ has to be
different from $\beta_w$.
Since $\cP$ is (2,3)-minimal, for every $u,w\in V$ there is 
$C^{\{u,w\}}=\ang{(u,w),\rel^{\{u,w\}}}\in\cC$. 
For $w\in W_{v,\al\beta}$ (we omit $\ov\beta$ from 
$W_{v,\al\beta,\ov\beta}$ here)  consider 
$\rel^{*,\{v,w\}}=\rel^{\{v,w\}}\cap(B_v\tm B_w)$,
$\rel^{\{v,w\},\cR}$ the set of all $\cR$-compatible pairs from 
$\rel^{\{v,w\}}$, $\rel'^{\{v,w\}}=\rel^{\{v,w\}}\fac\al$, and
$\rel'^{\{v,w\},\cR}=\rel^{\{v,w\},\cR}\fac\al$. 
By (S5) for $\cR$ we have that $\umax(\pr_v\rel'^{\{v,w\},\cR})$ is as-closed in 
$B^*_v\fac\al$, where $B^*_v=\pr_v\rel^{*,\{v,w\}}$; since 
$\typ(\al,\beta)=\two$, this implies $\pr_v\rel'^{\{v,w\},\cR}=B^*_v\fac\al$. 
Also, $\pr_v\rel^{\{v,w\},\cR}$, and therefore $\pr_v\rel'^{\{v,w\},\cR}$
are polynomially closed in $\pr_v\rel^{\{v,w\}},\pr_v\rel'^{\{v,w\}}$, 
respectively, by (S6), and $\pr_v\rel^{\{v,w\}},\pr_v\rel'^{\{v,w\}}$ are 
strongly chained by (S7). Therefore, by the Congruence 
Lemma~\ref{lem:affine-link} either 
$B^*_v\fac\al\tm\umax(\pr_w(\rel'^{\{v,w\},\cR})\sse\rel'^{\{v,w\},\cR}$
or $\rel'^{\{v,w\},\cR}$ is the graph of a mapping 
$\nu_w:\pr_w\rel'^{\{v,w\},\cR}\to B^*_v\fac\al$. 
Let $U\sse W_{v,\al\beta}$ be the set of variables for which 
the latter holds, and let $\al_w$ be the corresponding congruence of 
$\zA_w$, extension of the kernel of $\nu_w$.
Let $\beta'_v=\al, B'_v=B$ and $\beta'_w=\al_w, B'_w=\nu^{-1}_w(B)$ 
for $w\in U$, and $\beta'_w=\beta_w, B'_w=B_w$ for $w\in V-U$. 

Now we are in a position to define the new strategy. Let $\cR'$ be the following 
collection of relations. We omit subscript $\ov\beta$.
\begin{itemize}
\item[(R1)]
$\cR'=\{\rel'_{C,w,\gm\dl}\mid C=\ang{\bs,\rel}\in\cC, 
(w,\gm,\dl)\in\cW(\ov\beta')\}$;
\item[(R2)]
for every $C=\ang{\bs,\rel}\in\cC$, $(u,\gm,\dl)\in\cW(\ov\beta')$, 
\begin{itemize}
\item[(a)] if $(v,\al,\beta)\not\in\cW'$, $\rel'_{C,u,\gm\dl}=
\{\ba\in\rel_{C,u,\gm\dl}\mid$ there is a $\cR$-compatible solution $\vf$ of 
$\cP_{W_{v,\al\beta_v}}$, $\vf(v)\in B'_v$,  and
$\vf(w)=\ba[w]$ for 
$w\in\bs\cap W_{v,\al\beta_v}\cap W_{u,\gm\dl}\}$; 
\item[(b)]
if $(v,\al,\beta)\in\cW',(u,\gm,\dl)\not\in\cW'$, $\rel'_{C,u,\gm\dl}=
\{\ba\in\rel_{C,u,\gm\dl}\mid$ there is a $\cR$-compatible solution $\vf$ of 
$\cP_{W_{v,\al\beta_v}}\fac{\ov\mu^Y}$ with 
$Y=\Razm(\cP)-W_{u,\gm\dl}$, such that $\vf(v)\in B'_v\fac{\mu^Y_v}$, and
$\vf(w)=\ba[w]$ for $w\in\bs\cap W_{v,\al\beta_v}\cap W_{u,\gm\dl}\}$
\item[(c)] 
if $(v,\al,\beta),(u,\gm,\dl)\in\cW'$, $\rel'_{C,u,\gm\dl}=
\{\ba\in\rel_{C,u,\gm\dl}\mid$ there is a $\cR$-compatible solution $\vf$ of 
$\cP_{W_{v,\al\beta_v}}\fac{\ov\mu^Y}$ with 
$Y=\Razm(\cP)-\bs$, such that $\vf(v)\in B'_v\fac{\mu^Y_v}$, and
$\vf(w)=\ba[w]$ for $w\in\bs\cap W_{v,\al\beta_v}\cap W_{u,\gm\dl}\}$; 
\end{itemize}
\end{itemize}

Similar to $\rel^\cR,\cS^\cR_W,\cS^\cR_{W,Y}$ by 
$\rel^{\cR'},\cS^{\cR'}_W,\cS^{\cR'}_{W,Y}$ we denote the corresponding 
sets of $\cR'$-compatible tuples.
As is easily seen, the sets of both types are indeed subalgebras of 
$\rel,\cS_W,\cS_{W,Y}$. 

The following three statements show how relations $\rel'_{C,w,\gm\dl}$ 
from $\cR'$ are related to $\rel_{C,w,\gm\dl}$ from $\cR$. They amount 
to saying that either $\rel'_{C,w,\gm\dl}$ is the intersection of 
$\rel_{C,w,\gm\dl}$ with a block of a congruence of the projection of $\rel$, or 
$\umax(\rel'_{C,w,\gm\dl})=\umax(\rel_{C,w,\gm\dl})$. Recall that for 
congruences $\beta_w$, $w\in V$, and $U\sse V$ by $\ov\beta_U$ we 
denote the collection $(\beta_w)_{U\sse V}$.
Set $W=W_{v,\al\beta}$ for the rest of Section~\ref{sec:type-2}.

\begin{lemma}\label{lem:congruence-restriction}
Let $C=\ang{\bs,\rel}\in\cC$, and let 
$\cS'_W$ be the set of solutions of 
$\cP_W$ if $(v,\al,\beta)\not\in\cW'$, or the set of solutions of 
$\cP_W\fac{\ov\mu^{\Razm(\cP)-\bs}}$ if $(v,\al,\beta)\in\cW'$. 
For every $U\sse \bs\cap W$ there is a congruence $\tau_U$\label{not:tau-U} 
of $\pr_U\cS'_W=\pr_U\rel$ satisfying the following conditions:\\[1mm]
(1) either $\umax(\pr_U\cS'^{\cR'}_W)=\umax(\pr_U\cS'^\cR_W)$ or for a 
$\tau_U$-block $\rela$ it holds 
$\pr_U\cS'^{\cR'}_W=\pr_U\cS'^\cR_W\cap\rela$. \\[1mm]
(2) For any $U_1\sse U_2\sse W$ the congruence $\tau_{U_1}$ is the restriction 
of $\tau_{U_2}$, that is $(\ba,\bb)\in\tau_{U_1}$ if and only if for some
$\ba',\bb'\in\cS'_{U_2}$ with $\pr_{U_1}\ba'=\ba,\pr_{U_1}\bb'=\bb$ it
holds $(\ba',\bb')\in\tau_{U_2}$.\\[1mm]
(3) For any $U\sse \bs\cap W$ either 
$\tau_U\red{\pr_U\rel\cap\ov B}=\ov\beta_U\red{\pr_U\rel\cap\ov B}$, 
or the algebra $\pr_U\cS'^{\cR}_W\fac{\tau_U}$ is isomorphic to 
$\pr_v(\cS'_W\cap\ov B)\fac\al$.\\[1mm]
(4) For any $(w,\gm,\dl)\in\cW(\ov\beta)$, $X=W_{w,\gm\dl}$, 
$X'=\bs\cap W\cap X$, let $\tau=\tau_{X'}$. Then either 
$\umax(\pr_{X'}\rel'_{C,w,\gm\dl})=\umax(\pr_{X'}\rel_{C,w,\gm\dl})$, 
or for a $\tau$-block $\rela$ it holds 
$\pr_{X'}\rel'_{C,w,\gm\dl}\sse\pr_{X'}\rel_{C,w,\gm\dl}\cap\rela$ and 
$\umax(\pr_{X'}\rel'_{C,w,\gm\dl})$ is the set of u-maximal elements of  
$\umax(\pr_{X'}\rel_{C,w,\gm\dl})\cap \rela$.
\end{lemma}

If, according to item (3) of the lemma, 
$\tau_U\red{\pr_U\rel\cap\ov B}=\ov\beta_U\red{\pr_U\rel\cap\ov B}$,
we say that $\tau_U$ is the \emph{full congruence}\index{full congruence};
if the latter option of item (3) holds we say that $\tau_U$ is a 
\emph{maximal congruence}\index{maximal congruence}.

\begin{proof}
If $v\in U$ then set $\tau_U$ to be $\ov\beta'_U$. Otherwise consider 
$\relo=\pr_{U\cup\{v\}}\cS'_W$ as a subdirect product of $\zA_v$ and
$\pr_U\cS'_W$. This relation is chained by (S7) and 
$\pr_{U\cap\{v\}}\cS'^{\cR}_W$ is polynomially 
closed in $\relo$ by Lemma \ref{lem:poly-closed}(2); apply the Congruence 
Lemma~\ref{lem:affine-link} to it. Specifically, consider $\relo\fac\al$ as a 
subdirect product of $\pr_U\cS'_W$ and $\zA_v\fac\al$. If the first option of the 
Congruence Lemma~\ref{lem:affine-link} holds, set $\tau_U=\ov\beta_U$. If the 
second option is the case, choose $\tau_U$ to the congruence $\eta$ of 
$\pr_U\cS'_W$ identified in the Congruence Lemma~\ref{lem:affine-link}.

(1) In this case the result follows by the Congruence Lemma~\ref{lem:affine-link}.

(2) Obvious.

(3) If $v\in U$ then by item (1) 
$\pr_U\cS'^{\cR}_W\fac{\tau_U}=\pr_U\cS'^{\cR}_W\fac{\ov\beta'_U}$,
which is isomorphic to $\pr_v(\cS'_W\cap\ov B)\fac\al$. Otherwise consider 
relation $\relo$ as in the beginning of the proof. By the Congruence 
Lemma~\ref{lem:affine-link} if $\tau_U\ne\ov\beta_U$, 
$\tau_U\prec\ov\beta_U$. 
The result follows.

(4) If $\tau$ is the full congruence then by (S2) for $\cR$ we have 
$\umax(\pr_X\rel_{C,w,\gm\dl})=\umax(\pr_X\cS'^{\cR'}_W)$ and we have 
the first option. If $\tau$ is a maximal congruence then by (R2) and item (1)
there is a $\tau$-block $\rela$ such that 
$\pr_{X'}\rel'_{C,w,\gm\dl}\sse\rela\cap\pr_{X'}\rel_{C,w,\gm\dl}$. By
condition (S2) every $\ba\in\umax(\pr_{X'}\rel_{C,w,\gm\dl})$ can be extended
to a solution from $\cS'^\cR_W$. In particular, if 
$\ba\in\umax(\pr_{X'}\rel'_{C,w,\gm\dl})\sse\rela$ then such an extension 
$\vf$ has to satisfy $\vf(v)\in B'_v$. Since $S$ is as-maximal in
$\pr_{X'}(\rel_{C,w,\gm\dl})\fac\tau$ and 
$\umax(\pr_{X'}\rel'_{C,w,\gm\dl})\sse \umax(\pr_{X'}\rel_{C,w,\gm\dl})$,
we have $\umax(\pr_{X'}\rel'_{C,w,\gm\dl})=
\umax(\rela\cap\pr_{X'}\rel_{C,w,\gm\dl})$.
\end{proof}

For $C=\ang{\bs,\rel}\in\cC$ we use $\tau_C$\label{not:tau-C} to denote 
the congruence 
$\tau_{\bs\cap W}$. Also for $(w,\gm,\dl)\in\cW(\ov\beta)$ we use
$\tau_{C,w,\gm\dl}$ to denote the congruence 
$\tau_{\bs\cap W\cap W_{w,\gm\dl}}$. 

\begin{lemma}\label{lem:restricted-congruence}
In the notation above let $\gm_u,\dl_u\in\Con(\zA_u)$, $u\in W'=\bs\cap W$ 
be such that $(u,\gm_u,\dl_u)\in\cW(\ov\beta)$ and 
$(\al,\beta_v),(\gm_u,\dl_u)$ cannot be separated from each other. Then 
if $\tau_C$ is a maximal congruence, for any polynomial $f$ of $\rel$, 
$f(\ov\beta_{W'})\sse\tau_C$ if and only if $f(\dl_u)\sse\gm_u$ for any 
$u\in W'$. If $\gm_u,\dl_u$ are considered as congruences of $\pr_{W'}\rel$,
equal to $\gm_u\tm\prod_{x\in W'-\{u\}}\zo_x$,  
$\dl_u\tm\prod_{x\in W'-\{u\}}\zo_x$, respectively, this condition means
that $(\tau_C,\ov\beta_{W'})$ and $(\gm_u,\dl_u)$ cannot be separated.
\end{lemma}

\begin{proof}
Let $\cS'_W$ be defined as in Lemma~\ref{lem:congruence-restriction}
and $\tau_C$ a maximal congruence.
Take a polynomial $f$ of $\rel$. 
As $(\gm_{u_1},\dl_{u_1}),(\gm_{u_2},\dl_{u_2})$ cannot be separated
for any $u_1,u_2\in W'$, it suffices to consider just one variable $u\in W'$. 
Since $\cP$ is a block-minimal instance, the polynomial $f$ can be extended
from a polynomial on $\pr_{W'}\rel$ to a polynomial $f'$ of $\cS'_W$, and, 
in particular, to a polynomial $f''$ of $\pr_{W'\cup\{v\}}\cS'_W$. 
Since $\tau_C$ is maximal, by the Congruence Lemma~\ref{lem:affine-link}
the intervals $(\al,\beta_v)$ and $(\tau_C,\ov\beta_{W'})$ in the congruence
lattices of $\zA_v$ and $\pr_{W'}\rel$, respectively, cannot be 
separated in $\pr_{W'\cup\{v\}}\cS'_W$. 
Therefore $f''(\beta_v)\sse\al$ if and only if $f(\ov\beta_{W'})\sse\tau_C$.
Since $(\al,\beta_v)$ and $(\gm_u,\dl_u)$ cannot be separated in $\cP$, the first 
inclusion holds if and only if $f(\dl_u)\sse\gm_u$, and we infer the 
result. 
\end{proof}

\begin{corollary}\label{cor:cR-structure}
For any $(w,\gm,\dl)\in\cW(\ov\beta)$, $X=\bs\cap W_{w,\gm\dl}$, 
$X'=W\cap X$, let $\tau=\tau_{X'}$,  
$\tau'=\{(\ba,\bb)\mid \ba,\bb\in\pr_X\rel, (\pr_{X'}\ba,\pr_{X'}\bb)\in\tau\}$, 
and $\tau'=\{(\ba,\bb)\mid \ba,\bb\in\pr_X\rel, 
(\pr_{X'}\ba,\pr_{X'}\bb)\in\tau, \pr_{X-X'}\ba=\pr_{X-X'}\bb\}$. Then either 
$\umax(\rel'_{C,w,\gm\dl})=\umax(\rel_{C,w,\gm\dl})$, 
or for a $\tau'$-block $\reli$ it holds 
$\rel'_{C,w,\gm\dl}\sse\rel_{C,w,\gm\dl}\cap\reli$ and 
$\umax(\rel'_{C,w,\gm\dl})$ is the set of u-maximal elements of  
$\umax(\rel_{C,w,\gm\dl})\cap\reli$.
Moreover, $\rela\tm\umax(\pr_{X-X'}\rel_{C,w,\gm\dl})\sse
\rel_{C,w,\gm\dl}\fac{\tau''}$.
\end{corollary}  

\begin{proof}
By Lemma~\ref{lem:congruence-restriction}(4) either 
$\umax(\pr_{X'}\rel'_{C,w,\gm\dl})=\umax(\pr_{X'}\rel_{C,w,\gm\dl})$, 
or for a $\tau$-block $\rela$ it holds 
$\pr_{X'}\rel'_{C,w,\gm\dl}\sse\pr_{X'}\rel_{C,w,\gm\dl}\cap\rela$ and 
$\umax(\pr_{X'}\rel'_{C,w,\gm\dl})$ is the set of u-maximal elements of  
$\umax(\pr_{X'}\rel_{C,w,\gm\dl})\cap \rela$. Then considering 
$\rel_{C,w,\gm\dl}\fac{\tau''}$ as a subdirect product of
$\pr_{X'}\rel_{C,w,\gm\dl}\fac\tau$ and $\pr_{X-X'}\rel_{C,w,\gm\dl}$,
the interval $(\tau,\ov\beta_{X'})$ in $\pr_{X'}\rel\fac\tau$ can be separated
from interval $(\eta,\th)$ in $\Con(\zA_u)$ for any $u\in X-X'$ by 
Lemma~\ref{lem:restricted-congruence}. Then 
we use the Congruence Lemma~\ref{lem:affine-link} to conclude that
$\rela\tm\umax(\pr_{X-X'}\rel_{C,w,\gm\dl})\sse\rel_{C,w,\gm\dl}\fac{\tau''}$.
The result follows.
\end{proof}

Now we are in a position to prove that $\cR'$ is a $\ov\beta'$-strategy.

\begin{theorem}\label{the:restricting-strategy}
In the notation above, $\cR'$ is a $\ov\beta'$-strategy for 
$\ov B'$.
\end{theorem}

We give proofs of most difficult conditions from (S1)--(S7) in separate 
lemmas. 

\subsubsection{Condition (S6)}\label{sec:S6-1}

We start with condition (S6), as it will be needed for other conditions.

\begin{lemma}\label{lem:S6-1}
Condition (S6) for $\cR'$ holds. That is, for every $(w,\gm,\dl)\in\cW(\ov\beta)$ 
with $\bs\cap W_{w,\gm\dl}\ne\eps$ the set of $\cR'$-compatible tuples 
from $\rel_{C,w,\gm\dl}$ is polynomially closed in 
$\pr_{\bs\cap W_{w,\gm\dl}}\rel$.
\end{lemma}

\begin{proof}
Consider $C=\ang{\bs,\rel}\in\cC$. 
Let $f$ be a polynomial of $\rel$, and let $\ba,\bb\in\rel$ be tuples 
satisfying the conditions of polynomial closeness. Let $\bc\in\Sg{\ba,f(\bb)}$
be such that $\ba\sqq_{as}\bc$ in $\Sg{\ba,f(\bb)}$. By (S6) for $\cR$, $\bc$ is 
$\cR$-compatible. It suffices to show that $\pr_{\bs\cap W}\bc$ is in
the same $\tau_C$ block as $\pr_{\bs\cap W}\ba$. However, this is
straightforward, because $\pr_{\bs\cap W}\ba\eqc{\tau_C}\pr_{\bs\cap W}\bb$,
and as $f(\ba)=\ba$, we also have 
$\pr_{\bs\cap W}\ba\eqc{\tau_C}f(\pr_{\bs\cap W}\bb)$. Since
$\pr_{\bs\cap W}\bc\in\Sg{\pr_{\bs\cap W}\ba,f(\pr_{\bs\cap W}\bb)}$,
it follows $\pr_{\bs\cap W}\bc\eqc{\tau_C}\pr_{\bs\cap W}\ba$.
\end{proof}

\subsubsection{Condition (S1)}\label{sec:S1-1}

In this section we prove 

\begin{lemma}\label{lem:S1-1}
Condition (S1) for $\cR'$ holds. That is, the relations $\umax(\rel^{X,\cR'})$, 
where $\rel^{X,\cR'}$ consists of $\cR'$-compatible tuples from $\rel^X$ 
for $X\sse V$, $|X|\le 2$, form a nonempty $(2,3)$-strategy for $\cP^{\cR'}$.
\end{lemma}

We start with an auxiliary lemma.

\begin{lemma}\label{lem:module-intersection}
Let $\zA$ be an algebra and $\beta\in\Con(\zA)$ and $\zA'$ a 
subalgebra of a $\beta$-block. Let also $\al<\beta\red{\zA'}$ be a
congruence of $\zA'$ such that $\zA'\fac\al$ is a module.
Let $\zB,\zC$ be subalgebras of $\zA'$ such that $\zB\cap\zC\ne\eps$ and 
$\zB\cap\zC$ contains a u-maximal element of an $\al$-block of $\zA'$, 
$\zB\fac\al=\zA'\fac\al$, $\zC\fac\al=\zA'\fac\al$, and $\zB,\zC$ are 
polynomially closed in $\zA$. Then $(\zB\cap\zC)\fac\al=\zA'\fac\al$.
\end{lemma}

\begin{proof}
Let $a\in\zB\cap\zC$ be u-maximal in an $\al$-block of $\zA'$. As $\zA'\fac\al$
is a module, by Lemma~\ref{lem:u-max-congruence} it is u-maximal in the 
$\gm$-block $a^\gm$ for any $\gm\in\Con(\zA')$, $\al\le\gm$. Let $a'\in\zB$ and 
$a'\not\eqc\al a$; let also $\dl=\Cg{\{(a,a')\}\cup\al}$ in $\zA'$ and 
$\gm\in\Con(\zA')$ any such that $\al\le\gm\prec\dl$. 
For any $(\gm,\dl)$-trace $T$ that contains $a$ and any
$b\in T$, $a^\gm\ne b^\gm$, there is a polynomial $f$ of $\zA$ such 
that $f(a)=a$ and $f(a')=b$. 
Since $\zB$ is polynomially closed, for any $c\in\Sg{a,b}$ such that 
$a\sqq_{as}c$, we have $c\in\zB$, and $c$ is u-maximal in $\zA'$. In a 
similar way $c\in\zC$. Therefore it suffices to show that for any 
$\dl=\Cg{\{(a,a')\}\cup\al}$ in $\zA'$ for some $a'\in\zA'$ and any 
$\gm\in\Con(\zA')$, $\al\le\gm\prec\dl$ there is $b\in\zA'$ such that 
$\{a,b\}$ is a $(\gm,\dl)$-subtrace. Indeed, as is proved above 
$b\in\zB\cap\zC$, and therefore $(\zB\cap\zC)\fac\al$ is not contained
in a $\dl$-block. Since $\zA'\fac\al$ is a module, this means 
$(\zB\cap\zC)\fac\al=\zA'\fac\al$.

Since $\zA'\fac\al$ is a module, there is a $(\gm,\dl)$-subtrace $\{d,e\}$ 
such that $d\eqc\gm a$. Then by Lemma~\ref{lem:maximal-minimal}
there is also a polynomial $g$ such that $g(\zA)$ is an 
$(\gm,\dl)$-minimal set and $g(a)=a$. The result follows.
\end{proof}

We now can prove Lemma~\ref{lem:S1-1}.

\begin{proof}(of Lemma~\ref{lem:S1-1})\ \ \ \ 
We consider the collection of constraints 
$C^X=\ang{X,\rel^X}$, $X\sse V$, $|X|=2$, such that 
$\umax(\rel^{X,\cR})$ constitute a $(2,3)$-strategy for $\cP^\cR$. 
This collection exists by (S1) for $\cR$.
Let $\rel^{X,\cR'}$ denote the set of $\cR'$-compatible tuples from 
$\rel^{X,\cR}$. It suffices to show that for any tuple 
$(a,b)\in\umax(\rel^{X,\cR'})$, $X=\{x,y\}$, and any $w\not\in \{x,y\}$ 
there is $c\in\zA_w$ such that  $(a,c)\in\umax(\rel^{\{x,w\},\cR'}),
(b,c)\in\umax(\rel^{\{y,w\},\cR'})$; then the proof can be completed by
(S3) for $\cR'$ (to be proved later). By (S1) for $\cR$ there is $d\in\zA_w$ 
such that $(a,d)\in\umax(\rel^{\{x,w\},\cR}),(b,d)\in\umax(\rel^{\{y,w\},\cR})$. 

Consider the relation $\rel$ given by
$$
\rel(x,y,w)=\rel^{\{x,y\}}(x,y)\meet\rel^{\{x,w\}}(x,w)\meet
\rel^{\{y,w\}}(y,w),
$$
and let $\rel^\cR,\rel^{\cR'}$ be the set of $\cR$- and $\cR'$-compatible 
tuples from $\rel$, respectively. As is easily seen, element $c\in\zA_w$ 
satisfies the required conditions if and only if $(a,b,c)\in\rel^{\cR'}$ and
$(a,c)\in\umax(\pr_{xw}\rel^{\cR'})$, $(b,c)\in\umax(\pr_{yw}\rel^{\cR'})$.

By (S4) $\rel^\cR$ is a subalgebra of $\rel$; moreover, by (S1) the binary 
projections of $\rel^\cR$ contain $\umax(\rel^{\{x,y\},\cR})$, 
$\umax(\rel^{\{x,w\},\cR})$, $\umax(\rel^{\{y,w\},\cR})$, respectively.
Also, by Lemmas~\ref{lem:poly-closed}(2) and~\ref{lem:S6-1} 
$\rel^\cR,\rel^{\cR'}$ are polynomially closed in $\rel$. 
We consider a number of cases.

\medskip

{\sc Case 1.} $w\not\in W$. 

\smallskip

If, say, $x\not\in W$ then 
$\rel^{\{x,w\},\cR}=\rel^{\{x,w\},\cR'}$. If $x\in W$, 
then by construction and the Congruence Lemma~\ref{lem:affine-link} either 
$\umax(\rel^{\{x,w\},\cR})=\umax(\rel^{\{x,w\},\cR'})$ if $\tau_{x}$
is the full congruence, or 
$\umax(\rel^{\{x,w\},\cR}\cap(B'_x\tm B_w))\sse\umax(\rel^{\{x,w\},\cR'})$
otherwise. In either case $(a,d)\in\rel^{\{x,w\},\cR'}$. Similarly, 
$(b,d)\in\rel^{\{y,w\},\cR'}$. Therefore $(a,b,d)\in\rel^{\cR'}$, and by 
the Maximality Lemma~\ref{lem:to-max}(5) there is always $c$ such that 
$(a,c)\in\umax(\rel^{\{x,w\},\cR'}), (b,c)\in\umax(\rel^{\{y,w\},\cR'})$. 

\medskip

Therefore we may assume that $w\in W$.

If $x\in W$ (or $y\in W$, or $x,y\in W$) then let $\tau_x,\tau_{xw}$
(respectively, $\tau_y,\tau_{yw}$, or $\tau_{xy}$) be as in
Lemma~\ref{lem:congruence-restriction}. If $x\not\in W$ 
(or $y\not\in W$) then let $\tau_x=\beta_x$, $\tau_{xw}=\tau_w$
(more precisely, $\tau_{xw}=\beta_x\tm\tau_w$),
and $\tau_{xy}=\tau_y$ (respectively, $\tau_y=\beta_y,\tau_{yw}=\tau_w$).
We view all these congruences interchangeably: as congruences of $\rel$ in the 
natural way and their restrictions to $\cR$-compatible tuples as congruences 
of $\rel^\cR$, or as congruences of the corresponding projections of $\rel$ and
$\rel^\cR$.  By 
Lemma~\ref{lem:congruence-restriction} if $\tau'$ is one of these 
congruences, say, $\tau'=\tau_X$ for $X\sse\{x,y,w\}$, then $\tau'$ 
viewed as a congruence of $\rel^X$ is either the full congruence on 
$\rel^{X,\cR}$ or $\tau'\prec\ov\beta_X$. Therefore $\tau'$ 
viewed as a congruence of $\rel$ is either the full congruence on $\rel^\cR$, 
or $\tau'\prec(\beta_x\tm\beta_y\tm\beta_w)$ 
in $\Con(\rel)$; in the latter case we will say that $\tau'$ is maximal. Also let 
$\tau=\tau_{xy}\meet\tau_{xw}\meet\tau_{yw}$. Again by 
Lemma~\ref{lem:congruence-restriction} $\rel^\cR\fac\tau$ is a module.

Let $\rel^*_a,\rel^*_b$ be the 
sets of tuples $\ba\in\rel^\cR$ satisfying $\ba[x]=a,\ba[y]=b$, respectively; 
note that $(a,b,d)\in\rel^*_x\cap\rel^*_y$. We consider several cases of what 
the congruences introduced earlier can be.

\medskip

{\sc Case 2.} $\tau_x$ is the full congruence and $\tau_{xw}=\tau_w$, 
or $\tau_y$ is the full congruence and $\tau_{yw}=\tau_w$.

\smallskip

Suppose  $\tau_x$ is the full congruence, the other option is similar. 

\smallskip

{\sc Subcase 2.1} $\tau_{yw}$ is the full congruence.\\[2mm]
As $\tau_{yw}\le\tau_w$, if $\tau_{yw}$ is a full congruence, 
$\tau_w=\tau_{xw}$ is also the full congruence. Therefore,
$(a,d)\in\rel^{\{x,w\},\cR'},(b,d)\in\rel^{\{y,w\},\cR'}$, and we are done.

\smallskip

{\sc Subcase 2.2} $\tau_y$ is maximal while $\tau_w$ is full.\\[2mm]
Since any tuple $(b,x)\in\rel^{\{y,w\},\cR}$ also belongs to 
$\rel^{\{y,w\},\cR'}$ in this case, again 
$(a,d)\in\rel^{\{x,w\},\cR'},(b,d)\in\rel^{\{y,w\},\cR'}$  

\smallskip

{\sc Subcase 2.3} Both $\tau_y$ and $\tau_w$ are maximal.\\[2mm]
In this case, as  
$\tau_{yw}\le\tau_y\meet\tau_w$ and all three congruences are maximal,
we have $\tau_{yw}=\tau_y=\tau_w$. Also, $\tau_w$ (viewed as a congruence 
of $\pr_w\rel^\cR$) is the link congruence
with respect to $\pr_{yw}\rel^\cR$, and so $(b,d)\in\rel^{\{y,w\},\cR'}$.
This also implies, since $(a,b)\in\pr_{xy}\rel^\cR$, that 
$(a,d)\in\rel^{\{x,w\},\cR'}$.

\smallskip

{\sc Subcase 2.4} $\tau_w$ is maximal while $\tau_y$ is full.\\[2mm]
Let $B^*_z=\pr_z\rel^\cR$ for 
$z\in\{x,y,w\}$. By the Congruence Lemma~\ref{lem:affine-link} 
$\umax(B^*_x)\tm B^*_w\fac{\tau_w}\sse\pr_{xw}\rel^\cR$ and 
$\umax(B^*_y)\tm B^*_w\fac{\tau_w}\sse\pr_{yw}\rel^\cR$. Indeed,
otherwise $\pr_{xw}\rel^\cR\fac{\tau_w}$ would be the graph of a 
mapping $\nu:B^*_x\to B^*_w\fac{\tau_w}$ contradicting the assumption
that $\tau_x$ is the full congruence; similar for $B^*_y$.
Therefore $\rel^*_a\fac{\tau_w}=\rel^*_b\fac{\tau_w}=\rel^\cR\fac{\tau_w}$
and since $(a,b)$ is u-maximal in a $\tau_{xy}$-block, $d$ can be 
chosen such that $(a,b,d)$ is u-maximal in a $\tau_w$-block. Then  by 
Lemma~\ref{lem:module-intersection} 
$(\rel^*_a\cap\rel^*_b)\fac{\tau_w}=\rel^\cR\fac{\tau_w}$ and the
result follows.

\smallskip

{\sc Subcase 2.5.} $\tau_{yw}\prec\tau_y=\tau_w$ and $\tau_y,\tau_w$ 
are full congruences.\\[2mm] 
Again by the Congruence Lemma~\ref{lem:affine-link}
$\umax(B^*_x)\tm \pr_{yw}\rel^\cR\fac{\tau_{yw}}\sse\rel^\cR$ and
therefore $\rel^*_a\fac{\tau_{yw}}=\rel^\cR\fac{\tau_{yw}}$.
Also, let $\relo$ be the union of all $\tau_{yw}$-blocks of $\rel^\cR$ whose 
intersection with $\rel^*_b$ is nonempty, this is clearly a subalgebra. Also, 
let $\rel^{**}_a=\rel^*_a\cap\relo$. 
Note that, since $\tau_y$ is the full congruence, and therefore 
$\umax(B^*_y)\sse\rel^{\{y,w\},\cR'}$, we have 
$\rel^{\{y,w\},\cR'}\cap\pr_{yw}\relo\ne\eps$. Then 
$\rel^{**}_a\fac{\tau_{yw}}=\relo\fac{\tau_{yw}}$ and 
$\rel^*_b\fac{\tau_{yw}}=\relo\fac{\tau_{yw}}$. Then as before the result 
follows by Lemma~\ref{lem:module-intersection}.


\medskip

{\sc Case 3.} $\tau_x,\tau_y$ are maximal congruences.

\smallskip

In this case $\tau_{xw}=\tau_x,\tau_{yw}=\tau_y$, and depending on
whether or not $\tau_w$ is maximal we proceed as in Cases~2.2 or~2.3.

\medskip

If $\tau_w$ is a maximal congruence or $\tau_{xw}$ or $\tau_{yw}$ is 
the full congruence, the result follows from one of the previous cases. 
Therefore the only remaining case is

{\sc Case 4.} $\tau_{xw}$ and $\tau_{yw}$ are maximal, while 
$\tau_x,\tau_y,\tau_w$ are not. 

\smallskip

Let again $B^*_z=\pr_z\rel^\cR$ for $z\in\{x,y,w\}$. 
If $(v,\al,\beta)\not\in\cW'$ or $w\not\in\Razm(\cP)$, that is, 
$\mu^Y_w=\zz_w$ for any choice of $Y$ in (R1),(R2), then the required 
$c\in B^*_w$ exists, since $(a,b)$ can be extended to a solution from 
$\cS^{\cR'}_W$ or $\cS^{\cR'}_{W,Y}$, $Y=\Razm(\cP)-\{x,y\}$ by
construction. 

Suppose that $(v,\al,\beta)\in\cW'$ and $\mu^Y_w=\mu_w$ 
for $Y=\Razm(\cP)-\{x,y\}$. Let $\relo_x$ be a subalgebra of the product 
$\zA_x\tm\zA_w\tm\zA_v\fac\al$ that consists of all triples $(a',b',c')$ 
such that there is a solution $\vf\in\cS_{W,Y}$  with $\vf(x)=a',\vf(w)=b'$, 
and $\vf(v)\in c'$. Let $B^*_v=A_{\cR,v}$. By block-minimality $\relo_x$ 
is indeed a subdirect 
product and by (S2) for $\cR$ we have $\umax(\rel^{\{x,w\},\cR})
\sse\pr_{xw}(\relo_x\cap(B^*_x\tm B^*_w\tm B^*_v\fac\al)$ and 
$\pr_v(\relo_x\cap(B^*_x\tm B^*_w\tm B^*_v\fac\al)=B^*\fac\al$. 
Also, by Lemma~\ref{lem:poly-closed}(3) $\relo_x$ is polynomially
closed. Relation $\relo_y$ is defined in a similar way. Let also
$$
\relo(x,y,w,v)=\relo_x(x,w,v)\meet\relo_y(y,w,v),
$$
and $\relo'=\relo\cap\ov B$. Let $\relo^a=\{\ba\in\relo'\mid \ba[x]=a\}$,
$\relo^b=\{\ba\in\relo'\mid \ba[y]=b\}$, and 
$\al'=\beta_x\tm\beta_y\tm\beta_w\tm\al$. By the assumption that 
$\tau_x,\tau_y$ are full congruences 
$\relo^a\fac{\al'}=\relo^b\fac{\al'}=\relo'\fac{\al'}$. Therefore,
if we prove that $(a,b)\in\pr_{xy}\relo'$, we obtain the result by 
Lemma~\ref{lem:module-intersection}. Note that if this is the case, 
since $(a,b)\in\umax(\rel^{\{x,y\},\cR}$, there is also $(a,b,c',d')\in\relo'$ 
that is u-maximal in an $\al'$-block.

To this end consider the relations 
$$
\rela(x,y,w,v,v')=\relo_x(x,w,v)\meet\relo_y(y,w,v'),
$$
and $\rela'=\rela\cap\ov B$. In a similar way we define 
$\rela^a=\{\ba\in\rela'\mid \ba[x]=a\}$, 
$\rela^b=\{\ba\in\rela'\mid \ba[y]=b\}$. By (S1) there are $d\in B^*_w$ and
$e',e''\in B^*_v\fac\al$ such that $(a,b,d,e',e'')\in\rela'$. Also by construction 
(R2) there are $a'\in B^*_x, b'\in B^*_y$, $d_1,d_2\in B^*_w$, 
$d_1\eqc{\mu_w}d_2$ and $e\in B^*_v$, such that 
$(a,b',d_1,e,e)$, $(a',b,d_2,e,e)\in\rela'$. Recall that $B^*_v\fac\al$ is a 
module. Let $\dl$ be
the skew congruence of $B^*_v\fac\al\tm B^*_v\fac\al$, that is, 
a congruence such that $\Dl=\{(c_1,c_2),\mid c_1\eqc\al c_2\}$ is a $\dl$-block.
Let $\dl'$ be the congruence of $\rela'$ given by $\bc\eqc{\dl'}\bd$ if and only
if $\pr_{vv'}\bc\eqc\dl\pr_{vv'}\bd$. Note that $\dl'<\ov\beta$ 
in $\Con(\rela)$. Let $\Dl'$ be the $\dl'$-block corresponding to $\Dl$,
that is, $\Dl=\pr_{vv'}\Dl'$.
Since $\rela^a\cap\rela^b\ne\eps$ and $(a,b,d,e',e'')$ can be chosen to 
be u-maximal in a $\dl'$-block, and $\rela^a\cap\Dl'\ne\eps$, 
$\rela^b\cap\Dl'\ne\eps$, by Lemma~\ref{lem:module-intersection} we have 
$\rela^a\cap\rela^b\cap\Dl'\ne\eps$. The result follows.
\end{proof}

\subsubsection{Conditions (S2), (S3)}\label{sec:S2-S3-1}

In this section we prove that $\cR'$ satisfies conditions (S2) and (S3).
We first prove (S2) and then show what in the proof has to be changed 
to obtain a proof of (S3).

As before, let $W=W_{v,\al\beta}$.

\begin{lemma}\label{lem:S2-1}
$\cR'$ satisfies (S2). That is, for every 
$(w_1,\gm_1,\dl_1),(w_2,\gm_2,\dl_2)\in\cW(\ov\beta')$ 
(let $W_1=W_{w_1,\gm_1\dl_1},W_2=W_{w_2,\gm_2\dl_2}$) and every 
$\ba\in\umax(\pr_{\bs\cap W_1\cap W_2}\rel'_{C,w_1,\gm_1\dl_1})$ it holds:
if $(w_2,\gm_2,\dl_2)\not\in\cW'$ then $\ba$ extends
to an $\cR'$-compatible solution $\vf$ of $\cP_{W_2}$; otherwise
if $(w_1,\gm_1,\dl_1)\not\in\cW'$ then $\ba$ extends to an $\cR'$-compatible 
solution of $\cP_{W_2}\fac{\ov\mu^{Y_1}}$, where 
$Y_1=\Razm(\cP)-(W_1\cap W_2)$; 
and if $(w_1,\gm_1,\dl_1)\in\cW'$ then $\ba$ extends to an $\cR'$-compatible 
solution of $\cP_{W_2}\fac{\ov\mu^{Y_2}}$, where $Y_2=\Razm(\cP)-\bs$.
\end{lemma}

Take $C=\ang{\bs,\rel}$ and 
$(w_1,\gm_1,\dl_1),(w_2,\gm_2,\dl_2)\in\cW(\ov\beta')$ and let
$W_1=W_{w_1,\gm_1\dl_1}$, $W_2=W_{w_2,\gm_2\dl_2}$, and 
$U=\bs\cap W_1\cap W_2$. Let 
$\ba\in\umax(\pr_U\rel'_{C,w_1,\gm_1\dl_1})$. 
Depending on whether or not $(w_1,\gm_1,\dl_1),(w_2,\gm_2,\dl_2)\in\cW'$
we need to show that $\ba$ can be extended to a solution of 
$\cP_{W_2}\fac{\ov\mu^Y}$, where $Y$ is either empty, if 
$(w_2,\gm_2,\dl_2)\not\in\cW'$, or $Y=\Razm(\cP)-W_1$ if 
$(w_1,\gm_1,\dl_1)\not\in\cW'$ and $(w_2,\gm_2,\dl_2)\in\cW'$, and 
$Y=\Razm(\cP)-\bs$ if $(w_1,\gm_1,\dl_1),(w_2,\gm_2,\dl_2)\in\cW'$. 
The three cases are quite similar so we will unify the proofs as much as 
possible.

Let $\cP'_{W'}$\label{not:P-prime-W-prime} denote the
problem $\cP_{W'}\fac{\ov\mu^Y}$ for a set $W'\sse V$ and 
$\cS'_{W'}$\label{not:S-prime-W-prime} denote its set of solutions.
Then we need to show that $\ba\in\pr_U\cS'^{\cR'}_{W_2}$.

First, we will show that there are $\cR'$-compatible solutions of 
$\cP'_{W\cap W_2}$. This statement would be trivial, as 
$\cS^{\cR'}_{W\cap W_2}$ contains $\pr_{W\cap W_2}\cS^{\cR'}_W$,
if $Y$ was always empty. Otherwise a simple proof is needed. Note that this 
is the only statement where we consider the three options for $Y$
separately. Next, we identify the properties of $\cP_{W'}\fac{\ov\mu^Y}$
required for proving (S2) and show that the problem satisfies them
for all choices of $Y$. The rest of the proof will be given only using these 
properties; thus we use the same argument
in all the three cases of $Y$. 

If $\tau_{C',w_2,\gm_2\dl_2}$ is the 
full congruence of $\rel_{C',w_2,\gm_2\dl_2}\fac{\ov\mu^Y}$ for all 
$C'\in\cC$, then $\umax(\cS'^\cR_{W_2})=\umax(\cS'^{\cR'}_{W_2})$ 
and there is nothing to prove. Thus we will always assume that there is
$C'\in\cC$ such that $\tau_{C',w_2,\gm_2\dl_2}$ is maximal  
on $\rel_{C',w_2,\gm_2\dl_2}\fac{\ov\mu^Y}$.

\begin{lemma}\label{lem:W-cap-W2}
The relations $\cS'^\cR_{W\cap W_2}$ and $\cS'^{\cR'}_{W\cap W_2}$ 
are nonempty.
\end{lemma}

\begin{proof}
The set $\cS'^\cR_{W\cap W_2}$ is nonempty, as it contains 
$\pr_{W\cap W_2}\cS'^\cR_{W_2}$, which is nonempty by (S2) for $\cR$.
If $(v,\al,\beta)\not\in\cW'$, then $\cS'^{\cR'}_{W\cap W_2}$ 
contains $\pr_{W\cap W_2}\cS'^{\cR'}_W$ or its factor modulo $\ov\mu^Y$
(if $Y\ne\eps$), which is nonempty. So, let $(v,\al,\beta)\in\cW(\ov\beta)$.

Suppose $(w_2,\gm_2,\dl_2)\not\in\cW'$ and so $Y=\eps$,
and $C_1=\ang{\bs_1,\rel_1}\in\cC$ is a constraint such that 
$\tau_{C_1,w_2,\gm_2\dl_2}$ is nontrivial on $\rel_{C_1,w_2,\gm_2\dl_2}$ and 
$\bb\in\umax(\rel'_{C_1,w_2,\gm_2\dl_2})$. By (S2) for $\cR$ tuple $\bb$ can be 
extended to a solution $\vf$ of $\cP_W\fac{\ov\mu^{Y_1}}$, 
$Y_1=\Razm(\cP)-W_2$. Then $\vf(v)\in B'_v\fac{\mu^{Y_1}_v}$ and 
therefore for any $C_2=\ang{\bs_2,\rel_2}\in\cC$
and any $(u,\eta,\th)\in\cW(\ov\beta)$ we have 
$\vf(\bs_2\cap W\cap W_{u,\eta\th})\in
\pr_{\bs\cap W\cap W_{u,\eta\th}}\rel'_{C_2,u,\eta\th}\fac{\ov\mu^{Y_1}}$, 
that is,
$\vf(W\cap W_2)\in\cS'^{\cR'}_{W\cap W_2}\fac{\ov\mu^{Y_1}}=
\cS'^{\cR'}_{W\cap W_2}$. 

Suppose now that $(v,\al,\beta),(w_2,\gm_2,\dl_2)\in\cW'$. 
If $(w_1,\gm_1,\dl_1)\not\in\cW'$ and $Y=\Razm(\cP)-W_1$, 
then we apply the argument above to the problem
$\cP_W\fac{\ov\mu^{\Razm(\cP)-W_1}}$. If $(w_1,\gm_1,\dl_1)\in\cW'$ 
and $Y=\Razm(\cP)-\bs$, then we consider the problem 
$\cP_W\fac{\ov\mu^{\Razm(\cP)-\bs}}$.
\end{proof}

For any possible choice of $Y$ the sets 
$\cS'_{W_2},\cS'^\cR_{W_2},\cS'^{\cR'}_{W_2}$ satisfy the following 
conditions. For the rest of the proof we only need these conditions, and therfore
the argument is valid for all choices of $Y$. First of all if $Y\ne\eps$ the set
$W$ may need to be redefined. By the definition of $W$ for every $w\in W$
there are $\gm_w,\dl_w\in\Con(\zA_w)$ such that 
$\gm_w\prec\dl_w\le\beta_w$ and $(\al,\beta_v),(\gm_w,\dl_w)$ cannot
be separated in $\cP$. If $w\in Y$ and $\gm_w=\zz_w,\dl_w=\mu_w$ is
the only choice for $\gm_w,\dl_w$, variable $w$ is removed from $W$. 
Set $W_2$ does not have to be modified. The conditions needed are:
\begin{itemize}
\item[(X1)]
$\cP_{W_2,Y}$ is minimal by block-minimality.
\item[(X2)]
For every $(u,\eta,\th)\in\cW(\ov\beta)$ and $C'=\ang{\bs',\rel'}$ any
tuple from $\umax(\pr_Z\rel_{C',u,\eta\th})\fac{\ov\mu^Y}$, where
$Z=W_{u,\eta\th}\cap W_2\cap\bs'$, can be extended to a solution 
from $\cS'^\cR_{W_2}$; by (S2). 
\item[(X3)]
For every $(u,\eta,\th)\in\cW(\ov\beta)$ and $C'=\ang{\bs',\rel'}$ 
let $W'=W_{u,\eta\th}\cap W\cap W_2\cap\bs'$ and 
$\tau'_{C'}(u,\eta\th)$\label{not:tau-prime-u-eta-th} 
denote the restriction of $\tau_{C'}$ on $\rel^*_{C',u,\eta\th}=
\pr_{W'}\rel_{C',u,\eta\th}\fac{\ov\mu^Y}$, 
that is, $(\bb,\bc)\in\tau'_{C'}(u,\eta\th)$ if
there are $\bb',\bc'\in\rel_{C',v,\al\beta_v}\fac{\ov\mu^Y}$ such that 
$\pr_{W'}\bb'=\bb$, $\pr_{W'}\bc'=\bc$, and
$(\bb',\bc')\in\tau_{C'}$. By Lemma~\ref{lem:congruence-restriction} 
$\tau'_{C'}(u,\eta\th)$ is either the full congruence
on $\rel^*_{C',u,\eta\th}$, or a maximal one. In the latter case 
$\typ(\tau'_{C'}(u,\eta\th),\ov\beta_{W'})=\two$.
\item[(X4)]
For every $w\in W\cap W_2$ there are 
$\gm_w,\dl_w\in\Con(\zA_w\fac{\mu^Y_w})$ such that 
$\mu^Y_w\le\gm_w\prec\dl_w\le\beta_w$ and 
$(\al,\beta_v),(\gm_w,\dl_w)$ cannot be separated in $\cP$. 
\item[(X5)]
For every $(u,\eta,\th)\in\cW(\ov\beta)$, $C'=\ang{\bs',\rel'}$, and 
any $w\in W'=W_{u,\eta\th}\cap W\cap W_2\cap\bs'$, if 
$\tau'_{C'}(u,\eta\th)$ is maximal, then for every polynomial $f$ of 
$\cS'_{W_2}$ we have $f(\dl_w)\sse\gm_w$ if and only if 
$f(\ov\beta_{W'})\sse\tau'_{C'}(u,\eta\th)$; by 
Lemma~\ref{lem:restricted-congruence}.
\end{itemize}

We are now in a position to prove Lemma~\ref{lem:S2-1}.

\begin{proof}\ (of Lemma~\ref{lem:S2-1}) \ \ \ 
Recall that we need to prove that any 
$\ba\in\umax(\pr_U\rel'_{C,w_1,\gm_1\dl_1})$ can be extended to a 
solution from $\cS'^{\cR'}_{W_2}$. By (S2) $\ba$ can be extended to a
tuple $\bb\in\cS'^\cR_{W_2}$. We will prove that a required extesion 
$\bc\in\cS'^{\cR'}_{W_2}$ can be found such that 
$\pr_{W_2-W}\bc=\pr_{W_2-W}\bb$ (if $\bb$ is chosen well). Therefore,
we will need to study the relationship of tuples from 
$\pr_{W_2\cap W}\cS'^\cR_{W_2}$ extending $\ba$ and the tuples
from $\pr_{W_2-W}\cS'^\cR_{W_2}$. This is done in Claim~1. Since 
tuples from $\pr_{W_2\cap W}\cS'^\cR_{W_2}$ are only considered 
modulo congruences of the form $\tau_{C'}$, while $\ba$ is considered
without factoring, the main technical hurdle is to `disentangle' $\ba$
from the quotient of $\pr_{W_2\cap W}\cS'^\cR_{W_2}$. To this end
we consider a congruence $\tau$ on $\cS'^\cR_{W_2\cap W}$ that
isolates $\cR'$-compatible tuples from there, and then introduce two
auxiliary relations $\relo$ and~$\rela$.

We would like to define a congruence similar to $\tau_C$ on 
$\cS'_{W\cap W_2}$. It cannot be done in the same straightforward way, 
since $\cS'_{W\cap W_2}\ne\pr_{W\cap W_2}\cS'_W$,
so we define it as follows. 
For $C'=\ang{\bs',\rel'}$ and $(u,\eta,\th)\in\cW(\ov\beta)$ we extend 
the congruences $\tau'_{C'}(u,\eta\th)$ to congruences of $\cS'_{W\cap W_2}$ 
using  $\tau'_{C'}(u,\eta\th)\tm\prod_{x\in(W\cap W_2)-W'}\zo_x$, where
$W'=\bs'\cap W_{u,\eta\th}\cap W\cap W_2$.
Then the set $\cS'^{\cR'}_{W\cap W_2}$ of $\cR'$-compatible tuples 
from $\cS'^\cR_{W\cap W_2}$ is a block of
\begin{equation}\label{equ:tau}
\tau=\bigwedge_{C'\in\cC,(u,\eta,\th)\in\cW(\ov\beta)}\tau'_{C'}(u,\eta\th),
\end{equation}
let it be denoted by $\cS^*$. Note that $\tau$ is a congruence of 
$\cS'_{W\cap W_2}$, but the interval $(\tau,\ov\beta_{W\cap W_2})$ 
is not necessarily prime. By the observation above and 
Lemma~\ref{lem:as-type-2} $\cS'^\cR_{W\cap W_2}\fac\tau$ is term
equivalent to a module. We need to prove that there is $\bc\in\cS'^\cR_{W_2}$
such that $\vf(U)=\ba$ and $\vf(W\cap W_2)\in\cS^*$. In fact we 
prove a stronger statement, namely, that for any u-maximal tuple $\bb$ from
$\pr_{W_2-W}\cS'^\cR_{W_2}$ there is a solution $\vf$ such that 
$\vf(U)=\ba$, $\vf(W\cap W_2)\in\cS^*$ and also $\vf(W_2-W)=\bb$. 
However, to formulate it precisely we need two additional constructions.

We need to describe the $\tau$-blocks that contain tuples extending $\ba$. 
In order to do that we separate the variables of $\ba$ from the rest of
$W_2\cap W$ by making an extra copy of them, as follows. 
Let $W\cap W_2=\{\vc xk\}$ and $X=W\cap U=\{\vc x\ell\}$,
and $X'=\{\vc y\ell\}$. Let 
$$
\relo(\vc xk,\vc y\ell)=\cS'^\cR_{W\cap W_2}(\vc xk)\meet
\pr_X\cS'^\cR_{W\cap W_2}(\vc y\ell)\meet\bigwedge_{i=1}^\ell(x_i=y_i),
$$
and its factor $\relo'=\relo\fac{\tau'}$, where 
$\tau'=\tau\tm\zz_{\pr_{X'}\cS'^\cR_{W\cap W_2}}$. 
Let $\eta_1,\eta_2$ denote the link congruences of 
$\cS'^\cR_{W\cap W_2}\fac\tau$ and $\pr_{X'}\cS'^\cR_{W\cap W_2}$
with respect to $\relo'$, and let $\eta'$ denote the congruence of 
$\cS'^\cR_{W\cap W_2}$ such that $\bc\eqc{\eta'}\bd$ if and only if 
$\bc^\tau\eqc{\eta_1}\bd^\tau$. Then, as is easily seen, since 
$\pr_{W\cap W_2}\relo\fac\tau$ is a module,  $(\bb,\bc)\in\eta'$ if and only if 
there are $\bb',\bc'\in\cS'^\cR_{W\cap W_2}$ such that 
$(\bb,\bb'),(\bc,\bc')\in\tau$ and $\pr_X\bb'=\pr_X\bc'$.

Now, we introduce a similar construction for $\cS'_{W_2}$. The goal
is to compare the two.
Let $W_2=\{\vc xm\}$ (recall that $W\cap W_2=\{\vc xk\}$) and
define a relation $\rela(\vc xm,\vc y\ell)$ as follows:
$$
\rela(\vc xm,\vc y\ell)=\cS'^\cR_{W_2}(\vc xm)\meet
\pr_X\cS'^\cR_{W_2}(\vc y\ell)\meet\bigwedge_{i=1}^\ell(x_i=y_i),
$$
and let $\rela'=\rela\fac{\tau''}$, where $\tau''=\tau\tm
\zz_{\pr_{W_2-W}\cS'^\cR_{W_2}}\tm\zz_{\pr_{X'}\cS'^\cR_{W_2}}$. 
Similar to $\relo$, let $\th_1,\th_2$ be the link congruences of 
$\pr_{W\cap W_2}\cS'^\cR_{W_2}\fac\tau$ 
and $\pr_H\rela$, where $H=X'\cup\{x_{k+1}\zd x_m\}$,
with respect to $\rela'$, and let $\th'$ denote the congruence of 
$\pr_{W\cap W_2}\cS'^\cR_{W_2}$ given by $\bc\eqc{\th'}\bd$ if and only if 
$\bc^\tau\eqc{\th_1}\bd^\tau$. 
Then immediately by the definition $(\bb,\bc)\in\th'$ if and 
only if there are $\bb',\bc'\in\pr_{W\cap W_2}\cS'^\cR_{W_2}$ and 
$\bd\in\pr_H\rela$ such that $(\bb,\bb'),(\bc,\bc')\in\tau$, 
$\pr_X\bb'=\pr_X\bc'=\pr_X\bd$, and 
$(\pr_{(W\cap W_2)-X}\bb',\bd)$, $(\pr_{(W\cap W_2)-X}\bc',\bd)
\in\cS'^\cR_{W_2}$.

We are interested in congruences $\eta_1$ and $\th_1$. The first of them 
indicates which $\tau$-blocks extensions of $\pr_X\ba$ can belong to.
The second congruence also indicates to which $\tau$-blocks extensions
of $\ba$ to a solution from $\cS'^\cR_{W_2}$ can belong to. Clearly,
$\th_1\sse\eta_1$. We prove however, that in both cases the set of 
attainable $\tau$-blocks is the same. This essentially means that if a 
$\tau$-block can be extended to a solution from $\cS'^\cR_{W_2}$, it
can be extended in an almost arbitrary way. 

\medskip

{\sc Claim 1.}
(1) $\rela''=(\pr_{W\cap W_2}\rela)\fac\tau$ is a union of $\eta_1$-blocks;\\
(2) $\th_1=\eta_1\red{\rela''}$;\\
(3) let $D$ be a $\th_1$-block and $E$ the corresponding $\th_2$-block,
then $D\tm\umax(E)\sse\rela'$;\\
(4) for any $\bb\in\cS'^\cR_{W_2}$ such that $\pr_{X\cup(W_2-W)}\bb$ 
is u-maximal in a $\th_2$-block and any $\bb'\in\cS'^\cR_{W\cap W_2}$ 
such that $\bb'\eqc{\th'}\pr_{W\cap W_2}\bb$ there is $\bb''\in\rela$
such that $\pr_{W\cap W_2}\bb''\eqc\tau\bb'$ and 
$\pr_{X\cup(W_2-W)}\bb''=\pr_{X\cup(W_2-W)}\bb$.

\medskip

(1) It follows by Proposition~\ref{pro:umax-rectangular} for any 
$\th_1$-block $D'$ and a $\th_2$-block $E'$ with $\rela'\cap(D'\tm E')\ne\eps$,
that $D'\tm\umax(E')\sse\rela'$. 

Let $\bb\in\rela$, $D'$ be the $\th_1$-block containing 
$(\pr_{W\cap W_2}\bb)^\tau$ and $E'$ the corresponding $\th_2$-block. 
Then by the Maximality Lemma~\ref{lem:to-max}(5)
there is $\bb'\in\rela$ such that 
$\pr_{W\cap W_2}\bb'\eqc\tau\pr_{W\cap W_2}\bb$ and 
$\pr_H\bb'$ is u-maximal in $E'$. We assume that $\bb$ satisfies
this condition. Let also $D$ be the $\eta_1$-block containing 
$(\pr_{W\cap W_2}\bb)^\tau$. Note that $D'\sse D$.

Suppose there is $\bc\in D$ such that $(\bc',\pr_H\bb)\in\rela$ for
no $\bc'\in\bc$. We will derive a contradiction. Take some u-maximal 
tuple $\bc'$ from $\bc$. Since $\pr_{W\cap W_2}\bb\not\eqc\tau\bc'$, there is 
$C'=\ang{\bs',\rel'}\in\cC$ and $(u,\chi,\xi)\in\cW(\ov\beta)$ such that 
$(\pr_Z\bb,\pr_Z\bc')\not\in\tau^*$, where $\tau^*=\tau'_{C'}(u,\chi\xi)$ and
$Z=W_{u,\chi\xi}\cap W\cap W_2\cap \bs'$. Choose a pair $\bb,\bc$ 
in such a way that the number of such constraints and triples is minimal. 
We will find a polynomial $f$ of $\cS'_{W_2}$ such that (roughly speaking) 
$f(\pr_{W\cap W_2}\bb),f(\bc')\in\pr_{W\cap W_2}\cS'^\cR_{W_2}$, 
$f(\pr_H\bb)=\pr_H\bb$, and $f(\pr_{W\cap W_2}\bb),f(\bc')$ differ on 
fewer constraints and triples.
Let $\bb''=(\pr_Z\bb)^{\tau^*}$, $\bc''=(\pr_Z\bc')^{\tau^*}$. Since the 
interval $(\tau^*,\ov\beta_Z)$ has type \two,
$\{\bb'',\bc''\}$ is a $(\tau^*,\ov\beta_Z)$-subtrace. 
Let $(x,\gm_x,\dl_x)\in\cW(\ov\beta)$, $x\in Z$, be such that 
$(\al,\beta_v)$ and $(\gm_x,\dl_x)$ cannot be separated in $\rel^{\{v,x\}}$.

By (S7) for $\cR$ relation $\cS'_{W_2}$ is strongly chained. Since  
$\tau^*\prec\ov\beta''$ on $\pr_Z\cS'^\cR_{W_2}$ if we consider $\cS'_{W_2}$ 
as a subdirect product of $\pr_Z\cS'_{W_2}$ and $\zA_y$, $y\in W_2-Z$, by 
Lemma~\ref{lem:collapsing} there is a $(\tau^*,\ov\beta_Z)$-collapsing 
polynomial $f$ of $\cS'_{W_2}$ for $\ov\beta,\ov B$. By 
Lemma~\ref{lem:restricted-congruence} $(\tau^*,\ov\beta_Z)$ cannot be 
separated from $(\gm_x,\dl_x)$ for $x\in W\cap W_2$. For any 
$(z,\eta,\th)\in\cW(\ov\beta)$, $z\in W_2$, consider
$\pr_{Z\cup\{z\}}\cS'_{W_2}$. If $(\eta,\th)$ can be separated from 
$(\al,\beta_v)$ (or $(\al,\beta_v)$ can be separated from $(\eta,\th)$) then 
$(\tau^*,\ov\beta_Z)$ can be separated from $(\eta,\th)$ (or the 
other way round). In particular, by Lemma~\ref{lem:collapsing} $f$ can be 
chosen to satisfy the following conditions\\[2mm]
(a) $f(\pr_Z(\cS'_{W_2}))$ is a $(\tau^*,\ov\beta_Z)$-minimal set and
$f(\zA_y)$ is a $(\gm_y,\dl_y)$-minimal set for $y\in W\cap W_2$;\\[2mm]
(b) for every $z\in W_2-W$, $|f(B_z)|=1$; and\\[2mm]
(c) $f$ is idempotent.\\[2mm]
Since $\{\bb'',\bc''\}$ is a $(\tau^*,\ov\beta_Z)$-subtrace of 
$(\pr_Z\rel_{C',u,\chi\xi})\fac{\tau^*}$, by condition (Q2s) of being 
strongly chained (S7) for $\cR$ polynomial $f$ can be chosen such that\\[2mm]
(d) $\bb''=f(\bb'')$, $\bc''=f(\bc'')$. \\[2mm]
Moreover, let $\bb^*\in\pr_{W\cap W_2}\cS'^\cR_{W_2}$ be such that
there is a tuple $\bd\in\cS'^\cR_{W_2}$ such that 
$\pr_{W_2-W}\bd=\pr_{W_2-W}\bb$, $\pr_{W\cap W_2}\bd=\bb^*$,
$\pr_Z\bb^*\in\bb''$, and $\bd$ is u-maximal in $\cS'_{W_2}\cap\ov B$.
Such a tuple exists. Indeed, for any $(z,\gm',\dl')\in\cW(\ov\beta)$,
$z\in W_2-W$, and any $x\in W\cap W_2$, the 
interval $(\gm_x,\dl_x)$ and therefore $(\tau^*,\ov\beta'')$ 
can be separated from $(\gm',\dl')$, or the other way round. 
Lemma~\ref{lem:collapsing} also implies that for any 
$\gm',\dl'\in\Con(\cS'_{W_2-W})$, $\gm'\prec\dl'\le\ov\beta_{W_2-W}$, 
the interval  $(\tau^*,\ov\beta'')$ can be separated from $(\gm',\dl')$,
or the other way round. Therefore, by the Congruence 
Lemma~\ref{lem:affine-link} $\pr_{W\cap W_2}\cS'^\cR_{W_2}\fac\tau
\tm\umax(\pr_{W_2-W}\cS'^\cR_{W_2})\sse\cS'^\cR_{W_2}$.
By Lemma~\ref{lem:maximal-minimal}
polynomial $f$ can be chosen such that\\[2mm]
(e) $f(\bd)=\bd$.\\[2mm]
\indent
Let $\bc'$ be a u-maximal tuple from the $\tau$-block 
$\bc\sse\cS'^\cR_{W\cap W_2}$, and let $\bc^*=f(\bc')$, and 
$\bc^\dagger$ a tuple from $\Sg{\bb^*,\bc^*}$ such that 
$\bb^*\sqq_{as}\bc^\dagger$ and $\bc^*\eqc\tau\bc^\dagger$;
such a tuple $\bc^\dagger$ exists, because $\cS'^\cR_{W\cap W_2}\fac\tau$
is a module. Note that it suffices to prove that 
$\vr=(\bc^\dagger,\pr_{W_2-W}\bb)\in\cS'^\cR_{W_2}$. Indeed, since 
$\pr_Z\bc^*\eqc{\tau^*}\pr_Z\bc'$ and $\bc'$ 
agrees with $\bc^\dagger$ modulo  
$\tau'_{C''}(u,\chi'\xi')$ for every $C''\in\cC$ and $(u',\chi'\xi')\in\cW(\ov\beta)$, 
for which $\bb$ and $\bc'$ agree,
we obtain a contradiction with the choice of $\bb,\bc$.

To show that $\vr$ is an $\cR$-compatibile solution take any 
$C''=\ang{\bs'',\rel''}\in\cC$ and $(x,\chi',\xi')\in\cW(\ov\beta)$ and 
let $U''=\bs''\cap W_2\cap W_{x,\chi'\xi'}$. We show that 
$\vr(U'')\in\pr_{U''}\rel_{C'',x,\chi'\xi'}$; note that this proves not only 
$\cR$-compatibility, but also that $\vr$ is a solution: $(x,\chi',\xi')$
can be chosen to be $(w_2,\gm_2,\dl_2)$. Let 
$U'=W'\cap W$. Since $\bc'\in\umax(\cS'^\cR_{W\cap W_2})$, we have 
$\pr_{U'}\bc'\in\umax(\pr_{U'}\rel_{C'',x,\chi'\xi'})$. 
By (S2) for $\cR$ tuple $\pr_{U'}\bc'$ can be extended to an 
$\cR$-compatible solution $\sg$ 
from $\cS'_{W_2}$. By the choice of $f$, property (b),
$f(\pr_{U''-U'}\sg)=\pr_{U''-U'}\bd=\pr_{U''-U'}\bb$, and 
$f(\pr_{U'}\sg)=\pr_{U'}\bc^*$ by definition of $\bc^*$.  Since 
$\pr_{U'}\bb^*\sqq_{as}\pr_{U'}\bc^\dagger$ 
in $\Sg{\pr_{U'}\bb^*,\pr_{U'}\bc^*}$, we have 
$\pr_{U''}\bb^*\sqq_{as}\pr_{U''}\vr$ in 
$\Sg{\pr_{U''}\bb^*,\pr_{U''}\bc^*}$, implying 
by (S6) for $\cR$ that $\pr_{U''}\vr\in\rel_{C'',x,\chi'\xi'}$. 
As this is true for every constraint $C''$, $\vr$ is an $\cR$-compatible solution.

(2) The proof of item (1) shows in particular that for any 
$\bb\in\rela$ and $\bc\in(\pr_{W\cap W_2}\rela)\fac\tau$ such that
$\bc\eqc{\eta_1}(\pr_{W\cap W_2}\bb)^\tau$, we also
have $(\pr_{W\cap W_2}\bb)^\tau\eqc{\th_1}\bc$, as both $\bb^\tau$ and
$(\bc,\pr_H\bb)$ belong to $\rela'$. 

(3) follows from Proposition~\ref{pro:umax-rectangular}.
 
(4) follows from (1) and (3).

\medskip

We now can complete the proof of Lemma~\ref{lem:S2-1}.
If we show that $\ba$ can be extended to $\bb\in\rela$ such that 
$\pr_U\bb=\ba$ and such that $\pr_H\bb$ is u-maximal in a $\th_2$-block, 
Claim~2 implies (S2) for $\cR'$. Indeed, suppose there is 
$\bb\in\rela$ satisfying the above conditions. Since 
$\pr_X\ba\in\pr_X\cS'^{\cR'}_W\sse\pr_X\cS'^{\cR'}_{W\cap W_2}$, 
there is $\bc'\in\cS'^{\cR'}_{W\cap W_2}=\pr_{W\cap W_2}\rela$
such that $\pr_X\ba=\pr_X\bc'$. In particular, this means 
that $\bc=\bc'^\tau\in D$, where $D$ is the $\eta_1$-, and therefore 
$\th_1$-block containing $(\pr_{W\cap W_2}\bb)^\tau$, because 
$\pr_X\bc'=\pr_X\bb=\pr_X\ba$. Then as 
$(\bc,\pr_H\bb)\in\rela'$ by Claim~2(4), this
means that there is $\bd\in\umax(\cS'^\cR_{W_2})$ such that 
$\pr_U\bd=\ba$ and $\pr_{W\cap W_2}\bd\eqc\tau\bc'$; that is, 
$\bd\in\cS'^{\cR'}_{W_2}$.

Next we show such $\bb$ exists. Recall that $\rela^*$ denotes the $\tau$-block
of $\pr_{W\cap W_2}\relo=\cS'^\cR_{W\cap W_2}$ that contains 
$\cS'^{\cR'}_{W\cap W_2}$, and in particular, 
$\pr_{W\cap W_2}\cS'^{\cR'}_W$. First, observe that 
$$
\rela^*\cap\pr_{W\cap W_2}\rela=
\rela^*\cap\pr_{W\cap W_2}\cS'^\cR_{W_2}\ne\eps.
$$
Indeed, let $\bd\in\rela$ be such that 
$\pr_U\bd=\ba$. Then $\pr_{W\cap W_2}\bd\in\cS'^\cR_{W\cap W_2}$ and 
$\bd^\tau$ belongs to the same $\eta_1$-block as $\rela^*$, because 
$(\rela^*\fac\tau,\pr_X\ba)\in\relo'$. By Claim~2(1)
$\rela^*\fac\tau\in(\pr_{W\cap W_2}\rela)\fac\tau$, and as 
$\pr_{W\cap W_2}\rela\sse\cS'^\cR_{W\cap W_2}$, this proves the  
observation. 

Let $D$ be the $\th_1$-block containing $\rela^*\fac\tau$ and $E$ the 
corresponding $\th_2$-block. By what is proved above $\pr_H\bd\in E$. 
We now only need to show that $\bd$ can be chosen such that 
$\pr_H\bd\in \umax(E)$. Let $\pi$ be the congruence on 
$\pr_H\rela$ given by $\be_1\eqc\pi\be_2$ if and only if 
$\pr_{X'}\be_1\eqc{\tau_X}\pr_{X'}\be_2$. Then $\pr_H\rela\fac\pi$ 
is isomorphic 
to $\pr_X(\rel_{C,w_1,\gm_1\dl_1})\fac{\tau_X}$, in particular, it is a module
if $\tau_X$ is maximal, and 1-element otherwise. 
Let $\tau'_X$ denote the congruence of $\pr_U\rel_{C,w_1,\gm_1\dl_1}$
given by $\bc\eqc{\tau'_X}\bd$ if and only if $\pr_X\bc\eqc{\tau_X}\pr_X\bd$.
As is easily seen, if $G$ is a $\pi$-block of $\pr_H\rela$ then $\pr_UG$
is a $\tau'_X$-block of $\pr_U\rel_{C,w_1,\gm_1\dl_1}$. Therefore,
as $\ba$ is u-maximal in $\pr_U\rel'_{C,w_1,\gm_1\dl_1}$, by 
the Maximality Lemma~\ref{lem:to-max}(5) it can be extended to a 
u-maximal tuple $\ba'$ in a $\pi$-block $G$. Since $E\fac\pi$ is a module, by 
Lemma~\ref{lem:u-max-congruence} $\ba'$ is also u-maximal in $E$.
The result follows.
\end{proof}

Condition (S3) can be proved in a similar way, so we just indicate where
the proof of Lemma~\ref{lem:S2-1} has to be changed.

\begin{lemma}\label{lem:S3-1}
$\cR'$ satisfies condition (S3). That is, for every $C=\ang{\bs,\rel}\in\cC$,  
$\rel\cap\ov B'_\bs\ne\eps$ and for 
any $I\sse\bs$ any $\cR'$-compatible tuple $\ba\in\umax(\pr_I\rel)$ 
extends to an $\cR'$-compatible tuple $\bb\in\rel$.
\end{lemma}

\begin{proof}
If $\bs\cap W=\eps$ then Corollary~\ref{cor:cR-structure} implies that every 
u-maximal $\cR$-compatible tuple is also $\cR'$-compatible, and the result
follows from (S3) for $\cR$. Suppose that $W'=\bs\cap W\ne\eps$. Then,
as there is a $\cR'$-compatible tuple in $\umax(\pr_w\rel)$ for $w\in W'$,
the first part of (S3) follows from the second part.

To prove the second part we show that $\rel,\rel^\cR$ satisfy conditions
(X1)--(X5), and therefore the argument from the proof of 
Lemma~\ref{lem:S2-1} applies in this case as well. We take $\bs$ for
$W_2$ and set $Y=\eps$. Condition (X1) is trivial; condition (X2) follows 
from (S3) for $\cR$; and congruences $\tau'_{C'}(u,\eta\th)$ can be defined 
in the same way replacing $W_2$ with $\bs$, and they possess the same
properties required in (X3). Conditions (X4),(X5) follow from the definitions.
Now we use the same argument except that the proof of Claim~1 in this case 
is simpler, because we do not need to show something is a solution; just
that the required tuple belongs to a relation.
\end{proof}

\subsubsection{The remaining conditions}\label{sec:rest-1}

We can now complete the proof of Theorem~\ref{the:restricting-strategy}.

\begin{proof} (of Theorem~\ref{the:restricting-strategy}).
Conditions (S1)--(S3) and (S6) are proved in the preceding lemmas.

For (S4) observe that every $\rel'_{C,w,\gm\dl}$ is obtained as the intersection
of $\rel_{C,w,\gm\dl}$ with a block of $\tau_{C,w,\gm\dl}$, and therefore is a 
subalgebra. Also, since $\rel_{C,w,\gm\dl}\fac{\tau_{C,w,\gm\dl}}$ is a module, 
by Lemma~\ref{lem:u-max-congruence} 
$\umax(\rel'_{C,w,\gm\dl})\sse\umax(\rel_{C,w,\gm\dl})$ proving the first
part of (S4). 

To prove the rest of (S4) let $(w,\gm,\dl)\in\cW(\ov\beta)$, 
$U=W_{w,\gm\dl}$ and $\cS'_U$ be one of the sets $\cS_U$ if 
$(w,\gm,\dl)\not\in\cW'$, or $\cS_{U,Y}$ for $Y=\Razm(\cP)-\bs$ for some
$C=\ang{\bs,\rel}$ or $Y=\Razm(\cP)-W_{u,\eta\th}$ for some 
$(u,\eta,\th)\not\in\cW'$, if $(w,\gm,\dl)\in\cW'$. As in the proof of condition
(S2) we consider the congruence $\tau$ constructed as in (\ref{equ:tau}) 
with $U$ in place of $W_2$. Let $\relo$ be the $\tau$-block of 
$\pr_{W\cap U}\cS'_U$ containing $\cR'$-compatible tuples. By (S4) for $\cR$
there is a tuple $\ba\in\cS'^{\cR}_U$ that is in $\umax(\cS'_U)$. Since 
$\pr_{W\cap U}\cS'_U\fac\tau$ is a module, by the Maximality 
Lemma~\ref{lem:to-max}(4) and Lemma~\ref{lem:thin} there is an as-path 
$\ba=\ba_1\zd\ba_k$ in $\cS'^\cR_U$
such that $\pr_{W\cap U}\ba_k\in\relo$. The tuple $\ba_k$ belongs to 
$\cS'^{\cR'}_U$ and to $\umax(\cS'_U)$.  

For (S5), the existence of $A_{\cR',w}$ follows from (S3). Also, as in the
proof of (S4) $\umax(A_{\cR',w})\sse\umax(A_{\cR,w})$. The result now 
follows from (S5) for $\cR$.

Finally, (S7) follows from Lemma~\ref{lem:S7}.
\end{proof}

\subsection{Tightening non-affine factors}\label{sec:type-not-2}

Let $\cP=(V,\cC)$ be a block-minimal instance, let $\cR$ be a 
$\ov\beta$-strategy with respect to $\ov B$. Take $v\in V$ and 
$\al\in\Con(\zA_v)$ with $\al\prec\beta_v$ such that $\typ(\al,\beta_v)\ne\two$.
We tighten $\cR$ in two steps. In the first step we restrict $B_v$ to the 
subalgebra generated by an
as-component of $A_{\cR,v}$ obtaining a collection of relations that satisfies 
all the properties of a strategy except (S5) and (S6). In the second step we 
restrict the same domain to one $\al$-block and restore (S5) and (S6). 
Let $D$ be an as-component of $A_{\cR,v}\fac\al$, by (S5) $D$ is also an 
as-component of $B_v\fac\al$. Note that if 
$\typ(\al,\beta_v)\in\{\four,\five\}$ then by Lemma~\ref{lem:type23}
any as-component of $B_v\fac\al$ is a singleton and Step~2 is not needed.
Conditions (S5),(S6) in this case are proved as in Step~2.

\subsubsection{Step 1.}\label{sec:step-1}

Let $D'=\{a\in A_{\cR,v}\mid a^\al\in D\}$ and let $\wh D=\Sg{D'}$. We 
consider the problem $\cP'$ obtained from $\cP$ by restricting the domain 
of $v$ to $\wh D$ and the domain of $w\in V-\{v\}$ to $A_{\cR,w}$. We 
first show that $\cP'$ can be converted to a nonempty (2,3)-minimal 
instance that also satisfies some additional conditions.

In order to do that we introduce a family of binary relations, and then prove 
that this family is a (2,3)-strategy of $\cP'$. For $x,y\in V$, let 
\begin{eqnarray*}
\relo^x&=&\{a\in \amax(A_{\cR,x})\mid \text{ there is 
$d\in D$ such that } (d,a)\in\rel^{\{v,x\},\cR}\fac\al\},\\
\quad \text{and}&&\\
\relo^{xy}&=&\{(a,b)\in\amax(\rel^{\{x,y\},\cR})\mid\text{ there is 
$d\in D$ such that }\\
&& \qquad (d,a)\in\rel^{\{v,x\},\cR}\fac\al, (d,b)\in\rel^{\{v,y\},\cR}\fac\al\}.
\end{eqnarray*}\label{not:Qxy}
In particular $\relo^v=\amax(D')$.
We say that a tuple $\ba$ on a set $U\sse V$ (where $U$ can be, e.g.\ a 
subset of $\bs$ for a constraint $C=\ang{\bs,\rel}$, or a subset of 
$W_{w,\gm\dl}$ for some $(w,\gm,\dl)\in\cW(\ov\beta)$) 
is \emph{Q-compatible}\index{Q-compatible} if $(\ba[x],\ba[y])\in\relo^{xy}$ 
for any $x,y\in U$.

\begin{prop}\label{pro:3-tight}
(1) For any $x,y,z\in V$ and any $(a,b)\in\relo^{xy}$ there is 
$c\in\amax(A_{\cR,z})$
such that $(a,c)\in\relo^{xz}$ and $(b,c)\in\relo^{yz}$.\\ 
(2) For any $C=\ang{\bs,\rel}$ let $\rel^\cR$ denote the set of 
$\cR$-compatible tuples from $\rel$. For any $I\sse\bs$ and any Q-compatible 
$\ba\in\amax(\pr_I\rel^\cR)$, there is $\ba'\in\amax(\rel^\cR)$ that is 
Q-compatible,  and $\pr_I\ba'=\ba$.\\
(3) For any $(w,\gm,\dl)\in\cW(\ov\beta)$, any $U\sse W_{w,\gm\dl}$, 
and any $\ba\in\amax(\pr_U\cS'^\cR_{W_{w,\gm\dl}})$, where 
$\cS'^\cR_{W_{w,\gm\dl}}$ is the set of solutions of 
$(\cP_{W_{w,\gm\dl}})\fac{\ov\mu^Y}$ for some set $Y$ from the 
definition of block-minimality, there is 
$\ba'\in\amax(\cS'^\cR_{W_{w,\gm\dl}})$ that is Q-compatible and 
$\pr_U\ba'=\ba$.
\end{prop}

\begin{proof}
For $x,y\in V$ let $\relo^{vxy}$ denote the set of tuples $(d,a,b)$ such that
$(d,a)\in\rel^{\{v,x\},\cR}\fac\al,(d,b)\in\rel^{\{v,y\},\cR}\fac\al,
(a,b)\in\rel^{\{x,y\},\cR}$.

\smallskip

\noindent
{\sc Claim 1.} The set $\relo^x$ is as-closed in $A_{\rel,x}$, and $\relo^{xy}$ 
is as-closed if $\rel^{\{x,y\},\cR}$.

\smallskip

Let $(a,b)\in\relo^{xy}$. By the Maximality Lemma~\ref{lem:to-max}(3) either 
$\relo^{vxy}$ contains a subdirect 
product of $D$ and $\as(a,b)$, or $(D\tm\as(a,b))\cap\relo^{vxy}=\eps$. Since 
$(a,b)\in\relo^{xy}$ the former option holds. For the first part of the claim 
observe that $\pr_x\as(a,b)=\as(a)$.

\smallskip

\noindent
{\sc Claim 2.} For any $x,y\in V$, $\relo^{xy}$ is a subdirect product of 
$\relo^x\tm\relo^y$.

\smallskip

Let $a\in\relo^x$, then there is $d\in D$ with $(d,a)\in\rel^{\{v,x\},\cR}$. 
By (S1) for $\cR$, $\{\umax(\rel^{X,\cR})\}_{X\sse V,|X|=2}$ is (2,3)-strategy, and there 
is $b\in\zA_y$ with $(d,a,b)\in\relo^{vxy}$; then $(a,b)\in\relo^{xy}$.
\smallskip

We prove (2), the proof of (3) is basically identical, and we explain how to 
modify this proof to prove (1).

By induction on $i$ we prove that a Q-compatible tuple 
$\ba\in\amax(\pr_I\rel^\cR)$ can be found for any $I\sse\bs$, $|I|=i$. 
Moreover, for any Q-compatible $\ba'\in\amax(\pr_{I-\{u\}}\rel^\cR)$, a 
Q-compatible $\bb\in\amax(\pr_I\rel^\cR)$ can be found such that 
$\pr_{I-\{u\}}\bb=\ba'$. 

First we consider the case $u\ne v$. 
For $i=2$ the existence of $\ba$ follows from Claim~2. So, 
suppose such a tuple exists for any $I\sse\bs$ with $|I|\le i$.
Let $I\sse\bs$, $|I|=i+1$, $y\in I$, 
$I'=I-\{y\}$, and $\ba\in\amax(\pr_{I'}\rel^\cR)$ is Q-compatible. 
Let also $t\in I'$, $I''=I'-\{t\}$. Without loss of generality assume 
$I=\{\vc x{i+1}\}$, $y=x_{i+1}$, $t=x_i$. Consider the relation given by  
\begin{eqnarray}
\lefteqn{\relo(\vc xi,\vc zi) = \exists y\ \  \pr_{I''}\rel^\cR(\vc x{i-1})}
\label{equ:formula} \\
&& \meet\bigwedge_{j=1}^i \left(\rel^{\{x_j,v\},\cR}\fac\al (x_j,z_j)\meet
\rel^{\{v,x_{i+1}\},\cR}\fac\al (z_j,y)\meet
\rel^{\{x_j,x_{i+1}\},\cR}(x_j,y)\right).\nonumber
\end{eqnarray}
It suffices to prove that $\ba''=(\ba,\be)\in\relo$, where $\be[j]\in D$ for each
$j\in[i]$, since this would mean that there is a $c\in\zA_{\cR,x_{i+1}}$
with the required properties. 

Observe first that $\umax(\pr_{I'}\rel^\cR)\sse\pr_{I'}\relo$. Indeed, any
$\bb\in\pr_{I'}\rel^\cR$ by (S3) can be extended to $\bb'\in\pr_I\rel^\cR$;
then the values of the variables $z_j$ can be chosen by (2,3)-consistency and (S1).
This also implies that $\umax(\pr_{I''}\relo)=\umax(\pr_{I''}\rel^\cR)$.
Since $\al\prec\beta_v$, by (S7) for $\cR$ and the Congruence Lemma~\ref{lem:affine-link} 
for any $w\in V$ 
the relation $\rel^{\{v,w\},\cR}\fac\al$ either contains $D\tm\umax(\relo^w)$,
or $\rel^{\{v,w\},\cR}\fac\al\cap(D\tm\relo^w)$ is the graph of a 
mapping $\kp_w:\relo^w\to D$. Let the
set of variables for which the latter option holds be denoted by
$Z$. By construction, for any $w_1,w_2\in Z$, and any 
$(c,d)\in\relo^{w_1w_2}$, we have $\kp_{w_1}(c)=\kp_{w_2}(d)$. 
For any $w\in V-Z$, the set $\relo^w$ is as-closed, and 
$D\tm\umax(\relo^w)\sse\rel^{\{v,w\},\cR}\fac\al$.

Let $J=I'\cap Z$ and, if $J\ne\eps$, let $d=\kp_w(\ba[w])$ for any $w\in J$. 
If $J=\eps$, but $x_{i+1}\in Z$, then let $(\pr_{I''}\ba,c)$ be an extension of 
$\pr_{I''}\ba$ to a $Q$-compatible tuple from $\pr_{I\cup\{x_{i+1}\}}\rel$ 
and set $d=\kp_{x_{i+1}}(c)$. If $x_{i+1}\not\in Z$, then let $d$ be any 
element of $D$. Then $d$ is such that $(\ba[w],d)\in\rel^{\{w,v\},\cR}\fac\al$ 
for any $w\in I'$. 

Consider the tuple $\bb=(\ba,d\zd d)$; we show that it satisfies the conditions 
of the 2-Decomposition Theorem~\ref{the:quasi-2-decomp} with $X=I''$. Note that we cannot 
replace $I''$ with $I'$ here, because in order to apply the 2-Decomposition 
Theorem~\ref{the:quasi-2-decomp} $\ba$ has to be as-maximal 
$\pr_{I'}\relo$, which may not be true.
By what is observed before, $\pr_{I''}\ba\in\pr_{I''}\relo$, and for any 
$s_1,s_2\in I'$ we have $(\ba[s_1],\ba[s_2])\in\pr_{s_1s_2}\relo$.
We now show that for any of the remaining pairs of variables 
$x,z\in\{\vc xi,\vc zi\}$ $(\bb[x],\bb[z])\in\pr_{xz}\relo$. 
Let $x\in\{\vc xi\}$ and $z\in\{\vc zi\}$. If
$x\ne x_i$, then by the inductive hypothesis $\pr_{I''}\ba$ can be extended 
to some value $c$ of $y$ such that $(\ba[s],c)\in\relo^{sy}$ for any 
$s\in I''$. Then there is $c'\in\zA_{x_i}$ such that $(c',c)\in\relo^{x_iy}$.
Also, $(c,d)\in\relo^{yv}$ and $c'$ can be chosen such that 
$(c',d)\in\relo^{x_iv}$. If $x=x_i$, then find a value 
$c$ for $y$ such that $(\ba[x_i],c)\in\relo^{x_i,x_{i+1}}$, and then 
extend $c$ to a tuple on $I''\cup\{y\}$ by induction hypothesis. The values of
$z_j$ can be set using (2,3)-consistency. Finally, if $x,z\in\{\vc zi\}$, we 
proceed as in one of the previous cases.  

By the 2-Decomposition Theorem~\ref{the:quasi-2-decomp} there is $\bb'\in\relo$ such that
$\pr_{I''}\ba\sqq_{as}\pr_{I''}\bb'$ in $\relo'=\pr_{I''}\relo$,
$(\ba[x],\ba[z])\sqq_{as}(\bb'[x],\bb'[z])$ in $\rel^{\{x,z\},\cR}$ for any 
$x,z\in I'$, and $d\sqq_{as}\bb'[z_j]$ in $A_{\cR,v}\fac\al$ for any $j\in[i]$. 

Let $\lnk_1,\lnk_2$ be the link congruences of $\relo',\zA_{x_i}$ with
respect to $\pr_{I'}\relo$. Since $\umax(\pr_{I'}\rel^\cR)\sse\pr_{I'}\relo$, 
the link congruences of $\umax(\pr_{I''}\rel^\cR),\zA_{x_i}$ with respect to 
$\pr_{I'}\rel^\cR$ are smaller than $\lnk_1,\lnk_2$. Therefore the 
$\lnk_1$- and $\lnk_2$-blocks $A,B$ containing $\pr_{I''}\ba$ and $\ba[x_i]$, 
respectively, are such that $\relo''=\pr_{I'}\relo\cap(A\tm B)\ne\eps$. 
Choose $\ba'\in\pr_{I'}\rel^\cR\cap(A\tm B)$
such that $\ba\sqq_{as}\ba'$ in $\pr_{I'}\rel^\cR$ and $\ba'$ is
as-maximal in $\relo''$. As is easily seen, such a tuple
exists by the Rectangularity Corollary~\ref{cor:linkage-rectangularity}, because, since 
$\relo''$ is linked, any $\ba'\in\relo''$ such that
$\pr_{I''}\ba'$ is as-maximal in $A$ and $\ba'[x_i]$ is as-maximal in $B$ 
is as-maximal in $\relo''$. Now, consider
\begin{eqnarray*}
\lefteqn{\rela(\vc xi,\vc zi)=\exists y A(\vc x{i-1})\meet B(x_i)}\\
&&  \meet\bigwedge_{j=1}^i \rel^{\cR,\{x_j,v\}}\fac\al (x_j,z_j)
\meet\rel^{\cR,\{v,x_{i+1}\}}\fac\al (z_j,y)\meet 
\rel^{\cR,\{x_j,x_{i+1}\}}(x_j,y).
\end{eqnarray*}
By the same argument as before, there is $\bc\in\rela$ such that 
$\ba'\sqq_{as}\pr_{I'}\bc$ in $\pr_{I'}\rela$ and $\bc[z_j]\in D$ for $j\in[i]$. 
Since $\ba'$ is as-maximal
in $\pr_{I'}\rela$, we also have $\pr_{I'}\bc\sqq_{as}\ba'$ in $\pr_{I'}\rela$.
Therefore there is $c\in A_{\cR,x_{i+1}}$ such that 
$(\ba'[x],c)\in\relo^{xx_{i+1}}$ for every $x\in I'$. 
By the 2-Decomposition Theorem~\ref{the:quasi-2-decomp} applied to $\rel^\cR$, there is 
$\ba''\in\pr_I\rel$ such that $\pr_{I'}\ba''=\ba'$ and 
$c\sqq_{as}\ba''[x_{i+1}]$. Finally, since $\ba$ is as-maximal in 
$\pr_{I'}\rel^\cR$, there is an as-path from $\ba'$ to $\ba$ in 
$\pr_{I'}\rel^\cR$ and we complete the proof by the Maximality Lemma~\ref{lem:to-max}(4).  

Next, we consider the case $u=v$. Let $I=\{\vc xi,v\}$, 
$I'=\{\vc xi\}$ and $I''=\{\vc x{i-1}\}$. Note that by reordering the 
variables we may assume that if $x_i\in Z$ then 
$Z\cap\{\vc x{i-1}\}\ne\eps$. By the induction hypothesis 
there is $c\in A_{\cR,v}$ such that $(\pr_{I''}\ba,c)$ belongs to 
$\amax(\pr_{I''\cup\{v\}}\rel^\cR)$ and is Q-compatible, in particular,
$c^\al\in D$. We consider the tuple $(\ba,c^\al)$ and relation 
$\rel'=\pr_I\rel^\cR\fac\al$. We have $\ba\in\pr_{I'}\rel'$. For any $j\in[i-1]$ 
we have $(\ba[j],c)\in\relo^{x_jv}$ by the choice of $c$. If $x_i\not\in Z$
then there is $c'\eqc\al c$ such that $(\ba[x_i],c')\in\rel^{\{x_i,v\},\cR}$.
By (S3) for $\cR$, $(\ba[x_i],c')$ extends to a tuple from $\rel^\cR$, 
therefore $(\ba[x_i],c)\in\pr_{x_iv}\rel'$. If $x_i\in Z$, then there is also
some $j\in I''\cap Z$. Then since $(\ba[j],\ba[i])\in\relo^{x_jx_i}$,
we have $\kp_{x_j}(\ba[x_j])=\kp_{x_i}(\ba[x_i])=c$. Therefore 
there is $c'\eqc\al c$ such that $(\ba[x_i],c')\in\rel^{\{x_i,v\},\cR}$ 
and we continue as before.
By the 2-Decomposition Theorem~\ref{the:quasi-2-decomp} there is $\bb\in\rel'$ such that 
$\pr_{I'}\bb=\ba$ and $c\sqq_{as}\bb[v]$ in $A_{\cR,v}\fac\al$, that is, 
$\bb[v]\in D$. Therefore there is $\bc\in\pr_I\rel$ such that 
$\pr_{I'}\bc=\ba$ and $\bc[v]^\al\in D$, as required.

\smallskip

Item (3) can be proved in the same way, as $\cP$ is block-minimal.
For (1) we need to make two changes. First, we apply the argument above for 
$i=2$ and stop before the last application of the 2-Decomposition Theorem~\ref{the:quasi-2-decomp}.
Second, we need to consider the case when extending a pair from $\relo^{xy}$
by a value of $v$. More precisely, let $(a,b)\in\relo^{xy}$, $x,y\in V$, 
we need to find $c\in A_{\cR,v}$ such that $c^\al\in D$ and 
$(a,c)\in\relo^{xv}$, $(b,c)\in\relo^{yv}$. Let 
$$
\relo(x,y,v)=\rel^{\{x,y\},\cR}(x,y)\meet\rel^{\{x,v\},\cR}(x,v)\meet
\rel^{\{y,v\},\cR}(y,v).
$$ 
By (2,3)-consistency and construction $(a,b)\in\pr_{xy}\relo, 
(a,d_1)\in\pr_{xv}\relo,(b,d_2)\in\pr_{yv}\relo$, where 
$d_1^\al=d_2^\al\in D$. The relation $\relo'=\relo\fac\al$ then
satisfies the conditions: $(a,b)\in\pr_{xy}\relo', (a,d)\in\pr_{xv}\relo',
(b,d)\in\pr_{yv}\relo'$ where $d=d_1^\al=d_2^\al$. By the 2-Decomposition 
Theorem~\ref{the:quasi-2-decomp} there is a tuple $(a,b,d')\in\relo'$
such that $d\sqq_{as} d'$, that is, $d'\in D$. Therefore $\relo$ contains
a tuple $(a,b,c)$ for some $c$ with $c^\al=d'$. This $c$ is as required.
\end{proof}

Let $\cP''$ be the problem obtained from $\cP'$ as follows: establish
(2,3)-minimality and, for any $(w,\gm,\dl)\in\cW(\ov\beta)$, 
establish minimality of $\cP'_{W_{w,\gm\dl}}\fac{\ov\mu^Y}$,
where $Y$ is one of the sets specified in the definition of block-minimality, 
until the instance does not change any longer.
Let $A''_{\cR,w}$ be the domains of $w\in V$ for $\cP''$, for 
$C=\ang{\bs,\rel}\in\cC$ let $\rel''$ denote the corresponding constraint 
relation of $\cP''$; for $(w,\gm,\dl)\in\cW(\ov\beta)$ let 
$\cS''_{W_{w,\gm\dl},Y}$ denote the set of solutions of 
$\cP''_{W_{w,\gm\dl}}\fac{\ov\mu^Y}$; finally 
$\rel''_{C,w,\gm\dl}$ denote the set of tuples from $\rel_{C,w,\gm\dl}$
extending to a solution of $\cP''_{W_{w,\gm\dl}}\fac{\ov\mu^Y}$.

The next corollary follows straightforwardly from Proposition~\ref{pro:3-tight},
because establishing (2,3)-minimality or minimality never eliminates an
as-maximal Q-compatible tuples.

\begin{corollary}\label{cor:23-minimality}
The sets $A''_{\cR,w}$, $\rel''$, $\cS''_{W_{w,\gm\dl},Y}$, and
$\rel''_{C,w,\gm\dl}$ contain all the as-maximal Q-compatible tuples from
$A_{\cR,w}$, $\rel$, $\cS_{W_{w,\gm\dl},Y}$, and
$\rel_{C,w,\gm\dl}$, respectively.
\end{corollary}

Now we can show that the collection of relations
$$
\cR''=\{\rel''_{C,w,\gm\dl}\mid C\in\cC, (w,\gm,\dl)\in\cW(\ov\beta)\}
$$
is almost a $\ov\beta$-strategy.

\begin{theorem}\label{the:non-as-restriction}
The collection of relations $\cR''$ constructed as above satisfies all the
conditions of a $\ov\beta$-strategy with respect to $\ov B$,
except (S5) and (S6).
\end{theorem}

\begin{proof}
Condition (S1) follows from Corollary~\ref{cor:23-minimality}. 
Conditions (S2),(S3) follow from Corollary~\ref{cor:23-minimality} and 
construction (establishing block-minimality). 
For (S4) every relation $\rel''_{C,w,\gm\dl}$ is a subalgebra in 
$\rel_{C,w,\gm\dl}$ and therefore in 
$\pr_{\bs\cap W_{w,\gm\dl}}\rel$ by construction. Moreover, as 
by Proposition~\ref{pro:3-tight} $\rel''_{C,w,\gm\dl}$ contains an element 
as-maximal in $\rel_{C,w,\gm\dl}$, we have 
$\umax(\rel''_{C,w,\gm\dl})\sse\umax(\rel_{C,w,\gm\dl})$.
Finally, condition (S7) follows from (S7) for $\cR$.
\end{proof}

We will need another property of $\cP''$. Unfortunately, Q-compatible tuples are 
not very helpful in establishing properties (S5),(S6), since those are properties of 
u-maximal fragments of relations. Therefore we need to extend Q-compatibility 
to u-maximal elements. Similar to Q-compatibility we make 
the following definition. For $x,y\in V$, let 
\begin{eqnarray*}
\relp^x&=&\{a\in \umax(A_{\cR,x})\mid \text{ there is 
$d\in D$ such that } (d,a)\in\rel^{\{v,x\},\cR}\fac\al\},\\
\quad \text{and}&&\\
\relp^{xy}&=&\{(a,b)\in\umax(\rel^{\{x,y\},\cR})\mid\text{ there is 
$d\in D$ such that }\\
&& \qquad (d,a)\in\rel^{\{v,x\},\cR}\fac\al, (d,b)\in\rel^{\{v,y\},\cR}\fac\al\}.
\end{eqnarray*}\label{not:Pxy}
In particular $\relp^v=\umax(\wh D)$.
Note that these relations are different from $\relo^x,\relo^{xy}$ in that they
consist of u-maximal elements and pairs, rather than as-maximal as 
$\relo^x,\relo^{xy}$.
We say that a tuple $\ba$ on a set $U\sse V$ (where $U$ can be, e.g.\ a 
subset of $\bs$ for a constraint $C=\ang{\bs,\rel}$, or a subset of 
$W_{w,\gm\dl}$ for some $(w,\gm,\dl)\in\cW(\ov\beta)$) 
is \emph{P-compatible}\index{P-compatible} if $(\ba[x],\ba[y])\in\relp^{xy}$ 
for any $x,y\in U$.

Let $\cP'''$ be the instance obtained as follows: First, restrict the domains
and relations of $\cP$ to the sets of P-compatible tuples they contain;
let the new relation for $C=\ang{\bs,\rel}\in\cC$ be denoted by 
$\rel^\dagger$. Second, establish (2,3)-minimality and block-minimality
of the resulting instance. Let $\rel''', \rel'''_{C,w,\gm\dl},
\cS'''_{W_{w,\gm\dl},Y}$ denote the relations induced by $\cP'''$. 
Note that the domains and relations of $\cP'''$ are not necessarily subalgebras.

\begin{lemma}\label{lem:step-1-as-closeness}
Let $C=\ang{\bs,\rel}\in\cC$, $(w,\gm,\dl)\in\cW(\ov\beta)$.\\[2mm]
(1) Relations $\rel''', \rel'''_{C,w,\gm\dl}, \cS'''_{W_{w,\gm\dl},Y}$
are nonempty, and $\rel'''\sse\rel''$, $\rel'''_{C,w,\gm\dl}\sse\rel''_{C,w,\gm\dl}$,
$\cS'''_{W_{w,\gm\dl},Y}\sse\cS''_{W_{w,\gm\dl},Y}$.
\\[1mm]
(2) Relations $\rel'''_{C,w,\gm\dl}$ and $\cS'''_{W_{w,\gm,\dl,\ov\beta},Y}$ 
are as-closed in $\rel_{C,w,\gm\dl}$, and in
$\cS'^\cR_{W_{w,\gm,\dl}}$, respectively 
(recall that $\cS'^\cR_{W_{w,\gm,\dl,\ov\beta}}$ is the set of $\cR$-compatible 
tuples from $\cS_{W_{w,\gm,\dl,\ov\beta},Y}$).
\end{lemma}

\begin{proof}
(1) The inclusions $\rel'''\sse\rel''$, $\rel'''_{C,w,\gm\dl}\sse\rel''_{C,w,\gm\dl}$,
$\cS'''_{W_{w,\gm\dl,\ov\beta},Y}\sse\cS''_{W_{w,\gm\dl,\ov\beta},Y}$
follow from the construction, as $\rel^\dagger\sse\rel'$. On the other 
hand, as every Q-compatible tuple from $\rel$ belongs to $\rel^\dagger$,
the nonemptyness follows by Corollary~\ref{cor:23-minimality}.

(2) First we observe that for any $x,y\in V$ the relation $\relp^{xy}$, and 
therefore $\relp^x$, is as-closed in $\umax(\rel^{\{x,y\},\cR})$.
This can be done in the same way as in the proof of 
Proposition~\ref{pro:3-tight}. Let $(a,b)\in\relp^{xy}$ and 
$(a',b')\in\umax(\rel^{\{x,y\},\cR})$ be such that
$(a,b)\sqq_{as}(a',b')$. We need to find  $d\in D$ such that 
$(a',d)\in\rel^{\{x,v\},\cR}\fac\al$, $(b',d)\in\rel^{\{y,v\},\cR}\fac\al$. 
Let 
$$
\relo(x,y,v)=\rel^{\{x,y\},\cR}(x,y)\meet\rel^{\{x,v\},\cR}\fac\al(x,v)\meet
\rel^{\{y,v\},\cR}\fac\al(y,v),
$$ 
which is a subalgebra of $\zA_x\tm\zA_y\tm\zA_v\fac\al$. Since
$(a,b)\in\relp^{xy}$ there is $c\in D$ with $(a,b,c)\in\relo$.
By (S1) for $\cR$ we have $(a',b')\in\pr_{xy}\relo$, moreover, 
$(a,b)\sqq_{as}(a',b')$ in $\pr_{xy}\relo$. By the Maximality Lemma~\ref{lem:to-max}(4)
there is $d\in A_{\cR,v}$ such that $(a',b',d^\al)\in\relo$ and 
$(a,b,c)\sqq_{as}(a',b',d^\al)$ in $\relo$. Therefore $d^\al\in D$.

It suffices to prove the statement for relations of the form $\rel'''_{C,w,\gm\dl}$,
because relations $\cS'''_{W_{w,\gm,\dl},Y}$ are pp-definable
through $\rel'''_{C,w,\gm\dl}$, and we can use Lemma~\ref{lem:poly-closed}(2). 
Since $\relp^{xy}$ are as-closed, every relation 
$\rel^\dagger_{C,w,\gm\dl}$ is also as-closed in $\umax(\rel_{C,w,\gm\dl})$.
We prove by induction that this property is preserved as 
(2,3)-minimality and block-minimality is being established. The observation
about the relations $\rel^\dagger_{C,w,\gm\dl}$ establishes the base case.
For the induction step, let $\rel^{\ddagger}_{C,w,\gm\dl}$ and 
$\cS^{\ddagger}_{W_{w,\gm\dl},Y}$ denote the current 
state of the corresponding relations, and they are as-closed.
There are two cases for the induction step.

In the first case we make a step to enforce (2,3)-minimality, that is, for some
$x,y,z\in V$ we check whether or not some 
$(a,b)\in\rel^{\ddagger \{x,y\}}$ can be extended by $c\in B_z$ 
such that $(a,c)\in\rel^{\ddagger \{x,z\}}$ and 
$(b,c)\in\rel^{\ddagger \{y,z\}}$. Let $(a,b)\in\rel^{\ddagger \{x,y\}}$
be such that  there are $c\in B_z$ with $(a,c)\in\rel^{\ddagger \{x,z\}}$, 
$(b,c)\in\rel^{\ddagger \{y,z\}}$, and let $(a',b')\in\rel^{\ddagger \{x,y\}}$ 
such that $(a,b)\sqq_{as}(a',b')$ in $\rel^{\ddagger \{x,y\}}$. Then
there is $c'\in B_z$ such that $(a',c')\in\rel^{\{x,z\}}$, $(b',c')\in\rel^{\{y,z\}}$. 
Similar to part (1) let 
$$
\relo(x,y,z)=\Sg{\rel^{\ddagger \{x,y\}}}(x,y)\meet\rel^{\{x,v\},\cR}(x,v)\meet
\rel^{\{y,v\},\cR}(y,v),
$$ 
We have $\Sg{\rel^{\ddagger \{x,y\}}}=\pr_{xy}\relo$ by (S1) for $\cR$,
and $(a,b,c),(a',b',c')\in\relo$. By the Maximality Lemma~\ref{lem:to-max}(4) we may assume 
that $(a,b,c)\sqq_{as}(a',b',c')$ in $\relo$. Since 
$\rel^{\ddagger \{x,z\}}, \rel^{\ddagger \{x,z\}}$ are as-closed,
we have $(a',c')\in\rel^{\ddagger \{x,z\}}, (b',c')\in\rel^{\ddagger \{x,z\}}$

In the second case we solve a subproblem of the form 
$\cP^{\ddagger}_{W_{w,\gm\dl}}\fac{\ov\mu^Y}$.
Let $U=W_{w,\gm\dl}$, let $\cS^{\ddagger}_{W_{w,\gm\dl},Y}$ be the 
corresponding set of solutions. Take $C=\ang{\bs,\rel}\in\cC$ and 
$(u,\eta,\th)\in\cW(\ov\beta)$, let $U'=\bs\cap U\cap W_{u,\eta\th}$.
Suppose $\ba\in\rel^\ddagger_{C,u,\eta\th}$ be such that there is a solution
$\vf\in\cS^{\ddagger}_{W_{w,\gm\dl,\ov\beta},Y}$ extending 
$\pr_{U'}\ba$, and let $\bb\in\rel^\ddagger_{C,u,\eta\th}$
is such that $\ba\sqq_{as}\bb$ in $\rel^\ddagger_{C,u,\eta\th}$. We need to
show that $\pr_{U'}\bb$ is extendible to a solution from 
$\cS^{\ddagger}_{W_{w,\gm\dl},Y}$. Let 
$\psi\in\cS'^\cR_{W_{w,\gm\dl},Y}$ be a solution extending $\pr_{U'}\bb$. 
Since $\pr_{U'}\ba\sqq_{as}\pr_{U'}\bb$ by the Maximality 
Lemma~\ref{lem:to-max}(4) we may assume that $\vf\sqq_{as}\psi$ in 
$\cS'^\cR_{W_{w,\gm\dl,\ov\beta},Y}$. Since by the induction hypothesis 
$\cS^{\ddagger}_{W_{w,\gm\dl,\ov\beta},Y}$ is as-closed,
the result follows.
\end{proof}

\subsubsection{Step 2.}\label{sec:step-2}

In this step we tighten the `near-strategy' $\cR''$ in a way similar 
to that from Section~\ref{sec:type-2}. We start with showing that the domains 
of all variables in $W_{v,\al\beta}$ have to be tightened.

\begin{lemma}\label{lem:binary-connections}
For every $w\in W=W_{v,\al\beta_v}$ there is a congruence 
$\al_w\in\Con(\zA_w)$ with $\al_w\prec\beta_w$, and such that  
$\cS_{vw}\cap(B_v\tm B_w)$ is aligned with respect to $(\al,\al_w)$, 
that is, for any $(a_1,a_2),(b_1,b_2)\in\cS_{vw}\cap(B_v\tm B_w)$, 
$a_1\eqc\al b_1$ if and only if $a_2\eqc{\al_w}b_2$.
\end{lemma}

\begin{proof}
It suffices to show that the link congruences $\lnk_1,\lnk_2$ of $\relo=\cS_{vw}$ 
viewed as a subdirect product of $\zA_v\tm\zA_w$ are such that 
$\beta_v\meet\lnk_1\le\al$ and $\beta_w\meet\lnk_2<\beta_w$. Since 
$w\in W$ there are $\gm, \dl\in\Con(\zA_w)$ such that 
$\gm\prec\dl\le\beta_w$ and $(\al,\beta_v)$ and $(\gm,\dl)$ cannot be 
separated. By Lemma~\ref{lem:34-links} it follows that  
$\beta_v\meet\lnk_1\le\al$ and $\lnk_2\meet\dl\le\gm$. We set 
$\al_w=\beta_w\meet\lnk_2<\beta_w$. Since $B_w\fac{\al_w}$
is isomorphic to $B_v\fac\al$, $\al_w\prec\beta_w$.
\end{proof}

Let $\beta'_v=\al$, $\beta'_w=\al_w$ for $w\in W=W_{v,\al\beta_v}$, 
and $\beta'_w=\beta_w$ for $w\in V-W$. 
Lemma~\ref{lem:binary-connections} implies that there is an isomorphism
$\nu_w:B_v\fac{\beta'_v}\to B_w\fac{\beta'_w}$. Choose an
as-maximal $\beta'_v$-block $B$, an element of $D$ from Step~1  
and set $B'_v=B$, $B'_w=\nu_w(B)$ for $w\in V-W$, 
and $B'_w=B_w$ for $w\in V-W$. 
Let $\cP^*$ be the problem instance obtained from $\cP''$ as follows: 
first restrict the domain of $w\in W$ in $\cP''$ to $B'_w$, then establish the 
(2,3)-minimality of the resulting problem, and finally, establish the minimality 
of all problems of the form $\cP''_{W_{w,\gm\dl}}\fac{\ov\mu^Y}$ for 
$(w,\gm,\dl)\in\cW(\ov\beta)$, where $Y$ is a set specified in the definition
of block-minimality for $\cP$. 

Let $\cR^*$ be the following collection of relations; 
\begin{itemize}
\item[(T1)]
$\cR^*=\{\rel^*_{C,w,\gm\dl}\mid C=
\ang{\bs,\rel}\in\cC, (w,\gm,\dl)\in\cW(\ov\beta')\}$;
\item[(T2)]
for every $C=\ang{\bs,\rel}\in\cC$, $(u,\gm,\dl)\in\cW(\ov\beta')$,
$\rel^*_{C,u,\gm\dl}=\pr_{\bs\cap W_{u,\gm\dl}}\rel^*$,
where $\rel^*$ is the constraint relation of $\cP^*$ obtained from $\rel$.
\end{itemize}

\begin{lemma}\label{lem:step-2-relations}
(1) For every constraint $C=\ang{\bs,\rel}\in\cC$, and every 
$\ba\in\umax(\rel'')$ such that $\ba[u]\in\nu_u(D)$ for $u\in\bs\cap W$
there is a tuple $\bb\in\umax(\rel'')$ such that 
$\pr_{\bs-W}\bb=\pr_{\bs-W}\ba$ 
and $\bb[u]\in B'_u$ for $u\in\bs\cap W$.\\[1mm]
(2) Let $C=\ang{\bs,\rel}\in\cC$, $(w,\gm,\dl)\in\cW(\ov\beta')$,
$W'=W_{w,\gm\dl}$, and $\ba\in\umax(\rel''_{C,w,\gm\dl})$ such that 
$\ba[u]\in\nu_u(D)$ for $u\in \bs\cap W'\cap W$. Then there is 
$\bb\in\umax(\rel''_{C,w,\gm\dl}))$ such that 
$\pr_{(\bs\cap W')-W}\bb=\pr_{(\bs\cap W')-W}\ba$ and
$\bb[u]\in B'_u$ for $u\in\bs\cap W'\cap W$.\\[1mm]
(3) Let $(w,\gm,\dl)\in\cW(\ov\beta')$ and $W'=W_{w,\gm\dl}$.
Let $\cS''_{W'}$ be the set of solutions of $\cP''_{W'}\fac{\ov\mu^Y}$, where 
$Y=\eps$ if $(w,\gm,\dl)\not\in\cW'$, and is one of the sets specified in
the definition of block-minimality otherwise.
For every solution $\vf\in\umax(\cS''_{W'})$ such that 
$\vf[u]\in\nu_u(D)$ for $u\in W'\cap W$ there 
is a solution $\psi\in\umax(\cS''_{W'})$ such that 
$\psi(u)=\vf(u)$ for $u\in W'-W$ and $\psi(u)\in B'(u)$ for $u\in W'\cap W$.
\end{lemma}

\begin{proof}
(1) Let $U_1=\bs\cap W$ and $U_2=\bs- W$. If $U_1=\eps$ there is 
nothing to prove; assume $U_1\ne\eps$. It suffices to consider 
$\relo=\rel''\fac{\ov\al'}$ where $\al'_u=\beta'_u$ if $u\in U_1$ and 
$\al'_u=\zz_u$
otherwise. So, we assume $\beta'_u=\zz_u$ for all $u\in U_1$. Then for any 
$u_1,u_2\in U_1$ and any $\bd\in\rel$ such that 
$\bd[u_1]\in B_{u_1},\bd[u_2]\in B_{u_2}$ we have 
$\bd[u_2]=\nu_{u_2}\circ\nu_{u_1}^{-1}(\bd[u_1])$. Therefore we 
may assume that $|U_1|=1$, say, $U_1=\{u\}$. 

Considering $\rel''$ as a subalgebra of $\zA_u\tm\pr_{\bs-\{u\}}\rel$,
the result follows by the Congruence Lemma~\ref{lem:affine-link}. Indeed, since there is a 
$\al_u\beta_u$-collapsing polynomial $f$ of $\rel$, that is, 
$f(\ov\beta_{\bs-\{u\}})\sse\zz_{\bs-\{u\}}$, there are no 
$\eta,\th\in\Con(\pr_{\bs-\{u\}}\rel)$ 
with $\eta\prec\th\le\ov\beta_{\bs-\{u\}}$ such that $(\al_u,\beta_u)$ cannot
be separated from $(\eta,\th)$. 

(2) and (3) are proved in essentially the same way.
\end{proof}

To show that $\cP^*$ has the desirable properties, in particular, 
it is nonempty, we consider a collection of unary and binary relations similar 
to $\relo^x,\relo^{xy}$ from Step~1. For $x,y\in V$ let $\reli^x,\reli^{xy}$
denote the following sets: 
\begin{eqnarray*}
\reli^x &=& \{a\in \amax(A''_{\cR,x})\mid (a,c)\in\amax(\rel''^{\{x,v\},\cR}) 
\text{ for some $c\in B$}\};\\
\reli^{xy} &=& \{(a,b)\in\amax(\rel''^{\{x,y\},\cR})\mid 
(a,c)\in\amax(\rel''^{\{x,v\},\cR}),\\
&& \qquad  (b,c)\in\amax(\rel''^{\{y,v\},\cR}) 
\text{ for some $c\in B$}\}.
\end{eqnarray*}\label{not:Txy}

A tuple $\ba$ over a set of variables $U\sse V$ is said to be 
\emph{T-compatible}\index{T-compatible}
if for any $x,y\in U$, $(\ba[x],\ba[y])\in\reli^{xy}$. The following lemma
provides the main structural result necessary for proving that $\cR^*$ is a
$\ov\beta^*$-strategy.

\begin{lemma}\label{lem:step-2-nonempty}
Let $\rela$ be one of the relations $\rel$, $\rel_{C,w,\gm\dl}$, 
$\cS_{W_{w,\gm,\dl},Y}$, and $\rela'',\rela^*$ the corresponding relations
$\rel''$, $\rel''_{C,w,\gm\dl}$, or $\cS''_{W_{w,\gm,\dl},Y}$, and  
$\rel^*$, $\rel^*_{C,w,\gm\dl}$, or $\cS^*_{W_{w,\gm,\dl},Y}$,
respectively, for 
some $C=\ang{\bs,\rel}\in\cC$ and $(w,\gm,\dl)\in\cW(\ov\beta)$, where 
$Y$ is a set from the definition of block-minimality; and let $U$ be its set of 
coordinate positions.\\[2mm]
(1) Every Q-compatible $\ba\in\amax(\rela)$ such that $\ba[u]\in B'_u$ for
$u\in U\cap W$ is also T-compatible.\\[1mm]
(2) $\rela^*$ contains all the as-maximal T-compatible tuples from 
$\rela''$.\\[1mm]
(3) If $\rela^*=\rel^*_{C,w,\gm\dl}$, then $\umax(\rela^*)$ is as-closed in 
$\umax(\rel_{C,w,\gm\dl}\cap\ov B')$.
\end{lemma}

\begin{proof}
(1) By Corollary~\ref{cor:23-minimality}, if, say $x\in W$ and 
$(a,b)\in\relo^{xy}$ are such that $a\in B'_x$, then $(a,b)\in\reli^{xy}$, which 
can be proved as in the proof of Proposition~\ref{pro:3-tight}(1). Therefore, 
it suffices to prove that if $x,y\not\in W$,
then $T^{xy}=\relo^{xy}$. This can be done in the same way as in the proof of 
Proposition~\ref{pro:3-tight}. Let $(a,b)\in\relo^{xy}$, we need to find 
$c\in B'_v$ such that $(a,c)\in\relo^{xv}$, $(b,c)\in\relo^{yv}$. Let 
$$
\relo(x,y,v)=\rel^{\{x,y\},\cR}(x,y)\meet\rel^{\{x,v\},\cR}(x,v)\meet
\rel^{\{y,v\},\cR}(y,v),
$$ 
which is a subalgebra of $\zA_x\tm\zA_y\tm\zA_v$.
By (2,3)-consistency of $\cP''$ 
$\relo^{xy}\sse\pr_{xy}\relo, \relo^{xv}\sse\pr_{xv}\relo,
\relo^{yv}\sse\pr_{xv}\relo$.  Let $\relo'=\relo\fac\al$. 
Since each of $\rel^{\{x,y\},\cR}$, $\rel^{\{x,v\},\cR}$, $\rel^{\{y,v\},\cR}$ is 
polynomially closed in the corresponding constraint
relation $\rel^{\{x,y\}},\rel^{\{x,v\}}$, or $\rel^{\{y,v\}}$ of $\cP$, 
$\relo$ is polynomially closed in 
$$
\rel^{\{x,y\}}(x,y)\meet\rel^{\{x,v\}}(x,v)\meet\rel^{\{y,v\}}(y,v),
$$
as well, and so is $\relo'$ in 
$$
\rel^{\{x,y\}}(x,y)\meet\rel^{\{x,v\}}\fac\al(x,v)
\meet\rel^{\{y,v\}}\fac\al(y,v).
$$
Let $\relo^\dagger=\pr_{xy}(\relo'\cap(B_x\tm B_y\tm D))$, By the Congruence 
Lemma~\ref{lem:affine-link} either $\umax(\relo^\dagger)\tm D\sse\relo'$ or
there is $\eta\in\Con(\rel^{\{x,y\}})$ such that 
$\umax(\relo'\cap(B_x\tm B_y\tm D))$ is the graph
of a mapping $\tau:\umax(\relo^\dagger)\to D$. In the former case we are done,
because then $(a,b,B'_v)\in\relo'$, and therefore $(a,b,c)\in\relo$ for
some $c\in B'_v$. The latter case is not possible, because by the Congruence 
Lemma~\ref{lem:affine-link} there is $\th\in\Con(\rel^{\{x,y\}})$ such
that $\eta\prec\th\le\beta_x\tm\beta_y$, and $(\al,\beta_v)$ and 
$(\eta,\th)$ cannot be separated. This however is not the case, since 
$x,y\not\in W$.

(2) Similar to Proposition~\ref{pro:3-tight} and Corollary~\ref{cor:23-minimality}
it suffices to prove that for any $X\sse U$, any T-compatible 
$\ba\in\amax(\pr_X\rela'')$ can be extended to a T-compatible 
$\bb\in\amax(\rela'')$. 
In fact, since $\rela''$ contains all the Q-compatible tuples, and therefore all 
the T-compatible tuples from $\rela$, it suffices 
to prove the statement for $\rela$, rather than for $\rela''$. We show that for 
any $w\in U-X$ tuple $\ba$ can be extended to a T-compatible tuple 
$\bc\in\amax(\pr_{X\cup\{w\}}\rela)$. By Proposition~\ref{pro:3-tight} there is a 
Q-compatible $\bb'\in\rela$ with $\ba=\pr_X\bb'$. If $X\cap W\ne\eps$ or 
$(X\cup\{w\})\cap W=\eps$, we can set $\bc=\pr_{X\cup\{w\}}\bb'$. 

Suppose that $X\cap W=\eps$ and $w\in W$. Then we proceed similar to 
part (1). Let $X=\{\vc xk\}$, let 
$$
\relo(\vc xk,w)=\pr_X\rela(\vc xk)
\meet\bigwedge_{i=1}^k \rel^{\{x_i,w\},\cR}(x_i,w),
$$ 
By (2,3)-consistency of $\cP''$ and Proposition~\ref{pro:3-tight}
$\relo^{x_iw}\sse\pr_{x_iw}\relo$, and by Corollary~\ref{cor:23-minimality}
$\rela'\sse\pr_X\relo$, where $\rela'$ is the set of all T-compatible (equivalently, 
Q-compatible) tuples from $\pr_X\rela$. Let $\relo'=\relo\fac\al$. 
Since each of $\rel^{\{x_i,w\},\cR}$ is polynomially closed in $\rel^{\{x_i,w\}}$ 
and $\pr_X\rela(\vc xk)$ is polynomially closed in itself, by 
Lemma~\ref{lem:poly-closed}(2) $\relo$ is polynomially closed in 
$$
\pr_X\rela(\vc xk)\meet\bigwedge_{i=1}^k \rel^{\{x_i,w\}}(x_i,w),
$$ 
as well, and so is $\relo'$ in 
$$
\pr_X\rela(\vc xk)\meet\bigwedge_{i=1}^k \rel^{\{x_i,w\}}\fac{\al_w}(x_i,w).
$$ 
Let 
$\relo^\dagger=
\pr_X\left(\relo'\cap\left(\prod_{i=1}^k B_{x_i}\tm D\right)\right)$.
Now we can finish the proof in the same way as in part (1).

(3) Let $U=\bs\cap W_{w,\dl\gm}$. 
Observe first, that every tuple from $\umax(\rel^*_{C,w,\gm\dl})$ is 
P-compatible (see Section~\ref{sec:step-1}). If we prove that 
$\rel^*_{C,w,\gm\dl}$ contains every P-compatible tuple $\ba$ from 
$\rel''_{C,w,\gm\dl}$ such that $\ba[u]\in B'_u$ for every $u\in U\cap W$, 
by Lemma~\ref{lem:step-1-as-closeness}(2) the result follows. 
As in the proof of 
Lemma~\ref{lem:step-1-as-closeness} we proceed by induction on the
restriction of the problem $\cP''$ being converted to a (2,3)-minimal 
and block-minimal instance. 

Let $\rel''$ and $\rel''_{C,w,\gm\dl}$ denote the relations associated with the 
instance $\cP''$. Let $\rel^\dagger$ and $\rel^\dagger_{C,w,\gm\dl}$ 
denote the relations 
obtained from $\rel''$ and $\rel''_{C,w,\gm\dl}$ in the first step of
converting $\cP''$ to $\cP^*$, that is, restricting the domains. By 
Lemma~\ref{lem:step-2-relations}
relations $\rel^\dagger,\rel^\dagger_{C,w,\gm\dl}$ contain all the 
necessary P-compatible relations. This provides the base case. For the 
induction step we again consider two cases. We denote the current 
constraint relations by $\rel^\ddagger_{C,w,\gm\dl}$ and the ones from 
the (2,3)-strategy by $\rel^{\ddagger X}$. 

In the first case we enforce 
(2,3)-minimality for $x,y,z\in V$. Let $(a,b)\in\rel^{\ddagger \{x,y\}}$
be a P-compatible tuple. Then there is $c_1,c_2\in B_z$ such that 
$(a,c_1)\in\rel''^{\{x,z\}}$, $(b,c_2)\in\rel''^{\{y,z\}}$ are P-compatible.
As in the proof of item (1) of this lemma, we can argue that $c_1=c_2$
can be assumed. If $z\not\in W$, the pairs $(a,c_1),(b,c_2)$ are as required. 
Otherwise, $c_1,c_2$ can be chosen from $B'_z$ by 
Lemma~\ref{lem:step-2-relations}.

In the second case let $(u,\eta,\th)\in\cW(\ov\beta)$ and 
$X=W_{u,\eta\th}$; we solve a problem of the form 
$\cP''_X\fac{\ov\mu^Y}$, let $\cS^\ddagger_X$ be the set of solutions
of this problem. Let also $U'=\bs\cap X$. We need to show that 
for any P-compatible $\ba\in\umax(\rel^\ddagger_{C,w,\gm\dl})$ with
$\ba[u]\in B'_u$ for $u\in U$ the tuple $\pr_{U'}\ba$ can be extended
to a P-compatible solution $\vf\in\cS^\ddagger_X$. Since 
$\ba\in\umax(\rel''_{C,w,\gm\dl})$,
the tuple $\pr_{U'}\ba$ can be extended to a u-maximal solution 
$\vf\in\cS''_X$. If $U'\cap W\ne\eps$ or $X\cap W=\eps$, solution 
$\vf$ is as required. Otherwise by Lemma~\ref{lem:step-2-relations} 
$\vf$ can be chosen P-compatible and such that $\vf(x)\in B'_x$ for 
$x\in X\cap W$; that is $\vf\in\cS^\ddagger_X$ by the induction hypothesis.
\end{proof}
 
Now we are ready to prove that $\cR^*$ is a $\ov\beta'$-strategy.

\begin{theorem}\label{the:step-2-strategy}
$\cR^*$ is a $\ov\beta'$-strategy with respect to $\ov B'$.
\end{theorem}

\begin{proof}
(S1) follows directly from the construction, since the relations $\rel^{X,\cR^*}$
result from establishing (2,3)-minimality of $\cP^*$, and they are 
nonempty by Lemma~\ref{lem:step-2-nonempty}(1).
Conditions (S2) and (S3) are also by construction. Condition (S4) also holds by 
construction, as all the relations of the form 
$\rel^*_{C,w,\gm\dl}$ are subalgebras. Also, each of them contains a 
Q-compatible element, which is as-maximal in $\rel_{C,w,\gm\dl}$, 
implying that $\umax(\rel^*_{C,w,\gm\dl})\sse\umax(\rel_{C,w,\gm\dl})$.

For (S5) the existence of $A_{\cR^*,w}$ for $w\in V$ follows from the 
construction, and the as-closeness of $\umax(A_{\cR^*,w})$ follows from 
Lemma~\ref{lem:step-2-nonempty}(3). Condition (S6) follows from 
Lemma~\ref{lem:step-2-nonempty}(3) and (S6) for $\cR$ as well. 
Finally, condition (S7) holds by Lemma~\ref{lem:S7}.
\end{proof}

\bibliographystyle{plain}

\printindex
\twocolumn
\section*{Notation}

$[n]$, \pageref{not:n}\\
$\le$, \pageref{not:le}\\
$\sqq$, \pageref{not:sqq}\\
$\sqq_{as}$, \pageref{not:sqq-as}\\
$\sqq_{asm}$, \pageref{not:sqq-asm}\\
$\ang{\bs,\rel}$, \pageref{not:ang}\\
$\zz_\zA$, $\zo_\zA$, \pageref{not:zz-zo}\\
$\ov\al$, \pageref{not:ov-alpha}\\
$\ov\al_I$, \pageref{not:ov-alpha-I}\\
$\zeta(\al,\beta)$, \pageref{not:zeta}\\
$\tau_U$, \pageref{not:tau-U}\\
$\tau_C$, \pageref{not:tau-C}\\
$\tau'_{C'}(u,\eta\th)$, \pageref{not:tau-prime-u-eta-th}\\
$\ba[i]$, \pageref{not:bai}\\
$\ov A$, \pageref{not:ov-A}\\
$\ov A_I$, \pageref{not:ov-A-I}\\
$\amax(\zA)$, \pageref{not:amax}\\
$\as(a)$, \pageref{not:as}\\
$\Centr(\cP)$, \pageref{not:Centr}\\
$\Cg{\cdot}$, \pageref{not:Cg}\\
$f^\ba(x)$, \pageref{not:f-a}\\
$\Filt_\zA(a)$, \pageref{not:Filt}\\
$\Filt^{as}_\zA(a)$, \pageref{not:Filt-as}\\
$\Filt^{asm}_\zA(a)$, \pageref{not:Filt-asm}\\
$\cG(\zA)$, \pageref{not:cG-A}\\
$h_{ab}$, \pageref{not:h-ab}\\
$\lnk_1,\lnk_2$, \pageref{not:lnk}\\
$\maj$, \pageref{not:maj}\\
$\max(\zA)$, \pageref{not:max}\\
$\Razm(\cP)$, \pageref{not:MAX}\\
$p(\cP)$, \pageref{not:p-cP}\\
$\cP=(V,\cC)$, \pageref{not:cP}\\
$\cS_\cP$, \pageref{not:cS}\\
$\cP_W$, \pageref{not:cP-W}\\
$\cP'_{W'}$, \pageref{not:P-prime-W-prime}\\
$\cP_{(v,\zB)}$, \pageref{not:P-v-B}\\
$\cP_{C\to C'}$, \pageref{not:C-to-C}\\
$\cP\fac{\ov\al}$, \pageref{not:P-fac}\\
$\cP^\cR$, \pageref{not:P-cR}\\
$\relp^x,\relp^{xy}$, \pageref{not:Pxy}\\
$\Pol(\zA)$, $\Polo(\zA)$, \pageref{not:Pol}\\
$\pr_I\rel$, \pageref{not:pr}\\
$\relo^x,\relo^{xy}$, \pageref{not:Qxy}\\
$\relo_{ab}^\zA$, \pageref{not:Qab}\\
$\rel^X$, \pageref{not:R-X}\\
$\cR$, \pageref{not:cR}\\
$\rel[c], \rel^{-1}[c']$, \pageref{not:rel-of-a}\\
$\cS_\cP$, \pageref{not:cS}\\
$\cS_W$, \pageref{not:cS-W}\\
$\Sg B$, \pageref{not:Sg}\\
$\cS^\cR,\cS^{\cR'}$, \pageref{not:S-R-S-R-prime}\\
$\cS'_{W'}$, \pageref{not:S-prime-W-prime}\\
$\razm(\cP)$, \pageref{not:size}\\
$t_{ab}$, \pageref{not:1-ab}\\
$\tol_1,\tol_2$, \pageref{not:tol}\\
$\reli^x,\reli^{xy}$, \pageref{not:Txy}\\
$\umax(\zA)$, \pageref{not:umax}\\
$\cW^\cP(\ov\beta)$, \pageref{not:cW}\\
$W_{v,\al\beta,\ov\beta}$, \pageref{not:W-v-albeta-beta}\\
$\cW'(\ov\beta)$, \pageref{not:cWprime}\\
$\emin_\zA$, \pageref{not:emin}
\end{document}